\documentclass[10pt,a4paper]{article}
\usepackage[utf8]{inputenc}
\usepackage{amsmath}
\usepackage{amsfonts}
\usepackage{amssymb}
\usepackage{graphicx}
\usepackage{braket}
\usepackage[left=2cm,right=2cm,top=2cm,bottom=2cm]{geometry}

\date{}

\title{Non-deterministic, probabilistic, and quantum effects through
  the lens of event structures (Technical report)}

\author{Vítor Fernandes \and
Marc de Visme \and
Benoît Valiron}

\usepackage{proof}
\usepackage{amsthm}
\usepackage{xcolor}
\usepackage{stmaryrd}
\usepackage{mdframed}
\usepackage{mathtools}
\expandafter\def\csname opt@stmaryrd.sty\endcsname
{only,shortleftarrow,shortrightarrow}
\usepackage[textsize=scriptsize,textwidth=45]{todonotes}

\theoremstyle{definition}
\newtheorem{definition}{Definition}[section]
\theoremstyle{definition}
\newtheorem{example}[definition]{Example}
\theoremstyle{definition}
\newtheorem{lemma}[definition]{Lemma}
\theoremstyle{definition}
\newtheorem{proposition}[definition]{Proposition}
\theoremstyle{definition}
\newtheorem{theorem}[definition]{Theorem}
\theoremstyle{definition}

\theoremstyle{definition}
\newtheorem{remark}{Remark}

\usepackage{MnSymbol,cases,wasysym}
\usepackage[normalem]{ulem}
\usepackage{url,hyperref}

\usepackage[all]{xy}
\usepackage{wrapfig}
\usepackage{subcaption}

\newcommand{\adjoint}{\dagger}

\newcommand{\es}[1]{\mathrm{#1}}
\newcommand{\pes}[3]{(#1,\, #2,\, #3)}
\newcommand{\ppes}[4]{(#1,\, #2,\, #3,\, #4)}
\newcommand{\pespq}[2]{(#1,\, #2)}
\newcommand{\qpes}[4]{(#1,\, #2,\, #3,\, #4)}
\newcommand{\init}[1]{\mathcal{I}(#1)}
\newcommand{\minconflict}{\mathrel{\!\!\!\xymatrix@C=0.5cm{{}\ar@{~}[r]&{}}\!\!\!\!}}
\newcommand{\confmax}[1]{\mathcal{C}_{\text{max}}(#1)}
\newcommand{\confES}[1]{\mathcal{C}(#1)}
\newcommand{\qvar}[1]{ \text{qVar}(#1)}
\newcommand{\chain}{\mathrel{\ooalign{\raise.225em\hbox{$\rule{0.4cm}{0.5pt}$}\cr
  \hidewidth\hbox{$\subset\mkern-12mu$}\cr}}}
\newcommand{\cchain}[1]{\stackrel{#1}{\chain}\ }
\newcommand{\dropc}[4]{d^{(#1)}_{#2}[#3\, ;\, #4]}
\newcommand{\econc}[2]{#1 \text{ co } #2}
\newcommand\extconflict{\mathrel{\ooalign{\lower.525em\hbox{%
        $\stackrel{\mathclap{\normalfont\mbox{$\lrcorner
              \llcorner$}}}{\urcorner \ulcorner}$}}}}

\newcommand{\com}[1]{#1} 

\newcommand{\fvar}[1]{ FV(#1)}
\newcommand{\bvar}[1]{ BV(#1)}

\newcommand{\seq}[2]{#1\, ;\, #2}

\newcommand{\meas}[3]{\com{M}(#1, #2, #3)} 
\newcommand{\nd}[2]{#1\, \square\, #2} 
\newcommand{\probC}[3]{#1\, +_{#2}\, #3}
\newcommand{\conc}[2]{#1\, ||\, #2}

\newcommand{\qwhi}[2]{while\ \com{M}(#1, #2)} 

\newcommand{\rec}[2]{\mu #1 . #2}

\newcommand\uconvex{\mathrel{\ooalign{$\cup$\cr
  \hidewidth\raise.0ex\hbox{$-\mkern-.4mu$}\cr}}}
\newcommand{\xtwoheadrightarrow}[2][]{%
  \xrightarrow[#1]{#2}\mathrel{\mkern-14mu}\rightarrow
}

\newcommand{\mf}[1]{ \llbracket #1 \rrbracket }

\newcommand{\catfont}[1]{\mathsf{#1}}

\newcommand{\Subs}[2]{\catfont{Sub}_{}}

\newcommand{\sfunfont}[1]{\mathrm{#1}}
\newcommand{\Pow}{\sfunfont{P}}
\newcommand{\Dist}{\sfunfont{D}}

\newcommand\adjunct[2]{\xymatrix@=8ex{\ar@{}[r]|{\top}\ar@<1mm>@/^2mm/[r]^{{#2}}
& \ar@<1mm>@/^2mm/[l]^{{#1}}}}
\newcommand\adjunctop[2]{\xymatrix@=8ex{\ar@{}[r]|{\bot}\ar@<1mm>@/^2mm/[r]^{{#2}}
& \ar@<1mm>@/^2mm/[l]^{{#1}}}}
\newcommand\retract[2]{\xymatrix@=8ex{\ar@{}[r]|{}\ar@<1mm>@/^2mm/@{^{(}->}[r]^{{#2}}
& \ar@<1mm>@/^2mm/@{->>}[l]^{{#1}}}}

\newcommand{\ie}{\emph{i.e.}}

\input{tikzit.sty}

\tikzstyle{event}=[fill=white, draw=white, shape=rectangle, minimum size=0.5cm, tikzit shape=rectangle, tikzit draw=black, align=center]
\tikzstyle{command}=[fill=none, draw=none, shape=rectangle, tikzit draw=none, align=center]
\tikzstyle{bullet}=[fill=black, draw=black, shape=circle, inner sep=0pt, minimum size=3pt, tikzit shape=circle, align=center]
\tikzstyle{rectangle}=[fill=white, draw=black, shape=rectangle, align=center]
\tikzstyle{circle-1}=[fill=white, draw=white, shape=circle, align=center]
\tikzstyle{circle}=[fill=white, draw=black, shape=circle, radius=0.5cm, align=center]
\tikzstyle{ellipse}=[fill=white, draw=black, shape=ellipse, align=center]
\tikzstyle{gate}=[fill=white, draw=black, shape=rectangle, align=center]
\tikzstyle{state}=[fill=none, draw=none, shape=rectangle, align=center]

\tikzstyle{wiggle}=[decorate, decoration={snake, segment length=5mm, amplitude=0.5mm}, -, tikzit draw=red]
\tikzstyle{arrow}=[-, arrows={-Stealth[inset=0pt, angle=90:5pt, open]}]
\tikzstyle{op}=[->, draw=black, fill=none]
\tikzstyle{prob}=[decorate, decoration={snake, segment length=5mm, amplitude=0.5mm}, ->, tikzit draw=green]
\tikzstyle{traco}=[dashed, ->]
\tikzstyle{cchain}=[-, arrows={-Arc Barb[reversed]}]

\usepackage{tikz}
\usetikzlibrary{automata}
\usetikzlibrary{positioning}
\usetikzlibrary{arrows.meta}
\usetikzlibrary{trees}
\usetikzlibrary{automata}
\usetikzlibrary{positioning, calc}
\usetikzlibrary{shapes.geometric, arrows}
\usetikzlibrary{calc,decorations.pathmorphing,shapes}
\usetikzlibrary{scopes}
\usetikzlibrary{backgrounds}
\begin{document}
\maketitle

\tableofcontents

\newpage
\section{Introduction}
Concurrency is pervasive in modern computer architecture. Starting in the early 1960s, the study of
its semantics, both operational and denotational, and within different paradigms (from interleaving
to the so-called true concurrency) became a highly active research area with concrete implications
in language design.

In the interleaving paradigm, saying that two atomic actions $a$ and $b$ are in parallel is
interpreted as $a$ then $b$ or $b$ then $a$. On the other hand, from a true concurrent point of
view, the same command is interpreted as $a$ and $b$, which are causally unrelated. We focus on the
latter interpretation for which event structures~\cite{winskel82,winskel88} are a known model.

An event structure is a partial order with a conflict relation on events. If $a$ and $b$ are in
conflict, then they are incompatible events, \ie\ they cannot be performed in the same
computation. Furthermore, event structures are very flexible, and proof of that is the fact that
they have been used to study several computational effects: parallelism~\cite{winskel88},
probabilities~\cite{varacca06a,varacca07,visme19}, quantum effects~\cite{clairambault19,winskel14},
shared weak memory~\cite{castellan16}, etc.

Despite all the work around event structures on different computational effects, when the goal is to
provide denotational semantics to a programming language, they seem to play a secondary role. More
often than not, they serve as the backbone of some much more complex models, such as games and
strategies~\cite{castellan17,paquet20,clairambault19}. Some exceptions are the works of
Winskel~\cite{winskel88,winskel82}, in which he used event structures to give denotational semantics
to CCS~\cite{milner89}, and Marc de Visme~\cite{visme19}, in which two notions of conflict are used
in order to accommodate both probabilistic and non-deterministic choices in a probabilistic
extension of CCS~\cite{baier97}, who have used event structures as the primary model.

\paragraph{Contribution.}
In this paper, we aim at giving event structures the leading role as a computational model.  Our
work combines parallelism with three different algebraic effects: non-determinism, probabilities,
and quantum. For each algebraic effect, we propose a small imperative-style programming language
together with suitable operational semantics, wherein for the non-deterministic and quantum cases,
we used a simple labeled transition system -- or, in the probabilistic case, a labeled Segala
automaton~\cite{segala95,sokolova04}.

We rely on different flavors of event structures. For the non-deterministic case, we use the event
structures defined by Winskel~\cite{winskel88} as a base model. For the probabilistic case, we use
probabilistic event structures~\cite{winskel14}. For the quantum case, we consider a restriction of
the definition in~\cite{winskel14}, which we call Unitary event structures. This modification allows
us to extend~\cite[Theorem 3]{winskel14}, which states that quantum event structures without events
in conflict are probabilistic event structures when given an initial state, by dropping the
necessity of having an empty conflict relation.

We also show that the operational and denotational semantics are sound and adequate for the three
different algebraic effects considered. We do it by checking that the words created by the
operational semantics and the covering chains in event structures, which are essentially finite
sequences of events, coincide.

\paragraph{Remark.} The majority of the proofs are put in appendix to improve readability.

\newpage
\section{Event Structures}\label{sec:es}
In the imperative setting, the evaluation of a program is commonly accompanied by a memory that
changes accordingly the execution of said program, where each step performed by the computation is
not labeled.  On the other side we have a process algebra approach, in which states are dropped and
each step of the computation is labeled by the action that caused the occurrence of the computation.
Although we intend to model an imperative language, our approach is similar to the latter.  This
decision comes from the use of event structures.  By dropping the state we can use the usual
definitions of event structures~\cite{winskel84,winskel82}.  Since we want to model an imperative
language, we need to have the notion of state. Well, since we label the transitions we perform, we
can create a trace of the actions that were performed. By doing this, we can apply each instruction
in the trace to a given state.


Informally, an event structure~\cite{winskel88} is composed of a set of events, together with a
notion of causality given by a partial order on events: if $e \leq e'$ then $e'$ depends on $e$
(another way of interpreting $e \leq e'$ is $e'$ occurs after $e$), and a notion of conflict between
events: if $e \# e'$ then either $e$ occurs or $e'$ occurs, which is a behavior similar to a
non-deterministic choice.
\begin{definition}[Event Structures]\label{def:pes}
  Define an \textit{event structure} to be a structure $\es{E} = \pes{E}{\leq}{\#}$ consisting of a
  set $E$ of events, which are partially ordered by $\leq$, the \textit{causal dependency relation},
  and a binary, symmetric, irreflexive relation $\# \subseteq E \times E$, the \textit{conflict
    relation}, satisfying:
  \begin{itemize}
  \item $\{e' \mid e' \leq e\}$ is finite
  \item $e \# e' \leq e'' \Rightarrow e \# e''$
  \end{itemize}
  for all $e, e', e'' \in E$.
\end{definition}

Summing up, the first condition tells us that the downward closure of an event $e$ must be finite,
\ie\ the set of events that $e$ causally depends on needs to be finite, and the second condition
tells us that the conflict relation is hereditary.

\begin{definition}[Concurrent Event]
  Two events, $e,e'$ are said \emph{concurrent} iff
  $\neg (e \leq e') \wedge \neg (e' \leq e) \wedge \neg (e \# e')$. In other words, two events are
  concurrent when they are not causally dependent and are not in conflict.
\end{definition}

\begin{definition}[Configuration]
  A configuration is a subset of the set of events, $x \subseteq E$, that are
  \begin{align*}
    & \text{conflict-free: } \forall e, e' \in x\ .\ \neg (e \# e')\\
    & \text{down-closed: } \forall e,e'\ .\ e' \leq e \wedge e \in x \Rightarrow e' \in x
  \end{align*}
  We then denote by $\mathcal{C}^\infty(\es{E})$ the set of all configurations and by
  $\confES{\es{E}}$ the set of finite configurations.
\end{definition}

\begin{definition}[Covering chain]
  Let $\es{E} = \pes{E}{\leq}{\#}$ be a event structure, $e \in E$, and
  $x \in \confES{\es{E}}$.  Denote by $x \stackrel{e}{\chain\ } x \cup \{e\}$ if $e \not\in x$
  and $(x \cup \{e\}) \in \confES{\es{E}}$. A \emph{covering chain} on a configuration
  $x \in \confES{\es{E}}$ is a finite sequence of events $e_1 e_2 \dots e_n$ such that
  \[
    \emptyset \stackrel{e_1}{\chain\ } x_1 \stackrel{e_2}{\chain\ } x_2 \stackrel{e_3}{\chain\ } \dots
    \stackrel{e_n}{\chain\ } x_{n+1} = x
  \]
\end{definition}

\begin{definition}[Cover]
  Let $\es{E}$ be a event structure and $x,y \in \confES{\es{E}}$.  Say that $y$ covers
  $x$, pictured as $x \chain\ y$, if $x \subset y$ with nothing in between
  ($\not\exists z\ .\ x \subset z \subset y$).
\end{definition}

Later on we may find useful to say $y \cchain{e_1, \dots e_n} x_1, \dots x_n$ when
$y \chain\ x_1, \dots, x_n$.

\begin{definition}[Maximal configuration]
  Let $\es{E}$ be a event structure and $x \in \confES{\es{E}}$.  Say that $x$ is a
  maximal configuration iff $\not\exists y \in \confES{\es{E}}$ such that $x \chain\ y$.
  Denote by $\confmax{\es{E}}$ the set of maximal configurations.
\end{definition}

An interesting property of maximal configurations is that two distinct maximal configurations have
events in conflict. This is useful when showing the validity of sequential composition's definition
in the probabilistic realm (Lemma~\ref{lem:seq-es2}).
\begin{lemma}\label{lem:max-config-e-conflict}
  Let $\es{E}$ be an event structure and $x, x' \in \confmax{\es{E}}$. If $x \neq x'$ then
  $\exists e \in x, e' \in x'$ such that $e \# e'$.
\end{lemma}
\begin{proof}
  Let $x, x' \in \confmax{\es{E}}$ such that $x \neq x'$.

  Pick $e \in x \backslash x' = \{e \in x \mid e \not\in x'\}$ and let
  $M = \{a \in E \mid a \leq e \text{ and } a \not\in x'\}$ be the set of events that $e$ causally
  depends on and not in $x'$.  Clearly we have $M \neq \emptyset$, since $e \in M$. By
  Definition~\ref{def:pes}, $\{e' \mid e' \leq e\}$ is finite. Thus we can choose $a$ minimal in $M$
  with respect to $\leq$, \ie\ $\forall a' < a\ .\ a' \not\in M$ (in other words, choose the first
  event $a$ that $e$ causally depends on that is not in $x'$). Then every $a' \in E$ such that
  $a' < a$ is in $x'$, \ie\ $a' \in x'$. Thus, $x' \cup \{a\}$ is downward closed.  Since $x'$ is
  maximal, $x' \cup \{a\}$ cannot be a configuration, hence it must fail conflict-freeness:
  $\exists e' \in x'$ such that $a \# e'$.  By hereditary of conflict, $a \leq e$ and $a \# e'$, we
  have $e \# e'$.
\end{proof}

Later on we shall find useful to simplify how covering chains are represented.
We then let $\omega = e_1 e_2 \dots e_n$ and denote
$\emptyset \stackrel{e_1}{\chain\ } x_1 \stackrel{e_2}{\chain\ } x_2 \stackrel{e_3}{\chain\ } \dots
\stackrel{e_n}{\chain\ } x_{n+1} = x$ simply by $\emptyset \stackrel{w}{\chain\ } x$.

Since the causal relation is a partial order we know that it is transitive.  Furthermore, the
conflict relation is hereditary over events. Hence if we want to draw an event structure using these
two relations we would have to add a lot of redundant information, which would make the event
structure hard to understand.  To ease such task, we find it useful to use the notions of
\emph{immediate causality}, pictured by $\rightarrowtriangle$, and \emph{minimal conflict},
represented by $\minconflict$.  Let $\es{E}$ be an event structure such that $e_1, e_2 \in E$.  We
say $e_1 \rightarrowtriangle e_2$ iff $e_1 < e_2$ and $\not\exists e_3 \in E\, .\, e_1 < e_3 < e_2$.
We say $e_1 \minconflict e_2$ iff $e_1 \# e_2$ and whenever $e'_1 \leq e_1$, $e'_1 \# e'_2$, and
$e'_2 \leq e_2$ we have $e_1 = e'_1$ and $e_2 = e'_2$. Note that it is possible to deduce the causal
and conflict relations from the immediate causality and minimal conflict relations.

Example~\ref{ex:es-1} aims to get the reader familiarized with event structures.
\begin{example}\label{ex:es-1}
  Figure~\ref{fig:ex1-1} shows an event structure with four events, $a$, $b$, $c$, and $d$,
  where: $b$ causally depends on $a$, $c$ and $d$ are concurrent events which are in conflict
  with $a$, and consequently also with $b$.  Furthermore, note that $a$ is in minimal conflict
  with $c$ and $d$.  The set of configurations, \ie\ the set of possible computations, is
  $\set{\emptyset, \set{a}, \set{c}, \set{d}, \set{a,b}, \set{c,d}}$.  Furthermore note that the
  configuration $\set{c,d}$, which is composed of two concurrent events, has two possible
  covering chains: $\emptyset \cchain{d} \set{d} \cchain{c} \set{c,d}$ and
  $\emptyset \cchain{c} \set{c} \cchain{d} \set{c,d}$.
  \begin{figure}[ht!]
    \centering
    \begin{tikzpicture}
      \begin{pgfonlayer}{nodelayer}
        \node [style=event] (0) at (-2, 3) {$a$};
        \node [style=event] (1) at (-2, 2) {$b$};
        \node [style=event] (2) at (-0.5, 3) {$c$};
        \node [style=event] (3) at (0.25, 3) {$d$};
      \end{pgfonlayer}
      \begin{pgfonlayer}{edgelayer}
        \draw [style=arrow] (0) to (1);
        \draw [style=wiggle] (0) to (2);
        \draw [style=wiggle, bend left] (0) to (3);
      \end{pgfonlayer}
    \end{tikzpicture}
    \caption{Example of an event structure}
    \label{fig:ex1-1}
  \end{figure}
\end{example}

The last concept we will introduce in this section is that of a map of event structures, which can
be total or partial. Its usefulness will be demonstrated when showing that unitary event structures
with an initial state correspond to probabilistic event structures. Intuitively, a map of event
structures $f : \es{E} \rightarrow \es{E}'$ tells how the occurrence of an event $e$ in $\es{E}$
implies the occurrence of $f(e) = e'$ in $\es{E}'$. To be more precise, if $x$ in $\es{E}$, \ie\
$x \cchain{e}\ x \cup \{e\}$ then $f[x] \cchain{f(e)}\ f[x \cup \{e\}]$ in $\es{E}'$, whenever
$f(e)$ is defined.  In words, this means that if an event $e$ can be added to a configuration $x$ to
form a new configuration $x \cup \{e\}$ in $\es{E}$, then the same behavior holds in $\es{E}'$ when
$f(e)$ is defined. Formally, a map of event structures is defined as follows:
\begin{definition}[Map event structures]\label{def:map-es}
  Let $\es{E} = \pes{E}{\leq}{\#}, \es{E}' = \pes{E'}{\leq'}{\#'}$ be event structures.  A
  partial/total map $f$ from $\es{E}$ to $\es{E}'$ is a partial/total function
  $f : E \rightharpoonup E'$ such that:
  \begin{align*}
  \text{(Configuration Preserving)} & \forall x \in \confES{\es{E}} \Rightarrow f(x) \in \confES{\es{E}'} \\
  \text{(Locally injective)} & \forall (a \neq b) \in x \in \confES{\es{E}}, \text{ if $f$ is defined in both then }
    f(a) \neq f(b)
  \end{align*}
  where $f[x] = \set{ f(e) \mid e \in x,\ f(e) \text{ defined}}$
\end{definition}

The preserving configuration condition is self-explanatory. It suffices to recall what was
previously said, \ie\ that if an event $e$ is added to a configuration $x$ to form a new
configuration $x \cup \{e\}$ in $\es{E}$, then the same behavior holds in $\es{E}'$, when $f(e)$ is
defined. The locally injective condition assures that no two distinct events, under the same
configuration, are mapped to the same image event. This prevents concurrent events, or causally
dependent events from being mapped to the same event.  To understand why this condition is
necessary, let us explore an example illustrating a spurious map between event structures.
\begin{example}\label{ex:map-wrong-es}
  Consider the event structures in Figure~\ref{fig:ex4-2} linked by a broken map of event
  structures.
  \begin{figure}[ht!]
    \centering
    \begin{tikzpicture}
      \begin{pgfonlayer}{nodelayer}
        \node [style=event] (0) at (0, 1) {$a$};
        \node [style=event] (1) at (0, 0) {$c$};
        \node [style=event] (2) at (1, 1) {$b$};
        \node [style=event] (3) at (3.5, 1) {$a$};
        \node [style=event] (4) at (4.5, 1) {$b$};
        \node [style=none] (6) at (1.5, 0.5) {};
        \node [style=none] (7) at (3, 0.5) {};
        \node [style=none] (8) at (2.25, 0.75) {$f$};
      \end{pgfonlayer}
      \begin{pgfonlayer}{edgelayer}
        \draw [style=arrow] (0) to (1);
        \draw [style=op] (6.center) to (7.center);
      \end{pgfonlayer}
    \end{tikzpicture}
    \caption{Broken map of event structures}
    \label{fig:ex4-2}
  \end{figure} 

  The function $f$ maps $c$ to $a$ and applies the identity on $a$ and $b$.  Let us denote the event
  structure on the left by $\es{E}$ and the one on the right by $\es{E}'$. The set of configurations
  of $\es{E}$ is $\{\emptyset, \{a\}, \{b\}, \{a,c\}, \{a,b\}, \allowbreak \{a,b,c\}\}$ and the set of
  configurations of $\es{E}'$ is $\{\emptyset, \{a\}, \{b\}, \{a,b\}\}$.

  It is not difficult to check that the configurations of $\es{E}$ are preserved. However, the map
  becomes invalid because it does not respect the local injective condition. Let us consider two
  specific configurations from $\es{E}$: $\{a,b\}$ and $\{a,c\}$. The local injective condition is
  satisfied for $\{a,b\}$, since $a$ and $b$ belong to $\{a,b\}$ and are mapped to themselves. On
  the other hand, the configuration $\{a,c\}$ fails to satisfy the local injective condition, since
  $a$ and $c$ are mapped to $a$, \ie\ causally related events are being mapped to the same image
  event. Thus, different events are mapped to the same event, which contradicts the locally
  injective condition.

  Another way to understand the local injective condition is by looking at the size of the
  configurations, especially when the map is total. If the map is total, the size of the
  configurations must remain the same. In other words, if a configuration has $n$ events in
  $\es{E}$, the corresponding configuration in $\es{E}'$ must also have $n$ events.
\end{example}

We now develop an example that shows a valid map of event structures.
\begin{example}\label{ex:map-es-1}
  In Figure~\ref{fig:ex4-1}, we have a map of event structures $f$ that maps $a$ to itself and the
  conflicting events $c$ and $b$ to $d$. The map $f$ is total, and it is valid because it preserves
  configurations and obeys local injectivity. Each configuration in $\es{E}$ is mapped to a
  configuration in $\es{E}'$, and distinct events within the same configuration are mapped to
  distinct events. Furthermore, since $f$ is total, we notice that the size of covering chains are
  preserved.  Consider the covering chain $\emptyset \cchain{c} \set{c} \cchain{a} \set {a, c}$.
  The respective covering chain after applying $f$ is
  $\emptyset \cchain{d} \set{d} \cchain{a} \set{a, d}$.  This is only possible because the local
  injective condition ensures that different events in the same configuration must map to distinct
  events.
  \begin{figure}[ht!]
    \centering
    \begin{tikzpicture}
      \begin{pgfonlayer}{nodelayer}
        \node [style=event] (0) at (0, 1) {$a$};
        \node [style=event] (1) at (0, 0) {$b$};
        \node [style=event] (2) at (1, 1) {$c$};
        \node [style=event] (3) at (3.5, 1) {$a$};
        \node [style=event] (4) at (4.5, 1) {$d$};
        \node [style=event] (5) at (4.5, 0) {$e$};
        \node [style=none] (6) at (1.25, 0.5) {};
        \node [style=none] (7) at (3.25, 0.5) {};
        \node [style=none] (8) at (2.25, 0.75) {$f$};
      \end{pgfonlayer}
      \begin{pgfonlayer}{edgelayer}
        \draw [style=arrow] (0) to (1);
        \draw [style=wiggle] (2) to (1);
        \draw [style=arrow] (4) to (5);
        \draw [style=op] (6.center) to (7.center);
      \end{pgfonlayer}
    \end{tikzpicture}
    \caption{Map event structures}
    \label{fig:ex4-1}
  \end{figure}

  Although the event structure on the right has an event $e$ that is not in the image of $f$, this
  does not invalidate the map. The reason is that $e$ is not causally required by any event in the
  image of $f$.
\end{example}

With these definitions presented, we are now prepared to advance to the next stage, where we present
the language that we intend to model with event structures, \ie\ its syntax and respective
operational semantics in terms of a small-step and n-step.  After presenting the language, we
present the constructions made on event structures to capture the behavior of the language
operator's. Then we show how to interpret commands of the language using event structures and, at
last, we show that both semantics are sound and adequate.

\subsection{Language}\label{subsec:lan-1}
To present the language we consider a set of atomic actions $Act$ ranged over by $a$ (examples of
atomic actions are assignments, or unitary application, etc...).

The set of commands allowed by the language are given by the following grammar:
\[
  \com{C} ::= \com{skip} \mid a \in Act \mid \seq{\com{C}}{\com{C}} \mid \nd{\com{C}}{\com{C}} \mid
  \conc{\com{C}}{\com{C}}
\]

To define the operational semantics, we add a new command to the language, denoted by $\checkmark$,
that indicates the end of a computation.

We denote by $L = Act \cup \set{sk}$ the set of labels, which is ranged by $l$, and we consider a
\emph{terminal command}, denoted by $\checkmark$, representing the end of a computation.  The
small-step semantics is then defined as the smallest relation
$\xlongrightarrow{l} \subseteq C \times L \times (C \cup \set{\checkmark})$ obeying the rules in
Figure~\ref{fig:op-small1}. 
\begin{figure}[ht!]
  \centering
  \begin{align*}
    \com{skip} \xrightarrow{sk} \checkmark
    \hspace*{1cm}
    \com{a} \xrightarrow{a} \checkmark
    \hspace*{1cm}
    \infer{\seq{\com{C}_1}{\com{C}_2} \xrightarrow{l} \com{C}_2}{ \com{C}_1 \xrightarrow{l} \checkmark }
    \hspace*{1cm}
    \infer{
    \seq{\com{C}_1}{\com{C}_2} \xrightarrow{l} \seq{ \com{C}'_1  }{ \com{C}_2  }
    }
    {
    \com{C}_1 \xrightarrow{l} \com{C}'_1
    }
  \end{align*}
  \begin{align*}
    \infer{\nd{\com{C}_1}{\com{C}_2} \xrightarrow{l} \checkmark}{ \com{C}_1 \xrightarrow{l} \checkmark }
    \hspace*{1cm}
    \infer{
    \nd{\com{C}_1}{\com{C}_2} \xrightarrow{l} \com{C}'_1
    }
    {
    \com{C}_1 \xrightarrow{l} \com{C}'_1
    }
    \hspace*{1cm}
    \infer{\nd{\com{C}_1}{\com{C}_2} \xrightarrow{l} \checkmark}{ \com{C}_2 \xrightarrow{l} \checkmark }
    \hspace*{1cm}
    \infer{
    \nd{\com{C}_1}{\com{C}_2} \xrightarrow{l} \com{C}'_2
    }
    {
    \com{C}_2 \xrightarrow{l} \com{C}'_2
    }
  \end{align*}
  \begin{align*}
    \infer{\conc{\com{C}_1}{\com{C}_2} \xrightarrow{l} \com{C}_2}{ \com{C}_1 \xrightarrow{l} \checkmark }
    \hspace*{1cm}
    \infer{
    \conc{\com{C}_1}{\com{C}_2} \xrightarrow{l} \conc{ \com{C}'_1  }{ \com{C}_2  }
    }
    {
    \com{C}_1 \xrightarrow{l} \com{C}'_1
    }
    \hspace*{1cm}
    \infer{\conc{\com{C}_1}{\com{C}_2} \xrightarrow{l} \com{C}_1}{ \com{C}_2 \xrightarrow{l} \checkmark }
    \hspace*{1cm}
    \infer{
    \conc{\com{C}_1}{\com{C}_2} \xrightarrow{l} \conc{ \com{C}_1  }{ \com{C}'_2  }
    }
    {
    \com{C}_2 \xrightarrow{l} \com{C}'_2
    }
  \end{align*}
  \caption{Rules of the small-step operational semantics}
  \label{fig:op-small1}
\end{figure}

Define a word to be a sequence of labels:
\[
  \omega ::= l \mid l : \omega 
\]
where $l : \omega$ appends $l$ to the beginning of $\omega$.  A word can also be seen as an element
of $L^+$, \ie\ a possibly infinite sequence of labels without the empty sequence. Despite $L^+$
allows the possibility of having infinite words, by now we focus only on the finite words.

Define the $n$-step transition,
$\xrightarrow{\omega} \subseteq \com{C} \times L^+ \times (\com{C} \cup \{\checkmark\})$, where $n$ is the length of the words, as follows:
\begin{figure}[ht!]
  \centering
  \begin{align*}
    \infer{ \com{C} \xtwoheadrightarrow{l} \com{C'}  }{ \com{C} \xrightarrow{l} \com{C'}  }
    \hspace*{2cm}
    \infer{
    \com{C} \xtwoheadrightarrow{l : \omega'} \com{C'}
    }{
    \com{C} \xrightarrow{l} \com{C''}
    \qquad
    \com{C''} \xtwoheadrightarrow{\omega'} \com{C'}
    }
  \end{align*}
  \caption{Rules of the n-step operational semantics}
  \label{fig:op-nstep1}
\end{figure}

\begin{example}\label{ex:small-1}
  The initial program is $\nd{\seq{\com{a}}{\com{b}}}{\conc{\com{c}}{\com{d}}}$, from which we have
  three possible transitions: by $a$, $c$ or $d$. If we transit by $a$, we reach the command
  $\com{b}$, which we execute to finish the computation. Otherwise, we could either transit via $c$
  and then execute $d$, or transit via $d$ and then execute $c$, in order to finish the computation.

  With the support of Figure~\ref{fig:ex2-1} together with the above explanation, we can
  straightforwardly deduce the words that can be formed by the n-step semantics: $a$, $c$, $d$,
  $ab$, $cd$, and $dc$.
  \begin{figure}[ht!]
    \centering
    \begin{tikzpicture}
      \begin{pgfonlayer}{nodelayer}
        \node [style=event] (0) at (0, 0) {$\nd{\seq{\com{a}}{\com{b}}}{\conc{\com{c}}{\com{d}}}$};
        \node [style=event] (1) at (-1.75, -1.5) {$\com{b}$};
        \node [style=event] (2) at (2, -1.5) {$\com{c}$};
        \node [style=event] (3) at (0, -1.5) {$\com{d}$};
        \node [style=event] (4) at (-1.75, -3) {$\checkmark$};
        \node [style=event] (5) at (0, -3) {$\checkmark$};
        \node [style=event] (6) at (2, -3) {$\checkmark$};
        \node [style=none] (7) at (-1.25, -0.75) {$a$};
        \node [style=none] (8) at (-2, -2.25) {$b$};
        \node [style=none] (9) at (0.25, -0.75) {$c$};
        \node [style=none] (10) at (1.5, -0.75) {$d$};
        \node [style=none] (11) at (0.25, -2.25) {$d$};
        \node [style=none] (12) at (2.25, -2.25) {$c$};
      \end{pgfonlayer}
      \begin{pgfonlayer}{edgelayer}
        \draw [style=op] (0) to (1);
        \draw [style=op] (1) to (4);
        \draw [style=op] (0) to (3);
        \draw [style=op] (3) to (5);
        \draw [style=op] (0) to (2);
        \draw [style=op] (2) to (6);
      \end{pgfonlayer}
    \end{tikzpicture}
    \caption{Labeled transition system of $\nd{\seq{\com{a}}{\com{b}}}{\conc{\com{c}}{\com{d}}}$}
    \label{fig:ex2-1}
  \end{figure}
\end{example}

\begin{example}\label{ex:small-1-1}
  The initial program is $\conc{(\seq{\com{a}}{\com{b}})}{\com{c}}$, from which we have two possible
  transitions: either by $a$ or by $c$. If we transit by $c$ we go to the command
  $\seq{\com{a}}{\com{b}}$, where we execute $a$ followed by $b$ to complete the computation.  On
  the other hand, case we transition by $a$, we reach the command $\conc{\com{b}}{\com{c}}$, which
  allows two possible transitions: first $b$ and then $c$ or first $c$ and then $b$.

  With the support of Figure~\ref{fig:ex2-1} together with the above explanation, we can
  straightforwardly deduce the words that can be formed by the n-step semantics: $a$, $c$, $ab$,
  $ac$, $ca$, $abc$, $acb$, and $cab$.
  \begin{figure}[ht!]
    \centering
    \begin{tikzpicture}
      \begin{pgfonlayer}{nodelayer}
        \node [style=command] (0) at (0, 11) {$\conc{(\seq{\com{a}}{\com{b}})}{\com{c}}$};
        \node [style=command] (1) at (-1.25, 9.75) {$\conc{\com{b}}{\com{c}}$};
        \node [style=command] (2) at (-2.25, 8.5) {$\com{c}$};
        \node [style=command] (3) at (-2.25, 7.25) {$\checkmark$};
        \node [style=command] (4) at (-0.25, 8.5) {$\com{b}$};
        \node [style=command] (5) at (-0.25, 7.25) {$\checkmark$};
        \node [style=command] (6) at (1.25, 9.75) {$\seq{\com{a}}{\com{b}}$};
        \node [style=command] (7) at (1.25, 8.5) {$\com{b}$};
        \node [style=command] (8) at (1.25, 7.25) {$\checkmark$};
        \node [style=command] (10) at (0.75, 10.5) {$c$};
        \node [style=command] (11) at (-1, 10.5) {$a$};
        \node [style=command] (12) at (1.5, 9.25) {$a$};
        \node [style=command] (13) at (1.5, 8) {$b$};
        \node [style=command] (14) at (-0.5, 9.25) {$c$};
        \node [style=command] (15) at (0, 8) {$b$};
        \node [style=command] (16) at (-2, 9.25) {$b$};
        \node [style=command] (17) at (-2.5, 8) {$c$};
      \end{pgfonlayer}
      \begin{pgfonlayer}{edgelayer}
        \draw [style=op] (0) to (6);
        \draw [style=op] (6) to (7);
        \draw [style=op] (7) to (8);
        \draw [style=op] (0) to (1);
        \draw [style=op] (1) to (4);
        \draw [style=op] (4) to (5);
        \draw [style=op] (1) to (2);
        \draw [style=op] (2) to (3);
      \end{pgfonlayer}
    \end{tikzpicture}
    \caption{Labeled transition system of $\conc{(\seq{\com{a}}{\com{b}})}{\com{c}}$}
    \label{fig:ex2-1}
  \end{figure}
\end{example}

\subsection{Constructions on Event Structures}\label{subsec:constructions-es-1}
Having defined the language, \ie\ its syntax and operational semantics, we now focus on event
structures.  We need to define the constructions on event structures that captures the effects of
sequential composition, non-deterministic choice, and parallel composition.

To capture the behavior of the language's operators, we need to define them in terms of event
structures.

Let us begin with sequential composition. Consider $\seq{\com{a}}{\com{b}}$ to be the sequential
combination of two actions, $a$ and $b$. According to the rules in Figure~\ref{fig:op-small1} we
execute $\com{b}$ after $\com{a}$ has been executed, which with an event structure view means that
$b$ causally depends on $a$.  As a first attempt to define sequential composition of two events
structures, $\seq{\es{E}_1}{\es{E}_2}$, one might try to connect every event of $\es{E}_1$ with
every event of $\es{E}_2$. However, this approach fails to interpret programs like
$\seq{(\nd{\com{a}}{\com{b}})}{\com{c}}$, as show in Figure~\ref{fig:ill-seq}. This failure arises
because there are two ways to reach event $a$, which come from conflicting events. According to
the definition of event structures, the conflict relation is hereditary, and an event is not in
conflict with itself. Thus, we would end up with an invalid event structure.  To address this
issue, we introduce a `copy' for each event of $\es{E}_2$ regarding the different ways it can be
reached. For example, in the aforementioned program, we create two copies of $a$: one indicating
it was reached by executing event $b$, and another indicating it was reached by executing event
$c$, as can be seen in Figure~\ref{fig:good-seq}.
\begin{figure}[ht!]
  \centering
  \begin{subfigure}{0.4\textwidth}
    \centering
    \begin{tikzpicture}
      \begin{pgfonlayer}{nodelayer}
        \node [style=event] (0) at (-2, 7) {};
        \node [style=event] (1) at (-2, 7) {$\com{a}$};
        \node [style=event] (2) at (0, 7) {$\com{b}$};
        \node [style=event] (3) at (-1, 6) {$\com{c}$};
      \end{pgfonlayer}
      \begin{pgfonlayer}{edgelayer}
        \draw [style=wiggle] (1) to (2);
        \draw [style=arrow] (1) to (3);
        \draw [style=arrow] (2) to (3);
      \end{pgfonlayer}
    \end{tikzpicture}
    \caption{Unwanted sequential composition}
    \label{fig:ill-seq}
  \end{subfigure}
  \hfill
  \begin{subfigure}{0.4\textwidth}
    \centering
    \begin{tikzpicture}
      \begin{pgfonlayer}{nodelayer}
        \node [style=event] (0) at (-1, 12) {};
        \node [style=event] (1) at (-1, 12) {$\com{a}$};
        \node [style=event] (2) at (1, 12) {$\com{b}$};
        \node [style=event] (3) at (-1, 11) {$\com{c_a}$};
        \node [style=event] (4) at (1, 11) {$\com{c_b}$};
      \end{pgfonlayer}
      \begin{pgfonlayer}{edgelayer}
        \draw [style=wiggle] (1) to (2);
        \draw [style=arrow] (1) to (3);
        \draw [style=arrow] (2) to (4);
      \end{pgfonlayer}
    \end{tikzpicture}
    \caption{Good sequential composition}
    \label{fig:good-seq}
  \end{subfigure}
\end{figure}


To capture the intended behavior in event structures, we make use of maximal configurations, as
shown in the sequential definition of event structures.
\begin{definition}[Sequential Composition]\label{def:pes-seq1}
  Let $\es{E_1} = \pes{E_1}{\leq_1}{\#_1}$ and $\es{E_2} = \pes{E_2}{\leq_2}{\#_2}$ be event
  structures.  Define $\seq{\es{E_1}}{\es{E_2}} = \pes{E}{\leq}{\#}$ as:
  \begin{align*}
    & E = E_1 \uplus (E_2 \times \confmax{\es{E}_1}) \\
    & \leq\ = \set{e_1 \leq e'_1 \mid e_1 \leq_1 e'_1} \cup
      \set{(e_2,x) \leq (e'_2,x) \mid \ e_2 \leq_2 e'_2} \cup
      \set{ e_1 \leq (e_2,x) \mid e_1 \in x  } \\
    & \#\ = \set{ e \# e' \mid \exists (e_1 \leq e, e'_1 \leq e')\ .\ e_1 \#_1 e'_1 }
      \cup \set{ (e_2,x) \# (e'_2,x) \mid e_2 \#_2 e'_2  }
  \end{align*}
  where $E_2 \times \confmax{\es{E}_1} = \set{ (e,x) \mid e \in E_2,\ x \in \confmax{\es{E}_1}}$ and
  $\uplus$ denotes the disjoint union~\footnote{The proper definition of the disjoint union is
    $A \uplus B = \{(0,a) | a \in A\} \cup \{ (1,b) | b \in B \}$. For $R,S \in A \times B$, the
    disjoint union extends to a relation as $(i,e) R \uplus S (i',e')$ whenever $i=0=i'$ and
    $e R e'$ or $i=1=i'$ and $e S e'$. For the sake of keeping the notations readable, we will keep
    the $0$s and $1$s implicit.}.
\end{definition}

Note that we multiplied $E_2$ with the maximal configurations of $\es{E_1}$.  That is due maximal
configurations representing finished computations. In other words, the set of maximal configurations
gives all the possible ways to reach the first event of $\es{E_2}$.

\begin{lemma}\label{lem:seq-es1}
  Let $\es{E}_1$ and $\es{E}_2$ be event structures. $\seq{\es{E_1}}{\es{E_2}}$ is an event
  structure.
\end{lemma}

The absence of communication in the language considered simplifies the definition of parallel
composition in event structures, when compared to~\cite{winskel88}, since we do not need a mechanism
of synchronization.  In our case, we simply place `side-by-side' the two event structures.
\begin{definition}[Parallel Composition]\label{def:pes-conc1}
  Let $\es{E_1} = \pes{E_1}{\leq_1}{\#_1}$ and $\es{E_2} = \pes{E_2}{\leq_2}{\#_2}$ be event
  structures.  Define $\conc{\es{E_1}}{\es{E_2}} = \pes{E}{\leq}{\#}$ as:
  \begin{align*}
    & E = E_1 \uplus E_2  \\
    & \leq\ = \leq_1 \uplus \leq_2 \\
    & \# = \#_1 \uplus \#_2 
  \end{align*}
\end{definition}

\begin{lemma}\label{lem:conc-es1}
  Let $\es{E}_1$ and $\es{E}_2$ be event structures. $\conc{\es{E_1}}{\es{E_2}}$ is an event
  structure.
\end{lemma}

In Definition~\ref{def:pes-conc1} we use the disjoint union to ensure that whenever we interpret the
same command in parallel, \ie\ $\conc{\com{C}}{\com{C}}$, we have two copies of $\com{C}$ within the
event structure, as each action of $\com{C}$ can occur twice.

At last we have the non-deterministic composition of event structures.  Let us use the rules in
Figure~\ref{fig:op-small1} to give the intuition behind the definition.  Consider the command
$\nd{\com{a}}{\com{b}}$. According to the operational semantics, if we execute $a$ then we cannot
execute $b$ and if we execute $b$ we cannot execute $a$. If we abstract ourselves and instead of
$\nd{\com{a}}{\com{b}}$ we consider $\nd{\com{C}_1}{\com{C}_2}$, we notice that if we execute an
action from $\com{C}_1$, then it is no longer possible to execute any action of $\com{C}_2$ and
vice-versa.  To capture this behavior in event structures, we need to put all the events
corresponding to $\com{C}_1$ in conflict with all the events corresponding to $\com{C}_2$. Formally:
\begin{definition}[Non-deterministic Composition]\label{def:pes-nd1}
  Let $\es{E_1} = \pes{E_1}{\leq_1}{\#_1}$ and $\es{E_2} = \pes{E_2}{\leq_2}{\#_2}$ be event
  structures.
  Define $\nd{\es{E_1}}{\es{E_2}} = \pes{E}{\leq}{\#}$ as:
  \begin{align*}
    & E = E_1 \uplus E_2  \\
    & \leq\ = \leq_1 \uplus \leq_2 \\
    & \# = \#_1 \uplus \#_2 \cup \set{ e_1 \# e_2 \mid e_1 \in E_1,\, e_2 \in E_2}
                            \cup \set{ e_2 \# e_1 \mid e_1 \in E_1,\, e_2 \in E_2}
  \end{align*}
\end{definition}

Equivalently, we can define the partial order in Definition~\ref{def:pes-nd1} as follows:
\[
  e \leq e' =
  \begin{cases}
    e \leq_1 e' & \text{if } e,e' \in E_1 \\
    e \leq_2 e' & \text{if } e,e' \in E_2
  \end{cases}
\]

\begin{lemma}\label{lem:nd-es1}
  Let $\es{E}_1$ and $\es{E}_2$ be event structures. $\nd{\es{E_1}}{\es{E_2}}$ is an event
  structure.
\end{lemma}

We now have everything we need to interpret the language presented in Section~\ref{subsec:lan-1}
using event structures.
\begin{definition}\label{def:den-sem1}
  We interpret commands as event structures as follows:
  \begin{align*}
    & \mf{\com{skip}} = (\set{sk}, \set{sk \leq sk}, \emptyset) \\
    & \mf{\com{a}} = (\set{a}, \set{a \leq a}, \emptyset) \\
    & \mf{\seq{\com{C_1}}{\com{C_2}}} = \seq{\mf{\com{C_1}}}{\mf{\com{C_2}}} \\
    & \mf{\nd{\com{C_1}}{\com{C_2}}} = \nd{\mf{\com{C_1}}}{\mf{\com{C_2}}} \\
    & \mf{\conc{\com{C_1}}{\com{C_2}}} = \conc{\mf{\com{C_1}}}{\mf{\com{C_2}}} \\
  \end{align*}
\end{definition}

When the goal is to show the equivalence between the operational and the denotational semantics,
Definition~\ref{def:pes-fix-order-1} is not suitable, since sequential composition is not
left-monotone. This happens because the inclusion on the set of events is too restrict, \ie\ the
copies made by Definition~\ref{def:pes-seq1} are distinct.  Hence, we need to loose the inclusion on
the set of events, which originates \emph{sub-similar event structures}. We must first be explicit
about the form of the events we consider.  Events can either be pairs or not. If they are not pairs
we call them \emph{plain events}. If they are pairs we call them \emph{composite events}.  We shall
denote both plain and composite events by $a, \dots, e$.  Thus, examples of composite events are:
$(e, x_1)$, $((e, x_1), x_2)$, $(e,1)$, $((e,x),1)$.  Examples of plain events are: $e$, $e_1$,
$e^1$.  Composite events have an underlying plain event.  To extract a plain event from a composite
event, we recursively extract the left side until it is not a pair. Formally, $\pi(a,b) = \pi(a)$,
recursively. An equivalent, explicit definition of $\pi$ is as follows:
\begin{align*}
  \pi : &E \rightarrow E \\
        &e \mapsto
          \begin{cases}
            e & \text{ if } e \text{ is a plain event} \\
            \pi(e') & \text{ if } e = (e',x) \text{ for some } e'
          \end{cases}
\end{align*}
 
For example,
\[
  \pi(a) = a, \qquad
  \pi(a, x_1) = a, \qquad
  \pi((a, x_1), x_2) = a.
\]

Returning to the notion of sub-similar event structure.  Intuitively, an event structure $\es{E}_1$
is a sub-similar event structure of $\es{E}_2$ if plain events, the causal order and the conflict
relation of $\es{E}_1$ are preserved in $\es{E}_2$. Note that the idea of ignoring the ``copies''
created by sequential composition is captured by the condition on plain events.  We say that two
event structures $\es{E}_1$ and $\es{E}_2$ are \emph{similar} if each is a sub-similar event
structure of the other.

A sub-similar event structure is formally defined as follows:
\begin{definition}\label{def:pes-sub1}
  Let $\es{E_1} = \pes{E_1}{\leq_1}{\#_1}$ and $\es{E_2} = \pes{E_2}{\leq_2}{\#_2}$ be event
  structures.  Say $\es{E_1} \sqsubseteq \es{E_2}$ whenever exists an injective function
  $f: E_1 \rightarrow E_2$ such that $\forall e,e' \in E_1$:
  \begin{align*}
    & \pi(f(e)) = \pi(e) \\
    & e \leq_1 e' \Leftrightarrow f(e) \leq_2 f(e') \\
    & e \#_1 e' \Leftrightarrow f(e) \#_2 f(e')
  \end{align*}
  We say that two event structures $\es{E}_1, \es{E}_2$ are similar, denoted
  $\es{E}_1 \equiv \es{E}_2$, iff $\es{E}_1 \sqsubseteq \es{E}_2$ and
  $\es{E}_2 \sqsubseteq \es{E}_1$.
\end{definition}
Note that in Definition~\ref{def:pes-sub1}, when comparing the set of events, we ignore the
``copies'' of events. This comes as a consequence of Definition~\ref{def:pes-seq1} in which we make
``copies'' of the same event to distinguish the different ways an event can be reached. However, the
``copies'' denote the same event. Thus we want to forget the different ways they can be reached and
just focus on the event itself. Case we have not done that, sequential composition would not be
monotone, \ie\ if $\es{E_1} \sqsubseteq \es{E'_1}$ and $\es{E_2} \sqsubseteq \es{E'_2}$ then
$\seq{\es{E_1}}{\es{E_2}} \not\sqsubseteq \seq{\es{E'_1}}{\es{E'_2}}$. That is easily seen when the
number of maximal configurations of $\es{E_1}$ is greater than that of $\es{E'_1}$.

\begin{remark}\label{rem:equiv-1}
  It is clear that if $\es{E}_1 = \es{E}_2$ then $\es{E}_1 \equiv \es{E}_2$.
\end{remark}

To finish this section of definitions, we define the set of initial events and the removal of an
initial event from a event structure.
\begin{definition}[Set of initial events]\label{def:es-init-event}
  Let $\es{E} = \pes{E}{\leq}{\#}$ be a event structure.
  Define the set of initial events as follows:
  \[
    \init{\es{E}} = \set{ e' \mid \not\exists e \in E\ .\ e \leq e' \wedge e \neq e'}
  \]
\end{definition}

When removing an initial event from an event structure, not only the event itself but also all
conflicting events are eliminated. This decision aims to mimic, within event structures, what
happens in a transition using the small-step semantics.  In small-step semantics, once an action
triggers a transition, that same action cannot be executed again. Furthermore, if the transition
occurs within a non-deterministic program, only the program associated with the triggering action
continues, while the others are discarded.
\begin{definition}[Remove initial event]\label{def:rem-init1}
  Let $\es{E} = \pes{E}{\leq}{\#}$ be a event structure and $a \in \init{\es{E}}$.
  Define $\es{E} \backslash a = \pes{E'}{\leq'}{\#'}$ as
  \begin{align*}
    & E' = \set{e \in E \mid \neg (e \# a),\, e \neq a} \\
    & \leq'\, =\, \set{e \leq e' \mid e,e' \in E' } \\
    & \#'\, =\, \set{e \# e' \mid e,e' \in E' }    
  \end{align*}
\end{definition}

\begin{lemma}\label{lem:rem-init-es1}
  Let $\es{E}$ be an event structure and $a \in \init{\es{E}}$.  $\es{E} \backslash a$ is a event
  structure.
\end{lemma}

\subsection{Results}\label{subsec:results-1}
In this section we present the results obtained. For that, we interpret $\checkmark$ as the empty
event structure, \ie\ $\mf{\checkmark} = (\emptyset, \emptyset, \emptyset) = \emptyset$.

The following lemmas show that the sequential, parallel, and non-deterministic compositions are
monotonic w.r.t Definition~\ref{def:pes-sub1}.
\begin{lemma}\label{lem:seq-mono1}
  Let $\es{E_1}, \es{E}'_1, \es{E_2}, \es{E}'_2$ be event structures.  If
  $\es{E_1} \sqsubseteq \es{E'_1}$ and $\es{E_2} \sqsubseteq \es{E'_2}$ then
  $\seq{\es{E_1}}{\es{E_2}} \sqsubseteq \seq{\es{E'_1}}{\es{E'_2}}$.
\end{lemma}

\begin{lemma}\label{lem:nd-mono1}
  Let $\es{E_1}, \es{E}'_1, \es{E_2}, \es{E}'_2$ be event structures.  If
  $\es{E_1} \sqsubseteq \es{E'_1}$ and $\es{E_2} \sqsubseteq \es{E'_2}$ then
  $\nd{\es{E_1}}{\es{E_2}} \sqsubseteq \nd{\es{E'_1}}{\es{E'_2}}$.
\end{lemma}

\begin{lemma}\label{lem:conc-mono1}
  Let $\es{E_1}, \es{E}'_1, \es{E_2}, \es{E}'_2$ be event structures.  If
  $\es{E_1} \sqsubseteq \es{E'_1}$ and $\es{E_2} \sqsubseteq \es{E'_2}$ then
  $\conc{\es{E_1}}{\es{E_2}} \sqsubseteq \conc{\es{E'_1}}{\es{E'_2}}$.
\end{lemma}

Now, Lemma~\ref{lem:op-sim-1} tells us that whenever we have similar event structures, \ie\
$\es{E}_1 \equiv \es{E}_1'$ and $\es{E}_2 \equiv \es{E}_2'$, they will still be similar when
composed of sequentially, parallelly, and non-deterministically.
\begin{lemma}\label{lem:op-sim-1}
  Let $op \in \set{;\, ,\, ||,\, +}$ and $\es{E}_1$, $\es{E}_2$, $\es{E}_1'$ and $\es{E}_2'$ be
  event structures.  If $\es{E}_1 \equiv \es{E}_1'$ and $\es{E}_2 \equiv \es{E}_2'$ then
  $\es{E}_1\, op\, \es{E}_2 \equiv \es{E}_1'\, op\, \es{E}_2'$.
\end{lemma}

Lastly, we show how the removal of an initial event interacts with the operations associated with
sequential composition, concurrent composition, and non-deterministic composition.
\begin{lemma}\label{lem:seq-rem-init1}  
  Let $\es{E_1}, \es{E_2}$ be event structures.  Consider $\seq{\es{E_1}}{\es{E_2}}$ such that
  $l \in \init{\seq{\es{E_1}}{\es{E_2}}}$. Then
  $(\seq{\es{E_1}}{\es{E_2}}) \backslash l \equiv \seq{(\es{E_1}\backslash l)}{\es{E_2}}$.
\end{lemma}

\begin{lemma}\label{lem:nd-rem-init1}
  Let $\es{E_1}, \es{E_2}$ be event structures.  Consider $\nd{\es{E_1}}{\es{E_2}}$ such that
  $l \in \init{\nd{\es{E_1}}{\es{E_2}}}$. Then
  \[
    (\nd{\es{E_1}}{\es{E_2}}) \backslash l \equiv
    \begin{cases}
      \es{E_1} \backslash l & \text{ if } l \in \init{\es{E_1}} \\
      \es{E_2} \backslash l & \text{ if } l \in \init{\es{E_2}}
    \end{cases}
  \]
\end{lemma}

\begin{lemma}\label{lem:conc-rem-init1}
  Let $\es{E_1}, \es{E_2}$ be event structures.  Consider $\conc{\es{E_1}}{\es{E_2}}$ such that
  $l \in \init{\conc{\es{E_1}}{\es{E_2}}}$. Then
  $(\conc{\es{E_1}}{\es{E_2}}) \backslash l \equiv
  \conc{(\es{E_1}\backslash l)}{(\es{E_2} \backslash l)}$.
\end{lemma}

It is straightforward to check that:
\begin{lemma}\label{lem:conc-symmetric1}
  Let $\es{E}_1, \es{E}_2$ be event structures.  Then
  $\conc{\es{E}_1}{\es{E}_2} = \conc{\es{E}_2}{\es{E}_1}$.
\end{lemma}

We now show how the operational and denotational semantics are related via a soundness and adequacy
theorem. Recall that $\mf{\checkmark} = \pes{\emptyset}{\emptyset}{\emptyset}$.

The intuition behind Lemma~\ref{res:soundI-1} is that if the action $l$ triggers a transition from
$\com{C}$ to $\com{C}'$, then removing $l$ from $\mf{\com{C}}_{\gamma}$ is similar to have
$\mf{\com{C}'}_{\gamma}$.  This can be further clarified if we recall the intuition given when
defining the removal of an initial event from an event structure (Definition~\ref{def:rem-init1}).
\begin{lemma}[Soundness I]\label{res:soundI-1}
  If $\com{C} \xrightarrow{l} \com{C'}$ then $\mf{\com{C'}} \equiv \mf{\com{C}} \backslash l$.
\end{lemma}

With a result that establishes a relation between the small-step and denotational semantics, we now
focus on the big-step semantics. Recall that the big-step semantics creates words, which are
sequences of labels. On the other hand, event structures have the notion of covering chains, which
are finite sequences representing the execution of events. We can exploit this similarity to
formulate an equivalence between the big-step and denotational semantics.
\begin{theorem}[Soundness II]\label{res:soundII-1}
  If $\com{C} \xtwoheadrightarrow{\omega} \com{C'}$ then $\exists x \in \mathcal{C}(\mf{\com{C}})$
  such that $\emptyset \stackrel{\omega}{\chain\ } x$.
\end{theorem}

Since adequacy is the reverse of soundness, following a similar procedure to prove it makes
sense. Hence, we first need to establish a relation between denotational and small-step
semantics. To do this, we take advantage once again of the intuition behind removing an initial
event from an event structure.
\begin{lemma}[Adequacy I]\label{res:adI-1}
  Let $l \in \init{\mf{\com{C}}}$. Then $\exists \com{C'} \in (\com{C} \cup \{\checkmark\})$ s.t
  $\com{C} \xrightarrow{l} \com{C'}$ and $\mf{\com{C}} \backslash l \equiv \mf{\com{C'}}$.
\end{lemma}

With Lemma~\ref{res:adI-1}, we can turn our efforts to relate the denotational semantics and the
big-step semantics. To formulate the adequacy theorem, we reverse the reasoning used for
Theorem~\ref{res:soundII-1}. However, we need to be careful when selecting the
configurations. Specifically, we avoid the empty configuration because the big-step semantics does
not allow empty words. The adequacy theorem is then formulated as follows:
\begin{theorem}[Adequacy II]\label{res:adII-1}
  If $\emptyset \neq x \in \mathcal{C}(\mf{\com{C}})$ s.t.
  $\emptyset \stackrel{\omega}{\chain\ } x$ then $\exists \com{C'}$ s.t.
  $C \xtwoheadrightarrow{\omega} \com{C'}$.
\end{theorem}

Theorem~\ref{res:soundII-1} states that every word $\omega$ derived from
the n-step semantics corresponds to a covering chain, and consequently
to a configuration. Conversely, Theorem~\ref{res:adII-1} indicates that if
we have a non-empty covering chain $\omega$, then there exists a
command $\com{C}'$ reachable from $\com{C}$ by executing $\omega$.

\subsection{Introducing cyclic behavior}\label{subsec:cyclic-beh-1}
We now introduce cyclic behavior to the language in Section~\ref{subsec:lan-1}.  In order to avoid
the introduction of the notion of state in the language, the cyclic behavior will be given by
recursion.  In that way, we do not need to associate the notion of state to a command in the
operational semantics.  We can just keep recording the actions that are being made by the program.

Another thing to have in mind is that with cyclic behavior we open the door to infinite
computations.  However, covering chains are only defined in finite sequence of words and infinite
configurations are odd, because we would need to define precisely what it means to be an infinite
configuration.  Hence, the words that we formed with the n-step will be always finite, despite the
possibility of them being infinite.  We can justify this by saying that we are only concerned on the
`interesting words', \ie\ those who are finite.

To introduce recursion we need to add some restrictions when forming programs, since we do not want
to allow commands like: $\rec{X}{\seq{X}{\com{a}}}$ and $\rec{X}{\seq{\seq{\com{a}}{X}}{\com{b}}}$.

Let $X \subseteq Var$, with $Var$ a set of variables.  The syntax is now given by:
\[
  \com{C} ::= \com{skip} \mid a \in Act \mid \seq{\com{C}}{\com{C}} \mid \nd{\com{C}}{\com{C}} \mid
  \conc{\com{C}}{\com{C}} \mid \rec{X}{\com{C}} \mid X
\]
where $\com{skip}$ is a command that does nothing; $\com{a}$ is an atomic action from a
pre-determined set of atomic actions, denoted as $Act$; $\seq{\com{C}}{\com{C}}$ is the usual
sequential composition of programs; $\conc{\com{C}}{\com{C}}$ is the parallel composition of
commands; $\nd{\com{C}}{\com{C}}$ represents the non-deterministic choice; $\rec{X}{\com{C}}$ is the
recursive command; and $X \in Var$ with $Var$ a set of variables. Furthermore, we only consider
closed commands, \ie\ commands in which every variable $X$ is bound by a recursion $\mu X$ and in
sequential composition we only allow recursion to occur at right.

We define the set of free-variables and bound-variables as follows:
\begin{table}[ht!]
  \centering
  \begin{tabular}{l|l}
    $\fvar{\com{skip}} = \emptyset$ & $\bvar{\com{skip}} = \emptyset$ \\
    $\fvar{\com{a}} = \emptyset$ & $\bvar{\com{a}} = \emptyset$ \\
    $\fvar{\seq{\com{C}_1}{\com{C}_2}} =
    \fvar{\com{C}_1} \cup \fvar{\com{C}_2}$ & $\bvar{\seq{\com{C}_1}{\com{C}_2}} =
                                              \bvar{\com{C}_1} \cup \bvar{\com{C}_2}$ \\
    $\fvar{\conc{\com{C}_1}{\com{C}_2}} =
    \fvar{\com{C}_1} \cup \fvar{\com{C}_2}$ & $\bvar{\conc{\com{C}_1}{\com{C}_2}} =
                                              \bvar{\com{C}_1} \cup \bvar{\com{C}_2}$ \\
    $\fvar{\nd{\com{C}_1}{\com{C}_2}} =
    \fvar{\com{C}_1} \cup \fvar{\com{C}_2}$ & $\bvar{\nd{\com{C}_1}{\com{C}_2}} =
                                              \bvar{\com{C}_1} \cup \bvar{\com{C}_2}$ \\
    $\fvar{X} = \set{X}$ & $\bvar{X} = \emptyset$ \\
    $\fvar{\rec{X}{\com{C}}} = \fvar{C} \backslash \set{X}$ & $\bvar{\rec{X}{\com{C}}} = \set{X} \cup \bvar{\com{C}}$
  \end{tabular}
\end{table}

We restrict the sequential composition to those whose free-variables and bound-variables on the left
are empty, \ie\ $\seq{\com{C}_1}{\com{C}_2}$ if $\fvar{\com{C}_1} = \emptyset = \bvar{\com{C}_1}$.
With this restriction we forbid program like $\rec{X}{\seq{X}{\com{a}}}$,
$\rec{X}{\seq{\seq{\com{a}}{X}}{\com{b}}}$ (with the condition $\fvar{\com{C}_1} = \emptyset$) and
$\seq{(\rec{X}{\seq{\com{a}}{X}})}{\com{b}}$ (with the condition $\bvar{\com{C}_1} = \emptyset$).
We want to forbid these kind of programs in sequential composition, because if $\com{C}_1$ never
terminates then the sequential composition never terminates. This is also a restriction that comes
from the fact that covering chains are only defined in finite sequences and that infinite
configurations are odd in event structures.  Note however that we allow programs like
$\rec{X}{\conc{X}{a}}$ and $\rec{X}{\nd{X}{a}}$, since they do not block the computation.

Before adding a rule for the recursion command to Figure~\ref{fig:op-small1}, we need to define what
it means to substitute a variable $X$ by a command $\com{C}'$.  Inspired by~\cite{hindley08}, we
define substitution as follows:
\begin{definition}\label{fig:substitution-1}
  Let $X \in Var$ and $\com{C}, \com{C}'$ be commands. Define $\com{C}[X \leftarrow \com{C}']$,
  where we substitute every free occurrence of $X$ in $\com{C}$ by $\com{C}'$ (while changing bound
  variables to avoid clashes) by induction on $\com{C}$ as follows:
  \begin{align*}
    & \com{skip}[X \leftarrow \com{C'}] = \com{skip} \\
    & \com{a}[X \leftarrow \com{C'}] = \com{a} \\
    & (\seq{\com{C}_1}{\com{C}_2})[X \leftarrow \com{C'}] =
      \seq{\com{C}_1}{(\com{C}_2 [X \leftarrow \com{C'}])} \\
    & (\conc{\com{C}_1}{\com{C}_2})[X \leftarrow \com{C'}] =
      \conc{\com{C}_1[X \leftarrow \com{C'}]}{\com{C}_2 [X \leftarrow \com{C'}]} \\
    & (\nd{\com{C}_1}{\com{C}_2})[X \leftarrow \com{C'}] =
      \nd{\com{C}_1[X \leftarrow \com{C'}]}{\com{C}_2 [X \leftarrow \com{C'}]} \\
    & X [X \leftarrow \com{C}] = \com{C} \\
    & (\rec{X}{\com{C}})[X \leftarrow \com{C'}] =
      \rec{X}{\com{C}} \\
    & (\rec{Y}{\com{C}})[X \leftarrow \com{C'}] =
      \rec{Y}{\com{C}[X \leftarrow \com{C'}]} \text{ if } X \neq Y \text{ and } Y \not\in FV(\com{C'})
  \end{align*}
\end{definition}

We then add to Figure~\ref{fig:op-small1} the following rule for the recursion command, inspired
by~\cite[Fig.~1]{lopez04}:
\[
  \infer{
    \rec{X}{\com{C}} \xlongrightarrow{l} \com{C}'[X \leftarrow \rec{X}{\com{C}}]
  }{
    \com{C} \xlongrightarrow{l} \com{C}'
  }
\]
This rule tells us that if an action $l$ triggers a transition from $\com{C}$ to $\com{C}'$, then
from $\rec{X}{\com{C}}$ we transit with the same label to $\com{C}'[X \leftarrow \rec{X}{\com{C}}]$,
where we substitute in $\com{C}'$ all the occurrences of $X$ by $\rec{X}{\com{C}}$.

\begin{example}\label{ex:loop-1}
  Figure~\ref{fig:ex-loop-1} illustrates the behavior of a non-deterministic toss coin, which
  produces a possibly empty sequence of $a$'s that finishes with $sk$. To understand this we observe
  that the initial program has two possible transitions: (1) we execute $sk$ that terminates the
  computation; (2) we execute $a$, and we transit to a command equal to the initial one in which we
  have two possible transitions again.
  \begin{figure}[ht!]
    \centering
    \begin{tikzpicture}
      \begin{pgfonlayer}{nodelayer}
        \node [style=command] (0) at (0, 0) {$\rec{X}{(\nd{\com{skip}}{\seq{\com{a}}{X}})}$};
        \node [style=command] (1) at (-1, -1.25) {$\checkmark$};
        \node [style=command] (2) at (1, -1.25) {$\rec{X}{(\nd{\com{skip}}{\seq{\com{a}}{X}})}$};
        \node [style=command] (3) at (0, -2.5) {$\checkmark$};
        \node [style=command] (4) at (2, -2.5) {$\rec{X}{(\nd{\com{skip}}{\seq{\com{a}}{X}})}$};
        \node [style=command] (5) at (1, -3.75) {$\checkmark$};
        \node [style=command] (6) at (3, -3.75) {$\rec{X}{(\nd{\com{skip}}{\seq{\com{a}}{X}})}$};
        \node [style=command] (7) at (2, -5) {$\checkmark$};
        \node [style=command] (8) at (4, -5) {$\ddots$};
        \node [style=none] (11) at (-0.75, -0.5) {$sk$};
        \node [style=none] (12) at (0.75, -0.5) {$a$};
        \node [style=none] (13) at (0.25, -1.75) {$sk$};
        \node [style=none] (14) at (1.75, -1.75) {$a$};
        \node [style=none] (15) at (1.25, -3) {$sk$};
        \node [style=none] (16) at (2.75, -3) {$a$};
        \node [style=none] (17) at (2.25, -4.25) {$sk$};
        \node [style=none] (18) at (3.75, -4.25) {$a$};
      \end{pgfonlayer}
      \begin{pgfonlayer}{edgelayer}
        \draw [style=op] (0) to (1);
        \draw [style=op] (0) to (2);
        \draw [style=op] (2) to (3);
        \draw [style=op] (2) to (4);
        \draw [style=op] (4) to (5);
        \draw [style=op] (4) to (6);
        \draw [style=op] (6) to (7);
        \draw [style=op] (6) to (8);
      \end{pgfonlayer}
    \end{tikzpicture}
    \caption{Unrolling the execution of $\rec{X}{(\nd{\com{skip}}{\seq{\com{a}}{X}})}$}
    \label{fig:ex-loop-1}
  \end{figure}
\end{example}

On the event structure side, we want to use the Knaster-Tarski Theorem to build the least-fix point.
To define it, we will  use an order that does not ignore copies,  differently from what happens with
Definition~\ref{def:pes-sub1}.
\begin{definition}\label{def:pes-fix-order-1}
  Let $\es{E_1} = \pes{E_1}{\leq_1}{\#_1}$ and $\es{E_2} = \pes{E_2}{\leq_2}{\#_2}$ be event
  structures.  Say $\es{E_1} \trianglelefteq \es{E_2}$ if:
  \begin{align*}
    & E_1 \subseteq E_2 \\
    & \forall e,e'\ .\ e \leq_1 e'
      \Leftrightarrow
      e, e' \in E_1 \wedge e \leq_2 e' \\
    & \forall e,e'\ .\ e \#_1 e'
      \Leftrightarrow
      e, e' \in E_1 \wedge e \#_2 e'
  \end{align*}
\end{definition}

We now start to check the conditions of the Knaster-Tarski theorem. We first verify that
Definition~\ref{def:pes-fix-order-1} is a partial order with a least element.
\begin{lemma}\label{lem:po-1}
  $\trianglelefteq$ is a partial order.
\end{lemma}

\begin{lemma}\label{lem:po-least-elem-1}
  Define $\bot = \pes{\emptyset}{\emptyset}{\emptyset}$.  $\bot$ is the least element of
  $\trianglelefteq$.
\end{lemma}

We then define what it means to be a least upper bound in terms of event structures, demonstrate
that it is an event structure, and prove that it is indeed a least upper bound of an $\omega$-chain
of event structures.
\begin{definition}\label{def:lub-1}
  Let $\es{E}_1 \trianglelefteq \dots \trianglelefteq \es{E}_n \trianglelefteq \dots$ be a
  $\omega$-chain. Let $\es{E}^{\omega} = \pes{E^\omega}{\leq^\omega}{\#^\omega}$ be its least upper
  bound where:
  \begin{itemize}
  \item $E^\omega = \cup_{n \in \omega} E_n$
  \item $\leq^\omega = \cup_{n \in \omega} \leq_n$
  \item $\#^\omega = \cup_{n \in \omega} \#_n$
  \end{itemize}
\end{definition}

\begin{lemma}\label{lem:lub-es-1}
  $\es{E}^\omega$ is an event structure.
\end{lemma}

\begin{lemma}\label{lem:lub-1}
  Let $\es{E}_1 \trianglelefteq \dots \trianglelefteq \es{E}_n \trianglelefteq \dots$ be a
  $\omega$-chain. Then $\es{E}^\omega$ is its least upper bound.
\end{lemma}

Now we define what it means for an operator in event structures to be monotone and continuous with
respect to Definition~\ref{def:pes-fix-order-1}. Intuitively, an operator is monotone if it
preserves the ordering. For example, for the parallel composition we say that it is monotone if
$\es{E}_1 \trianglelefteq \es{E}'_1$ and $\es{E}_2 \trianglelefteq \es{E}'_2$ then
$\conc{\es{E}_1}{\es{E}_2} \trianglelefteq \conc{\es{E}'_1}{\es{E}'_2}$. An operator is continuous
if the action of the operator on the least upper bounds of event structures is the same as the least
upper bound of the action of the operator on the event structures. For example, for the parallel
composition we say that it is continuous if
$\conc{\bigsqcup_n \es{E}_{n}}{\bigsqcup_m \es{E}_{m}} = \bigsqcup_{n,m}
(\conc{\es{E}_n}{\es{E}_m})$. Formally~\cite[Definition 2.8]{winskel82}:
\begin{definition}\label{def:cont-1}
  Let $op$ be an $n$-ary operation on the class of event structures $\mathbb{E}$.  Say $op$ is
  \textit{monotonic} iff when for event structures we have
 \begin{align*}
   \es{E}_1 \trianglelefteq \es{E}'_1, \dots, \es{E}_n \trianglelefteq \es{E}'_n \text{ then }
   op(\es{E}_1, \dots, \es{E}_n) \trianglelefteq op(\es{E}'_1, \dots, \es{E}'_n)
 \end{align*}
 
 Say $op$ is continuous iff for all countable chains
 \begin{align*}
   & \es{E}_{11} \trianglelefteq \es{E}_{12} \trianglelefteq \dots
     \trianglelefteq \es{E}_{1i} \trianglelefteq \dots \\
  & \vdots \\
   & \es{E}_{n1} \trianglelefteq \es{E}_{n2} \trianglelefteq \dots
     \trianglelefteq \es{E}_{ni} \trianglelefteq \dots
 \end{align*}
 we have
 \begin{align*}
   op \left( \bigsqcup_i \es{E}_{1i}, \dots, \bigsqcup_i \es{E}_{ni} \right)
   =
   \bigsqcup_i op(\es{E}_{1i}, \dots, \es{E}_{ni})
 \end{align*}
 where $\bigsqcup$ denotes the least upper bound w.r.t $\trianglelefteq$.
\end{definition}

To the previous definition, we add that an operation is continuous iff it is continuous in each
argument separately~\cite{winskel82}.

We now show that the operators of the language are monotone with respect to
Definition~\ref{def:pes-fix-order-1}.  Recall that the sequential composition is only right monotone
because of the restriction imposed in the syntax, which requires the free-variables and
bounded-variables of the left command to be empty.
\begin{lemma}\label{lem:seq-fix-mono-1}
  Let $\es{E}, \es{E}_1, \es{E}_2$ be event structures.  If $\es{E}_1 \trianglelefteq \es{E}_2$ then
  $\seq{\es{E}}{\es{E}_1} \trianglelefteq \seq{\es{E}}{\es{E}_2}$.
\end{lemma}

\begin{lemma}\label{lem:conc-fix-mono-1}
  Let $\es{E_1}, \es{E}'_1, \es{E_2}, \es{E}'_2$ be event structures.  If
  $\es{E_1} \trianglelefteq \es{E'_1}$ and $\es{E_2} \trianglelefteq \es{E'_2}$ then
  $\conc{\es{E_1}}{\es{E_2}} \trianglelefteq \conc{\es{E'_1}}{\es{E'_2}}$.
\end{lemma}

\begin{lemma}\label{lem:nd-fix-mono-1}
  Let $\es{E_1}, \es{E}'_1, \es{E_2}, \es{E}'_2$ be event structures.  If
  $\es{E_1} \trianglelefteq \es{E'_1}$ and $\es{E_2} \trianglelefteq \es{E'_2}$ then
  $\nd{\es{E_1}}{\es{E_2}} \trianglelefteq \nd{\es{E'_1}}{\es{E'_2}}$.
\end{lemma}

The following lemma from~\cite[Lemma~2.9]{winskel82} is helpful to show that the sequential, the
concurrent, and the non-deterministic operators are continuous. Intuitively, a function is
continuous on event structures if it is monotonic and act continuously on the component sets of
events ordered by inclusion.
\begin{lemma}\label{lem:cont-1}
  Let $op$ be a unary operation on the class of event structures $\mathbb{E}$.  Then $op$ is
  continuous iff
 \begin{enumerate}
  \item $op$ is monotonic
  \item if $\es{E}_1 \trianglelefteq \dots \trianglelefteq \es{E}_n \trianglelefteq \dots$ is an
    $\omega$-chain then each event of $op(\bigsqcup_n \es{E}_n)$ is an event of
    $\bigsqcup_n op (\es{E}_n)$.
 \end{enumerate}
\end{lemma}

We proceed to show that the operators are continuous. Note that since the sequential composition is
only right-monotone, then it will only be right-continuous.
\begin{lemma}\label{lem:seq-cont-1}
  $\bigsqcup_m(\seq{\es{E}}{\es{E}_m}) = \seq{\es{E}}{\bigsqcup_m \es{E}_m}$.
\end{lemma}

\begin{lemma}\label{lem:conc-cont-1}
  $\bigsqcup_{n,m}(\conc{\es{E}_n}{\es{E}_m}) = \conc{\bigsqcup_n \es{E}_n}{\bigsqcup_m \es{E}_m}$.
\end{lemma}

\begin{lemma}\label{lem:nd-cont-1}
  $\bigsqcup_{n,m}(\nd{\es{E}_n}{\es{E}_m}) = \nd{\bigsqcup_n \es{E}_n}{\bigsqcup_m \es{E}_m}$.
\end{lemma}

Lemma~\ref{lem:fix-prop-1} is a version of the Kleene fixed-point theorem for the case of event
structures. Intuitively, it tells us how to build a least fixed-point for a continuous operator.
\begin{lemma}\label{lem:fix-prop-1}
  Let $\Gamma$ be a continuous operation on the class of event structures $\mathbb{E}$.  Let
  $\bot = (\emptyset, \emptyset, \emptyset) \in \mathbb{E}$.  Define $fix(\Gamma)$ to be the least
  upper bound of the chain
  $\bot \trianglelefteq \Gamma(\bot) \trianglelefteq \dots \trianglelefteq \Gamma^{n}(\bot)
  \trianglelefteq \dots$.  Then $\Gamma(fix(\Gamma)) = fix(\Gamma)$.
\end{lemma}

We now update Definition~\ref{def:den-sem1} to accommodate recursion.
\begin{definition}\label{def:fix-den-sem-1}
  Define an environment to be a function $\gamma : Var \rightarrow \mathbb{E}$ from variables to
  event structures. For a command $\com{C}$ and an environment $\gamma$ define
  $\mf{\com{C}}_\gamma$~\footnote{whenever a command is not recursive, for simplicity we may drop
    the environment $\gamma$ from $\mf{-}_{\gamma}$} as follows:
  \begin{align*}
    & \mf{\com{skip}}_\gamma = (\set{sk}, \set{sk \leq sk}, \emptyset) \\
    & \mf{\com{a}}_\gamma = (\set{a}, \set{a \leq a}, \emptyset) \\
    & \mf{\seq{\com{C_1}}{\com{C_2}}}_\gamma = \seq{\mf{\com{C_1}}_\gamma}{\mf{\com{C_2}}_\gamma} \\
    & \mf{\conc{\com{C_1}}{\com{C_2}}}_\gamma = \conc{\mf{\com{C_1}}_\gamma}{\mf{\com{C_2}}_\gamma} \\
    & \mf{\nd{\com{C_1}}{\com{C_2}}}_\gamma = \nd{\mf{\com{C_1}}_\gamma}{\mf{\com{C_2}}_\gamma} \\
    & \mf{X}_\gamma = \gamma(X) \\
    & \mf{\rec{X}{\com{C}}}_\gamma = fix(\Gamma^{\com{C}, \gamma})
  \end{align*}
  where $\Gamma^{\com{C}, \gamma} : \mathbb{E} \rightarrow \mathbb{E}$ is given by
  $\Gamma^{\com{C}, \gamma}(\es{E}) = \mf{\com{C}}_{\gamma(X \leftarrow \es{E})}$.
\end{definition}

Since $\Gamma^{\com{C}, \gamma}$ is a new operator, we need to show that it is continuous. To do
this, it is useful to know that $fix$ is continuous~\cite{abramsky94}.
\begin{lemma}\label{lem:gamma-cont-1}
  $\Gamma^{\com{C}, \gamma}$ is continuous.
\end{lemma}

With the introduction of the recursive command, we need the two following lemmas for
Lemma~\ref{res:soundI-1}.
\begin{lemma}[Substitution lemma]\label{lem:1-1}
  $\mf{\com{C}'[X \leftarrow \mf{\rec{X}{\com{C}}}_\gamma]}_\gamma =
  \mf{\com{C}'}_{\gamma{(X \leftarrow \mf{\rec{X}{\com{C}}}_\gamma})}$
\end{lemma}

\begin{lemma}\label{lem:2-1}
  $\mf{\rec{X}{\com{C}}}_\gamma = \mf{\com{C}}_{\gamma(X \leftarrow \mf{\rec{X}{\com{C}}}_\gamma)}$
\end{lemma}

We present an example where we illustrate a fragment of the event structure that is obtained by
interpreting a recursive command.
\begin{example}\label{ex:es-recursive}
  Let us reuse the command in Example~\ref{ex:loop-1}.  A fragment of the event structure that is
  generated by interpreting $\mf{\rec{X}{(\nd{\com{skip}}{\seq{\com{a}}{X}})}}_{\gamma}$ is given in
  Figure~\ref{fig:es-recursive}.
  \begin{figure}[ht!]
    \centering
    \begin{tikzpicture}[tikzfig]
      \begin{pgfonlayer}{nodelayer}
        \node [style=event] (0) at (-3, 0) {$\com{sk}_0$};
        \node [style=event] (1) at (-1, 0) {$\com{a}_0$};
        \node [style=event] (2) at (-1, -2) {$\com{sk}_1$};
        \node [style=event] (3) at (1, -2) {$\com{a}_1$};
        \node [style=event] (4) at (1, -4) {$\com{sk}_2$};
        \node [style=event] (5) at (3, -4) {$\com{a}_2$};
        \node [style=event] (6) at (3, -6) {$\com{sk}_3$};
        \node [style=event] (7) at (5, -6) {$\com{a}_3$};
        \node [style=event] (8) at (5, -8) {\vdots};
        \node [style=none] (9) at (5, -7.5) {};
      \end{pgfonlayer}
      \begin{pgfonlayer}{edgelayer}
        \draw [style=wiggle] (0) to (1);
        \draw [style=wiggle] (2) to (3);
        \draw [style=wiggle] (4) to (5);
        \draw [style=wiggle] (6) to (7);
        \draw [style=arrow] (1) to (2);
        \draw [style=arrow] (1) to (3);
        \draw [style=arrow] (3) to (4);
        \draw [style=arrow] (3) to (5);
        \draw [style=arrow] (5) to (6);
        \draw [style=arrow] (5) to (7);
        \draw [style=arrow] (7) to (9.center);
      \end{pgfonlayer}
    \end{tikzpicture}
    \caption{Unrolling the event structure of $\mf{\rec{X}{(\nd{\com{skip}}{\seq{\com{a}}{X}})}}_{\gamma}$}
    \label{fig:es-recursive}
  \end{figure}

  Note that the events have in subscript a natural number indicating how many times the recursive
  command has been unfolded.
\end{example}

To show the equivalence between the operational and the denotational semantics, we reuse what was
done in Section~\ref{subsec:results-1}. The only lemmas in which we need to add the proof for the
recursion case are the following:
\begin{lemma}[Soundness I]\label{res:soundI-fix-1}
  If $\com{C} \xrightarrow{l} \com{C'}$ then
  $\forall \gamma,\ \mf{\com{C'}}_\gamma \equiv \mf{\com{C}}_\gamma \backslash l$.
\end{lemma}

\begin{lemma}[Adequacy I]\label{res:adI-fix-1}
  Let $l \in \init{\mf{\com{C}}}$. Then $\exists \com{C'} \in (\com{C} \cup \{\checkmark\})$ s.t
  $\com{C} \xrightarrow{l} \com{C'}$ and $\mf{\com{C}} \backslash l \equiv \mf{\com{C'}}$.
\end{lemma}

For Theorem~\ref{res:soundII-1} and Theorem~\ref{res:adII-1} we only need to adapt $\mf{-}$ to
$\mf{-}_\gamma$.

We now delve into examples that show soundness and adequacy in practice.
\begin{example}\label{ex:eq-sem-1}
  The event structure in Example~\ref{ex:es-1} corresponds to the command in
  Example~\ref{ex:small-1}.

  To see how the semantics relate, recall the configurations in Example~\ref{ex:es-1} and the words
  in Example~\ref{ex:small-1}.

  Let us select the words $c d$ and $d c$.  It is straightforward to see that each word corresponds
  to a covering chain, $\emptyset \cchain{d} \set{d} \cchain{c} \set{d, c}$ and
  $\emptyset \cchain{c} \set{c} \cchain{d} \set{d, c}$, respectively. Both covering chains
  correspond to the configuration $\set{d, c}$.

  Conversely, the configuration $\set{d, c}$ is obtained by two covering chains:
  $\emptyset \cchain{d} \set{d} \cchain{c} \set{d, c}$ and
  $\emptyset \cchain{c} \set{c} \cchain{d} \set{d, c}$. It is straightforward to see that each
  covering chain corresponds to the words $d c$ and $c d$, respectively.
\end{example}

\begin{example}
  Figure~\ref{fig:ex3-1} shows the event structure corresponding to the interpretation of
  $\mf{\conc{(\seq{\com{a}}{\com{b}})}{\com{c}}}$.  The set of configurations is
  $\set{\emptyset, \set{a}, \set{c}, \set{a, b}, \set{a, c}, \set{a, b, c}}$, where we note that in
  the presence of concurrent events, a configuration has more than one possible covering chain.
  
  To see the equivalence between both semantics through an example, recall the words that can be
  formed by the n-step in Example~\ref{ex:small-1}: $a$, $c$, $ab$, $ac$, $ca$, $abc$, $acb$, and
  $cab$.

  Each word corresponds to a covering chain, which represents a configuration. For example the words
  $ac$ and $ca$ correspond to the covering chains
  $\emptyset \cchain{a} \set{a} \cchain{c} \set{a, c}$ and
  $\emptyset \cchain{c} \set{c} \cchain{a} \set{a, c}$, respectively. These covering chains
  correspond to the configuration $\set{a, c}$. Conversely, for each covering chain, there exists a
  corresponding word.
  \begin{figure}[ht!]
    \centering
    \begin{tikzpicture}
      \begin{pgfonlayer}{nodelayer}
        \node [style=event] (0) at (0, 6) {$a$};
        \node [style=event] (1) at (0, 4.5) {$b$};
        \node [style=event] (2) at (1.5, 6) {$c$};
      \end{pgfonlayer}
      \begin{pgfonlayer}{edgelayer}
        \draw [style=arrow] (0) to (1);
      \end{pgfonlayer}
    \end{tikzpicture}
    \caption{Event structure of $\mf{\conc{(\seq{\com{a}}{\com{b}})}{\com{c}}}$}
    \label{fig:ex3-1}
  \end{figure}
\end{example}

\begin{example}\label{ex:es-com-rec}
  Recall the command used in Example~\ref{ex:loop-1} and Example~\ref{ex:es-recursive},
  $\rec{X}{(\nd{\com{skip}}{\seq{\com{a}}{X}})}$.  From the former, we know that the words formed by
  the big-step semantics are a possibly empty sequence of $\com{a}$'s that finishes with
  $\com{sk}$. Consider the word $\omega = a\ a\ sk$, which gives
  $\rec{X}{(\nd{\com{skip}}{\seq{\com{a}}{X}})} \xtwoheadrightarrow{\omega} \checkmark$.  From
  Example~\ref{ex:es-recursive}, we can deduce that we have a configuration
  \[
    x = \{a_0, a_1, sk_2\} \in
    \confES{\mf{\rec{X}{(\nd{\com{skip}}{\seq{\com{a}}{X}})}}_{\gamma}}
  \]
  such that $\emptyset \cchain{\omega}\ x$.  On the other way around, from
  \[
    \emptyset \cchain{a_0}\ \{a_0\} \cchain{a_1}\ \{a_0, a_1\} \cchain{sk_2}\ \{a_0, a_1, sk_2\}
  \]
  \ie\ $\emptyset \cchain{\omega}\ \{a_0, a_1, sk_2\}$, we have that
  $\rec{X}{(\nd{\com{skip}}{\seq{\com{a}}{X}})} \xtwoheadrightarrow{\omega} \checkmark$.

  Furthermore, we can extend this argument as follows. Let $\omega = a^n\ sk$, with $n$ being a
  finite number and $a^n$ denoting a sequence of $n$ computations triggered by the action $a$. We
  have that 
  \[
    \rec{X}{(\nd{\com{skip}}{\seq{\com{a}}{X}})} \xtwoheadrightarrow{\omega} \checkmark
  \]
  On the event structure side, we have a configuration
  \[
    x = \cup_{i=0}^{n} \{a_i\} \cup \{sk_{n+1}\} \in
    \confES{\mf{\rec{X}{(\nd{\com{skip}}{\seq{\com{a}}{X}})}}_{\gamma}}
  \]
  with $\emptyset \cchain{a_0}\ \dots \cchain{sk_{n+1}}\ \cup_{i=0}^{n} \{a_i\} \cup \{sk_{n+1}\}$.
\end{example}

\newpage
\section{Probabilistic Event Structures}\label{sec:pes}
This section describes how to endow event structures with a probabilistic behavior. Probabilistic
event structures~\cite{winskel14} are event structures together with a function on configurations,
$v : \confES{\es{E}} \rightarrow [0,1]$, which assigns $1$ for the empty set, jointly with a
\emph{monotonicity} condition. The intuition behind $v(x)$ is the probability of reaching at least
configuration $x$.

The monotonicity condition in~\cite[Definition~1]{winskel14} is described in terms of \emph{drop
  functions}, which will be necessary when showing that unitary event structures with an initial
state corresponds to probabilistic event structures.  Define drop functions
$\dropc{n}{v}{y}{x_1, \dots x_n} \in \mathbb{R}$ for $y, x_1, \dots x_n \in \confES{\es{E}}$ with
$y \subseteq x_1, \dots, x_n$, by induction, taking
\begin{align*}
  & \dropc{0}{v}{y}{} = v(y) \\
  & \dropc{n}{v}{y}{x_1, \dots, x_n} = \dropc{n-1}{v}{y}{x_1, \dots, x_{n-1}}
    - \dropc{n-1}{v}{x_n}{x_1 \cup x_n, \dots, x_{n-1} \cup x_n}
\end{align*}
for $n>0$.  Intuitively, a drop function $\dropc{n}{v}{y}{x_1, \dots x_n}$ indicates the probability
of reaching at least configuration $y$ without reaching any of the $x_1, \dots, x_n$, with
$y \subseteq x_1, \dots, x_n$.

Accordingly to~\cite[Proposition~1]{winskel14}, the drop function can be described based on the
inclusion-exclusion principle for sets~\footnote{The inclusion-exclusion principle for sets is a
  method to count the number of elements in the union of possibly overlapping finite sets}
(Equation~\ref{eq:prob-es-condition}).  We make use of such result to present the definition of
probabilistic event structures.
\begin{definition}[Probabilistic event structure]\label{def:prob-es}
  Let $\es{E} = \pes{E}{\leq}{\#}$ be an event structure. A configuration-valuation on $\es{E}$ is a
  function $v : \mathcal{C}(\es{E}) \rightarrow [0,1]$ such that $v(\emptyset) = 1$ and
  $\forall y, x_1, \dots x_n \in \confES{\es{E}}$ such that $y \subseteq x_1, \dots, x_n$
  \begin{align}
    v(y) - \sum_{\emptyset \neq I \subseteq \set{1, \dots, n}} (-1)^{|I|+1}
    v
    \left(
    \bigcup_{i \in I} x_i
    \right)
    \geq 0\label{eq:prob-es-condition}
  \end{align}
  where $v(x) = 0$ whenever $x \not\in \confES{\es{E}}$.
  
  A probabilistic event structure, $\es{P} = \pespq{\es{E}}{v}$, comprises an event structure
  $\es{E} = \pes{E}{\leq}{\#}$ together with a configuration-valuation
  $v : \mathcal{C}(\es{E}) \rightarrow [0,1]$.
\end{definition}

From Equation~\ref{eq:prob-es-condition}, there are some conclusions that we can
make~\cite{visme19}:
\begin{itemize}
\item The configuration-valuation is decreasing:
  \[
    x \subseteq y \implies v(x) \geq v (y)
  \]
\item Events in conflict have a conditional probability whose sum is less than or equal to one:
  \[
    \forall 1 \leq i \leq n,\ x \cchain{e_i} x_i \text{ and }
    \forall 1 \leq i < j \leq n,\ 
    e_i \# e_j \implies \sum_{i = 1}^n \dfrac{v(x \cup \set {e_i})}{v(x)} \leq 1
  \]
\end{itemize}

The first conclusion taken from Equation~\ref{eq:prob-es-condition} is trivial. Intuitively it says
that as further we advance in the computation, the probabilities either stay the same or decrease.
The second conclusion can be illustrated by considering two conflicting events $e$ and $e'$. Now
consider $x \cchain{e}\ x \cup \{e\}$ and $x \cchain{e'}\ x \cup \{e'\}$. Clearly we have
$x \subseteq x \cup \{e\}, x \cup \{e'\}$. From Equation~\ref{eq:prob-es-condition} we then derive
the following:
\begin{align*}
  v(x) - v(x \cup \{e\}) - v(x \cup \{e'\}) \geq 0
  & \Leftrightarrow
  v(x) \geq v(x \cup \{e\}) + v(x \cup \{e'\}) \\
  & \Leftrightarrow
  1 \geq \dfrac{v(x \cup \{e\}) + v(x \cup \{e'\})}{v(x)} \\
  & \Leftrightarrow
  \dfrac{v(x \cup \{e\}) + v(x \cup \{e'\})}{v(x)} \leq 1
\end{align*}
This can be interpreted as a conditional probability: given that configuration $x$ has occurred, the
probability of subsequently executing either $e$ or $e'$ cannot exceed $1$.

Example~\ref{ex:prob-es-2} illustrates probabilistic event structures.
\begin{example}\label{ex:prob-es-2}
  Figure~\ref{fig:ex1-2} shows a probabilistic event structure very similar to the event structure
  in Figure~\ref{fig:ex1-1}, the only difference being the addition of a new event $\tau$, for which
  the events $a$, $c$, and $d$ are causally dependent. The event $\tau$ is used to indicate that the
  events that are causally immediate to it, \ie\ $\tau \rightarrowtriangle a$,
  $\tau \rightarrowtriangle c$, and $\tau \rightarrowtriangle d$ arose from a probabilistic
  choice. Consequently, their probabilities are not $1$, as can be seen by the
  configuration-valuation. Furthermore, in the drawing of the event structure, we subscript the
  probability of each event to make it easier to understand where probabilities came from.

  The set of configurations is composed of
  $\set{\emptyset, \set{\tau}, \set {\tau, a}, \set{\tau, c} \set{\tau, d}, \set{\tau, a,b},
    \set{\tau, c,d}}$, where $\set{\tau, a, b}$ and $\set{\tau, c, d}$ are maximal configurations
  with probability $p$ and $1-p$, respectively.
  
  \begin{figure}[ht!]
    \centering
    \begin{minipage}[ht!]{0.3\textwidth}
      \begin{tikzpicture}
	\begin{pgfonlayer}{nodelayer}
          \node [style=event] (0) at (-2.25, 3) {$a_{\textcolor{blue}{p}}$};
          \node [style=event] (1) at (-2.25, 2) {$b_{\textcolor{blue}{1}}$};
          \node [style=event] (2) at (-0.5, 3) {$c_{\textcolor{blue}{(1-p)}}$};
          \node [style=event] (3) at (0.75, 3) {$d_{\textcolor{blue}{(1-p)}}$};
          \node [style=event] (4) at (-0.5, 4) {$\tau_{\textcolor{blue}{1}}$};
	\end{pgfonlayer}
	\begin{pgfonlayer}{edgelayer}
          \draw [style=arrow] (0) to (1);
          \draw [style=wiggle] (0) to (2);
          \draw [style=wiggle, bend right] (0) to (3);
          \draw [style=arrow] (4) to (0);
          \draw [style=arrow] (4) to (2);
          \draw [style=arrow] (4) to (3);
	\end{pgfonlayer}
      \end{tikzpicture}
    \end{minipage}    
    \begin{minipage}{0.2\textwidth}
      \[
        v(x) =
        \begin{cases}
          p & \text{ if } a \in x \\
          1-p & \text{ if } c \in x \text{ or } d \in x \\
          1 & \text{ otherwise}
        \end{cases}
      \]
    \end{minipage}
    \caption{Example of a probabilistic event structure}
    \label{fig:ex1-2}
  \end{figure}
\end{example}

\subsection{Language}\label{subsec:lan-2}
The set of commands allowed by the language are given by the following grammar (where $p \in ]0,1[$):
\[
  \com{C} ::= \com{skip} \mid a \in Act \mid \seq{\com{C}}{\com{C}} \mid \probC{\com{C}}{p}{\com{C}} \mid
  \conc{\com{C}}{\com{C}}
\]

In the design of this language we made two choices: the first was to substitute the
non-deterministic operator by the probabilistic operator and the second concerns the intervals for
which $p$ ranges.  The justification for the former is related with the chosen probabilistic event
structure.  In sum, Winskel probabilistic event structures are not suitable to model a language that
posses both non-deterministic and probabilistic operators, as explained in~\cite{visme19}.
Regarding the latter, the intervals chosen are influenced by Definition~\ref{def:rem-init2}, since
it is no reasonable to remove an initial event when its probability is zero.

We extend the set of labels with a new label $\tau$, \ie\ $L' = L \cup \set{\tau}$ and let it be
ranged by $l'$.  Similarly to process algebra, $\tau$ will be used to denote an invisible
transition.

We fix
$\Dist(X) = \set{\psi : X \rightarrow [0,1] \mid sup(X) \text{ finite}, \sum_{x \in X} \psi(x) = 1}$
as being the probabilistic finite support functor and we We define the small-step transition step
(labeled Segala automaton),
$\rightarrow \subseteq \com{C} \times \Dist(L' \times (\com{C} \cup \{ \checkmark \}))$, as the
smallest relation obeying the following rules:
\begin{figure}[ht!]
  \centering
  \begin{align*}
    \com{skip} \rightarrow 1 \cdot (sk, \checkmark)
    \hspace*{1cm}
    \com{a} \rightarrow 1 \cdot (a, \checkmark)
    \hspace*{1cm}
    \probC{\com{C}_1}{p}{\com{C}_2} \rightarrow p \cdot (\tau, \com{C}_1) + (1-p) \cdot (\tau, \com{C}_2)
  \end{align*}
  \begin{align*}
    \infer{\seq{\com{C}_1}{\com{C}_2} \rightarrow 1 \cdot (l, \com{C}_2)}
    {
    \com{C}_1 \rightarrow 1 \cdot (l, \checkmark)
    }
    \hspace*{1cm}
    \infer{
    \seq{\com{C}_1}{\com{C}_2} \rightarrow 1 \cdot (l, \seq{ \com{C}'_1}{ \com{C}_2})
    }
    {
    \com{C}_1 \rightarrow 1 \cdot (l, \com{C}'_1)
    }
    \hspace*{1cm}
    \infer{
    \seq{\com{C}_1}{\com{C}_2} \rightarrow \sum_i p_i \cdot (\tau, \seq{\com{C}_i}{\com{C}_2})
    }{
    \com{C}_1 \rightarrow \sum_i p_i \cdot (\tau, \com{C}_i)
    }
  \end{align*}
  \begin{align*}
    \infer{\conc{\com{C}_1}{\com{C}_2} \rightarrow 1 \cdot (l,\com{C}_2)}
    {
    \com{C}_1 \rightarrow 1 \cdot (l, \checkmark)
    }
    \hspace*{1cm}
    \infer{
    \conc{\com{C}_1}{\com{C}_2} \rightarrow 1 \cdot (l, \conc{ \com{C}'_1  }{ \com{C}_2})
    }
    {
    \com{C}_1 \rightarrow 1 \cdot (l, \com{C}'_1)
    }
    \hspace*{1cm}
    \infer{
    \conc{\com{C}_1}{\com{C}_2} \rightarrow \sum_i p_i \cdot (\tau, \conc{\com{C}_i}{\com{C}_2})
    }{
    \com{C}_1 \rightarrow \sum_i p_i \cdot (\tau, \com{C}_i)
    }
  \end{align*}
  \begin{align*}
    \infer{\conc{\com{C}_1}{\com{C}_2} \rightarrow 1 \cdot (l,\com{C}_1)}
    {
    \com{C}_2 \rightarrow 1 \cdot (l, \checkmark)
    }
    \hspace*{1cm}
    \infer{
    \conc{\com{C}_1}{\com{C}_2} \rightarrow 1 \cdot (l, \conc{ \com{C}_1  }{ \com{C}'_2})
    }
    {
    \com{C}_2 \rightarrow 1 \cdot (l, \com{C}'_2)
    }
    \hspace*{1cm}
    \infer{
    \conc{\com{C}_1}{\com{C}_2} \rightarrow \sum_j p_j \cdot (\tau, \conc{\com{C}_1}{\com{C}_j})
    }{
    \com{C}_2 \rightarrow \sum_j p_j \cdot (\tau, \com{C}_j)
    }
  \end{align*}  
  \caption{Rules of the probabilistic small-step operational semantics}
  \label{fig:op-small2}
\end{figure}

Define a word to be a sequence of labels:
\[
  \omega ::= l' \mid l' : \omega 
\]
where $l' : \omega$ appends $l'$ to the beginning of $\omega$.  A word can also be seen as an
element of $(L')^+$, \ie\ a possibly infinite sequence of labels without the empty sequence. Despite
$(L')^+$ allows the possibility of having infinite words, by now we focus only on the finite words.

Define the $n$-step transition,
$\xrightarrow{\omega_p} \subseteq \com{C} \times \Dist((L')^{+} \times (\com{C} \cup \{\checkmark\}))$,
where $n$ is the length of the words, as follows:
\begin{figure}[ht!]
  \centering
  \begin{align*}
    \infer{
    \com{C} \twoheadrightarrow \sum_i p_i (l', \com{C}_i)
    }{
    \com{C} \rightarrow \sum_i p_i (l', \com{C}_i)
    }
    \hspace*{2cm}
    \infer{
    \com{C} \twoheadrightarrow \sum_i p_i \left( \sum_j p_j \cdot (l' : \omega_{ij}, \com{C}_{ij}) \right)
    }{
    \com{C} \rightarrow \sum_i p_i (l', \com{C}_i)
    \qquad
    \forall i\,
    \com{C}_i \twoheadrightarrow \sum_j p_j \cdot (\omega_{ij}, \com{C}_{ij})
    }
  \end{align*}
  \caption{Rules of the n-step operational semantics}
  \label{fig:op-nstep2}
\end{figure}

The left rule represents the execution of a single step in a computation, while the right rule
represents multiple steps of the computation. The latter rule can be understood as follows: if
$\com{C}$ transits to $\sum_i p_i \cdot (l', \com{C}_i)$ and for each $\com{C}_i$ we transit to
$\sum_j p_j \cdot (\omega_{ij}, \com{C}_{ij})$, then by appending $l'$ to each $\omega_{ij}$, we can
transit from $\com{C}$ to
$\sum_i p_i \left( \sum_j p_j \cdot (l' : \omega_{ij}, \com{C}_{ij}) \right)$. In this transition,
for each $i$, we multiply the probabilities obtained from the small-step transition with the
probabilities obtained from the n-step transition.

\begin{example}\label{ex:small-2}
  In Figure~\ref{fig:ex2-2} we use straight arrows to denote a transition from a command to a
  distribution, which we denote by $\bullet$, labeled by the triggering action and wiggly arrows to
  represent a transition from a distribution to a command labeled by the associated probability.

  From $\probC{(\seq{\com{a}}{\com{b}})}{p}{(\conc{\com{c}}{\com{d}})}$ we transit with $\tau$ to
  the distribution $p \cdot \seq{\com{a}}{\com{b}} + (1-p) \cdot \conc{\com{c}}{\com{d}}$, which
  transits with probability $p$ to $\seq{\com{a}}{\com{b}}$ and with probability $1-p$ to
  $\conc{\com{c}}{\com{d}}$. For the former, by executing first $a$ and then $b$ we reach the end of
  the computation.  For the latter, since it is a concurrent program, to finish the computation we
  can either execute first $c$ and then $d$ or we can execute first $d$ and then $c$.

  Based on Figure~\ref{fig:ex2-2} and following the rules in Figure~\ref{fig:op-nstep2}, we can
  deduce that with probability $p$ the word $\tau a b$ leads to a final computation and the same
  behavior is captured with probability $1-p$ with the words $\tau c d$ and $\tau d c$.
  \begin{figure}[ht!]
    \centering
    \begin{tikzpicture}
      \begin{pgfonlayer}{nodelayer}
        \node [style=command] (0) at (0, 0) {$\probC{(\seq{\com{a}}{\com{b}})}{p}{(\conc{\com{c}}{\com{d}})}$};
        \node [style=command] (1) at (0, -1) {$\bullet$};
        \node [style=command] (2) at (1.25, -2) {$\conc{\com{c}}{\com{d}}$};
        \node [style=command] (3) at (-1.25, -2) {$\seq{\com{a}}{\com{b}}$};
        \node [style=command] (4) at (-1.25, -3) {$\bullet$};
        \node [style=command] (7) at (0.5, -4) {$\com{d}$};
        \node [style=command] (8) at (1.75, -4) {$\com{c}$};
        \node [style=command] (9) at (0.5, -6) {$\checkmark$};
        \node [style=command] (10) at (1.75, -6) {$\checkmark$};
        \node [style=command] (11) at (1.75, -5) {$\bullet$};
        \node [style=command] (12) at (0.5, -5) {$\bullet$};
        \node [style=command] (18) at (0.5, -3) {$\bullet$};
        \node [style=command] (19) at (1.75, -3) {$\bullet$};
        \node [style=command] (20) at (-1.25, -4) {$\com{b}$};
        \node [style=command] (21) at (-1.25, -5) {$\bullet$};
        \node [style=command] (22) at (-1.25, -6) {$\checkmark$};
        \node [style=none] (23) at (0.25, -0.5) {$\tau$};
        \node [style=none] (24) at (-0.75, -1.25) {$p$};
        \node [style=none] (25) at (0.75, -1.25) {$1-p$};
        \node [style=none] (26) at (-1.5, -2.5) {$a$};
        \node [style=none] (27) at (-1.5, -3.5) {$1$};
        \node [style=none] (28) at (-1.5, -4.5) {$b$};
        \node [style=none] (29) at (-1.5, -5.5) {$1$};
        \node [style=none] (30) at (0.5, -2.5) {$c$};
        \node [style=none] (31) at (0.25, -3.5) {$1$};
        \node [style=none] (32) at (0.25, -4.5) {$d$};
        \node [style=none] (33) at (0.25, -5.5) {$1$};
        \node [style=none] (34) at (1.75, -2.5) {$d$};
        \node [style=none] (35) at (2, -3.5) {$1$};
        \node [style=none] (36) at (2, -4.5) {$c$};
        \node [style=none] (37) at (2, -5.5) {$1$};
      \end{pgfonlayer}
      \begin{pgfonlayer}{edgelayer}
        \draw [style=op] (0) to (1);
        \draw [style=op] (3) to (4);
        \draw [style=op] (8) to (11);
        \draw [style=op] (7) to (12);
        \draw [style=prob] (1) to (3);
        \draw [style=prob] (1) to (2);
        \draw [style=prob] (11) to (10);
        \draw [style=prob] (12) to (9);
        \draw [style=op] (2) to (19);
        \draw [style=op] (2) to (18);
        \draw [style=prob] (19) to (8);
        \draw [style=prob] (18) to (7);
        \draw [style=prob] (21) to (22);
        \draw [style=prob] (4) to (20);
        \draw [style=op] (20) to (21);
      \end{pgfonlayer}
    \end{tikzpicture}
    \caption{Segala Automaton of $\probC{(\conc{\com{a}}{\com{b}})}{p}{\com{c}}$}
    \label{fig:ex2-2}
  \end{figure}
\end{example}

\begin{example}
  In Figure~\ref{fig:ex2-2} we use straight arrows to denote a transition from a configuration to a
  distribution, which we denote by $\bullet$, labeled by the triggering action and wiggly arrows to
  represent a transition from a distribution to a configuration labeled by the associated
  probability.

  From $\probC{(\conc{\com{a}}{\com{b}})}{p}{\com{c}}$ we transit with $\tau$ to the distribution
  $p \cdot \conc{\com{a}}{\com{b}} + (1-p) \cdot \com{c}$, which transits with probability $p$ to
  $\conc{\com{a}}{\com{b}}$ and with probability $1-p$ to $\com{c}$. By executing $\com{c}$ the
  computation finishes. On the other side we have two possible transitions: either we transit with
  $a$, leading to the distribution $1 \cdot \com{b}$, which terminates after executing $\com{b}$, or
  we transit with $b$ which goes to the distribution $1 \cdot \com{a}$ that after executing
  $\com{a}$ terminates the computation.  The words that lead
  $\probC{(\conc{\com{a}}{\com{b}})}{p}{\com{c}}$ to $\checkmark$ are $\tau a b$, $\tau b a$ with
  probability $p$ and $\tau c$ with probability $1-p$.
  \begin{figure}[ht!]
    \centering
    \begin{tikzpicture}
      \begin{pgfonlayer}{nodelayer}
        \node [style=command] (0) at (0, 0) {$\probC{(\conc{\com{a}}{\com{b}})}{p}{\com{c}}$};
        \node [style=command] (1) at (0, -1) {$\bullet$};
        \node [style=command] (2) at (-1, -2) {$\conc{\com{a}}{\com{b}}$};
        \node [style=command] (3) at (1, -2) {$\com{c}$};
        \node [style=command] (4) at (1, -3) {$\bullet$};
        \node [style=command] (6) at (1, -4) {$\checkmark$};
        \node [style=command] (7) at (-1.75, -4) {$\com{c}$};
        \node [style=command] (8) at (-0.5, -4) {$\com{c}$};
        \node [style=command] (9) at (-1.75, -6) {$\checkmark$};
        \node [style=command] (10) at (-0.5, -6) {$\checkmark$};
        \node [style=command] (11) at (-0.5, -5) {$\bullet$};
        \node [style=command] (12) at (-1.75, -5) {$\bullet$};
        \node [style=command] (13) at (-0.75, -1.5) {$p$};
        \node [style=command] (14) at (1, -1.5) {$1-p$};
        \node [style=command] (15) at (1.25, -3.5) {$1$};
        \node [style=command] (16) at (-0.25, -5.5) {$1$};
        \node [style=command] (17) at (-2, -5.5) {$1$};
        \node [style=command] (18) at (-1.75, -3) {$\bullet$};
        \node [style=command] (19) at (-0.5, -3) {$\bullet$};
        \node [style=command] (20) at (-0.25, -3.5) {$1$};
        \node [style=command] (21) at (-2, -3.5) {$1$};
        \node [style=command] (22) at (0.25, -0.5) {$\tau$};
        \node [style=command] (23) at (-1.75, -2.5) {$a$};
        \node [style=command] (24) at (-0.5, -2.5) {$b$};
        \node [style=command] (25) at (1.25, -2.5) {$c$};
        \node [style=command] (26) at (-2, -4.5) {$b$};
        \node [style=command] (27) at (-0.25, -4.5) {$a$};
      \end{pgfonlayer}
      \begin{pgfonlayer}{edgelayer}
        \draw [style=op] (0) to (1);
        \draw [style=op] (3) to (4);
        \draw [style=op] (8) to (11);
        \draw [style=op] (7) to (12);
        \draw [style=prob] (1) to (3);
        \draw [style=prob] (1) to (2);
        \draw [style=prob] (4) to (6);
        \draw [style=prob] (11) to (10);
        \draw [style=prob] (12) to (9);
        \draw [style=op] (2) to (19);
        \draw [style=op] (2) to (18);
        \draw [style=prob] (19) to (8);
        \draw [style=prob] (18) to (7);
      \end{pgfonlayer}
    \end{tikzpicture}
    \caption{Segala Automaton of $\probC{(\conc{\com{a}}{\com{b}})}{p}{\com{c}}$}
    \label{fig:ex2-2}
  \end{figure}
\end{example}

\subsection{Constructions on Probabilistic Event Structures}
The constructions on probabilistic event structures are an extension of the ones defined previously.
Hence, the explanation of the sequential and parallel composition will be focused on the valuation
and we detail more the probabilistic choice.

Let $\es{P}_1$ and $\es{P}_2$ be two probabilistic event structures.  For the valuation of the
sequential composition we note the following: either the configuration belongs to
$\confES{\es{P}_1}$ and in that case the valuation of the sequential composition equals the
valuation of $\es{P}_1$, or the configuration has elements of both probabilistic event
structures. In that case, we multiply valuation of a maximal configuration in $\es{P}_1$ with the
valuation of a configuration in $\es{P}_2$ whose events are reached by the maximal configuration of
$\es{P}_1$.
\begin{definition}[Sequential Composition]\label{def:pes-seq2}
  Let $\es{P_1} = \ppes{E_1}{\leq_1}{\#_1}{v_1}$ and $\es{P_2} = \ppes{E_2}{\leq_2}{\#_2}{v_2}$ be
  probabilistic event structures.  Define $\seq{\es{P_1}}{\es{P_2}} = \ppes{E}{\leq}{\#}{v}$ as:
  \begin{align*}
    & E = E_1 \uplus (E_2 \times \confmax{\es{P}_1}) \\
    & \leq\ = \set{e_1 \leq e'_1 \mid e \leq_1 e'} \cup
      \set{(e_2,x) \leq (e'_2,x) \mid \ e_2 \leq_2 e'_2} \cup
      \set{ e_1 \leq (e_2,x) \mid e_1 \in x  } \\
    & \#\ = \set{ e \# e' \mid \exists (e_1 \leq e, e'_1 \leq e')\ .\ e_1 \#_1 e'_1 }
      \cup \set{ (e_2,x) \# (e'_2,x) \mid e_2 \#_2 e'_2  } \\
    & \forall x \in \mathcal{C}(\seq{\es{P}_1}{\es{P}_2})\, .\, v(x) =
      \begin{cases}
        v_1(x) & \text{ if } x \in \mathcal{C}(\es{P}_1) \\
        v_1(x_1) \cdot v_{2}(x_2) & \text{ if } x = x_1 \cup (x_2 \times \{x_1\}) \\
        & \hspace*{1cm} 
                                    \text{ with } x_1 \in \confmax{\es{P}_1},\,
                                    x_2 \in \mathcal{C}(\es{P}_{2})
      \end{cases}
  \end{align*}
  where $E_2 \times \confmax{\es{P}_1} = \set{ (e,x) \mid e \in E_2,\ x \in \confmax{\es{P}_1}}$ and
  $x_2 \times \{x_1\} = \{(e_2, x_1) \mid e_2 \in x_2\}$.
\end{definition}

\begin{lemma}\label{lem:seq-es2}
  Let $\es{P_1}$ and $\es{P_2}$ be probabilistic event structures.  $\seq{\es{P}_1}{\es{P}_2}$ is a
  probabilistic event structure.
\end{lemma}

We now turn our attention to the probabilistic choice. As it can be deduced, the probabilistic
choice of two probabilistic event structures is going to be very similar to the non-deterministic
choice of event structures in Definition~\ref{def:pes-nd1}. Indeed when composing two probabilistic
events structures $\es{P}_1$ and $\es{P}_2$ under the probabilistic operator, the events of
$\es{P}_1$ will all be in conflict with the events of $\es{P}_2$ (recall that the meaning of
$\probC{\com{C}_1}{p}{\com{C}_2}$ is to execute $\com{C}_1$ with probability $p$ or execute
$\com{C}_2$ with probability $1-p$). Furthermore, each event of $\es{P}_1$ and $\es{P}_2$ will
causally depend on an event representing the invisible action $\tau$. In other words, the invisible
action $\tau$ should be the initial event. Regarding the valuations, if the configuration obtained
by removing $\tau$ belongs to $\confES{\es{P}_1}$, then we multiply by $p$ the valuation in
$\es{P}_1$. If the configuration obtained by removing $\tau$ belongs to $\confES{\es{P}_2}$, then we
multiply by $(1-p)$. Otherwise, the value of the valuation is $1$. Formally:
\begin{definition}[Probabilistic choice]\label{def:pes-prob2}
  Let $\es{P}_1 = \ppes{E_1}{\leq_1}{\#_1}{v_1}$ and $\es{P}_2 = \ppes{E_2}{\leq_2}{\#_2}{v_2}$ be
  probabilistic event structures.  Define $\probC{\es{P}_1}{p}{\es{P}_2} = \ppes{E}{\leq}{\#}{v}$
  as:
  \begin{align*}
    & E = \set{\tau} \uplus (E_1 \uplus E_2)  \\
    & \leq\ = \set{\tau \leq e \mid e \in E} \cup \leq_1 \uplus \leq_2 \\
    & \# = \#_1 \uplus \#_2 \cup \set{ e_1 \# e_2 \mid e_1 \in E_1,\, e_2 \in E_2}
                            \cup \set{ e_2 \# e_1 \mid e_1 \in E_1,\, e_2 \in E_2}\\
    & \forall x \in \mathcal{C}(\probC{\es{P}_1}{p}{\es{P}_2})\, .\, v(x) =
      \begin{cases}
        p \cdot v_1(x \backslash \tau) & \text{ if } x \backslash \tau \in \mathcal{C}(\es{P}_1) \\
        (1-p) \cdot v_2(x \backslash \tau) & \text{ if } x \backslash \tau \in \mathcal{C}(\es{P}_2) \\
        1 & \text{ if } x = \set{\tau} \vee x = \emptyset
      \end{cases}
  \end{align*}
\end{definition}

\begin{lemma}\label{lem:pes-prob2}
  Let $\es{P_1}$ and $\es{P_2}$ be probabilistic event structures. $\probC{\es{P}_1}{p}{\es{P}_2}$
  is a probabilistic event structure.
\end{lemma}

Note that, in the valuation of Definition~\ref{def:pes-prob2}), the configurations considered for
$\es{P}_1$ and $\es{P}_2$ are disjoint. This follows from the fact that, in
$\probC{\es{P}_1}{p}{\es{P}_2}$, every event in $\es{P}_1$ is in conflict with every event in
$\es{P}_2$.  Since configurations are conflict-free by definition, no configuration
$x \in \confES{\probC{\es{P}_1}{p}{\es{P}_2}}$ can contain events from both $\es{P}_1$ and
$\es{P}_2$.  Hence, the configurations $x \backslash \tau$ used in the valuation for $\es{P}_1$ and
$\es{P}_2$ are disjoint.

Furthermore, from Definition~\ref{def:pes-prob2}, it is also possible to notice that
$\probC{\es{P}}{p}{\es{P}} \neq \es{P}$. This is a consequence of introducing the $\tau$ event in
$\probC{\es{P}}{p}{\es{P}}$.

\begin{remark}
  Another way of representing $\probC{\es{P}_1}{p}{\es{P}_2}$ is by putting the probabilities
  explicit on both sides, \ie\ $p \cdot \es{P}_1 + (1-p) \cdot \es{P}_2$.  That leaves us with
  $\probC{\es{P}_1}{p}{\es{P}_2} = p \cdot \es{P}_1 + (1-p) \cdot \es{P}_2$
\end{remark}

\begin{remark}\label{rem:probc-n}
  Definition~\ref{def:pes-prob2} can be generalized to allow the probabilistic composition of $n$
  probabilistic event structures. Consider we have a finite number $n$ of probabilistic event
  structures. Let $1 \leq i \leq n$ and define $\sum_i p_i \cdot \es{P}_i$, with $\sum_i p_i = 1$,
  as follows:
  \begin{align*}
    & E = \set{\tau} \uplus \biguplus_i E_i  \\
    & \leq\ = \set{\tau \leq e \mid e \in E} \cup \biguplus_i \leq_i \\
    & \# = \biguplus_i \#_i \cup \set{ e_i \# e_j \mid e_i \in E_i,\, e_j \in E_j, 1 \leq i \neq j \leq n} \\
    & \hspace*{1.65cm} \cup \set{ e_j \# e_i \mid e_i \in E_i,\, e_j \in E_j, 1 \leq i \neq j \leq n} \\
    & \forall x \in \mathcal{C}\left(\sum_i p_i \cdot \es{P}_i\right)\, .\, v(x) =
      \begin{cases}
        p_i \cdot v_i(x \backslash \tau) & \text{ if } x \backslash \tau \in \mathcal{C}(\es{P}_i) \\
        1 & \text{ if } x = \set{\tau} \vee x = \emptyset
      \end{cases}
  \end{align*}

  This will be useful when showing the equivalence between the operational and the denotational
  semantics.
\end{remark}

For the parallel composition, and by considering the intuition that the parallel composition places
``side-by-side'' the two event structures, the valuation is determined by multiplying the valuations
obtained from projecting the configuration in $\conc{\es{P}_1}{\es{P}_2}$ onto the corresponding
configurations of $\es{P}_1$ and $\es{P}_2$.
\begin{definition}[Parallel Composition]\label{def:pes-conc2}
  Let $\es{P}_1 = \ppes{E_1}{\leq_1}{\#_1}{v_1}$ and $\es{P}_2 = \ppes{E_2}{\leq_2}{\#_2}{v_2}$ be
  probabilistic event structures.  Define $\conc{\es{P}_1}{\es{P}_2} = \ppes{E}{\leq}{\#}{v}$ as:
  \begin{align*}
    & E = E_1 \uplus E_2  \\
    & \leq\ = \leq_1 \uplus \leq_2 \\
    & \# = \#_1 \uplus \#_2 \\
    & \forall x \in \mathcal{C}(\conc{\es{P}_1}{\es{P}_2})\, .\,
      v(x) = v_1(x \cap E_1) \cdot v_2(x \cap E_2)
  \end{align*}
\end{definition}

\begin{lemma}\label{lem:conc-es2}
  Let $\es{P_1}$ and $\es{P_2}$ be probabilistic event structures.  $\conc{\es{P}_1}{\es{P}_2}$ is a
  probabilistic event structure.
\end{lemma}

We can now interpret the language shown in Section~\ref{subsec:lan-2} as probabilistic event
structures.
\begin{definition}\label{def:den-sem2}
  We interpret commands as probabilistic event structures as follows:
  \begin{align*}
    & \mf{\com{skip}} = (\set{sk}, \set{sk \leq sk}, \emptyset, v(\set{sk}) = 1) \\
    & \mf{\com{a}} = (\set{a}, \set{a \leq a}, \emptyset, v(\set{a}) = 1) \\
    & \mf{\probC{\com{C_1}}{p}{\com{C_2}}} = \probC{\mf{\com{C}_1}}{p}{\mf{\com{C}_2}} \\
    & \mf{\seq{\com{C_1}}{\com{C_2}}} = \seq{\mf{\com{C_1}}}{\mf{\com{C_2}}} \\
    & \mf{\conc{\com{C_1}}{\com{C_2}}} = \conc{\mf{\com{C_1}}}{\mf{\com{C_2}}} \\
  \end{align*}
\end{definition}

From Section~\ref{sec:es}, we can deduce that an extension of Definition~\ref{def:pes-fix-order-1}
to the probabilistic setting is unsuitable for relating the operational and the denotational
semantics. Hence, we extend Definition~\ref{def:pes-sub1} to the probabilistic realm.  Recall what
was said regarding plain and composite events in Section\ref{subsec:constructions-es-1} above
Definition~\ref{def:pes-sub1}.
\begin{definition}\label{def:pes-sub2}
  Let $\es{P}_1 = \ppes{E_1}{\leq_1}{\#_1}{v_1}$ and $\es{P}_2 = \ppes{E_2}{\leq_2}{\#_2}{v_2}$ be
  probabilistic event structures.  Say $\es{P}_1 \sqsubseteq \es{P}_2$ whenever exists an injective
  function $f: E_1 \rightarrow E_2$ such that $\forall e, e' \in E_1$:
  \begin{align*}
    & \pi(f(e)) = \pi(e) \\    
    & e \leq_1 e' \implies f(e) \leq_2 f(e') \\
    & e \#_1 e' \implies f(e) \#_2 f(e') \\
    & \forall x \in \mathcal{C}(\es{P}_1), y \in \mathcal{C}(\es{P}_2)\, .\,
      f[x] \subseteq y \Rightarrow v_1(x) \geq v_2(y) 
  \end{align*}
  We say that two probabilistic event structures $\es{P}_1, \es{P}_2$ are similar, denoted
  $\es{P}_1 \equiv \es{P}_2$, iff $\es{P}_1 \sqsubseteq \es{P}_2$ and
  $\es{P}_2 \sqsubseteq \es{P}_1$.
\end{definition}

To define $\sqsubseteq$ in the probabilistic setting, we based ourselves on the fact that, given two
configurations $x$ and $y$ such that $x \subseteq y$ the probability of $x$ must be greater or equal
to the probability of $y$, \ie\ $v(x) \geq v(y)$.  This comes from the intuition that the
probability tends to decrease as long as the computation continues. 

The valuation condition in Definition~\ref{def:pes-sub2} uses $v_1(x) \geq v_2(y)$ instead of
$v_1(x) \geq v_2(f[x])$, because $f[x]$ may not be a configuration in $\es{P}_2$.  For instance, let
$\es{P}_1 = a$ and $\es{P}_2 = b \rightarrowtriangle a$. It is straightforward to see that
$a \sqsubseteq b \rightarrowtriangle a$, via the identity function. However,
$id[\{a\}] \not\in \confES{\es{P}_2}$, since in $\es{P}_2$ the event $a$ causally depends on $b$,
and therefore cannot occur on its own.

Furthermore, note that in Definition~\ref{def:pes-sub2}, the conditions on causal and conflict
relations use an implication rather than an equivalence (recall that in the classical case, the
definition of sub-similar event structures used an equivalence -- Definition~\ref{def:pes-sub1}).
This adjustment is necessary because the interaction between probabilistic choice and parallel
composition can duplicate events, introducing new causal and conflict dependencies.  For example, in
Figure~\ref{fig:probES-lem-1}, the event $c$ is concurrent with all others, while in
Figure~\ref{fig:probES-lem-2} its copies $c'$ and $c''$ now causally depend on $\tau$ and are each
in conflict with events from the opposite branch (\ie\ $c'$ is in conflict with $b$ and $c''$ is in
conflict with $a$). Hence, if in Definition~\ref{def:pes-sub2} we kept the ``equivalence'', then
Lemma~\ref{lem:sound-aux-conc-2} would not hold and, consequently we would not be able to obtain
soundness neither adequacy.
\begin{figure}[ht!]
  \centering
  \begin{minipage}{0.49\textwidth}
    \begin{tikzpicture}[tikzfig]
      \begin{pgfonlayer}{nodelayer}
        \node [style=event] (0) at (8, 0) {$\com{c}_{\textcolor{blue}{1}}$};
        \node [style=event] (1) at (1, -2) {$\com{a}_{\textcolor{blue}{p}}$};
        \node [style=event] (2) at (5.5, -2) {$\com{b}_{\textcolor{blue}{(1-p)}}$};
        \node [style=event] (5) at (3.25, 0) {$\tau_{\textcolor{blue}{1}}$};
      \end{pgfonlayer}
      \begin{pgfonlayer}{edgelayer}
        \draw [style=wiggle] (1) to (2);
        \draw [style=arrow] (5) to (1);
        \draw [style=arrow] (5) to (2);
      \end{pgfonlayer}
    \end{tikzpicture}
  \end{minipage}
  \begin{minipage}{0.49\textwidth}
    \begin{align*}
      v(x) =
      \begin{cases}
        p & \text{ if } \com{a} \in x \\
        1-p & \text{ if } \com{b} \in x \\
        1 & \text{ otherwise}
      \end{cases}
    \end{align*}
  \end{minipage}
  \caption{Probabilistic event structure of
    $\conc{(\probC{\com{a}}{p}{\com{b}})}{\com{c}}$}
  \label{fig:probES-lem-1}
\end{figure}

\begin{figure}[ht!]
  \centering
  \begin{minipage}{0.49\textwidth}
    \begin{tikzpicture}[tikzfig, node distance=12mm and 18mm]
      \begin{pgfonlayer}{nodelayer}
        \node [style=event] (tau) {$\tau_{\textcolor{blue}{1}}$};
        \node [style=event, below left=of tau]  (a)  {$\com{a}_{\textcolor{blue}{p}}$};
        \node [style=event, below right=of tau] (b)  {$\com{b}_{\textcolor{blue}{(1-p)}}$};
        \node [style=event, below=of a]         (c') {$\com{c'}_{\textcolor{blue}{p}}$};
        \node [style=event, below=of b]         (c''){$\com{c''}_{\textcolor{blue}{(1-p)}}$};
      \end{pgfonlayer}
      \begin{pgfonlayer}{edgelayer}
        \draw [style=arrow] (tau) -- (a);
        \draw [style=arrow] (tau) -- (b);
        \draw [style=arrow] (tau) -- (c');
        \draw [style=arrow] (tau) -- (c'');
        \draw [style=wiggle] (a) -- (b);
        \draw [style=wiggle, bend right=12] (a) -- (c'');
        \draw [style=wiggle, bend left=12]  (c') -- (b);
        \draw [style=wiggle] (c') -- (c'');
      \end{pgfonlayer}
    \end{tikzpicture}
  \end{minipage}\hfill
  \begin{minipage}{0.49\textwidth}
    \begin{align*}
      v(x) =
      \begin{cases}
        p     & \text{if } \com{a} \text{ or } \com{c'} \in x,\\
        1-p   & \text{if } \com{b} \text{ or } \com{c''} \in x,\\
        1     & \text{otherwise.}
      \end{cases}
    \end{align*}
  \end{minipage}
  \caption{Probabilistic event structure of $\probC{(\conc{a}{c})}{p}{(\conc{b}{c})}$}
  \label{fig:probES-lem-2}
\end{figure}

To remove the initial event of a probabilistic event structure, we need to guarantee that the
probability of said event is not zero. Because if the event had probability zero, removing it would
lead to a division by zero.  Furthermore, this is the reason why $p \in ]0,1[$ in the probabilistic
operator.
\begin{definition}[Remove initial event]\label{def:rem-init2}
  Let $\es{P} = \ppes{E}{\leq}{\#}{v}$ be a probabilistic event structure and $a \in \init{\es{P}}$,
  s.t $v(\set{a}) \neq 0$.  Define $\es{P} \backslash a = (E', \leq', \#', v')$ as
  \begin{align*}
    & E' = \set{e \in E \mid \neg (e \# a),\, e \neq a} \\
    & \leq'\, =\, \set{e \leq e' \mid e,e' \in E' } \\
    & \#'\, =\, \set{e \# e' \mid e,e' \in E' } \\
    & \forall x \in \mathcal{C}(\es{P} \backslash a)\, .\, v'(x) = \dfrac{v(x \cup \set{a})}{v(\set{a})}
  \end{align*}
\end{definition}

\begin{lemma}\label{lem:rem-init-es2}
  Let $\es{P}$ be a probabilistic event structure.  $\es{P} \backslash a$ is a probabilistic event
  structure.
\end{lemma}

\subsection{Results}
In this section we present the results obtained in the probabilistic realm. For this, we interpret
the special command $\checkmark$ as the empty probabilistic event structure whose valuation of the
empty configuration is $1$, \ie\
$\mf{\checkmark} = (\emptyset, \emptyset, \emptyset, v(\emptyset) = 1)$.

We begin by showing that the sequential, parallel, and probabilistic operators are monotonic with
respect to Definition~\ref{def:pes-sub2}.
\begin{lemma}\label{lem:seq-mono2}
  Let $\es{P}_1, \es{P}'_1, \es{P}_2, \es{P}'_2$ be probabilistic event structures.  If
  $\es{P}_1 \sqsubseteq \es{P}'_1$ and $\es{P}_2 \sqsubseteq \es{P}'_2$ then
  $\seq{\es{P}_1}{\es{P}_2} \sqsubseteq \seq{\es{P}'_1}{\es{P}'_2}$.
\end{lemma}

\begin{lemma}\label{lem:probc-mono2}
  Let $\es{P}_1, \es{P}'_1, \es{P}_2, \es{P}'_2$ be probabilistic event structures.  If
  $\es{P}_1 \sqsubseteq \es{P}'_1$ and $\es{P}_2 \sqsubseteq \es{P}'_2$ then
  $\probC{\es{P}_1}{p}{\es{P}_2} \sqsubseteq \probC{\es{P}'_1}{p}{\es{P}'_2}$.
\end{lemma}

\begin{lemma}\label{lem:conc-mono2}
  Let $\es{P}_1, \es{P}'_1, \es{P}_2, \es{P}'_2$ be probabilistic event structures.  If
  $\es{P}_1 \sqsubseteq \es{P}'_1$ and $\es{P}_2 \sqsubseteq \es{P}'_2$ then
  $\conc{\es{P}_1}{\es{P}_2} \sqsubseteq \conc{\es{P}'_1}{\es{P}'_2}$.
\end{lemma}

We now show how probabilistic event structures composed by the sequential and concurrent operators
interact with the removal of initial events.
\begin{lemma}\label{lem:seq-rem-init2}  
  Let $\es{P}_1$ and $\es{P}_2$ be probabilistic event structures.  Consider
  $\seq{\es{P}_1}{\es{P}_2}$ such that $l \in \init{\seq{\es{P}_1}{\es{P}_2}}$. Then
  $(\seq{\es{P}_1}{\es{P}_2}) \backslash l \equiv \seq{(\es{P}_1\backslash l)}{\es{P}_2}$.
\end{lemma}

\begin{lemma}\label{lem:conc-rem-init2}  
  Let $\es{P}_1$ and $\es{P}_2$ be probabilistic event structures. Consider
  $\conc{\es{P}_1}{\es{P}_2}$ such that $l \in \init{\conc{\es{P}_1}{\es{P}_2}}$. Then
  \[
    (\conc{\es{P}_1}{\es{P}_2}) \backslash l \equiv
    \begin{cases}
      \conc{(\es{P}_1 \backslash l)}{\es{P}_2} & \text{ if } l \in \init{\es{P}_1} \\
      \conc{\es{P}_1}{(\es{P}_2 \backslash l)} & \text{ if } l \in \init{\es{P}_2}
    \end{cases}
  \]
\end{lemma}

It is straightforward to see that:
\begin{lemma}\label{lem:conc-symmetric2}
  Let $\es{P}_1, \es{P}_2$ be probabilistic event structures.  Then
  $\conc{\es{P}_1}{\es{P}_2} = \conc{\es{P}_2}{\es{P}_1}$.
\end{lemma}

From Definition~\ref{def:es-init-event}, we can infer that, in a probabilistic event structure
interpreting a command, any configuration consisting solely of an initial event has valuation
$1$. The following lemma confirms this intuition.
\begin{lemma}\label{lem:prob-init}
  Let $\com{C}$ be a command and $l \in \init{\mf{\com{C}}}$.  Then $v(\set{l}) = 1$.
\end{lemma}

We start by presenting two useful lemmas to prove small-soundness (Lemma~\ref{res:soundI-2}).  The
first lemma states that sequentially composing a sum of probabilistic event structures with another
probabilistic event structure is equal to first sequentially composing the probabilistic event
structures and then using the probabilistic operator.
\begin{lemma}\label{lem:sound-aux-seq-2}
  Let $\es{P}, \sum_i p_i \cdot \es{P}_i$ be probabilistic event structures.  Then
  $\seq{\left( \sum_i p_i \cdot \es{P}_i \right)}{\es{P}} =
  \sum_i p_i \cdot (\seq{\es{P}_i}{\es{P}})$
\end{lemma}

The second lemma expresses a certain ``lax'' linearity condition of concurrent composition.
\begin{lemma}\label{lem:sound-aux-conc-2}
  Let $\es{P}, \sum_i p_i \cdot \es{P}_i$ be probabilistic event structures.
  \begin{enumerate}
  \item
    $\conc{\left( \sum_i p_i \cdot \es{P}_i \right)}{\es{P}} \sqsubseteq
    \sum_i p_i \cdot (\conc{\es{P}_i}{\es{P}})$
  \item $x \in \confmax{\conc{\left( \sum_i p_i \cdot \es{P}_i \right)}{\es{P}}}$ iff
    $x \in \confmax{\sum_i p_i \cdot (\conc{\es{P}_i}{\es{P}})}$ 
  \end{enumerate}
\end{lemma}

The following lemma helps proving soundness (Theorem~\ref{res:soundII-2}) and adequacy
(Theorem~\ref{res:adII-2}).  Intuitively it states that if the invisible action $\tau$ triggers a
transition from $\com{C}$, which is sequentially or concurrently composed, to
$\sum_i p_i \com{C}_i$, then any maximal configuration of $\mf{\com{C}}$ is a maximal configuration
in $\sum_i p_i \mf{\com{C}_i}$.  Furthermore, there exists a command $\mf{\com{C}_i}$ in
$\sum_i p_i \mf{\com{C}_i}$ where the valuation of the maximal configuration matches the valuation
of the maximal configuration in $\mf{\com{C}}$.
\begin{lemma}\label{lem:sound-aux-seq-conc-2}
  Let $\com{C} = \seq{\com{C}_1}{\com{C}_2}$ or $\com{C} = \conc{\com{C}_1}{\com{C}_2}$.
  If $\com{C} \rightarrow \sum_i p_i \cdot (\tau, \com{C}_i)$ then $x \in \confmax{\mf{\com{C}}}$ and
  $x \in \confmax{\sum_i p_i \cdot \mf{\com{C}_i}}$ such that $\exists \mf{\com{C}_i}\, .\, v(x) = v_i(x)$.
\end{lemma}

The following lemma tells us that the big-step semantics always progresses, \ie\ the system always
has at least one valid transition to make.
\begin{lemma}~\label{lem:prog-2}
  For any $\com{C}$, exists $\sum_i p_i (\omega_i, \com{C}_i)$ such that
  $\com{C} \twoheadrightarrow \sum_i p_i (\omega_i, \com{C}_i)$.
\end{lemma}

We can now establish a relation between the small-step and denotational semantics. Recall that in
Figure~\ref{fig:op-small2}, transitions are labeled either by an atomic or invisible action. We can
use that to establish a relation between the small-step and denotational semantics. Intuitively, if
a transition from $\com{C}$ is triggered by an atomic action to $\com{C}'$, then we should observe
the same behavior as in Lemma~\ref{res:soundI-1}. Otherwise, if the invisible action triggers the
transition, we transition from $\com{C}$ to $\sum_i p_i \com{C}_i$. Having in mind
Lemma~\ref{lem:sound-aux-conc-2}, we expect $\mf{\com{C}}$ to be a sub-probabilistic event structure
of $\sum_i p_i \mf{\com{C}_i}$.
\begin{lemma}[Soundness I]\label{res:soundI-2}
  \begin{itemize}
  \item If $\com{C} \rightarrow 1 \cdot (l, \com{C}')$ then
    $\mf{\com{C'}} \equiv \mf{\com{C}} \backslash l$
  \item If $\com{C} \rightarrow \sum_i p_i (\tau, \com{C}_i)$ then
    $\mf{\com{C}} \sqsubseteq \sum_i p_i \mf{\com{C}_i}$
  \end{itemize}
\end{lemma}

At this point, we can prove soundness, establishing a relation between the big-step and the
denotational semantics. Recall that the big-step semantics of a command $\com{C}$ returns a
probability distribution of possible outcomes.  Each outcome is a pair, formed by a word $\omega$
(describing the sequence of events that occurred), and either a command $\com{C}'$, with
instructions yet to execute, or the special symbol $\checkmark$, to indicate termination, \ie\
\[
  \com{C} \twoheadrightarrow \sum_i p_i (\omega_i, \checkmark) + \sum_j p_j (\omega_j, \com{C}_j)
\]
where, unlike Figure~\ref{fig:op-nstep2}, we separate computations that terminate from those that do
not.  Having in mind our experience with soundness in the non-deterministic case
(Theorem~\ref{res:soundII-1}), if a word leads to a terminal computation then there exists a maximal
configuration that matches the word. Consequently, the probability associated with that terminal
command equals the valuation of the maximal configuration.
\begin{theorem}[Soundness II]\label{res:soundII-2}
  If $\com{C} \twoheadrightarrow p_0 (\omega_0, \checkmark) + \sum_k p_k (\omega_{k}, \com{C}_k)$
  then exists $x_0 \in \confmax{\mf{\com{C}}}$ such that
  $\emptyset \stackrel{\omega_{0}}{\chain\ } x_{0}$ and $p_{0} = v(x_{0})$.
\end{theorem}

We now center our attention to show adequacy. Similar to what was done previously, we check if a
relation exists between the denotational and small-step semantics. To do this, we will take
advantage of the initial events. If an initial event represents an atomic action, we should be in a
similar situation as in the small-adequacy lemma in the non-deterministic case
(Lemma~\ref{res:adI-1}). Otherwise, if the initial event represents an invisible action, we deal
with a transition triggered by a probabilistic command. In this case, we need to adapt the previous
situation to the probabilistic case, considering Lemma~\ref{lem:sound-aux-conc-2} and the similar
case in small-soundness (Lemma~\ref{res:soundI-2}).
\begin{lemma}[Adequacy I]\label{res:adI-2}
  Let $l' \in \init{\mf{\com{C}}}$.
  \begin{enumerate}
  \item If $l' \neq \tau$ then
    $\exists \com{C'} \in (\com{C} \cup \set{\checkmark})\, .\,
    \com{C} \rightarrow 1 \cdot (l', \com{C'})$ and $\mf{\com{C}} \backslash l' \equiv \mf{\com{C'}}$.
  \item If $l' = \tau$ then
    $\exists \com{C'}, \com{C''}\, .\, \exists e' \in \init{\mf{\com{C'}}}\, .\,
    \com{C} \rightarrow p \cdot \com{C'} + (1-p) \cdot \com{C''}$ and
    $\mf{\com{C}} \sqsubseteq p \cdot \mf{\com{C'}} + (1-p) \cdot \mf{\com{C''}}$, with
    $p = v(\set{\tau, e'})$.
  \end{enumerate}
\end{lemma}

Note that in Lemma~\ref{res:adI-2}, for the case where $l' = \tau$, we only consider the existence
of an initial event of a command $\com{C}'$. This is because to know the probability of $\com{C}''$
to occur, we only need to know the probability $p$ and subtract it from $1$.  Furthermore, the
result states $\mf{\com{C}} \sqsubseteq p \cdot \mf{\com{C'}} + (1-p) \cdot \mf{\com{C''}}$ instead
of $\mf{\com{C}} \equiv p \cdot \mf{\com{C'}} + (1-p) \cdot \mf{\com{C''}}$.  The use of
sub-similarity instead of similarity arises with a similar argument as the one given by changing
from ``equivalence'' to ``implication'' in the definition of sub-similar probabilistic event
structures Definition~\ref{def:pes-sub2}. For instance, consider the following small-step transition
$\conc{(\probC{a}{p}{b})}{c} \longrightarrow p \cdot (\tau, \conc{a}{c}) + (1-p) \cdot (\tau,
\conc{b}{c})$. Here, the action $c$ is duplicated on both branches of the probabilistic choice.  Now
recall Figure~\ref{fig:probES-lem-1}, corresponding to $\mf{\conc{(\probC{a}{p}{b})}{c}}$, and
Figure~\ref{fig:probES-lem-1}, corresponding to
$p \cdot \mf{\conc{a}{c}} + (1-p) \cdot \mf{\conc{b}{c}}$.  In the former figure, $c$ is concurrent
with every event, while in the latter figure, the copies of $c$, $c'$ and $c''$ causally depend on
$\tau$, $c'$ is in conflict with $b$, and $c''$ is in conflict with $a$.  These new relations are
not present in $\mf{\conc{(\probC{a}{p}{b})}{c}}$.  Because the duplication of events in the
probabilistic branches introduces new causal and conflict relations, the probabilistic event
structure obtained after the transition cannot be similar to the original one, it can only be
sub-similar. Thus, the usage of $\sqsubseteq$ in the second case of Lemma~\ref{res:adI-2}.

We are ready to demonstrate the equivalence between the denotational and big-step semantics. The
reasoning is similar to the one presented in Theorem~\ref{res:soundII-2}.
\begin{theorem}[Adequacy II]\label{res:adII-2}
  For all $x_0 \in \confmax{\mf{\com{C}}}$, if
  $\emptyset \stackrel{\omega_{x_0}}{\chain\ } x_0$ then we have
  $\com{C} \twoheadrightarrow v(x_0) (\omega_0, \checkmark) + \sum_k p_k (\omega_{k}, \com{C}_k)$,
  for some $\omega_k, p_k, \com{C}_k$.
\end{theorem}

In Lemma~\ref{res:soundI-2} and Lemma~\ref{res:adI-2} we see the usefulness of introducing the label
$\tau$. It helps us identifying the situations where a transition occurred due to the probabilistic
command and when it did not.

Theorem~\ref{res:soundII-2} assures us that whenever any execution of the program leads to a
terminal command, we have a maximal configuration who matches the word and the respective
probability. Theorem~\ref{res:adII-2} tells us that for every maximal configuration of a command
$\com{C}$ and for every covering chain of that configuration, there is an execution of the program
leading to a terminal command who matches the covering chain and the its respective probability.

\subsection{Introducing cyclic behavior}\label{subsec:cyclic-beh-2}
We now introduce cyclic behavior to the language in Section~\ref{subsec:lan-2}.  In order to avoid
the introduction of the notion of state in the language, the cyclic behavior will be given by
recursion.  In that way, we do not need to associate the notion of state to a command in the
operational semantics.  We can just keep recording the actions that are being made by the program.

Another thing to have in mind is that with cyclic behavior we open the door to infinite
computations.  However, covering chains are only defined in finite sequence of words and infinite
configurations are odd, because we would need to define precisely what it means to be an infinite
configuration.  Hence, the words that we formed with the n-step will be always finite, despite the
possibility of them being infinite.  We can justify this by saying that we are only concerned on the
``interesting words'', \ie\ those who are finite.

To introduce recursion we need to add some restrictions when forming programs, since we do not want
to allow commands like: $\rec{X}{\seq{X}{\com{a}}}$ and $\rec{X}{\seq{\seq{\com{a}}{X}}{\com{b}}}$.

Let $X \subseteq Var$, with $Var$ a set of variables.  The syntax is now given by:
\[
  \com{C} ::= \com{skip} \mid a \in Act \mid \seq{\com{C}}{\com{C}} \mid \probC{\com{C}}{p}{\com{C}} \mid
  \conc{\com{C}}{\com{C}} \mid \rec{X}{\com{C}} \mid X
\]

We define the set of free-variables and bound-variables as follows:
\begin{table}[h!]
  \centering
  \begin{tabular}{l|l}
    $\fvar{\com{skip}} = \emptyset$ & $\bvar{\com{skip}} = \emptyset$ \\
    $\fvar{\com{a}} = \emptyset$ & $\bvar{\com{a}} = \emptyset$ \\
    $\fvar{\seq{\com{C}_1}{\com{C}_2}} =
    \fvar{\com{C}_1} \cup \fvar{\com{C}_2}$ & $\bvar{\seq{\com{C}_1}{\com{C}_2}} =
                                              \bvar{\com{C}_1} \cup \bvar{\com{C}_2}$ \\
    $\fvar{\conc{\com{C}_1}{\com{C}_2}} =
    \fvar{\com{C}_1} \cup \fvar{\com{C}_2}$ & $\bvar{\conc{\com{C}_1}{\com{C}_2}} =
                                              \bvar{\com{C}_1} \cup \bvar{\com{C}_2}$ \\
    $\fvar{\probC{\com{C}_1}{p}{\com{C}_2}} =
    \fvar{\com{C}_1} \cup \fvar{\com{C}_2}$ & $\bvar{\probC{\com{C}_1}{p}{\com{C}_2}} =
                                              \bvar{\com{C}_1} \cup \bvar{\com{C}_2}$ \\
    $\fvar{X} = \set{X}$ & $\bvar{X} = \emptyset$ \\
    $\fvar{\rec{X}{\com{C}}} = \fvar{C} \backslash \set{X}$ & $\bvar{\rec{X}{\com{C}}} = \set{X} \cup \bvar{\com{C}}$
  \end{tabular}
\end{table}

We restrict the sequential composition to those whose free-variables and bound-variables on the left
are empty, \ie\ $\seq{\com{C}_1}{\com{C}_2}$ if $\fvar{\com{C}_1} = \emptyset = \bvar{\com{C}_1}$.
With this restriction we forbid program like $\rec{X}{\seq{X}{\com{a}}}$,
$\rec{X}{\seq{\seq{\com{a}}{X}}{\com{b}}}$ (with the condition $\fvar{\com{C}_1} = \emptyset$) and
$\seq{(\rec{X}{\seq{\com{a}}{X}})}{\com{b}}$ (with the condition $\bvar{\com{C}_1} = \emptyset$).
We want to forbid these kind of programs in sequential composition, because if $\com{C}_1$ never
terminates then the sequential composition never terminates. This is also a restriction that comes
from the fact that covering chains are only defined in finite sequences and that infinite
configurations are odd in event structures.  Note however that we allow programs like
$\rec{X}{\conc{X}{a}}$ and $\rec{X}{\nd{X}{a}}$, since they do not block the computation.

Inspired by~\cite{hindley08}, we define substitution as follows:
\begin{definition}\label{fig:substitution-2}
  Let $X \in Var$ and $\com{C}, \com{C}'$ be commands. Define $\com{C}[X \leftarrow \com{C}']$,
  where we substitute every free occurrence of $X$ in $\com{C}$ by $\com{C}'$ (while changing bound
  variables to avoid clashes) by induction on $\com{C}$ as follows:
    \begin{align*}
    & \com{skip}[X \leftarrow \com{C'}] = \com{skip} \\
    & \com{a}[X \leftarrow \com{C'}] = \com{a} \\
    & (\seq{\com{C}_1}{\com{C}_2})[X \leftarrow \com{C'}] =
      \seq{\com{C}_1}{(\com{C}_2 [X \leftarrow \com{C'}])} \\
    & (\conc{\com{C}_1}{\com{C}_2})[X \leftarrow \com{C'}] =
      \conc{\com{C}_1[X \leftarrow \com{C'}]}{\com{C}_2 [X \leftarrow \com{C'}]} \\
    & (\probC{\com{C}_1}{p}{\com{C}_2})[X \leftarrow \com{C'}] =
      \probC{\com{C}_1[X \leftarrow \com{C'}]}{p}{\com{C}_2 [X \leftarrow \com{C'}]} \\
    & X [X \leftarrow \com{C}] = \com{C} \\
    & (\rec{X}{\com{C}})[X \leftarrow \com{C'}] =
      \rec{X}{\com{C}} \\
    & (\rec{Y}{\com{C}})[X \leftarrow \com{C'}] =
      \rec{Y}{\com{C}[X \leftarrow \com{C'}]} \text{ if } X \neq Y \text{ and } Y \not\in FV(\com{C'})
  \end{align*}
\end{definition}

We add to Figure~\ref{fig:op-small2} the following rules for the recursion command:
\begin{align*}
  \infer{\rec{X}{\com{C}} \rightarrow 1 \cdot (l, \com{C}'[X \leftarrow \rec{X}{
  \com{C}}])}{\com{C} \rightarrow 1 \cdot (l, \com{C}')
  }
  \hspace*{1cm}
  \infer{\rec{X}{\com{C}} \rightarrow \sum_i p_i \cdot (\tau, \com{C}_i[X \leftarrow \rec{X}{\com{C}}])}{
  \com{C} \rightarrow \sum_i p_i \cdot (\tau, \com{C}_i)
  }
\end{align*}
Similarly to what was done in Figure~\ref{fig:op-small2}, here we also distinguish if the action
that triggered the transition was an atomic or invisible action. Consequently, we define a rule for
each case. Both rules follow the same reasoning as the one explained in the non-deterministic
case. However, in the case where the transition is triggered by $\tau$, we need to substitute the
occurrence of $X$ by $\rec{X}{\com{C}}$ in each $\com{C}_i$ that is in
$\sum_i p_i \cdot (\tau, \com{C}_i)$.

Now we develop an illustrative example of recursion in the probabilistic setting.

\begin{example}\label{ex:loop-2}  
  Figure~\ref{fig:ex-loop-2} illustrates a probabilistic coin toss scenario where each time we toss
  the coin, it executes with probability $p$ the command $\com{skip}$ or continues the tossing with
  probability $1-p$.  To understand this behavior, focus on the initial command.  From there, we
  transit to a distribution formed by the commands $\com{skip}$ and
  $\rec{X}{(\probC{X}{p}{\com{skip}})}$, which is the same as the initial command.  From this
  distribution we transit to $\com{skip}$ with probability $p$ or to
  $\rec{X}{(\probC{X}{p}{\com{skip}})}$ with probability $1-p$, enabling us to repeat the process.
  \begin{figure}[ht!]
    \centering
    \begin{tikzpicture}
      \begin{pgfonlayer}{nodelayer}
        \node [style=command] (0) at (0, 0) {$\rec{X}{(\probC{\com{skip}}{p}{X})}$};
        \node [style=command] (1) at (-1, -2) {$\com{skip}$};
        \node [style=command] (3) at (0, -1) {$\bullet$};
        \node [style=none] (4) at (0.25, -0.5) {$\tau$};
        \node [style=none] (5) at (-0.75, -1.5) {$p$};
        \node [style=none] (6) at (1, -1.5) {$1-p$};
        \node [style=command] (7) at (1, -2) {$\rec{X}{(\probC{\com{skip}}{p}{X})}$};
        \node [style=command] (8) at (0, -4) {$\com{skip}$};
        \node [style=command] (9) at (1, -3) {$\bullet$};
        \node [style=none] (10) at (1.25, -2.5) {$\tau$};
        \node [style=none] (11) at (0.25, -3.5) {$p$};
        \node [style=none] (12) at (2, -3.5) {$1-p$};
        \node [style=command] (13) at (2, -4) {$\rec{X}{(\probC{\com{skip}}{p}{X})}$};
        \node [style=command] (14) at (1, -6) {$\com{skip}$};
        \node [style=command] (15) at (2, -5) {$\bullet$};
        \node [style=none] (16) at (2.25, -4.5) {$\tau$};
        \node [style=none] (17) at (1.25, -5.5) {$p$};
        \node [style=none] (18) at (3, -5.5) {$1-p$};
        \node [style=command] (19) at (3, -6) {$\ddots$};
      \end{pgfonlayer}
      \begin{pgfonlayer}{edgelayer}
        \draw [style=op] (0) to (3);
        \draw [style=prob] (3) to (1);
        \draw [style=op] (7) to (9);
        \draw [style=prob] (9) to (8);
        \draw [style=prob] (3) to (7);
        \draw [style=op] (13) to (15);
        \draw [style=prob] (15) to (14);
        \draw [style=prob] (9) to (13);
        \draw [style=prob] (15) to (19);
      \end{pgfonlayer}
    \end{tikzpicture}
    \caption{Fragment of the execution of $\rec{X}{(\probC{\com{skip}}{p}{X})}$}
    \label{fig:ex-loop-2}
  \end{figure}
\end{example}

On the event structure side, we want to use the Knaster-Tarski Theorem to build the least-fix point.
To define it, we will  use an order that does not ignore copies,  differently from what happens with
Definition~\ref{def:pes-sub2}.
\begin{definition}\label{def:pes-fix-order-2}
  Let $\es{P}_1 = \ppes{E_1}{\leq_1}{\#_1}{v_1}$ and $\es{P}_2 = \ppes{E_2}{\leq_2}{\#_2}{v_2}$ be
  probabilistic event structures.  Say $\es{P}_1 \trianglelefteq \es{P}_2$ if:
  \begin{align*}
    & E_1 \subseteq E_2 \\
    & \forall e,e'\ .\ e \leq_1 e'
      \Leftrightarrow
      e, e' \in E_1 \wedge e \leq_2 e' \\
    & \forall e,e'\ .\ e \#_1 e'
      \Leftrightarrow
      e, e' \in E_1 \wedge e \#_2 e' \\
    & \forall x \in \confES{P}_1\, .\, v_1(x) = v_2(x)
  \end{align*}
\end{definition}

The following lemmas confirm that Definition~\ref{def:pes-fix-order-2} is a partial order and that
it possesses a least element.
\begin{lemma}\label{lem:po-2}
  $\trianglelefteq$ is a partial order.
\end{lemma}

\begin{lemma}\label{lem:po-least-elem-2}
  Define $\bot = \ppes{\emptyset}{\emptyset}{\emptyset}{v_{\bot}(\emptyset) = 1}$.  $\bot$ is the least
  element of $\trianglelefteq$.
\end{lemma}

We extend Definition~\ref{def:lub-1} to the probabilistic realm, \ie\ we define what it means for a
probabilistic event structure to be a least upper bound in a chain of probabilistic event
structures.
\begin{definition}\label{def:lub-2}
  Let $\es{P}_1 \trianglelefteq \dots \trianglelefteq \es{P}_n \trianglelefteq \dots$ be a
  $\omega$-chain. Let $\es{P}^{\omega} = \ppes{E^\omega}{\leq^\omega}{\#^\omega}{v^\omega}$ be its
  least upper bound where:
  \begin{itemize}
  \item $E^\omega = \cup_{n \in \omega} E_n$
  \item $\leq^\omega = \cup_{n \in \omega} \leq_n$
  \item $\#^\omega = \cup_{n \in \omega} \#_n$
  \item $\forall x \in \confES{\es{P}^\omega}\, ,\,
    \exists n \in \omega\, .\, x
    \in \confES{\es{P}_n}\, .\, v^\omega(x) = v_n(x)$
  \end{itemize}
\end{definition}

We then show that the structure in Definition~\ref{def:lub-2} is in fact a probabilistic event
structure and a least upper bound in a chain of probabilistic event structures.
\begin{lemma}\label{lem:lub-es-2}
  $\es{P}^\omega$ is a probabilistic event structure.
\end{lemma}

\begin{lemma}\label{lem:lub-2}
  Let $\es{P}_1 \trianglelefteq \dots \trianglelefteq \es{P}_n \trianglelefteq \dots$ be a
  $\omega$-chain. Then $\es{P}^\omega$ is its least upper bound.
\end{lemma}

We have that sequential and concurrent composition, as well as probabilistic choice are monotone and
continuous with respect to Definition~\ref{def:pes-fix-order-2}. Recall that monotonicity and
continuity definitions are given by Definition~\ref{def:cont-1}.
\begin{lemma}\label{lem:seq-fix-mono-2}
  Let $\es{P}, \es{P}_1, \es{P}_2$ be probabilistic event structures.  If
  $\es{P}_1 \trianglelefteq \es{P}_2$ then
  $\seq{\es{P}}{\es{P}_1} \trianglelefteq \seq{\es{P}}{\es{P}_2}$.
\end{lemma}

\begin{lemma}\label{lem:conc-fix-mono-2}
  Let $\es{P}_1, \es{P}'_1, \es{P}_2, \es{P}'_2$ be probabilistic event structures.  If
  $\es{P}_1 \trianglelefteq \es{P}'_1$ and $\es{P}_2 \trianglelefteq \es{P}'_2$ then
  $\conc{\es{P}_1}{\es{P}_2} \trianglelefteq \conc{\es{P}'_1}{\es{P}'_2}$.
\end{lemma}

\begin{lemma}\label{lem:probc-fix-mono-2}
  Let $\es{P}_1, \es{P}'_1, \es{P}_2, \es{P}'_2$ be probabilistic event structures.  If
  $\es{P}_1 \trianglelefteq \es{P}'_1$ and $\es{P}_2 \trianglelefteq \es{P}'_2$ then
  $\probC{\es{P}_1}{p}{\es{P}_2} \trianglelefteq \probC{\es{P}'_1}{p}{\es{P}'_2}$.
\end{lemma}

\begin{lemma}\label{lem:seq-cont-2}
  $\bigsqcup_m(\seq{\es{P}}{\es{P}_m}) = \seq{\es{P}}{\bigsqcup_m \es{P}_m}$.
\end{lemma}

\begin{lemma}\label{lem:conc-cont-2}
  $\bigsqcup_{n,m}(\conc{\es{P}_n}{\es{P}_m}) = \conc{\bigsqcup_n \es{P}_n}{\bigsqcup_m \es{P}_m}$.
\end{lemma}

\begin{lemma}\label{lem:probc-cont-2}
  $\bigsqcup_{n,m}(\probC{\es{P}_n}{p}{\es{P}_m}) =
  \probC{\bigsqcup_n \es{P}_n}{p}{\bigsqcup_m \es{P}_m}$.
\end{lemma}

When adapting Lemma~\ref{lem:cont-1} (helps to show that operators defined in event structures are
continuous) and Lemma~\ref{lem:fix-prop-1} (a version of the Kleene fixed-point theorem for the case
of event structures) to the probabilistic realm, we notice that the proofs we need to do are
analogous to the ones done in the non-deterministic case. Hence, we omit their formulation in this
section.

We now extend Definition~\ref{def:den-sem2} with the recursion and let $\mathbb{P}$ denote the class
of probabilistic event structures.
\begin{definition}\label{def:fix-den-sem-2}
  Define an environment to be a function $\gamma : Var \rightarrow \mathbb{P}$ from variables to
  probabilistic event structures. For a command $\com{C}$ and an environment $\gamma$ define
  $\mf{\com{C}}_\gamma$ as follows:
  \begin{align*}
    & \mf{\com{skip}}_\gamma = (\set{sk}, \set{sk \leq sk}, \emptyset, v(\set{sk}=1)) \\
    & \mf{\com{a}}_\gamma = (\set{a}, \set{a \leq a}, \emptyset, v(\set{a}=1)) \\
    & \mf{\seq{\com{C_1}}{\com{C_2}}}_\gamma = \seq{\mf{\com{C_1}}_\gamma}{\mf{\com{C_2}}_\gamma} \\
    & \mf{\probC{\com{C_1}}{p}{\com{C_2}}}_\gamma = \probC{\mf{\com{C_1}}_\gamma}{p}{\mf{\com{C_2}}_\gamma} \\
    & \mf{\conc{\com{C_1}}{\com{C_2}}}_\gamma = \conc{\mf{\com{C_1}}_\gamma}{\mf{\com{C_2}}_\gamma} \\
    & \mf{X}_\gamma = \gamma(X) \\
    & \mf{\rec{X}{\com{C}}}_\gamma = fix(\Gamma^{\com{C}, \gamma})
  \end{align*}
  where $\Gamma^{\com{C}, \gamma} : \es{P} \rightarrow \es{P}$ is given by
  $\Gamma^{\com{C}, \gamma}(\es{P}) = \mf{\com{C}}_{\gamma(X \leftarrow \es{P})}$.
\end{definition}

We now show that $\Gamma^{\com{C}, \gamma}$ is continuous. 
\begin{lemma}\label{lem:gamma-cont-2}
  $\Gamma^{\com{C}, \gamma}$ is continuous.
\end{lemma}

To prove a lemma that is helpful for showing soundness (Lemma~\ref{res:soundI-fix-2}), we need the
probabilistic versions of the substitution lemma (Lemma~\ref{lem:1-1}) and
Lemma~\ref{lem:2-1}. Fortunately, the formulation of these probabilistic versions is the same as the
original ones. However, the proof of the substitution lemma needs to adjusted to accommodate the
probabilistic command. Hence, we show here its formulation (Lemma~\ref{lem:1-2}).
\begin{lemma}\label{lem:1-2}
  $\mf{\com{C}'[X \leftarrow \mf{\rec{X}{\com{C}}}_\gamma]}_\gamma =
  \mf{\com{C}'}_{\gamma{(X \leftarrow \mf{\rec{X}{\com{C}}}_\gamma})}$
\end{lemma}

To prove Theorem~\ref{res:soundII-fix-2}, we found it helpful to have the following lemma, whose
reasoning is similar to that of Lemma~\ref{lem:sound-aux-seq-conc-2}.
\begin{lemma}\label{lem:fix-conf-max-2}
  If
  $\rec{X}{\com{C}} \rightarrow \sum_i p_i \cdot (\tau, \com{C}_i[X \leftarrow \rec{X}{\com{C}}])$
  then $x \in \confmax{\mf{\rec{X}{\com{C}}}_\gamma}$ and
  $x \in \confmax{\sum_i p_i \cdot \mf{\com{C}_i[X \leftarrow \rec{X}{\com{C}}]}_\gamma}$ such that
  $\exists \mf{\com{C}_i}_\gamma\, .\, v(x) = v_i(x)$.
\end{lemma}

With the introduction of recursion, the formulation of the equivalence between operational and
denotational semantics is as in the case without recursion. We proceed to show the formulation of
the respective lemmas and theorems, since the proof needs to be completed to contemplate the
recursion case.
\begin{lemma}[Soundness I]\label{res:soundI-fix-2}
  \begin{itemize}
  \item If $\com{C} \rightarrow 1 \cdot (l, \com{C}')$ then
    $\mf{\com{C'}}_\gamma \equiv \mf{\com{C}}_\gamma \backslash l$
  \item If $\com{C} \rightarrow \sum_i p_i (\tau, \com{C}_i)$ then
    $\mf{\com{C}}_\gamma \sqsubseteq \sum_i p_i \mf{\com{C}_i}_\gamma$
  \end{itemize}
\end{lemma}

\begin{theorem}[Soundness II]\label{res:soundII-fix-2}
  If $\com{C} \twoheadrightarrow p_0 (\omega_0, \checkmark) + \sum_k p_k (\omega_{k}, \com{C}_k)$
  then exists $x_0 \in \confmax{\mf{\com{C}}_\gamma}$ such that
  $\emptyset \stackrel{\omega_{0}}{\chain\ } x_{0}$ and $p_{0} = v(x_{0})$.  
\end{theorem}


\begin{lemma}[Adequacy I]\label{res:adI-fix-2}
  Let $l' \in \init{\mf{\com{C}}_\gamma}$.
  \begin{enumerate}
  \item If $l' \neq \tau$ then
    $\exists \com{C'} \in (\com{C} \cup \set{\checkmark})\, .\,
    C \rightarrow 1 \cdot (l', \com{C'})$ and
    $\mf{\com{C}}_\gamma \backslash l' \equiv \mf{\com{C'}}_\gamma$.
  \item If $l' = \tau$ then
    $\exists \com{C'}, \com{C''}\, .\, \exists e' \in \init{\mf{\com{C'}}_\gamma}\, .\,
    \com{C} \rightarrow p \cdot \com{C'} + (1-p) \cdot \com{C''}$ and
    $\mf{\com{C}}_\gamma \sqsubseteq p \cdot \mf{\com{C'}}_\gamma + (1-p) \cdot \mf{\com{C''}}_\gamma$, with
    $p = v(\set{\tau, e'})$.
  \end{enumerate}
\end{lemma}

\begin{theorem}[Adequacy II]\label{res:adII-fix-2}
  For all $x_0 \in \confmax{\mf{\com{C}}_\gamma}$, if
  $\emptyset \stackrel{\omega_{x_0}}{\chain\ } x_0$ then we have
  $\com{C} \twoheadrightarrow v(x_0) (\omega_0, \checkmark) + \sum_k p_k (\omega_{k}, \com{C}_k)$,
  for some $\omega_k, p_k, \com{C}_k$.
\end{theorem}


The following example shows soundness and adequacy in practice, for commands without recursive
behavior.
\begin{example}\label{ex:eq-sem-2}
  The probabilistic event structure in Example~\ref{ex:prob-es-2} corresponds to the command in
  Example~\ref{ex:small-2}.

  To see how both semantics relate with each other, recall the maximal configurations in
  Example~\ref{ex:prob-es-2} and the words that lead to the end of a computation in
  Example~\ref{ex:small-2}.

  Similarly to what was shown in Example~\ref{ex:eq-sem-1}, it is straightforward to see that each
  word corresponds to a covering chain and vice-versa.  What is left to verify is the probability.
  From Example~\ref{ex:small-2} we know that the word $\tau a b$ has probability $p$, which is the
  same probability of the corresponding covering chain.  Similarly, the words $\tau c d$ and
  $\tau d c$ have probability $1-p$, which equals the probability of the respective covering chains.

  Conversely, if we pick a covering chain of a maximal configuration, we quickly notice that its
  probability and the probability of the respective word is the same.
\end{example}

\newpage
\section{Unitary Event Structures}\label{sec:qes}
A quantum event
structure~\cite[Definition~2]{winskel14} is an event structure together with a function that maps
events to unitary operators or projections on a finite-dimensional Hilbert space
$\mathcal{H}$. Additionally, operators of concurrent events must commute in quantum event structures
and the latter support the notion of an initial state considered to be a density operator.

The definition of quantum event structures~\cite[Definition~2]{winskel14} does not mention
probabilities, even though probabilities arise as a result of measuring quantum states. So, how do
probabilities arise in quantum event structures? The key lies in the combination of the initial
state together with quantum operators associated with the events in a configuration (recall that a
configuration captures the notion of a computation in event structures). By applying the sequence of
operators associated with the events in a configuration to the initial state, we evolve the state
along that computational path. This process allows us to extract the probabilities related to
measurement outcomes. Winskel formalizes this idea~\cite[Theorem~3]{winskel14}. He shows that a
quantum event structure without conflicting events -- also called an elementary quantum event
structure -- can be transformed into a corresponding probabilistic event structure.

To drop the elementary condition imposed by Winskel in~\cite[Theorem~3]{winskel14}, we add two new
conditions to Winskel's definition of quantum event structures~\cite [Definition~2]{winskel14} (in
Section~\ref{subsec:quantum-to-prob} we show how these additional constraints allow us to remove the
elementary condition). First, we impose the minimal conflict to be transitive.  Secondly, the sum of
the operators of events in minimal conflict, or an event itself, should be a unitary operator. The
intuition behind the restrictions is to consider events in minimal conflicts as measurements.  We
call the resultant structure \emph{unitary event structures}.  Why do we add these restrictions?
The reason is to ensure that Equation~\ref{eq:prob-es-condition} holds when considering unitary
event structures, jointly with an initial state, as probabilistic event structures. Recall from
Equation~\ref{eq:prob-es-condition} that conflicting events have a conditional probability no
greater than one. Now, let us consider two events in conflict, $e$ and $e'$, such that
$Q(e) = H = Q(e')$.  Taking $\ket{0}$ as the initial state, the configuration-valuations of $\{e\}$
and $\{e'\}$ are $v(\{e\}) = 1 = v(\{e'\})$. It is straightforward to check that
Equation~\ref{eq:prob-es-condition} does not hold for $\emptyset \subseteq \{e\}, \{e'\}$:
\begin{align*}
  v(\emptyset) - v(\{e\}) - v(\{e'\}) \geq 0
  \Leftrightarrow
  1 - 1 - 1 \geq 0
  \Leftrightarrow
  -1 \geq 0
\end{align*}
Now, consider $Q(e) = P_0$ and $Q(e') = P_1$, \ie\ projections onto $\ket{0}$ and $\ket{1}$,
respectively. For any initial state $\rho$, their configuration-valuations are given by
$v(\{e\}) = p_0^{\rho}$ and $v(\{e'\}) = p_1^{\rho}$, \ie\ the probabilities of obtaining $\ket{0}$
or $\ket{1}$ when measuring $\rho$. Note that $p_0^{\rho} + p_1^{\rho} = 1$.  Now
Equation~\ref{eq:prob-es-condition} holds for $\emptyset \subseteq \{e\}, \{e'\}$:
\begin{align*}
  v(\emptyset) - v(\{e\}) - v(\{e'\}) \geq 0
  \Leftrightarrow
  1 - p_0^{\rho} - p_1^{\rho} \geq 0
  \Leftrightarrow
  1 - 1 \geq 0
  \Leftrightarrow
  0 \geq 0
\end{align*}
However, by being complacent with the satisfaction of Equation~\ref{eq:prob-es-condition} we lose
the ability to interpret non-deterministic choice in unitary event structures. For instance, the
program $\nd{\nd{H}{X}}{Z}$ would be represented as $e_H \minconflict e_X \minconflict e_Z$. Since
minimal conflict is not transitive, this fails to form a unitary event structure.

To define unitary event structures, we use the equivalence class of an event $e$, which is composed
of itself and by the events in which $e$ is in minimal conflict, \ie\
$ [e] = \set{ e' \mid e = e',\ e \minconflict e'}$.
\begin{definition}[Unitary Event Structure]\label{def:qes}
  A unitary event structure over a finite-dimensional Hilbert space $\mathcal{H}$, is a pair
  $(\es{E},\ Q : E \rightarrow Op(\mathcal{H}))$ comprised of an event structure
  $\es{E} = \pes{E}{\leq}{\#}$, where $Q$ maps events $e \in E$ to either projections or unitary
  operators on $\mathcal{H}$ such that:
  \begin{itemize}
  \item $\forall e_1, e_2 \in E,\, \econc{e_1}{e_2} \Rightarrow Q(e_1)Q(e_2) = Q(e_2)Q(e_1)$
  \item $\minconflict$ is transitive
  \item $\forall e \in E,\ \sum_{e' \in [e]} Q(e')$ is unitary
  \end{itemize}
  Given a finite configuration, $x \in \confES{\es{E}}$, define the operator $A_x$ as the
  composition $Q_{e_n}Q_{e_{n-1}} \dots Q_{e_2}Q_{e_1}$ for some covering chain
  $ \emptyset \stackrel{e_1}{\chain}\ x_1 \stackrel{e_2}{\chain}\ x_2 \dots \stackrel{e_n}{\chain}\
  x_n$ in $\confES{\es{E}}$, with $x_n = x$.  We additionally set
  $A_{\emptyset} = Id : \mathcal{H} \to \mathcal{H}$ for the initial configuration.  An
  \textit{initial state} is given by a density operator $\rho$ on $\mathcal{H}$.
\end{definition}

We now show that the operator $A_x$ in Definition~\ref{def:qes} is well-defined. For that, we need
an auxiliary definition.  Let $\alpha = \dots, e, e', \dots$ be a finite event sequence. Denote
$Q(\alpha)$ be the ordered product of operators along $\alpha$, accordingly to
Definition~\ref{def:qes}, \ie\ $Q(\alpha) = \dots Q(e') Q(e) \dots$ Furthermore, we also need a
lemma stating that the operators for any two sequence of events that differ from swapping adjacent
concurrent events are the same.
\begin{lemma}\label{lem:Ax-well-defined}
  Let $\alpha = \dots, e, e', \dots$ and $\beta = \dots, e', e, \dots$ be two finite event
  sequences that differ only by swapping adjacent concurrent events $\econc{e}{e'}$. Then
  $\dots Q(e') Q(e) \dots = \dots Q(e) Q(e') \dots$.
\end{lemma}
\begin{proof}
  Let $\alpha = \dots, e, e', \dots$ and $\beta = \dots, e', e, \dots$ be two finite event sequences
  such that $\econc{e}{e'}$.  Let $Q(\alpha) = \dots Q(e') Q(e) \dots$ and
  $Q(\beta) = \dots Q(e) Q(e') \dots$. Since $\econc{e}{e'}$ then $Q(e') Q(e) = Q(e) Q(e')$. Thus
  \[
    Q(\alpha) = \dots Q(e') Q(e) \dots= \dots Q(e) Q(e') \dots = Q(\beta)
  \]
\end{proof}

As stated by Winskel~\cite{winskel14}, any two covering chains of a configuration $x$ are
Mazurkiewicz trace equivalent, \ie\ obtainable, one from the other, by successively interchanging
concurrent events.  This fact together with Lemma~\ref{lem:Ax-well-defined} are crucial for showing
that $A_x$ is well-defined.
\begin{proposition}\label{prop:Ax-well-defined}
  Let $x \in \confES{\es{E}}$. Then $A_x$ is well-defined, \ie\ $A_x$ is independent from the chosen
  covering chain of $x$.
\end{proposition}
\begin{proof}
  Let $x \in \confES{\es{E}}$ and $\alpha, \beta$ be two covering chains of $x$. From
  Winskel~\cite{winskel14}, we know that any two covering chains of a configuration $x$ are
  obtainable, one from another, by successively swapping concurrent events (since covering chains
  are by definition finite, then the interchanges performed are also finite). By applying
  Lemma~\ref{lem:Ax-well-defined} at each swap, the associated operator is preserved at each
  step. Hence $Q(\alpha) = Q(\beta)$, and therefore $A_x$ is independent of the chosen covering
  chain.  The base case, $A_{\emptyset} = Id$ holds by definition.
\end{proof}

Example~\ref{ex:qes-1} is designed for the reader to get used to unitary event structures.
\begin{example}\label{ex:qes-1}
  Figure~\ref{fig:ex1-3}, depicts a unitary event structure composed of the events $a$,
  $\tau_0$, $\tau_1$, $d$, and $b$, where the subscript is the respective associated quantum
  operator. $a$ is the initial event, followed by $\tau_0$, which leads to $b$, and $\tau_1$, which
  leads to $d$.  Note that $\tau_0$ and $\tau_1$ are in conflict (specifically, minimal conflict)
  and that their associated quantum operations form a measurement. Furthermore, since the conflict
  relation is hereditary, $d$ and $b$ are in conflict.

  The set of configurations is
  $\set{\emptyset, \set{a}, \set{a, \tau_0}, \set{a, \tau_1}, \set{a, \tau_0, b}, \set{a, \tau_1,
      d}}$.  As said in Definition~\ref{def:qes}, from a configuration $x$ we have as operator
  $A_x$. For example, if we consider the maximal configurations $\set{a, \tau_0, b}$ and
  $\set{a, \tau_1, d}$, the respective operators are $Z P_0 H$ and $X P_1 H$. The former applies the
  Hadamard gate to qubit $1$, projects it to $\ket{0}$, and then applies the $Z$ gate. After
  applying the Hadamard gate, the latter projects the qubit to $\ket{1}$ and then applies the $X$
  gate.
  \begin{figure}[ht!]
    \centering
    \begin{minipage}{0.3\textwidth}
      \begin{tikzpicture}
        \begin{pgfonlayer}{nodelayer}
          \node [style=event] (0) at (0, 0) {$a_{\textcolor{blue}{H(1)}}$};
          \node [style=event] (1) at (-0.9, -1) {${\tau_{0_{\textcolor{blue}{P_0(1)}}}}$};
          \node [style=event] (2) at (0.9, -1) {${\tau_{1_{\textcolor{blue}{P_0(1)}}}}$};
          \node [style=event] (3) at (-0.9, -2.2) {$b_{\textcolor{blue}{Z(1)}}$};
          \node [style=event] (5) at (0.9, -2.2) {$d_{\textcolor{blue}{X(1)}}$};
        \end{pgfonlayer}
        \begin{pgfonlayer}{edgelayer}
          \draw [style=arrow] (0) to (1);
          \draw [style=arrow] (1) to (3);
          \draw [style=arrow] (0) to (2);
          \draw [style=wiggle] (1) to (2);
          \draw [style=arrow] (2) to (5);
        \end{pgfonlayer}
      \end{tikzpicture}
    \end{minipage}
    \begin{minipage}{0.3\textwidth}
      $
      \begin{array}{l}
        Q(a) = H,\\ 
        Q(\tau_0) = P_0,\ 
        Q(\tau_1) = P_1,\\
        Q(b) = Z(1),\\
        Q(d) = X
      \end{array}
      $
    \end{minipage}
    \caption{Example of an unitary event structure}
    \label{fig:ex1-3}
  \end{figure}
\end{example}

\subsection{Language}\label{subsec:lan-3}
We adapt the language shown in Section~\ref{subsec:lan-1} to the quantum setting. For that we need
some preliminaries. We consider at our disposal a finite number of qubits $N$, whose associated
space is $\mathbb{C}^{2^N}$. Each qubit is denoted by a natural number $n$ and we let $\tilde{n}$ be
a sequence of numbers, denoting multiple qubits. We will need the notion of a partial density
operator, which is a density operator whose trace is less or equal to one.  We denote by
$\mathcal{H}$ its associated space and we denote by $\mathcal{D}_{\leq 1}(\mathcal{H})$ the set of
partial density operators. We shall use $\rho$ to represent a partial density operator.  The set of
actions is now composed by a set of unitary gates $\mathcal{U}$ together with a set of projections
$\set{P_0^n, P_1^n}$ in which $P_0^n$ and $P_1^n$ represent the projection of qubit $n$ into
$\ket{0}$ and $\ket{1}$, respectively.  The set of labels is then
$L' = L \cup \set{P_0^n,\, P_1^n}$, with $L = Act \cup \{sk\}$.

The set of commands allowed by the language are given by the following grammar:
\[
  \com{C} ::= \com{skip} \mid U(\vec{n}) \mid \seq{\com{C}}{\com{C}} \mid \meas{n}{\com{C}_1}{\com{C}_2}
  \mid \conc{\com{C}}{\com{C}}
\]
where $U(\vec{n})$ applies the unitary gate $U$ to the qubits presented in $\vec{n}$ and
$\meas{n}{\com{C}_1}{\com{C}_2}$ represents the measurement of a qubit $n$ and if the measurement
was made by $P_0^n$ then we execute $\com{C}_1$, else if the measurement was made by $P_1^n$ then we
execute $\com{C}_2$. Note that the behavior of $\meas{n}{\com{C}_1}{\com{C}_2}$ is similar to that of a
classical if clause.

The set of qubits being used in a command $\com{C}$ is defined as follows:
\begin{align*}
  & \qvar{\com{skip}} = \emptyset \\
  & \qvar{\com{U}(\vec{n})} = \vec{n} \\
  & \qvar{\meas{n}{\com{C}_1}{\com{C}_2}} = \set{n} \cup \qvar{\com{C}_1} \cup \qvar{\com{C}_2} \\
  & \qvar{\seq{\com{C}_1}{\com{C}_2}} = \qvar{\com{C}_1} \cup \qvar{\com{C}_2} \\ 
  &\qvar{\conc{\com{C}_1}{\com{C}_2}} = \qvar{\com{C}_1} \cup \qvar{\com{C}_2} 
\end{align*}

We restrict the parallel operator to only compose commands with disjoint variables, \ie\
$\conc{\com{C}_1}{\com{C}_2}$ iff $\qvar{\com{C}_1} \cap \qvar{\com{C}_2} = \emptyset$.

To define the operational semantics, we add a new symbol, denoted by $\checkmark$, that indicates
the end of a computation.  We define the small-step transition step
$\xrightarrow{a} \subseteq \com{C} \times L' \times (\com{C} \cup \{ \checkmark \})$, as the
smallest relation obeying the following rules:
\begin{figure}[ht!]
  \centering
  \begin{align*}
    \com{skip} \xrightarrow{sk} \checkmark
    \hspace*{1cm}
    \com{U}(\vec{n}) \xrightarrow{U(\vec{n})} \checkmark
    \hspace*{1cm}
    \meas{n}{\com{C}_1}{\com{C}_2} \xrightarrow{P_0^n} \com{C}_1
    \hspace*{1cm}
    \meas{n}{\com{C}_1}{\com{C}_2} \xrightarrow{P_1^n} \com{C}_2
  \end{align*}
  \begin{align*}
    \infer{\seq{\com{C}_1}{\com{C}_2} \xrightarrow{l} \com{C}_2}{ \com{C}_1 \xrightarrow{l} \checkmark }
    \hspace*{1cm}
    \infer{
    \seq{\com{C}_1}{\com{C}_2} \xrightarrow{l'} \seq{ \com{C}'_1  }{ \com{C}_2  }
    }
    {
    \com{C}_1 \xrightarrow{l'} \com{C}'_1
    }
  \end{align*}
  \begin{align*}
    \infer{\conc{\com{C}_1}{\com{C}_2} \xrightarrow{l} \com{C}_2}{ \com{C}_1 \xrightarrow{l} \checkmark }
    \hspace*{1cm}
    \infer{
    \conc{\com{C}_1}{\com{C}_2} \xrightarrow{l'} \conc{ \com{C}'_1  }{ \com{C}_2  }
    }
    {
    \com{C}_1 \xrightarrow{l'} \com{C}'_1
    }
    \hspace*{1cm}
    \infer{\conc{\com{C}_1}{\com{C}_2} \xrightarrow{l} \com{C}_1}{ \com{C}_2 \xrightarrow{l} \checkmark }
    \hspace*{1cm}
    \infer{
    \conc{\com{C}_1}{\com{C}_2} \xrightarrow{l'} \conc{ \com{C}_1  }{ \com{C}'_2  }
    }
    {
    \com{C}_2 \xrightarrow{l'} \com{C}'_2
    }
  \end{align*}
  \caption{Rules of the small-step operational semantics}
  \label{fig:op-small3}
\end{figure}

Define a word to be a sequence of labels:
\[
  \omega ::= l' \mid l' : \omega 
\]
where $l' : \omega$ appends $l'$ to the beginning of $\omega$.  A word can also be seen as an
element of $(L')^+$, \ie\ a possibly infinite sequence of labels without the empty sequence. Despite
$(L')^+$ allows the possibility of having infinite words, by now we focus only on the finite words.

Define the $n$-step transition,
$\xrightarrow{\omega} \subseteq \com{C} \times (L')^+ \times (\com{C} \cup \{\checkmark\})$, where
$n$ is the length of the words, as follows:
\begin{figure}[ht!]
  \centering
  \begin{align*}
    \infer{ \com{C} \xtwoheadrightarrow{l} \com{C'}  }{ \com{C} \xrightarrow{l} \com{C'}  }
    \hspace*{2cm}
    \infer{
    \com{C} \xtwoheadrightarrow{l' : \omega'} \com{C'}
    }{
    \com{C} \xrightarrow{l} \com{C''}
    \qquad
    \com{C''} \xtwoheadrightarrow{\omega'} \com{C'}
    }
  \end{align*}
  \caption{Rules of the n-step operational semantics}
  \label{fig:op-nstep3}
\end{figure}

\subsection{Constructions on Unitary Event Structures}\label{subsec:constructions-es-3}
To define the constructions on unitary event structures, we extend the definitions of sequential and
parallel composition from Section~\ref{subsec:constructions-es-1} to include the corresponding mapping of
events to unitary or projection operators. Additionally, we define the measurement composition by
making slight adjustments to the definition of non-deterministic composition provided in
Section~\ref{subsec:constructions-es-1}.

The sequential composition of two unitary event structures is defined as follows:
\begin{definition}[Sequential composition]\label{def:pes-seq3}
  Let $\es{U}_1 = \qpes{E_1}{\leq_1}{\#_1}{Q_1}$ and $\es{U}_2 = \qpes{E_2}{\leq_2}{\#_2}{Q_2}$ be unitary
  event structures.  Define $\seq{\es{U}_1}{\es{U}_2} = \qpes{E}{\leq}{\#}{Q}$ as:
  \begin{align*}
    & E = E_1 \uplus (E_2 \times \confmax{\es{U}_1}) \\
    & \leq\ = \set{e_1 \leq e'_1 \mid e \leq_1 e'} \cup
      \set{(e_2,x) \leq (e'_2,x) \mid \ e_2 \leq_2 e'_2} \cup
      \set{ e_1 \leq (e_2,x) \mid e_1 \in x  } \\
    & \#\ = \set{ e \# e' \mid \exists (e_1 \leq e, e'_1 \leq e')\ .\ e_1 \#_1 e'_1 }
      \cup \set{ (e_2,x) \# (e'_2,x) \mid e_2 \#_2 e'_2 } \\
    & Q(e) =
      \begin{cases}
        Q_1(e) & \text{ if } e \in E_1 \\
        Q_{2}(e_2) & \text{ if } e = (e_2,x) \in E_2 \times \confmax{\es{U}_1}
      \end{cases}
  \end{align*}
  where $E_2 \times \confmax{\es{U}_1} = \set{ (e,x) \mid e \in E_2,\ x \in \confmax{\es{U}_1}}$.
\end{definition}

\begin{lemma}\label{lem:seq-es3}
  Let $\es{U}_1$ and $\es{U}_2$ be unitary event structures. $\seq{\es{U}_1}{\es{U}_2}$ is a unitary
  event structure.
\end{lemma}

When defining the parallel composition, we must consider the restriction in Definition~\ref{def:qes}
requiring the operators associated with concurrent events to commute.  Since in
Definition~\ref{def:pes-conc1} every event in $\es{E}_1$ is concurrent with every event of
$\es{E}_2$, then when extending this definition to the quantum case, the operators associated to
$\es{U}_1 = (\es{E}_1, Q_1)$ and $\es{U}_2 = (\es{E}_2, Q_2)$ must commute for every event, \ie\
$\forall e_1 \in E_1, e_2 \in E_2.\ [Q_1(e_1), Q_2(e_2)] = 0$.
\begin{definition}[Parallel composition]\label{def:pes-conc3}
  Let $\es{U}_1 = \qpes{E_1}{\leq_1}{\#_1}{Q_1}$ and $\es{U}_2 = \qpes{E_2}{\leq_2}{\#_2}{Q_2}$ be
  unitary event structures, such that
  $\forall e_1 \in E_1, e_2 \in E_2.\ [Q_1(e_1), Q_2(e_2)] = 0$.  Define
  $\conc{\es{U}_1}{\es{U}_2} = \qpes{E}{\leq}{\#}{Q}$ as:
  \begin{align*}
    & E = E_1 \uplus E_2  \\
    & \leq\ = \leq_1 \uplus \leq_2 \\
    & \# = \#_1 \uplus \#_2 \\
    & Q(e) =
      \begin{cases}
        Q_1(e) & \text{ if } e \in E_1 \\
        Q_2(e) & \text{ if } e \in E_2 \\
      \end{cases}
  \end{align*}
\end{definition}

\begin{lemma}\label{lem:conc-es3}
  Let $\es{U}_1$ and $\es{U}_2$ be unitary event structures. $\conc{\es{U}_1}{\es{U}_2}$ is a
  unitary event structure.
\end{lemma}

When defining the measurement command in terms of unitary event structures, we must consider the
restriction in Definition~\ref{def:qes} that requires the sum of operators associated with events in
minimal conflict to be the identity.
\begin{definition}[Measurement composition]\label{def:pes-meas3}
  Let $\es{U}_1 = \qpes{E_1}{\leq_1}{\#_1}{Q_1}$ and $\es{U}_2 = \qpes{E_2}{\leq_2}{\#_2}{Q_2}$ be unitary
  event structures. Define $\meas{n}{\es{U}_1}{\es{U}_2} = \qpes{E}{\leq}{\#}{Q}$ as:
  \begin{align*}
    & E = \set{\tau_0^n,\, \tau_1^n} \uplus E_1 \uplus E_2 \\
    & \leq = \set{\tau_0^n \leq e \mid (e = \tau_0^n \vee e \in E_1)}
        \cup \set{\tau_1^n \leq e \mid (e = \tau_1^n \vee e \in E_2)}
      \cup \leq_1 \uplus \leq_2 \\
    & \# = \set{e \# e' \mid (e = \tau_0^n \vee e \in E_1), (e' = \tau_1^n \vee e' \in E_2)} \\
    & \hspace*{0.8cm}
      \cup \set{e' \# e \mid (e = \tau_0^n \vee e \in E_1), (e' = \tau_1^n \vee e' \in E_2)}
      \cup \#_1 \uplus \#_2 \\
    & Q(e) =
      \begin{cases}
        P_0(n) & \text{ if } e = \tau_0^n \\
        P_1(n) & \text{ if } e = \tau_1^n \\
        Q_1(e) & \text{ if } e \in E_1 \\
        Q_2(e) & \text{ if } e \in E_2
      \end{cases}
  \end{align*}
\end{definition}

\begin{lemma}\label{lem:pes-meas3}
  Let $\es{U}_1$ and $\es{U}_2$ be unitary event structures. $\meas{n}{\es{U}_1}{\es{U}_2}$ is a
  unitary event structure.
\end{lemma}

Note that, in Definition~\ref{def:pes-meas3}, the fact that $Q(\tau_0^n) = P_0(n)$ and
$Q(\tau_1^n) = P_1(n)$, ensures us that $Q(\tau_0^n) + Q(\tau_1^n) = P_0(n) + P_1(n) = Id$.

\begin{remark}
  As previously stated in the operational semantics, the measurement command can be seen as
  $\nd{\seq{P_0(n)}{\com{C}_1}}{\seq{P_1(n)}{\com{C}_2}}$. To transport such relation to the setting
  of unitary event structures, we find helpful to write $\meas{n}{\es{U}_1}{\es{U}_2}$ as
  $\nd{\seq{\es{P}_0^n}{\es{U}_1}}{\seq{\es{P}_1^n}{\es{U}_2}}$, where
  \begin{align*}
    & \es{P}_0^n = \qpes{\set{\tau_0^n}}{\set{\tau_0^n \leq \tau_0^n}}{\emptyset}{Q(\tau_0^n) = P_0(n)} \\
    & \es{P}_1^n = \qpes{\set{\tau_1^n}}{\set{\tau_1^n \leq \tau_1^n}}{\emptyset}{Q(\tau_1^n) = P_1(n)}
  \end{align*}

  Note that by doing this, we are abusing notation, since $\es{P}_0^n$ and $\es{P}_1^n$ are quantum
  event structures and not unitary event structures, because the quantum operators associated to the
  events $\tau_0^n$ and $\tau_1^n$ are not unitary. However, this notational abuse allows the
  measurement command to connect to the operational counterpart directly.
\end{remark}

In Definition~\ref{def:pes-meas3}, the events associated with projections include a superscript
indicating the qubit on which the projection is performed. Whenever it is clear from the context, we
will drop the superscript.

We are now ready to interpret commands of the language in Section~\ref{subsec:lan-3} as unitary
event structures.
\begin{definition}\label{def:den-sem3}
  We interpret commands as unitary event structures as follows:
  \begin{align*}
    & \mf{\com{skip}} = (\set{sk}, \set{sk \leq sk}, \emptyset, Q(sk) = Id) \\
    & \mf{\com{U}(\vec{n})} = (\set{U_{\vec{n}}}, \set{U_{\vec{n}} \leq U_{\vec{n}}}, \emptyset,
      Q(U_{\vec{n}}) = U(\vec{n})) \\
    & \mf{\meas{n}{\com{C}_1}{\com{C}_2}} =
      \seq{\es{P}_0^n}{\mf{\com{C}_1}} + \seq{\es{P}_1^n}{\mf{\com{C}_2}} \\
    & \mf{\seq{\com{C_1}}{\com{C_2}}} = \seq{\mf{\com{C_1}}}{\mf{\com{C_2}}} \\
    & \mf{\conc{\com{C_1}}{\com{C_2}}} = \conc{\mf{\com{C_1}}}{\mf{\com{C_2}}} \\
  \end{align*}
\end{definition}

We extend Definition~\ref{def:pes-sub1} to the quantum setting. In this case, we need to ensure that
the quantum operator is preserved across unitary event structures, \ie\ if
$\es{U}_1 \sqsubseteq \es{U}_2$ then the quantum operators of the events from $\es{U}_2$ that are
also events in $\es{U}_1$ should be the same. Recall what was said regarding plain and composite
events in Section~\ref{subsec:constructions-es-1} above Definition~\ref{def:pes-sub1}.
\begin{definition}\label{def:pes-sub3}
  Let $\es{U}_1 = \qpes{E_1}{\leq_1}{\#_1}{Q_1}$ and $\es{U}_2 = \qpes{E_2}{\leq_2}{\#_2}{Q_2}$ be
  unitary event structures.  Say $\es{U}_1 \sqsubseteq \es{U}_2$ whenever exists an injective
  function $f: E_1 \rightarrow E_2$ such that $\forall e,e' \in E_1$:
  \begin{align*}
    & \pi(f(e)) = \pi(e) \\
    & e \leq_1 e' \Leftrightarrow f(e) \leq_2 f(e') \\
    & e \#_1 e' \Leftrightarrow f(e) \#_2 f(e') \\
    & Q_1(e) = Q_2(f(e))
  \end{align*}
  We say that two event structures $\es{U}_1, \es{U}_2$ are similar, denoted
  $\es{U}_1 \equiv \es{U}_2$, iff $\es{U}_1 \sqsubseteq \es{U}_2$ and
  $\es{U}_2 \sqsubseteq \es{U}_1$.
\end{definition}

We now extend the removal of an initial event to the quantum setting. Intuitively, we need to
preserve the quantum operators of the events that were not removed.
\begin{definition}[Remove initial event]\label{def:rem-init3}
  Let $\es{U} = \qpes{E}{\leq}{\#}{Q}$ be a unitary event structure and $a \in \init{\es{U}}$.
  Define $\es{U} \backslash a = \qpes{E'}{\leq'}{\#'}{Q'}$ as
  \begin{align*}
    & E' = \set{e \in E \mid \neg (e \# a),\, e \neq a} \\
    & \leq'\, =\, \set{e \leq e' \mid e,e' \in E' } \\
    & \#'\, =\, \set{e \# e' \mid e,e' \in E' } \\
    & Q' = Q|_{E'}
  \end{align*}
\end{definition}
\begin{lemma}\label{lem:rem-init-es3}
  Let $\es{U}$ be a unitary event structure and $a \in \init{\es{U}}$.  $\es{U} \backslash a$ is a
  unitary event structure.
\end{lemma}

\subsection{Results}\label{subsec:results-3}
Here we present the results obtained.  Similarly to the previous subsection, we will just list what
was proved.  We postpone the addition of the proofs as well as the examples for some results for
future versions of the document.

For this section, we interpret $\checkmark$ as the empty unitary event structure, \ie\
$\mf{\checkmark} = \qpes{\emptyset}{\emptyset}{\emptyset}{! : \emptyset \rightarrow
  Op(\mathcal{H})}$.

We start by showing that the sequential, parallel, and measurement operators are monotonic with
respect to Definition~\ref{def:pes-sub3}.
\begin{lemma}\label{lem:seq-mono3}
  Let $\es{U}_1 = \qpes{E_1}{\leq_1}{\#_1}{Q_1}$ and $\es{U}_2 = \qpes{E_2}{\leq_2}{\#_2}{Q_2}$ be
  unitary event structures.  If $\es{U}_1 \sqsubseteq \es{U}'_1$ and
  $\es{U}_2 \sqsubseteq \es{U}'_2$ then
  $\seq{\es{U}_1}{\es{U}_2} \sqsubseteq \seq{\es{U}'_1}{\es{U}'_2}$.
\end{lemma}

\begin{lemma}\label{lem:conc-mono3}
  Let $\es{U}_1 = \qpes{E_1}{\leq_1}{\#_1}{Q_1}$ and $\es{U}_2 = \qpes{E_2}{\leq_2}{\#_2}{Q_2}$ be
  unitary event structures.  If $\es{U}_1 \sqsubseteq \es{U}'_1$ and
  $\es{U}_2 \sqsubseteq \es{U}'_2$ then
  $\conc{\es{U}_1}{\es{U}_2} \sqsubseteq \conc{\es{U}'_1}{\es{U}'_2}$.
\end{lemma}

\begin{lemma}\label{lem:meas-mono3}
  Let $\es{U}_1 = \qpes{E_1}{\leq_1}{\#_1}{Q_1}$ and $\es{U}_2 = \qpes{E_2}{\leq_2}{\#_2}{Q_2}$ be
  unitary event structures.  If $\es{U}_1 \sqsubseteq \es{U}'_1$ and
  $\es{U}_2 \sqsubseteq \es{U}'_2$ then
  $\meas{n}{\es{U}_1}{\es{U}_2} \sqsubseteq \meas{n}{\es{U}'_1}{\es{U}'_2}$.
\end{lemma}

Next, we show how unitary event structures composed by the sequential, concurrent, and measurement
operators interact with the removal of initial events.
\begin{lemma}\label{lem:seq-rem-init3}
  Let $\es{U}_1 = \qpes{E_1}{\leq_1}{\#_1}{Q_1}$ and $\es{U}_2 = \qpes{E_2}{\leq_2}{\#_2}{Q_2}$ be
  unitary event structures. Consider $\seq{\es{U}_1}{\es{U}_2}$ such that
  $l \in \init{\seq{\es{U}_1}{\es{U}_2}}$. Then
  $(\seq{\es{U}_1}{\es{U}_2}) \backslash l \equiv \seq{(\es{U}_1\backslash l)}{\es{U}_2}$.
\end{lemma}

\begin{lemma}\label{lem:meas-rem-init3}
  Let $\es{U}_1 = \qpes{E_1}{\leq_1}{\#_1}{Q_1}$ and $\es{U}_2 = \qpes{E_2}{\leq_2}{\#_2}{Q_2}$ be
  unitary event structures. Consider $\meas{n}{\es{U}_1}{\es{U}_2}$ such that
  $l \in \init{\meas{n}{\es{U}_1}{\es{U}_2}}$. Then
  \[
    (\meas{n}{\es{U}_1}{\es{U}_2}) \backslash l \equiv
    \begin{cases}
      \es{U}_1 & \text{ if } l = \tau_0^n \\
      \es{U}_2 & \text{ if } l = \tau_1^n
    \end{cases}
  \]
\end{lemma}

\begin{lemma}\label{lem:conc-rem-init3}
  Let $\es{U}_1 = \qpes{E_1}{\leq_1}{\#_1}{Q_1}$ and $\es{U}_2 = \qpes{E_2}{\leq_2}{\#_2}{Q_2}$ be
  unitary event structures.  Consider $\conc{\es{U}_1}{\es{U}_2}$ such that
  $l \in \init{\conc{\es{U}_1}{\es{U}_2}}$. Then
  $(\conc{\es{U}_1}{\es{U}_2}) \backslash l \equiv
  \conc{(\es{U}_1\backslash l)}{(\es{U}_2 \backslash l)}$.
\end{lemma}

It is straightforward to see that:
\begin{lemma}\label{lem:conc-symmetric3}
  Let $\es{U}_1, \es{U}_2$ be unitary event structures.  Then
  $\conc{\es{U}_1}{\es{U}_2} = \conc{\es{U}_2}{\es{U}_1}$.
\end{lemma}
\begin{proof}
  It follows directly from Definition~\ref{def:pes-conc3}.
\end{proof}

We now establish a relation between the operational and denotational semantics by a soundness and
adequacy theorem.  The formulation of the following lemmas and theorems is equal to the ones in
Section~\ref{subsec:results-1}. We thus suggest that the reader recalls what was said about these
lemmas and theorems. It is useful to recall that $\bot = \mf{\checkmark}$.
\begin{lemma}[Soundness I]\label{res:soundI-3}
  If $\com{C} \xrightarrow{l} \com{C'}$ then $\mf{\com{C'}} \equiv \mf{\com{C}} \backslash l$.
\end{lemma}

\begin{theorem}[Soundness II]\label{res:soundII-3}
  If $\com{C} \xtwoheadrightarrow{\omega} \com{C'}$ then $\exists x \in \mathcal{C}(\mf{\com{C}})$
  such that $\emptyset \stackrel{\omega}{\chain\ } x$.
\end{theorem}




\begin{lemma}[Adequacy I]\label{res:adI-3}
  Let $l \in \init{\mf{\com{C}}}$. Then $\exists \com{C'} \in (\com{C} \cup \{\checkmark\})$ s.t
  $\com{C} \xrightarrow{l} \com{C'}$ and $\mf{\com{C}} \backslash l \equiv \mf{\com{C'}}$.
\end{lemma}

\begin{theorem}[Adequacy II]\label{res:adII-3}
  If $ \emptyset \neq x \in \mathcal{C}(\mf{\com{C}})$ s.t.
  $\emptyset \stackrel{\omega}{\chain\ } x$ then $\exists \com{C'}$ s.t.
  $\com{C} \xtwoheadrightarrow{\omega} \com{C'}$.
\end{theorem}
  





\subsection{Unitary Event Structures with initial state}\label{subsec:quantum-to-prob}
This section shows how the restrictions added to Definition~\ref{def:qes} (minimal conflict being
transitive and the sum of the operators in minimal conflict, or an event itself, should be a unitary
operator), when compared to Winskel's definition of quantum event
structures~\cite[Definition~2]{winskel14}, allow to extend a result obtained by
Winskel~\cite[Theorem~3]{winskel14}. Furthermore, in this section, we show most of the proofs after
stating a result, since these are not so evident.

According to~\cite[Theorem~3]{winskel14}, an elementary quantum event structure, \ie\ an event
structure without conflicting events, paired with an initial state $\rho$ and a valuation function
defined as $v(x) = \text{Tr}(A_x^\dagger A_x \rho)$, corresponds to a probabilistic event
structure. Recall that by Proposition~\cite[Proposition~5]{winskel14} it suffices to verify
Equation~\ref{eq:prob-es-condition} for $y \cchain{e_1} x_1,\, \dots,\, y \cchain{e_n} x_n$ to check
if a structure is a probabilistic event structure or not. Hence, we focus only on the relation
between these events. These events are either in conflict or are concurrent. The reason why Winskel
considers elementary quantum event structures lies in the fact that in a probabilistic event
structure, the sum of the probability of events in conflict is less than or equal to one.  For
example, consider two conflicting events $a$ and $b$, such that $v(\{a\}) = 1$ and $v(\{b\}) = 1$.
In this case, the sum condition in Definition~\ref{def:prob-es} fails when taking $y$ as the empty
set and $x_1$, $x_2$ as $\set{a}$ and $\set{b}$, respectively. Through some calculations, the sum
simplifies to $v(\emptyset) - (v(\set{a}) + v(\set{b})) = 1 - (1 + 1) = -1$, which does not meet the
criteria of being greater than or equal to $0$.

We aim to remove the restriction of considering only elementary quantum event structures
from~\cite[Theorem~3]{winskel14}. Hence, based on the previous example, we must restrict the
probabilities associated with events in conflict. To do this, Definition~\ref{def:qes} includes the
condition that ``$\minconflict$ is transitive''. Recall that the intuition behind this condition is
that events in minimal conflict correspond to projections that form a measurement. Since these
projections are orthogonal, they are inherently captured by conflicting events. For instance, if we
project onto $\ket{0}$, we cannot project onto $\ket{1}$. Furthermore, the probabilities associated
with projections are complementary, \ie\ if the probability of projecting onto $\ket{0}$ is $p$,
then the probability of projecting onto $\ket{1}$ must be $1-p$. By applying this new information to
the previous example of the valuation of conflicting events, we have that
$v(\emptyset) - (v(\set{a}) + v(\set{b})) = 1 - (p + (1-p)) = 1 - 1 = 0$, which meets the criteria
of being greater than or equal to $0$. Here, the event $a$ represents the projection to $\ket{0}$
and the event $b$ represents the projection to $\ket{1}$. Furthermore, Definition~\ref{def:qes} also
demands the sum of events in minimal conflict or the event itself to be a unitary. This is useful to
ensure that we are only considering unitary operations.

To achieve our goal, the difference between our definition and Winskel's quantum event structures
lies in the additional restrictions we impose. Therefore, we use as basis the proof outlined
in~\cite[Theorem~3]{winskel14}. This entails that we need only to consider the case in which all the
events are mapped to projections such that either all events are in conflict or there are events in
conflict. For the former, we have everything. For the latter, we need extra machinery, which we show
here.

Recall that according to~\cite[Proposition~5]{winskel14}, to show that a structure is a
probabilistic event structure, we only need to show that the condition in
Definition~\ref{def:prob-es} holds for $y \cchain{e_1} x_1,\, \dots,\, y \cchain{e_n} x_n$.  We then
build a unitary event structure formed by the events of
$y \cchain{e_1} x_1,\, \dots,\, y \cchain{e_n} x_n$.
\begin{definition}\label{def:E-tilde}
  Let $\es{U} = \qpes{E}{\leq}{\#}{Q}$ be a unitary event structure and $y \in \confES{\es{U}}$.
  Define $\tilde{\es{U}}_y = \qpes{\tilde{E}}{\tilde{\leq}}{\tilde{\#}}{\tilde{Q}}$ as follows:
  \begin{align*}
    & \tilde{E} = \set{e \mid y \cchain{e} y \cup \set{e}} \\
    & \tilde{\leq} = \set{ (e, e) \mid e \in \tilde{E}} \\
    & \tilde{\#} = \# \cap (\tilde{E} \times \tilde{E}) \\
    & \tilde{Q} = Q|_{\tilde{E}}
  \end{align*}
\end{definition}
\begin{lemma}\label{lem:E-tilde}
  $\tilde{\es{U}}_y = \qpes{\tilde{E}}{\tilde{\leq}}{\tilde{\#}}{\tilde{Q}}$ is a unitary event
  structure.
\end{lemma}

Next, we merge the conflicting events from Definition~\ref{def:E-tilde}. This gives us an empty
conflict relation, hence an elementary unitary event structure (recall that
$ [e] = \set{ e' \mid e = e',\ e \minconflict e'}$).
\begin{definition}\label{def:E-hat}
  Let $\tilde{\es{U}}_y$ be a unitary event structure.  Define
  $\hat{\es{U}} = \qpes{\hat{E}}{\hat{\leq}}{\hat{\#}}{\hat{Q}}$ as follows:
  \begin{align*}
    & \hat{E} = \set{[e] \mid e \in \tilde{\es{U}}} \\
    & \hat{\leq} = \set{([e],[e]) \mid [e] \in \hat{E}} \\
    & \hat{\#} = \emptyset \\
    & \hat{Q}([e]) =
      \begin{cases}
        Q(e) & \text{ if } |[e]| = 1 \\
        \sum_{e' \in [e]} Q(e') & \text{ if } |[e]| > 1 
      \end{cases}
  \end{align*}  
\end{definition}
\begin{lemma}\label{lem:E-hat}
  $\hat{\es{U}} = \qpes{\hat{E}}{\hat{\leq}}{\hat{\#}}{\hat{Q}}$ is a unitary event structure.
\end{lemma}

We then define a map of event structures, $proj_{\tilde{\es{E}}}$, between the underlying event
structures of $\tilde{\es{U}}_y$ and $\hat{\es{U}}$. Note that $proj_{\tilde{\es{E}}}$ is a total
map of event structures. Hence, by recalling Example~\ref{ex:map-es-1}, for any configuration
$x \in \confES{\tilde{\es{E}}_y}$ we have $|x| = |proj_{\tilde{\es{E}}}[x]|$.  Formally, the
definition of the map of event structures is as follows:
$proj_{\tilde{\es{E}}} : \pes{\tilde{E}}{\tilde{\leq}}{\tilde{\#}} \rightarrow
\pes{\hat{E}}{\hat{\leq}}{\hat{\#}}$, where
\begin{align*}
  proj_{\tilde{\es{E}}} :
  \tilde{E} &\rightarrow \hat{E} \\
  e &\mapsto [e]
\end{align*}

\begin{lemma}\label{lem:map-es}
  $proj_{\tilde{\es{E}}} :
\pes{\tilde{E}}{\tilde{\leq}}{\tilde{\#}} \rightarrow
\pes{\hat{E}}{\hat{\leq}}{\hat{\#}}$
 is a map of events structures.
\end{lemma}

We now prove that our intuition about a feature of $proj_{\tilde{\es{E}}}$ is true, \ie\ that for a
given configuration $x \in \confES{\tilde{\es{E}}_y}$ we have $|x| = |proj_{\tilde{\es{E}}}[x]|$.
\begin{lemma}\label{lem:map-proj}
  Consider $\tilde{\es{E}}_y$, $\hat{\es{E}}$, and $proj_{\tilde{\es{E}}}$.
  If $x \in \confES{\tilde{\es{E}}_y}$ then $|proj_{\tilde{\es{E}}}[x]| = |x|$.
\end{lemma}
\begin{proof}
  Let $x \in \confES{\tilde{\es{E}}_y}$.  We know that $proj_{\tilde{\es{E}}}$ is total, hence
  \[
    proj_{\tilde{\es{E}}}[x] =
    \set{ proj_{\tilde{\es{E}}}(e) \mid e \in x } =
    \set{ [e] \mid e \in x}
  \]
  Furthermore, $proj_{\tilde{\es{E}}}$ is locally injective.
  Hence if $e_1, \dots, e_n \in x$ then
  \[
    proj_{\tilde{\es{E}}}(e_1), \dots, proj_{\tilde{\es{E}}}(e_n) \in proj_{\tilde{\es{E}}}[x]
  \]
  
  Thus $|x| = |proj_{\tilde{\es{E}}}[x]|$.
\end{proof}

The next lemma states that for any $\hat{x} \in \confES{\hat{\es{E}}}$ and
$\tilde{x} \in \confES{\tilde{\es{E}}}$, the set of configurations $\tilde{x}$ satisfying
$proj_{\tilde{\es{E}}}[\tilde{x}] = \hat{x}$ consists precisely of those configurations formed by
selecting one event from each equivalence class that makes up $\hat{x}$.  Furthermore, the
equivalence classes in $\hat{x}$ are composed of single events or events in minimal
conflict. Consequently, multiple configurations in $\confES{\tilde{\es{E}}}$ can be mapped to the
same configuration in $\confES{\hat{\es{E}}}$.
\begin{lemma}\label{lem:help-sum}
  Consider $\tilde{\es{E}}_y$, $\hat{\es{E}}$. Let
  $\hat{x} = \set{[e_1], \dots, [e_n]} \in \confES{\hat{\es{E}}}$ and
  $\tilde{x} \in \confES{\tilde{\es{E}}_y}$.
  Then $\set{\tilde{x} \mid proj_{\tilde{\es{E}}}[\tilde{x}] = \hat{x}} =
  \set{ \set{\tilde{e}_1, \dots, \tilde{e}_n} \mid \forall i, \tilde{e}_i \in [e_i] }$.
\end{lemma}
\begin{proof}
  We have two cases:
  \begin{itemize}
  \item $\set{\tilde{x} \mid proj_{\tilde{\es{E}}}[\tilde{x}] = \hat{x}} \subseteq
    \set{ \set{\tilde{e}_1, \dots, \tilde{e}_n} \mid \forall i, \tilde{e}_i \in [e_i] }$

    Let $\tilde{x} = \set{\tilde{e}_1, \dots, \tilde{e}_n} \in \confES{\tilde{\es{E}}_y}$.  By
    Definition~\ref{def:map-es}, $proj_{\tilde{\es{E}}}[\tilde{x}] = \hat{x} \in \confES{\hat{\es{E}}}$.
    Furthermore,
    $proj_{\tilde{\es{E}}}[\tilde{x}]
    = proj_{\tilde{\es{E}}}(e_1), \dots, proj_{\tilde{\es{E}}}(e_n)
    = \set{ [e_1], \dots, [e_n] }$.
    By definition of $[e]$, we know that $\forall i\, .\, \tilde{e}_i \in \tilde{x}$ we have
    $\tilde{e}_i \in [e_i]$.
    Hence, we are done.
    
  \item $\set{ \set{\tilde{e}_1, \dots, \tilde{e}_n} \mid \forall i, \tilde{e}_i \in [e_i] }
    \subseteq \set{\tilde{x} \mid proj_{\tilde{\es{E}}}[\tilde{x}] = \hat{x}}$

    Let $\set{\tilde{e}_1, \dots, \tilde{e}_n} \in
    \set{ \set{\tilde{e}_1, \dots, \tilde{e}_n} \mid \forall i, \tilde{e}_i \in [e_i] }$.
    We need to show that $\set{\tilde{e}_1, \dots, \tilde{e}_n} \in \confES{\tilde{\es{E}}_y}$.
    \begin{enumerate}
    \item $\forall \tilde{e}, \tilde{e}' \in \set{\tilde{e}_1, \dots, \tilde{e}_n}\, .\,
      \neg (\tilde{e} \# \tilde{e}')$

      Let $\tilde{e}, \tilde{e}' \in \set{\tilde{e}_1, \dots, \tilde{e}_n}$.  Then we know that
      $\tilde{e} \in [e]$ and $\tilde{e}' \in [e']$.  By definition of $[e]$, we have that
      $\neg (\tilde{e} \minconflict \tilde{e}')$, which by Definition~\ref{def:E-tilde} means
      $\neg (\tilde{e} \# \tilde{e}')$.
      
    \item $\forall \tilde{e}, \tilde{e}'\, .\,
      \tilde{e}' \tilde{\leq} \tilde{e} \wedge \tilde{e} \in
      \set{\tilde{e}_1, \dots, \tilde{e}_n} \Rightarrow
      \tilde{e}' \in \set{\tilde{e}_1, \dots, \tilde{e}_n}$

      By Definition~\ref{def:E-tilde}, the causal relation is the equality.  Hence this condition
      trivially holds.
    \end{enumerate}

    Since $\set{\tilde{e}_1, \dots, \tilde{e}_n} \in \confES{\tilde{\es{E}}_y}$, it lacks to show
    that
    $proj_{\tilde{\es{E}}_y}[\set{\tilde{e}_1, \dots, \tilde{e}_n}] \in \confES{\tilde{\es{E}}_y} =
    \hat{x}$.  That comes directly from applying $proj_{\tilde{\es{E}}_y}$ to
    $\set{\tilde{e}_1, \dots, \tilde{e}_n} \in \confES{\tilde{\es{E}}_y}$, as follows:
    \[
      proj_{\tilde{\es{E}}_y}[\set{\tilde{e}_1, \dots, \tilde{e}_n}]
      = \set{proj_{\tilde{\es{E}}_y}(\tilde{e}_1), \dots, proj_{\tilde{\es{E}}_y}(\tilde{e}_n)}
      = \set{[e_1], \dots, [e_n]}
      = \hat{x}
    \]
  \end{itemize}
\end{proof}

To better understand the mapping of event structures, $proj_{\tilde{\es{E}}}$,
Lemma~\ref{lem:map-proj}, and Lemma~\ref{lem:help-sum} consider the following example.
\begin{example}\label{ex:tilde-U-to-hat-U}
  Let us consider the following unitary event structure (where we omit the associated quantum
  operators):
  \begin{figure}[ht!]
    \centering
    \begin{tikzpicture}[tikzfig]
      \begin{pgfonlayer}{nodelayer}
        \node [style=event] (0) at (0, 1) {$e_1$};
        \node [style=event] (1) at (-1, -1) {$e_2$};
        \node [style=event] (2) at (1, -1) {$e_3$};
        \node [style=event] (3) at (3, -1) {$e_4$};
        \node [style=event] (4) at (1, -3) {$e_5$};
        \node [style=event] (5) at (3, -3) {$e_6$};
      \end{pgfonlayer}
      \begin{pgfonlayer}{edgelayer}
        \draw [style=wiggle] (2) to (3);
        \draw [style=arrow] (0) to (2);
        \draw [style=arrow] (2) to (4);
        \draw [style=arrow] (3) to (5);
        \draw [style=arrow, bend left=15] (0) to (3);
        \draw [style=arrow] (0) to (1);
      \end{pgfonlayer}
    \end{tikzpicture}
    \caption{Unitary event structure $\es{U}$}
    \label{fig:unit-es-U}
  \end{figure}

  From Figure~\ref{fig:unit-es-U} we can deduce all possible covering chains and consequently all
  possible configurations.
  \begin{figure}[ht!]
    \centering
    \begin{tikzpicture}[tikzfig]
      \begin{pgfonlayer}{nodelayer}
        \node [style=event] (0) at (-13, 0) {$\emptyset$};
        \node [style=event] (1) at (-10, 0) {$\{e_1\}$};
        \node [style=event] (2) at (-5, 0) {$\{e_1, e_2\}$};
        \node [style=event] (3) at (1, 2) {$\{e_1, e_2, e_3\}$};
        \node [style=event] (4) at (1, -2) {$\{e_1, e_2, e_4\}$};
        \node [style=event] (5) at (1, -3.75) {$\{e_1, e_4, e_6\}$};
        \node [style=event] (6) at (1, 4) {$\{e_1, e_3, e_5\}$};
        \node [style=event] (7) at (8.5, 1.25) {$\{e_1, e_2, e_3, e_5\}$};
        \node [style=event] (8) at (8.5, -0.75) {$\{e_1, e_2, e_4, e_6\}$};
        \node [style=event] (9) at (-5, 2) {$\{e_1, e_3\}$};
        \node [style=event] (10) at (-5, -2) {$\{e_1, e_4\}$};
        \node [style=none] (11) at (-11.5, 0.5) {$e_1$};
        \node [style=none] (12) at (-7.5, 2) {$e_3$};
        \node [style=none] (13) at (-7.25, 0.5) {$e_2$};
        \node [style=none] (14) at (-7.75, -1.75) {$e_4$};
        \node [style=none] (15) at (-2, 2.5) {$e_2$};
        \node [style=none] (16) at (-1.75, -2.25) {$e_2$};
        \node [style=none] (17) at (6, 3.25) {$e_2$};
        \node [style=none] (18) at (5.5, -3.25) {$e_2$};
        \node [style=none] (19) at (-2.5, 1.25) {$e_3$};
        \node [style=none] (20) at (-2, -0.5) {$e_4$};
        \node [style=none] (21) at (-3, 3.75) {$e_5$};
        \node [style=none] (22) at (4.5, 2) {$e_5$};
        \node [style=none] (23) at (4.25, -1.25) {$e_6$};
        \node [style=none] (24) at (-2.25, -3.75) {$e_6$};
      \end{pgfonlayer}
      \begin{pgfonlayer}{edgelayer}
        \draw [style=cchain] (0) to (1);
        \draw [style=cchain] (1) to (2);
        \draw [style=cchain, bend left=15] (1) to (9);
        \draw [style=cchain, bend right=15] (1) to (10);
        \draw [style=cchain] (2) to (3);
        \draw [style=cchain] (2) to (4);
        \draw [style=cchain] (9) to (3);
        \draw [style=cchain, bend left=15] (9) to (6);
        \draw [style=cchain, bend right=15] (10) to (5);
        \draw [style=cchain] (10) to (4);
        \draw [style=cchain] (3) to (7);
        \draw [style=cchain] (4) to (8);
        \draw [style=cchain, bend right=15] (5) to (8);
        \draw [style=cchain, bend left=15] (6) to (7);
      \end{pgfonlayer}
    \end{tikzpicture}    
    \caption{Covering chains of $\es{U}$}
    \label{fig:cov-chain-unit-es-U}
  \end{figure}

  Let us now focus on Figure~\ref{fig:frag-cov-chain-unit-es-U}, which is a specific part of
  Figure~\ref{fig:cov-chain-unit-es-U}.
    \begin{figure}[ht!]
    \centering
    \begin{tikzpicture}[tikzfig]
      \begin{pgfonlayer}{nodelayer}
        \node [style=event] (1) at (-10, 0) {$\{e_1\}$};
        \node [style=event] (2) at (-5, 0) {$\{e_1, e_2\}$};
        \node [style=event] (9) at (-5, 2) {$\{e_1, e_3\}$};
        \node [style=event] (10) at (-5, -2) {$\{e_1, e_4\}$};
        \node [style=none] (12) at (-7.5, 2) {$e_3$};
        \node [style=none] (13) at (-7.25, 0.5) {$e_2$};
        \node [style=none] (14) at (-7.75, -1.75) {$e_4$};
      \end{pgfonlayer}
      \begin{pgfonlayer}{edgelayer}
        \draw [style=cchain] (1) to (2);
        \draw [style=cchain, bend left=15] (1) to (9);
        \draw [style=cchain, bend right=15] (1) to (10);
      \end{pgfonlayer}
    \end{tikzpicture}    
    \caption{Fragment of the covering chains of $\es{U}$}
    \label{fig:frag-cov-chain-unit-es-U}
  \end{figure}

  From Figure~\ref{fig:frag-cov-chain-unit-es-U} we can build the unitary event structure
  $\tilde{U}$, shown in Figure~\ref{fig:unit-es-tilde-U}.
  \begin{figure}[ht!]
    \centering
    \begin{tikzpicture}[tikzfig]
      \begin{pgfonlayer}{nodelayer}
        \node [style=event] (1) at (-1, -1) {$e_2$};
        \node [style=event] (2) at (1, -1) {$e_3$};
        \node [style=event] (3) at (3, -1) {$e_4$};
      \end{pgfonlayer}
      \begin{pgfonlayer}{edgelayer}
        \draw [style=wiggle] (2) to (3);
      \end{pgfonlayer}
    \end{tikzpicture}
    \caption{Unitary event structure $\tilde{\es{U}}_y$}
    \label{fig:unit-es-tilde-U}
  \end{figure}

  By applying the mapping of event structures $proj_{\tilde{\es{E}}}$ we obtain:
  \begin{align*}
    e_3, e_4 &\stackrel{proj_{\tilde{\es{E}}}}{\mapsto} [e_{3,4}] \\
    e_2 &\stackrel{proj_{\tilde{\es{E}}}}{\mapsto} [e_2]
  \end{align*}
  which gives the unitary event structure in Figure~\ref{fig:unit-es-hat-U}.
  \begin{figure}[ht!]
    \centering
    \begin{tikzpicture}[tikzfig]
      \begin{pgfonlayer}{nodelayer}
        \node [style=event] (1) at (-1, -1) {$[e_2]$};
        \node [style=event] (2) at (1, -1) {$[e_{3,4}]$};
      \end{pgfonlayer}
    \end{tikzpicture}
    \caption{Unitary event structure $\hat{\es{U}}_y$}
    \label{fig:unit-es-hat-U}
  \end{figure}

  Note that all the events of $\tilde{\es{U}}$ were mapped to $\hat{\es{U}}$. Hence the map
  $proj_{\tilde{\es{E}}}$ is total. Furthermore, the events $e_3$ and $e_4$ are mapped to the same
  event $[e_{3,4}]$.

  Now we verify that the size of a configuration $\tilde{x} \in \confES{\tilde{\es{E}}}$ remains
  unchanged after applying the map $proj_{\tilde{\es{E}}}$.  The set of configurations of
  $\tilde{\es{E}}$ is
  $\confES{\tilde{\es{E}}} = \{\emptyset, \{e_2\}, \{e_3\}, \{e_4\}, \{e_3, e_2\}, \{e_4, e_2\}\}$
  and the set of configurations of $\hat{\es{E}}$ is
  $\confES{\hat{\es{E}}} = \{\emptyset, \{[e_2]\}, \{[e_{3,4}]\}, \{[e_{3,4}], [e_2]\}\}$.  By
  applying $proj_{\tilde{\es{E}}}$ to each configuration of $\confES{\tilde{\es{E}}}$ we have:
  \begin{align*}
    \emptyset &\stackrel{proj_{\tilde{\es{E}}}}{\mapsto} \emptyset \\
    \{e_2\} &\stackrel{proj_{\tilde{\es{E}}}}{\mapsto} [e_2] \\
    \{e_3\}, \{e_4\} &\stackrel{proj_{\tilde{\es{E}}}}{\mapsto} [e_{3,4}] \\
    \{e_3, e_2\}, \{e_4, e_2\} &\stackrel{proj_{\tilde{\es{E}}}}{\mapsto} \{[e_{3,4}], [e_2]\}
  \end{align*}
  We can then conclude that the map $proj_{\tilde{\es{E}}}$ preserves the size of the configurations
  in $\confES{\tilde{\es{E}}}$.

  Now it lacks to see Lemma~\ref{lem:help-sum} in action.  For this we focus on the case where
  \[
    \{e_3, e_2\}, \{e_4, e_2\} \stackrel{proj_{\tilde{\es{E}}}}{\mapsto} \{[e_{3,4}], [e_2]\}
  \]

  First we expand the equivalence classes: $\{\{e_3, e_4\}, \{e_2\}\}$.  Next, we calculate the set
  of configurations formed by picking one event from each equivalence class of
  $\{[e_{3,4}], [e_2]\}$:
  \[
    \set{ \set{\tilde{e}_1, \dots, \tilde{e}_n} \mid \forall i, \tilde{e}_i \in [e_i] }
    =
    \{\{e_3, e_2\}, \{e_4, e_2\}\}
  \]
  From this, we can see that the configurations we obtained are precisely the ones that map to
  $\{[e_{3,4}], [e_2]\}$ after applying $proj_{\tilde{\es{E}}}$.
\end{example}

To prove our claim, we need one more auxiliary lemma.  The idea behind the following lemma is: if we
can show that $\dropc{n}{\tilde{v}}{\emptyset}{\set{e_1}, \dots, \set{e_n}} \geq 0$, where
$\tilde{v}(\tilde{x}) = \dfrac{v(y \cup \tilde{x})}{v(y)}$, then it follows directly that
$\dropc{n}{v}{y}{y \cup \set{e_1}, \dots, y \cup \set{e_n}} \geq 0$.  In other words, the condition
in Definition~\ref{def:prob-es} is satisfied.
\begin{lemma}\label{lem:qes-to-ppes-aux}
  Consider $\tilde{\es{U}}_y$ together with an initial state $\rho_y$, such that all the events of
  $\tilde{\es{U}}_y$ are projections.  For any $\tilde{x} \in \confES{\tilde{\es{U}}_y}$ let
  $\tilde{v}(\tilde{x}) = \dfrac{v(y \cup \tilde{x})}{v(y)}$.  Then
  $\dropc{n}{\tilde{v}}{\emptyset}{\set{e_1}, \dots, \set{e_n}} \geq 0$.
\end{lemma}
\begin{proof}
  To show $\dropc{n}{\tilde{v}}{\emptyset}{\set{e_1}, \dots, \set{e_n}} \geq 0$, it is helpful to
  consider $\hat{\es{E}}$ and $proj_{\tilde{\es{E}}}$.

  Recall that $\hat{\es{E}}$ does not have events in conflict.  Hence, by~\cite[Corollary~3]{winskel14},
  we have that $\hat{\es{U}}$ with $\rho_y$ and
  $\hat{v}(\hat{x}) = Tr(A_{\hat{x}}^{\adjoint} A_{\hat{x}} \rho_y)$ is a probabilistic event
  structure.

  We need to show that
  $\hat{v}(\hat{x}) =
  Tr(A_{\hat{x}}^{\adjoint} A_{\hat{x}} \rho_y) =
  \sum_{
    \substack{\tilde{y} \in \confES{\tilde{\es{U}}} \\
      proj_{\tilde{\es{E}}_y}[\tilde{y}] = \hat{x}}} \tilde{v}(\tilde{y})$.

  We note that $\hat{Q}([e]) = U = \sum_{\tilde{e} \in [e]} \tilde{Q}(\tilde{e})$ when $|[e]| > 1$
  and that for a configuration $\hat{x} \in \confES{\hat{\es{E}}}$, the operator
  $A_{\hat{x}} = \hat{Q}([e_n]) \cdots \hat{Q}([e_1]) = \prod_{[e] \in \hat{x}} \hat{Q}([e])$.

  Let us expand $\prod_{[e] \in \hat{x}} \hat{Q}([e])$:
  \begin{align*}
    \prod_{[e] \in \hat{x}} \hat{Q}([e])
    =& \prod_{[e] \in \hat{x}}
       \left( \sum_{\tilde{e} \in [e]} \tilde{Q}(\tilde{e}) \right)
       = \prod_{i=1}^n
       \left( \sum_{\tilde{e}_i \in [e_i]} \tilde{Q}(\tilde{e}_i) \right)
       = \sum_{\substack{
       \tilde{e}_1, \cdots, \tilde{e}_n \in \tilde{E}_y \\
    \forall i, \tilde{e}_i \in [e_i]
    }}
    \left( \prod_{i=1}^n \tilde{Q}(\tilde{e}_i) \right) \\
    =& \sum_{\substack{
       \tilde{e}_1, \cdots, \tilde{e}_n \in \tilde{E}_y \\
    \forall i, \tilde{e}_i \in [e_i]
    }}
    \left(
    \prod_{\tilde{e} \in \set{\tilde{e}_1, \cdots, \tilde{e}_n}} \tilde{Q}(\tilde{e})
    \right) 
    \stackrel{
    \text{(Lemma~\ref{lem:help-sum})}
    }{=} \sum_{\substack{
          \tilde{x} \in \confES{\tilde{\es{U}}_y} \\
    proj_{\tilde{\es{E}}_y}[\tilde{x}] = \hat{x}
    }} \left(
    \prod_{\tilde{e} \in \tilde{x}} \tilde{Q}(\tilde{e})
    \right)
    = \sum_{\substack{
    \tilde{x} \in \confES{\tilde{\es{U}}_y} \\
    proj_{\tilde{\es{E}}_y}[\tilde{x}] = \hat{x}
    }} A_{\tilde{x}} 
  \end{align*}
  
  Then it follows directly that:
  \begin{align*}
    \hat{v}(\hat{x})
    &= Tr(A_{\hat{x}}^{\adjoint} A_{\hat{x}} \rho_y)
      = Tr \left(
      \left(
      \sum_{\substack{
      \tilde{y} \in \confES{\tilde{\es{U}}_y} \\
    proj_{\tilde{\es{E}}_y}[\tilde{y}] = \hat{x}}} A_{\tilde{y}}^{\adjoint}
    \right)
    \left(
    \sum_{\substack{
    \tilde{y}' \in \confES{\tilde{\es{U}}_y} \\
    proj_{\tilde{\es{E}}_y}[\tilde{y}'] = \hat{x}}} A_{\tilde{y}'}
    \right)
    \rho_y
    \right) \\
    &= Tr \left(
    \sum_{\substack{
    \tilde{y} \in \confES{\tilde{\es{U}}} \\
    proj_{\tilde{\es{E}}_y}[\tilde{y}] = \hat{x}
    }}
    \sum_{\substack{
    \tilde{y}' \in \confES{\tilde{\es{U}}} \\
    proj_{\tilde{\es{E}}_y}[\tilde{y}'] = \hat{x}
    }}
    A_{\tilde{y}}^{\adjoint} A_{\tilde{y}'}
    \rho_y
    \right)
    \stackrel{(\star)}{=}
      Tr \left(
      \sum_{\substack{
      \tilde{y} \in \confES{\tilde{\es{U}}} \\
    proj_{\tilde{\es{E}}_y}[\tilde{y}] = \hat{x}}} A_{\tilde{y}}^{\adjoint} A_{\tilde{y}}
    \rho_y
    \right) \\
    &= \sum_{\substack{
      \tilde{y} \in \confES{\tilde{\es{U}}} \\
    proj_{\tilde{\es{E}}_y}[\tilde{y}] = \hat{x}}}
    \left(
    Tr(A_{\tilde{y}}^{\adjoint} A_{\tilde{y}} \rho_y)
    \right)
    = \sum_{
      \substack{\tilde{y} \in \confES{\tilde{\es{U}}} \\
    proj_{\tilde{\es{E}}_y}[\tilde{y}] = \hat{x}}} \tilde{v}(\tilde{y}) 
  \end{align*}

  Where in step $(\star)$ we note that if $\tilde{y} \neq \tilde{y}'$ then
  $A_{\tilde{y}}^\adjoint A_{\tilde{y}'} = 0$.  That is straightforward to see because
  $proj_{\tilde{\es{E}}_y}[\tilde{y}] = \hat{x} = proj_{\tilde{\es{E}}_y}[\tilde{y}']$. Hence it
  exists $\tilde{e} \in \tilde{y}$ and $\tilde{e}' \in \tilde{y}'$ such that
  $\tilde{Q}(\tilde{e}) \cdot \tilde{Q}(\tilde{e}') = 0$. In other words, $\tilde{e}$ and
  $\tilde{e}'$ are in conflict.

  Now we are ready to show that
  $\dropc{n}{\tilde{v}}{\emptyset}{\set{e_1}, \dots, \set{e_n}} \geq 0$.

  By\cite[Proposition~1]{winskel14},
  \begin{align*}
    \dropc{n}{\tilde{v}}{\emptyset}{\set{e_1}, \dots, \set{e_n}}
    &= \tilde{v}(\emptyset) - \sum_{\emptyset \neq I \subseteq \set{1, \dots, n}}
      (-1)^{|I|+1} \tilde{v} \left( \bigcup_{i \in I} \set{e_i} \right) \\
    &= \sum_{I \subseteq \set{1, \dots, n}} (-1)^{|I|} \tilde{v} \left(\bigcup_{i \in I} \set{e_i} \right) 
  \end{align*}

  Now we note that it exists events that are in conflict, and since the union of events that are in
  conflict do not form a configuration, we have that its valuation is zero.
  We can then remove those terms from the sum.
  \begin{align*}
    & \sum_{I \subseteq \set{1, \dots, n}} (-1)^{|I|} \tilde{v} \left(\bigcup_{i \in I} \set{e_i} \right) \\
    &= \sum_{\substack{
      I \subseteq \set{1, \dots, n} \\
    \forall i,j \in I\, .\, \econc{e_i}{e_j}
    }}
    (-1)^{|I|} \tilde{v} \left(\bigcup_{i \in I} \set{e_i} \right) \\
    &= \sum_{\substack{
      I \subseteq \set{1, \dots, n} \\
    \bigcup_{i \in I} \set{e_i} \in \confES{\tilde{\es{U}}}
    }}
    (-1)^{|I|} \tilde{v} \left(\bigcup_{i \in I} \set{e_i} \right) \\
  \end{align*}

  On the other side, on $\hat{\es{U}}$, we have:
  \begin{align*}
    \dropc{k}{\hat{v}}{\emptyset}{\hat{x}_1, \dots, \hat{x}_k}
    &= \hat{v}(\emptyset) -
      \sum_{\emptyset \neq J \subseteq \set{1, \dots, k}}
      (-1)^{|J| + 1} \hat{v}(\bigcup_{j \in J} \hat{x}_j) \\
    &= \sum_{J \subseteq \set{1, \dots, k}} (-1)^{|J|} \hat{v}(\bigcup_{j \in J} \hat{x}_j) \\
    &= \sum_{J \subseteq \set{1, \dots, k}} (-1)^{|J|} \hat{v}(\bigcup_{j \in J} \set{\hat{e}_j}) \\
    &= \sum_{\substack{
      J \subseteq \set{1, \dots, k}\\
    \tilde{y} \in \confES{\tilde{\es{U}}_y}\\
    proj_{\tilde{\es{E}}_y}[\tilde{y}] = \bigcup_{j \in J} \set{\hat{e}_j}
    }} (-1)^{|J|} \tilde{v}(\tilde{y}) \\
    &= \sum_{\substack{
      J \subseteq \set{1, \dots, k}\\
    I \subseteq \set{1, \dots, n}\\
    \bigcup_{i \in I} \set{e_i} \in \confES{\tilde{\es{U}}_y}\\
    proj_{\tilde{\es{E}}_y}[\bigcup_{i \in I} \set{e_i}] = \bigcup_{j \in J} \set{\hat{e}_j}
    }} (-1)^{|J|} \tilde{v}(\bigcup_{i \in I} \set{e_i})
  \end{align*}

  By Lemma~\ref{lem:map-proj} we know that
  $|proj_{\tilde{\es{E}}_y}[\bigcup_{i \in I} \set{e_i}]| = |\bigcup_{i \in I} \set{e_i}| = |I|$.
  Furthermore we also know that $|\bigcup_{j \in J} \set{e_j}| = |J|$.  Since
  $proj_{\tilde{\es{E}}_y}[\bigcup_{i \in I} \set{e_i}] = \bigcup_{j \in J} \set{\hat{e}_j}$, again
  by Lemma~\ref{lem:map-proj} we have
  $|proj_{\tilde{\es{E}}_y}[\bigcup_{i \in I} \set{e_i}]| = |\bigcup_{j \in J} \set{\hat{e}_j}|$.
  Thus $|I| = |J|$.  Hence
  \begin{align*}
    &\sum_{\substack{
      J \subseteq \set{1, \dots, k}\\
    I \subseteq \set{1, \dots, n}\\
    \bigcup_{i \in I} \set{e_i} \in \confES{\tilde{\es{U}}_y}\\
    proj_{\tilde{\es{E}}_y}[\bigcup_{i \in I} \set{e_i}] = \bigcup_{j \in J} \set{\hat{e}_j}
    }} (-1)^{|J|} \tilde{v}(\bigcup_{i \in I} \set{e_i})
    =
    \sum_{\substack{
    J \subseteq \set{1, \dots, k}\\
    I \subseteq \set{1, \dots, n}\\
    \bigcup_{i \in I} \set{e_i} \in \confES{\tilde{\es{U}}_y}\\
    proj_{\tilde{\es{E}}_y}[\bigcup_{i \in I} \set{e_i}] = \bigcup_{j \in J} \set{\hat{e}_j}
    }} (-1)^{|I|} \tilde{v}(\bigcup_{i \in I} \set{e_i}) \\
    &=
      \sum_{\substack{
      I \subseteq \set{1, \dots, n}\\
    \bigcup_{i \in I} \set{e_i} \in \confES{\tilde{\es{U}}_y}
    }} \left(
    \sum_{\substack{
    J \subseteq \set{1, \dots, k}\\
    proj_{\tilde{\es{E}}_y}[\bigcup_{i \in I} \set{e_i}] = \bigcup_{j \in J} \set{\hat{e}_j}
    }} (-1)^{|I|} \tilde{v}(\bigcup_{i \in I} \set{e_i})
    \right)
    \stackrel{(\star)}{=}
    \sum_{\substack{
    I \subseteq \set{1, \dots, n}\\
    \bigcup_{i \in I} \set{e_i} \in \confES{\tilde{\es{U}}_y}
    }}  (-1)^{|I|} \tilde{v}(\bigcup_{i \in I} \set{e_i}) \\
    &=
      \dropc{n}{\tilde{v}}{\emptyset}{\tilde{x}_1, \dots, \tilde{x}_n}
  \end{align*}

  In step $(\star)$ the sum no longer depends on $J$, hence we drop it.

  We shown that
  $ \dropc{n}{\tilde{v}}{\emptyset}{\tilde{x}_1, \dots, \tilde{x}_n} =
  \dropc{k}{\hat{v}}{\emptyset}{\hat{x}_1, \dots, \hat{x}_k}$.

  Hence $\dropc{n}{\tilde{v}}{\emptyset}{\tilde{x}_1, \dots, \tilde{x}_n} \geq 0$.  
\end{proof}

Now we are ready to extend~\cite[Theorem~3]{winskel14}.
\begin{proposition}\label{lem:winskel-extension}
  Let $\es{U} = \qpes{E}{\leq}{\#}{Q}$ be a unitary event structure with initial state $\rho$.  For
  each $x \in \confES{\es{U}}$ let $v(x) = Tr(\rho_x) = Tr(A_x^\adjoint A_x \rho)$.  Then
  $\es{U} = \ppes{E}{\leq}{\#}{v}$ is a probabilistic event structure.
\end{proposition}
\begin{proof}
  From~\cite[Proposition~5]{winskel14} we need to show $\dropc{n}{v}{y}{x_1, \dots x_n}$ when
  $y \cchain{e_1, \dots, e_n} x_1, \dots, x_n$.  We then identify the following cases:
  \begin{enumerate}
  \item $v(\emptyset) = 1$
  \item $\exists e_i \in e_1, \dots, e_n$ such that $Q(e_i)$ is a unitary
  \item $\forall e_i \in e_1, \dots, e_n$ we have $Q(e_i)$ is a projection
    \begin{enumerate}
    \item all the events are concurrent
    \item all events are in conflict
    \item there are events in conflict
    \end{enumerate}
  \end{enumerate}

  The proof of $1$, $2$, and $3.$(a) can be found in~\cite[Theorem~3]{winskel14}.  Thus, we only
  show the proof of $3.$(b) and $3.$(c).
  \begin{itemize}
  \item[$3.$(b)] Case every event is in conflict we know that the sum of the associated quantum
    operators is a unitary. Hence we are in case $2.$ and consequently
    $\dropc{n}{v}{y}{x_1, \dots x_n} = 0$.

  \item[$3$.(c)] Case there are events in conflict :
    \begin{align*}
      &\dropc{n}{v}{y}{x_1, \dots x_n} \\
      =& v(y) -
         \sum_{\emptyset \neq I \subseteq \set{1, \dots, n}} (-1)^{|I| + 1} v(\bigcup_{i \in I} x_i) \\
      =& v(y) - \left(
         \sum_{\substack{
         \emptyset \neq I \subseteq \set{1, \dots, n} \\ |I| = 1
      }}
      (-1)^{|I| + 1}
      v(\bigcup_{i \in I} x_i)
      +
      \sum_{\substack{
      \emptyset \neq I \subseteq \set{1, \dots, n} \\ |I| > 1
      }}
      (-1)^{|I| + 1}
      v(\bigcup_{i \in I} x_i)
      \right) \\
      =& v(y) - \sum_{i=1}^n v(x_i) -
         \sum_{\substack{
         \emptyset \neq I \subseteq \set{1, \dots, n} \\ |I| > 1
      }}
      (-1)^{|I| + 1}
      v(\bigcup_{i \in I} x_i)
    \end{align*}

    Focus on
    \[
      \sum_{\substack{
          \emptyset \neq I \subseteq \set{1, \dots, n} \\ |I| > 1
        }}
      (-1)^{|I| + 1}
      v(\bigcup_{i \in I} x_i)
    \]

    Since $x_i = \set{e_i} \cup y$ then
    \[
      \sum_{\substack{
          \emptyset \neq I \subseteq \set{1, \dots, n} \\ |I| > 1
        }}
      (-1)^{|I| + 1}
      v(\bigcup_{i \in I} x_i) =
      \sum_{\substack{
          \emptyset \neq I \subseteq \set{1, \dots, n} \\ |I| > 1
        }}
      (-1)^{|I| + 1}
      v(\bigcup_{i \in I} \set{e_i} \cup y)
    \]

    We know that with $|I| > 1$ we are not considering singletons. Hence we are either making the
    union of events that are concurrent or are in conflict. W.l.o.g consider
    $v(\set{e_j, e_k} \cup y)$ with $1 \leq j \neq k \leq n$. Case $e_j \# e_k$ then we know that
    $Q(e_j) \cdot Q(e_k) = 0 = Q(e_k) \cdot Q(e_j)$ and consequently we have
    $v(\set{e_j, e_k} \cup y) = Tr(A_y^\adjoint \cdot (Q(e_j) \cdot Q(e_k))^\adjoint \cdot Q(e_j)
    \cdot Q(e_k) \cdot A_y \rho) = 0$.  On the other side, case $\econc{e_j}{e_k}$ then we know that
    $Q(e_j) \cdot Q(e_k) = Q(e_k) \cdot Q(e_j)$ and consequently $v(\set{e_j, e_k} \cup y) \geq 0$.

    When events are in conflict their contribution to the sum is null, hence we can discard them. As
    a consequence, the sum is composed of elements that are concurrent.
    Hence we have
    \begin{align*}
      & v(y) - \sum_{i=1}^n v(x_i) -
        \sum_{\substack{
        \emptyset \neq I \subseteq \set{1, \dots, n} \\ |I| > 1
      }}
      (-1)^{|I| + 1}
      v(\bigcup_{i \in I} x_i) \\
      =& v(y) - \sum_{i=1}^n v(x_i) -
         \sum_{\substack{
         \emptyset \neq I \subseteq \set{1, \dots, n} \\ |I| > 1 \\ \bigcup_{i \in I} x_i \in \confES{\es{U}}
      }}
      (-1)^{|I| + 1}
      v(\bigcup_{i \in I} x_i) \\ 
      =& v(y) -
         \sum_{\substack{
         \emptyset \neq I \subseteq \set{1, \dots, n} \\ |I| = 1
      }}
      (-1)^{|I| + 1}
      v(\bigcup_{i \in I} x_i)
      - \sum_{\substack{
      \emptyset \neq I \subseteq \set{1, \dots, n} \\ |I| > 1 \\ \forall (i \neq j) \in I\, .\, \econc{e_i}{e_j}
      }}
      (-1)^{|I| + 1}
      v(\bigcup_{i \in I} x_i) \\
      =& v(y) -
         \sum_{\substack{
         \emptyset \neq I \subseteq \set{1, \dots, n} \\ \forall (i \neq j) \in I\, .\, \econc{e_i}{e_j}
      }}
      (-1)^{|I| + 1}
      v(\bigcup_{i \in I} x_i)
    \end{align*}
  \end{itemize}

  Despite removing the valuations from ill-configurations, we are not in case $3.$(a) since there
  are still events in conflict.
  We thus resort to Lemma~\ref{lem:qes-to-ppes-aux}.
  Concretely:
  \begin{align*}
    & \dropc{n}{\tilde{v}}{\emptyset}{\tilde{x}_1, \dots, \tilde{x}_n} \geq 0 \\
    \Leftrightarrow & \tilde{v}(\emptyset) -
                      \sum_{\substack{
                      \emptyset \neq I \subseteq \set{1, \dots, n} \\
    \bigcup_{i \in I} \set{e_i} \in \confES{\tilde{E}_y}
    }} (-1)^{|I| + 1} \tilde{v}(\bigcup_{i \in I} \set{e_i}) \geq 0 \\
    \Leftrightarrow & \sum_{\substack{
                      I \subseteq \set{1, \dots, n} \\
    \bigcup_{i \in I} \set{e_i} \in \confES{\tilde{E}_y}
    }} (-1)^{|I|} \tilde{v}(\bigcup_{i \in I} \set{e_i}) \geq 0 \\
    \Leftrightarrow & \sum_{\substack{
                      I \subseteq \set{1, \dots, n} \\
    \bigcup_{i \in I} \set{e_i} \in \confES{\tilde{E}_y}
    }} (-1)^{|I|} \dfrac{v(y \cup \bigcup_{i \in I} \set{e_i})}{v(y)} \geq 0 \\
    \Leftrightarrow & \sum_{\substack{
                      I \subseteq \set{1, \dots, n} \\
    \bigcup_{i \in I} \set{e_i} \in \confES{\tilde{E}_y}
    }} (-1)^{|I|} v(y \cup \bigcup_{i \in I} \set{e_i}) \geq 0 \\
    \Leftrightarrow& \dropc{n}{v}{y}{x_1, \dots, x_n} \geq 0
  \end{align*}

  With all cases proved, we have that $\es{U} = \ppes{E}{\leq}{\#}{v}$ is a probabilistic event
  structure.
\end{proof}

Hence we proved that any unitary event structure with an initial state is a probabilistic event
structure.

We can now show soundness and adequacy having in mind initial states.  The intuition behind the
n-step in Section~\ref{sec:qes} is: given a command $\com{C}$ and a list of instructions, which is a
word, $a: \omega'$ we reach a command $\com{C'}$ in n-steps.  If we give an initial state to
$\com{C}$, the evolution of the state will correspond to the application of the word to the initial
state.  We define a state to be a partial density operator, \ie\ a density operator whose trace is
less or equal to one and denote the set of partial density operator as
$\mathcal{D}_{\leq 1}(\mathcal{H})$.  We now define how a word is applied to a state:

\begin{definition}
  Let $\omega$ be a word and $\rho$ a partial density operator.
  Define $\omega(\rho)$ inductively as follows:
  \[
    \omega(\rho) =
    \begin{cases}
      a \rho a^\adjoint & \text{ if } \omega = a \\
      \omega' (a \rho a^\adjoint) & \text{ if } \omega = a:\omega' \\
    \end{cases}
  \]
\end{definition}

Note that when applying a word $\omega$ to a state $\rho$, the first action to be applied on $\rho$
is the head of $\omega$.

Following Lemma~\ref{lem:winskel-extension} we can define a new denotational semantics for quantum
event structures who takes into consideration initial states.
\begin{definition}\label{def:den-sem-init-state}
  We define
  ($\mf{-}_{\_} : \com{C} \times \mathcal{D}_{\leq 1}(\mathbb{C}^2) \rightarrow
  \pes{\es{U}}{\rho}{v}$, where $\es{U}$ is a unitary event structure) as follows:
  \begin{align*}   
    & \mf{\com{skip}}_\rho = \pes{\mf{\com{skip}}}{\rho}{v(\set{sk}) = 1} \\
    & \mf{\com{U}(\vec{n})}_\rho = \pes{\mf{\com{U}(\vec{n})}}{\rho}{v(\set{U_{\vec{n}}}) = 1} \\
    & \mf{\meas{n}{\com{C}_1}{\com{C}_2}}_\rho = \pes{\mf{\meas{n}{\com{C}_1}{\com{C}_2}}}{\rho}{v} \\
    & \mf{\seq{\com{C_1}}{\com{C_2}}}_\rho = \pes{\mf{\seq{\com{C_1}}{\com{C_2}}}}{\rho}{v}\\
    & \mf{\conc{\com{C_1}}{\com{C_2}}}_\rho = \pes{\mf{\conc{\com{C_1}}{\com{C_2}}}}{\rho}{v}
  \end{align*}
\end{definition}

At this point we can establish an equivalence between the semantics with or without initial state.
However we note that for the former we first need to show the equivalence without initial state.

Consider the case without initial state. We observe that both the operational and denotational
semantics presented in this section closely resemble those developed in
Section~\ref{sec:es}. Consequently, the results obtained in Section~\ref{subsec:results-1}
can be straightforwardly adapted to the quantum setting. It is worth to emphasize that removing an
initial element form $\meas{n}{\es{U}_1}{\es{U}_2}$ is equal to $\es{U}_1$ or to $\es{U}_2$ if the
event removed is $\tau_0^n$ or $\tau_1^n$, respectively.

To show the equivalence of the semantics with an initial state we state the following.
\begin{theorem}[Soundness]\label{res:soundII-3-init-state}
  Let $\rho$ be an initial state.  If $\com{C} \xtwoheadrightarrow{\omega} \com{C'}$ then
  $\exists x \in \mathcal{C}(\mf{\com{C}}_\rho)$ such that $\emptyset \stackrel{\omega}{\chain\ } x$ and
  $v(x) = \omega(\rho)$.
\end{theorem}
\begin{theorem}[Adequacy]\label{res:adII-3-init-state}
  Let $\rho$ be an initial state.  If $(x \neq \emptyset) \in \mathcal{C}(\mf{\com{C}}_\rho)$ s.t.
  $\emptyset \stackrel{\omega}{\chain\ } x$ then $\exists \com{C'}$ s.t.
  $\com{C} \xtwoheadrightarrow{\omega} \com{C'}$ and $v(x) = Tr(\omega(\rho))$.
\end{theorem}

To prove the above statements we make use of Theorem~\ref{res:soundII-3} and
Theorem~\ref{res:adII-3}, respectively.  What is left to show is that $v(x) =
Tr(\omega(\rho))$. However that comes freely because the operations applied on
$\xtwoheadrightarrow{\omega}$ and on $\emptyset \cchain{\omega} x$ are the same.

\begin{example}\label{ex:small-3}
  In Figure~\ref{fig:ex2-3}, we have the labeled transition system of
  $\seq{\com{H}(n)}{\meas{n}{\com{X}(n)}{\com{Z}(n)}}$. This program applies first the Hadamard gate
  to qubit $n$ and then measures it. If the measurement was made by $P_0^n$ then we apply the $X$
  gate to qubit $n$ and we are done. On the other side, if the measurement was performed by $P_1^n$
  we apply the $Z$ gate to qubit $n$ finishing the computation.  With the help of
  Figure~\ref{fig:ex2-3}, it is straightforward to see that the words that lead to a terminal
  command are: $H(n) P_0^n X(n)$ and $H(n) P_1^n Z(n)$.  By applying each word to the state
  $\rho = \ket{0}\bra{0}$, we obtain the following possible final states:
  $(H(n) P_0^n X(n)) (\ket{0}\bra{0}) = \dfrac{1}{2} \ket{1}\bra{1}$ and
  $(H(n) P_1^n Z(n)) (\ket{0}\bra{0}) = \dfrac{1}{2} \ket{1}\bra{1}$.
  \begin{figure}[ht!]
    \centering
    \begin{tikzpicture}
      \begin{pgfonlayer}{nodelayer}
        \node [style=command] (0) at (0, 0) { $\seq{\com{H}(n)}{\meas{n}{\com{X}(n)}{\com{Z}(n)}}$ };
        \node [style=command] (1) at (0, -1.25) {$\meas{n}{\com{X}(n)}{\com{Z}(n)}$};
        \node [style=command] (2) at (-1, -2.5) {$\com{X}(n)$};
        \node [style=command] (3) at (1, -2.5) {$\com{Z}(n)$};
        \node [style=command] (4) at (-1, -3.75) {$\checkmark$};
        \node [style=command] (5) at (1, -3.75) {$\checkmark$};
        \node [style=command] (6) at (0.5, -0.75) {$H(n)$};
        \node [style=command] (7) at (-1, -2) {$P_0^n$};
        \node [style=command] (8) at (1, -2) {$P_1^n$};
        \node [style=command] (9) at (-1.5, -3.25) {$X(n)$};
        \node [style=command] (10) at (1.5, -3.25) {$Z(n)$};
      \end{pgfonlayer}
      \begin{pgfonlayer}{edgelayer}
        \draw [style=op] (0) to (1);
        \draw [style=op] (1) to (2);
        \draw [style=op] (2) to (4);
        \draw [style=op] (1) to (3);
        \draw [style=op] (3) to (5);
      \end{pgfonlayer}
    \end{tikzpicture}
    \caption{Labeled transition system of $\seq{\com{H}(n)}{\meas{n}{\com{X}(n)}{\com{Z}(n)}}$}
    \label{fig:ex2-3}
  \end{figure}
\end{example}

\subsection{Introducing cyclic behavior}

Differently from what was done in Section~\ref{subsec:cyclic-beh-1} and
Section~\ref{subsec:cyclic-beh-2}, here the cyclic behavior will not be given by recursion. Instead
it will be given by a while loop. By doing this we still manage to keep the intended philosophy in
the operational semantics, because the while loop is defined in terms of a measurement. In other
words, Section~\ref{subsec:lan-3} already has all that we need to implement the while loop.

The set of commands allowed by the language are given by the following grammar:
\[
  \com{C} ::= \com{skip} \mid U(\vec{n}) \mid \seq{\com{C}}{\com{C}} \mid \meas{n}{\com{C}_1}{\com{C}_2}
  \mid \conc{\com{C}}{\com{C}} \mid \qwhi{n}{\com{C}}
\]
where $U(\vec{n})$ applies the unitary gate $U$ to the qubits presented in $\vec{n}$, the parallel
composition is disjoint~\footnote{$\conc{\com{C}_1}{\com{C}_2}$ being disjoint means that
  $\com{C}_1$ and $\com{C}_2$ do not share any qubit} $\meas{n}{\com{C}_1}{\com{C}_2}$ represents
the measurement of a qubit $n$ such that if the measurement is made by $P_0^n$ then we execute
$\com{C}_1$, else if the measurement is made by $P_1^n$ then we execute $\com{C}_2$, and
$\qwhi{n}{\com{C}}$ is a while loop that stops the computation if the measurement is made by
$P_0^n$. Note that the behavior of $\meas{n}{\com{C}_1}{\com{C}_2}$ is similar to that of a
classical if clause.

\begin{remark}
  In this section, we used a while loop instead of a recursive command for cyclic behavior, unlike
  Sections~\ref{subsec:lan-1} and \ref{subsec:lan-2}. The reason to opt by a while loop in this
  section comes from the behavior of a measurement resembling an if-then-else command. Furthermore
  projections decide if the computation stops or continues, allowing us to implement the while loop
  without needing a notion of state. On the other hand, implementing the while loop in
  Sections~\ref{subsec:lan-1} and \ref{subsec:lan-2} would require a notion of state associated with
  the command, which would obliged us to change the operational semantics we have designed without
  loops.
\end{remark}

The set of qubits being used in a command $\com{C}$ is defined as follows:
\begin{align*}
  & \qvar{\com{skip}} = \emptyset \\
  & \qvar{\com{U}(\vec{n})} = \vec{n} \\
  & \qvar{\meas{n}{\com{C}_1}{\com{C}_2}} = \set{n} \cup \qvar{\com{C}_1} \cup \qvar{\com{C}_2} \\
  & \qvar{\seq{\com{C}_1}{\com{C}_2}} = \qvar{\com{C}_1} \cup \qvar{\com{C}_2} \\ 
  & \qvar{\conc{\com{C}_1}{\com{C}_2}} = \qvar{\com{C}_1} \cup \qvar{\com{C}_2} \\
  & \qvar{\qwhi{n}{\com{C}}} = \set{n} \cup \qvar{\com{C}}
\end{align*}

We add the following rules to Figure~\ref{fig:op-small3}.
\begin{align*}
  \qwhi{n}{\com{C}} \xrightarrow{P_0^n} \checkmark
  \qquad
  \qwhi{n}{\com{C}} \xrightarrow{P_1^n} \seq{\com{C}}{\qwhi{n}{\com{C}}}
\end{align*}

\begin{example}\label{ex:loop-3}
  Figure~\ref{fig:ex-loop-3} illustrates the behavior of a quantum toss coin, which, similarly
  to Example~\ref{ex:loop-1}, produces a possibly empty sequence $H(n) P_1^n$ that finishes with
  $P_0^n$. To understand this we observe that the initial program has two possible transitions:
  (1) transits through $P_0^n$ and the computation finishes; (2) transits through $P_1^n$ to
  $\seq{\com{H}(n)}{\qwhi{n}{\com{H}(n)}}$, which executes $\com{H}(n)$ to transit to
  $\qwhi{n}{\com{H}(n)}$, which is the same command as the initial one.
  \begin{figure}[ht!]
    \centering
    \begin{tikzpicture}
      \begin{pgfonlayer}{nodelayer}
        \node [style=command] (0) at (0, 0) {$\qwhi{n}{\com{H}(n)}$};
        \node [style=command] (1) at (-1.75, -1.25) {$\checkmark$};
        \node [style=command] (2) at (1, -1.25) {$\seq{\com{H}(n)}{\qwhi{n}{\com{H}(n)}}$};
        \node [style=none] (3) at (-1.75, -0.75) {$P_0^n$};
        \node [style=none] (4) at (1, -0.75) {$P_1^n$};
        \node [style=command] (5) at (1, -2.5) {$\qwhi{n}{\com{H}(n)}$};
        \node [style=command] (6) at (-1, -3.75) {$\checkmark$};
        \node [style=command] (7) at (2, -3.75) {$\seq{\com{H}(n)}{\qwhi{n}{\com{H}(n)}}$};
        \node [style=none] (8) at (-1, -3.25) {$P_0^n$};
        \node [style=none] (9) at (2, -3.25) {$P_1^n$};
        \node [style=command] (10) at (2, -5) {$\qwhi{n}{\com{H}(n)}$};
        \node [style=command] (11) at (0, -6.25) {$\checkmark$};
        \node [style=command] (12) at (3, -6.25) {$\seq{\com{H}(n)}{\qwhi{n}{\com{H}(n)}}$};
        \node [style=none] (13) at (0, -5.75) {$P_0^n$};
        \node [style=none] (14) at (3, -5.75) {$P_1^n$};
        \node [style=none] (15) at (1.5, -2) {$H(n)$};
        \node [style=none] (16) at (2.5, -4.5) {$H(n)$};
        \node [style=command] (17) at (3, -7.5) {$\vdots$};
        \node [style=none] (18) at (3.5, -7) {$H(n)$};
      \end{pgfonlayer}
      \begin{pgfonlayer}{edgelayer}
        \draw [style=op] (0) to (1);
        \draw [style=op] (0) to (2);
        \draw [style=op] (5) to (6);
        \draw [style=op] (5) to (7);
        \draw [style=op] (10) to (11);
        \draw [style=op] (10) to (12);
        \draw [style=op] (2) to (5);
        \draw [style=op] (7) to (10);
        \draw [style=op] (12) to (17);
      \end{pgfonlayer}
    \end{tikzpicture}
    \caption{Fragment of the execution of $\qwhi{n}{\com{H}(n)}$}
    \label{fig:ex-loop-3}
  \end{figure}
\end{example}

\begin{definition}\label{def:pes-fix-order-3}  
  Let $\es{U}_1 = \qpes{E_1}{\leq_1}{\#_1}{Q_1}$ and $\es{U}_2 = \qpes{E_2}{\leq_2}{\#_2}{Q_2}$ be
  unitary event structures.  Say $\es{U}_1 \trianglelefteq \es{U}_2$ if:
  \begin{align*}
    & E_1 \subseteq E_2 \\
    & \forall e,e'\ .\ e \leq_1 e'
      \Leftrightarrow
      e, e' \in E_1 \wedge e \leq_2 e' \\
    & \forall e,e'\ .\ e \#_1 e'
      \Leftrightarrow
      e, e' \in E_1 \wedge e \#_2 e' \\
    & \forall e \in E_1\, .\, Q_1(e) = Q_2(e)
  \end{align*}
\end{definition}

The following lemmas confirm that Definition~\ref{def:pes-fix-order-3} is a partial order and that
it has a least element.
\begin{lemma}\label{lem:po-3}
  $\trianglelefteq$ is a partial order.
\end{lemma}

\begin{lemma}\label{lem:po-least-elem-3}
  Define $\bot = \qpes{\emptyset}{\emptyset}{\emptyset}{! : \emptyset \rightarrow Op(\mathcal{H})}$.
  $\bot$ is the least element of $\trianglelefteq$.
\end{lemma}

We now extend Definition~\ref{def:lub-1} to the quantum setting.
\begin{definition}\label{def:lub-3}  
  Let $\es{U}_1 \trianglelefteq \dots \trianglelefteq \es{U}_n \trianglelefteq \dots$ be a
  $\omega$-chain. Let $\es{U}^{\omega} = \qpes{E^\omega}{\leq^\omega}{\#^\omega}{Q^\omega}$ be its
  least upper bound where:
  \begin{itemize}
  \item $E^\omega = \cup_{n \in \omega} E_n$
  \item $\leq^\omega = \cup_{n \in \omega} \leq_n$
  \item $\#^\omega = \cup_{n \in \omega} \#_n$
  \item $Q^\omega(e) \Leftrightarrow \exists n \in \omega\, .\, e \in E_n$ and $Q_n(e) = Q^\omega(e)$
  \end{itemize}
\end{definition}

We then show that the structure in Definition~\ref{def:lub-3} is a unitary event structure and a
least upper bound in a chain of unitary event structures.
\begin{lemma}\label{lem:lub-es-3}
  $\es{U}^\omega$ is a unitary event structure.
\end{lemma}

\begin{lemma}\label{lem:lub-3}
  Let $\es{U}_1 \trianglelefteq \dots \trianglelefteq \es{U}_n \trianglelefteq \dots$ be a
  $\omega$-chain. Then $\es{U}^\omega$ is its least upper bound.
\end{lemma}

We have that sequential, concurrent, and measurement composition are monotone and continuous with
respect to Definition~\ref{def:pes-fix-order-3}. Recall that the monotone and continuous definitions
are given by Definition~\ref{def:cont-1}.
\begin{lemma}\label{lem:seq-fix-mono-3}
  Let $\es{U}, \es{U}_1, \es{U}_2$ be unitary event structures.  If
  $\es{U}_1 \trianglelefteq \es{U}_2$ then
  $\seq{\es{U}}{\es{U}_1} \trianglelefteq \seq{\es{U}}{\es{U}_2}$.
\end{lemma}

\begin{lemma}\label{lem:conc-fix-mono-3}
  Let $\es{U}_1, \es{U}'_1, \es{U}_2, \es{U}'_2$ be unitary event structures.  If
  $\es{U}_1 \trianglelefteq \es{U}'_1$ and $\es{U}_2 \trianglelefteq \es{U}'_2$ then
  $\conc{\es{U}_1}{\es{U}_2} \trianglelefteq \conc{\es{U}'_1}{\es{U}'_2}$.
\end{lemma}

\begin{lemma}\label{lem:meas-fix-mono-3}
  Let $\es{U}_1, \es{U}'_1, \es{U}_2, \es{U}'_2$ be unitary event structures.  If
  $\es{U}_1 \trianglelefteq \es{U}'_1$ and $\es{U}_2 \trianglelefteq \es{U}'_2$ then
  $\meas{n}{\es{U}_1}{\es{U}_2} \trianglelefteq \meas{n}{\es{U}'_1}{\es{U}'_2}$.
\end{lemma}

\begin{lemma}\label{lem:seq-cont-3}
  $\bigsqcup_m(\seq{\es{U}}{\es{U}_m}) = \seq{\es{U}}{\bigsqcup_m \es{U}_m}$.
\end{lemma}

\begin{lemma}\label{lem:conc-cont-3}
  $\bigsqcup_{n,m}(\conc{\es{U}_n}{\es{U}_m}) = \conc{\bigsqcup_n \es{U}_n}{\bigsqcup_m \es{U}_m}$.
\end{lemma}

\begin{lemma}\label{lem:meas-cont-3}
  $\bigsqcup_{n,m}(\meas{q}{\es{U}_n}{\es{U}_m}) = \meas{q}{\bigsqcup_n \es{U}_n}{\bigsqcup_m \es{U}_m}$.
\end{lemma}

When adapting Lemma~\ref{lem:cont-1} (helps to show that operators defined in event structures are
continuous) and Lemma~\ref{lem:fix-prop-1} (a version of the Kleene fixed-point theorem for the case
of event structures) to the quantum realm, we notice that the proofs we need to do are analogous to
the ones done in the non-deterministic case. Hence, we omit their formulation here.

We now extend Definition~\ref{def:den-sem3} with the while command and let $\mathbb{U}$ denote the
class of unitary event structures.
\begin{definition}\label{def:den-fix-sem3}
  We interpret commands as unitary event structures as follows:
  \centering
  \begin{align*}
    & \mf{\com{skip}} = (\set{sk}, \set{sk \leq sk}, \emptyset, Q(sk) = Id) \\
    & \mf{\com{U}_{\vec{n}}} = (\set{U_{\vec{n}}}, \set{U_{\vec{n}} \leq U_{\vec{n}}}, \emptyset,
      Q(U_{\vec{n}}) = U(\vec{n})) \\
    & \mf{\meas{n}{\com{C}_1}{\com{C}_2}} =
      \seq{\es{P}_0^n}{\mf{\com{C}_1}} + \seq{\es{P}_1^n}{\mf{\com{C}_2}} \\
    & \mf{\seq{\com{C_1}}{\com{C_2}}} = \seq{\mf{\com{C_1}}}{\mf{\com{C_2}}} \\
    & \mf{\conc{\com{C_1}}{\com{C_2}}} = \conc{\mf{\com{C_1}}}{\mf{\com{C_2}}} \\
    & \mf{\qwhi{n}{\com{C}}} = fix(\Gamma^n)
  \end{align*}
  where $\Gamma^n : \es{U} \rightarrow \es{U}$ is given by
  $\Gamma^n(\es{U}) = \es{P}_0^n + \seq{\es{P}_1^{n}}{\es{U}}$.
\end{definition}

Furthermore, note that
$\mf{\qwhi{n}{\com{C}}} = \meas{n}{\checkmark}{\mf{\seq{\com{C}}{\qwhi{n}{\com{C}}}}} = \es{P}_0^n +
\seq{\es{P}_1^{n}}{\mf{\seq{\com{C}}{\qwhi{n}{\com{C}}}}}$ and $\bot = \mf{\checkmark}$. These facts
will be useful when showing the equivalence between the semantics.

We note that $\Gamma^n$ is continuous because it is composed of continuous functions.

To show the equivalence between the operational and the denotational semantics, we reuse what was
done in Section~\ref{subsec:results-3}. The only lemmas in which we need to add the proof for the
recursion case are the following:
\begin{lemma}[Soundness I]\label{res:soundI-fix-3}
  If $\com{C} \xrightarrow{l} \com{C'}$ then $\mf{\com{C'}} \equiv \mf{\com{C}}\backslash l$.
\end{lemma}

    


\begin{lemma}[Adequacy I]\label{res:adI-fix-3}
  Let $l \in \init{\mf{\com{C}}}$. Then $\exists \com{C'} \in (\com{C} \cup \{\checkmark\})$ s.t
  $\com{C} \xrightarrow{l} \com{C'}$ and $\mf{\com{C}} \backslash l \equiv \mf{\com{C'}}$.
\end{lemma}


      


Similarly to what was done in Section~\ref{subsec:results-3}, we can consider the equivalence
between semantics with an initial state. However, doing it is very similar to what we already have,
hence we postpone it.

\begin{example}
  The unitary event structure in Example~\ref{ex:qes-1} corresponds to the interpretation of the
  command in Example~\ref{ex:small-3}.

  To see the equivalence between both semantics, recall the maximal configurations in
  Example~\ref{ex:qes-1} and the words used in Example~\ref{ex:small-3}.  It is trivial to see that
  for each word we have a corresponding covering chain, and vice-versa.

  It lacks to verify the probability when an initial state is given.  Consider that the initial
  state is $\rho = \ket{0}\bra{0}$. Applying the word $H(n) P_0^n X(n)$ to $\rho$ yields a
  probability of $0.5$, which matches the probability of the respective covering chain.  Similarly,
  when we apply the word $H(n) P_1^n Z(n)$ to $\rho$, we obtain a probability of $0.5$, once again
  matching the probability of the respective covering chain.

  Conversely, if we obtained the probability from the trace of $A_x \rho$, where $x$ is a
  configuration from a covering chain, we observe that applying the respective word to $\rho$ gives
  the same probability. Concretely, the covering chain of $\set{H_1, \tau_0^1, X_1}$ is
  $\emptyset \cchain{H_1} \set{H_1} \cchain{\tau_0^1} \set{H_1, \tau_0^1} \cchain{X_1} \set{H_1,
    \tau_0^1, X_1}$. The associated operator $A_x$ is $X(1) P_0^1 H(1)$. By applying $A_x$ to $\rho$
  we obtain the state $\ket{1}\bra{1}$ with probability $0.5$, which corresponds to the probability
  of applying the respective word to $\rho$.

  


\end{example}

\begin{example}
  Figure~\ref{fig:ex3-3} shows the event structure corresponding to the interpretation of
  $\mf{\seq{\com{H}(n)}{\meas{n}{\com{skip}}{\com{X}(n)}}}$. The set of configurations is
  $\set{\emptyset, \set{H_n}, \set{H_n, \tau_0^n}, \set{H_n, \tau_1^n}, \set{H_n, \tau_0^n, sk},
    \set{H_n, \tau_1^n, X_n}}$.
  
  To see the equivalence between both semantics through an example, we first derive the words that
  can be formed by the n-step in Example~\ref{ex:small-3}: $H(n)$, $H(n) P_0^n$, $H(n) P_1^n$,
  $H(n) P_0^n sk$ and $H(n) P_1^n X(n)$.

  Each word corresponds to a covering chain, which represents a configuration. For example the words
  $H(n) P_0^n sk$ and $H(n) P_1^n X(n)$ correspond to the covering chains
  $\emptyset \cchain{H_n} \set{H_n} \cchain{\tau_0^n} \set{H_n, \tau_0^n} \cchain{sk} \set{H_n,
    \tau_0^n, sk} = x_1$ and
  $\emptyset \cchain{H_n} \set{H_n} \cchain{\tau_1^n} \set{H_n, \tau_1^n} \cchain{X_n} \set{H_n,
    \tau_1^n, X_n} = x_2$, respectively.

  Furthermore, given as initial state $\rho = \ket{0}\bra{0}$, we have the following probabilities:
  $v(x_1) = 0.5$ and $v(x_2) = 0.5$, which correspond to the probabilities obtained by respectively
  applying the words $H(n) P_0^n sk$ and $H(n) P_1^n X(n)$ to the same state, as shown in
  Example~\ref{ex:small-3}.
  \begin{figure}[ht!]
    \centering
    \begin{minipage}{0.3\textwidth}
      \begin{tikzpicture}
        \begin{pgfonlayer}{nodelayer}
          \node [style=event] (0) at (0, 0) {$H_n$};
          \node [style=event] (1) at (-0.75, -1) {$\tau_0^n$};
          \node [style=event] (2) at (0.75, -1) {$\tau_1^n$};
          \node [style=event] (3) at (-0.75, -2) {$sk$};
          \node [style=event] (4) at (0.75, -2) {$X_n$};
        \end{pgfonlayer}
        \begin{pgfonlayer}{edgelayer}
          \draw [style=arrow] (0) to (1);
          \draw [style=arrow] (1) to (3);
          \draw [style=arrow] (0) to (2);
          \draw [style=arrow] (2) to (4);
          \draw [style=wiggle] (1) to (2);
        \end{pgfonlayer}
      \end{tikzpicture}
    \end{minipage}
    \begin{minipage}{0.3\textwidth}
      \begin{align*}
        & Q(H_n) = H(n),\\
        & Q(\tau_0^1) = P_0^n,\ 
          Q(\tau_1^1) = P_1^n,\\
        & Q(sk) = Id(n),\
          Q(X_n) = X(n)
      \end{align*}
    \end{minipage}
    \caption{Event structure of $\mf{\seq{\com{H}(n)}{\meas{n}{\com{skip}}{\com{X}(n)}}}$}
    \label{fig:ex3-3}
  \end{figure}
\end{example}

\section{Related Work}\label{sec:rel-work}
Most work on event structures extend them to different computational effects and when they give
denotational semantics for a language, most of the languages include notions of communication, which
are absent in the languages we consider.

In the classical setting, Winskel used event structures to give denotational semantics to
CCS~\cite{winskel82,winskel88}. In the probabilistic setting, Varacca and Yoshida used a
probabilistic version of event structures~\cite{varacca06} to interpret a probabilistic
$\pi$-calculus~\cite{varacca06a}. Marc de Visme later adapted Winskel's probabilistic event
structures~\cite{winskel14}, equivalent to Varacca's definition, to furnish a probabilistic
CCS~\cite{baier97} with a denotational semantics. In the quantum setting event structures have only
been used as the backbone for game semantics~\cite{clairambault19}.

A closer approach to ours is found in Castellan’s work~\cite{castellan16}, where event structures
interpret a simple imperative and concurrent language in the context of weak memory models. His goal
was to capture execution paths generated by compilers during code optimization, missed by
interleaving semantics. Interestingly, his definition of sequential and parallel composition are
similar to ours.

\section{Conclusion}

In this paper, we discussed how Winskel's event structures can be tamed as a model of computation
for representing sequences of actions, with causal and conflicting relationships, and even refined
Winskel's notion of quantum event structure to better match the probabilistic ones. We show how
Winskel's event structures support non-deterministic, probabilistic and quantum effects.

\newpage
\bibliographystyle{alphaurl}
\bibliography{biblio}

\newpage
\appendix
\section*{Proofs of Section~\ref{sec:es}}\label{app-chap:proofs-es}
\subsubsection*{Proof of Lemma~\ref{lem:seq-es1}}
\begin{proof}
  Let 
  \begin{align*}
    & \es{E}_1 = \pes{E_1}{\leq_1}{\#_1} \\
    & \es{E}_2 = \pes{E_2}{\leq_2}{\#_2} \\
    & \seq{\es{E}_1}{\es{E}_2} = \pes{E}{\leq}{\#}
  \end{align*}

  We need to show that $\leq$ is a partial order and that $\#$ is symmetric and irreflexive.
  \begin{itemize}
  \item $\leq$ is a partial order:
    \begin{itemize}
    \item Reflexivity ($e \leq e$): we have two cases
      \begin{enumerate}
      \item Case $e \in E_1$. This entails $e \leq_1 e$. Since $\leq_1$ is a partial order we are
        done.
      \item Case $(e,x) \in E_2 \times \confmax{\es{E}_1}$. This entails $e \leq_2 e$. Since
        $\leq_2$ is a partial order we are done.
      \end{enumerate}
    \item Transitivity ($e \leq e'$ and $e' \leq e''$ then $e \leq e''$): if $e, e', e'' \in E_1$ or
      $e, e', e'' \in E_2$, we are done since $\leq_1$ and $\leq_2$ are partial orders. We then have
      two more cases:
      \begin{enumerate}
      \item Case $e \leq (e', x)$ and $(e',x) \leq (e'', x)$. From $e \leq (e', x)$ we have that
        $e \in x$. From $(e',x) \leq (e'', x)$ we know that $e' \leq_2 e''$ and since $x$ is the
        same for both events, ten $e \in x$. Thus $e \leq (e'', x)$.
      \item Case $e \leq e'$ and $e' \leq (e'',x)$. From $e \leq e'$ we know that $e' \in x$.  From
        $e \leq e'$ we know that $e \leq_1 e'$. Since $x$ is a maximal configuration in $\es{E}_1$
        and $e \leq_1 e'$ then $e \in x$. Thus $e \leq (e'', x)$
      \end{enumerate}
    \item Antisymmetry ($e \leq e'$ and $e' \leq e$ then $e'=e$): if $e, e', e'' \in E_1$ or
      $e, e', e'' \in E_2$, we are done since $\leq_1$ and $\leq_2$ are partial orders. Since it is
      not possible to have $(e',x) \leq e$ we are done.
    \end{itemize}

    Hence $\leq$ is a partial order.

  \item $\#$ is symmetric and irreflexive:
    \begin{itemize}
    \item Symmetric (if $e \# e'$ then $e' \# e$): we have two cases:
      \begin{enumerate}
      \item $(e_2, x) \# (e'_2, x)$ then $(e'_2, x) \# (e_2, x)$. From $(e_2, x) \# (e'_2, x)$ we
        have that $e_2 \#_2 e'_2$, which is symmetric since $\#_2$ is a partial order. Hence
        $e'_2 \#_e e_2$ and $(e'_2, x) \# (e_2, x)$.
      \item $e \# e'$ then $e' \# e$. $e \# e'$ entails the existence of $e_1 \leq e$ and
        $e'_1 \leq e'$ such that $e_1 \#_1 e'_1$. Since $\#_1$ is a partial order, it is
        symmetric. then $e'_1 \#_1 e_1$. Furthermore, since $e_1 \leq e$ and $e'_1 \leq e'$ then
        $e' \# e$.
      \end{enumerate}
    \item Irreflexive ($\neg(e \# e)$): we show this by showing that the elements of $\#$ are of the
      form $e \# e'$ with $e \neq e'$. We then have four cases:
      \begin{enumerate}
      \item Case $e, e' \in E_1$. Then we have $e_1 \leq e$, $e'_1 \leq e'$, and $e_1 \#_1 e'_1$.
        Furthermore, we have $e_1 \leq_1 e$ and $e'_1 \leq_1 e'$.  Since $e_1 \#_1 e'_1$ and
        $e_1 \leq_1 e$ we have $e'_1 \#_1 e$.  Since $e'_1 \#_1 e$ and $e'_1 \leq_1 e'$ we have
        $e \#_1 e'$.  By irreflexivity of $\#_1$ we have that $e \neq e'$.
      \item Case $(e,x), (e',x') \in E_2 \times \confmax{\es{E}_1}$. If $x \neq x'$ then it is clear
        that $(e,x) \neq (e', x')$. If $x=x'$ then we have $e_1 \leq (e,x)$ and $e'_1 \leq (e', x)$,
        which entails that $e_1, e'_1 \in x$. However, $e_1 \leq_1 e'_1$, hence $e_1, e'_1 \not\in x$.
        Thus we have a contradiction, and consequently $(e,x) \neq (e', x)$.
      \item Case $e \in E_1$ and $e' \in E_2 \times \confmax{\es{E}_1}$. They are clearly different.
      \item Case $e \in E_2 \times \confmax{\es{E}_1}$ and $e' \in E_1$. They are clearly different.
      \end{enumerate}
    \end{itemize}

    Hence $\#$ is symmetric and irreflexive.
  \end{itemize}

  We need to show that $\forall e, e', e'' \in E$:
  \begin{enumerate}
  \item $\set{e' \mid e' \leq e}$ is finite
    \begin{enumerate}
    \item Case $e \in E_1$ we are done.
    \item Case $e \in E_2 \times \confmax{\es{E}_1}$.  Then $e = (e_2, x_1)$ with $e_2 \in E_2$ and
      $x_1 \in \confmax{\es{E}_1}$.  We know that
      $\set{e' \mid e' \leq (e_2, x_1)} = \set{(e'_2,x_1) \mid (e'_2, x_1) \leq (e_2, x_1)} \cup
      \set{e_1 \mid e_1 \leq (e_2, x_1),\, e_1 \in x_1}$. Both sets are finite because $\es{E}_2$ is
      an event structure and $x_1$ is finite, respectively. Since both sets are finite and the union
      of finite sets is finite, then $\set{e' \mid e' \leq (e_2, x_1)}$ is finite.
    \end{enumerate}
  \item $e \# e' \leq e'' \Rightarrow e \# e''$
    \begin{enumerate}
    \item Case $e, e', e'' \in E_1$ or $a, a', e, e', e'' \in E_2 \times \confmax{\es{E}_1}$ we are
      done.
    \item Case $e, e' \in E_1$ and $e'' \in E_2 \times \confmax{\es{E}_1}$.  We want to show that
      $e \# e''$. Hence we have to show that
      $\exists(e_1 \leq e,\ e'_1 \leq e')\ .\ e_1 \#_1 e'_1$.  Since $e, e' \in E_1$ then
      $e_1, e'_1 \in E_1$.  Let $e_1 = e$ and $e'_1 = e'$.  Hence we have
      $e_1 \leq e \Leftrightarrow e \leq e \Leftrightarrow e \leq_1 e$.  By the initial assumption,
      $e' \leq e'' \Leftrightarrow e'_1 \leq e''$. It lacks to show that $e_1 \#_1 e'_1$. That
      follows directly from the initial assumption
      $e \# e' \Leftrightarrow e \#_1 e' \Leftrightarrow e_1 \#_1 e'_1$.
    \end{enumerate}
  \end{enumerate}
\end{proof}

\subsubsection*{Proof of Lemma~\ref{lem:conc-es1}}
\begin{proof}
  Let
  \begin{align*}
    & \es{E}_1 = \pes{E_1}{\leq_1}{\#_1} \\
    & \es{E}_2 = \pes{E_2}{\leq_2}{\#_2} \\
    & \conc{\es{E}_1}{\es{E}_2} = \pes{E}{\leq}{\#}
  \end{align*}

  We need to show that $\leq$ is a partial order and that $\#$ is symmetric and irreflexive.
  \begin{itemize}
  \item $\leq$ is a partial order:
    \begin{itemize}
    \item Reflexivity ($e \leq e$): we have two cases: $e_1 \leq_1 e'_1$ or $e_2 \leq_2 e'_2$. Since
      $\leq_1$ and $\leq_2$ are partial orders, we are done.
    \item Transitivity ($e \leq e'$ and $e' \leq e''$ then $e \leq e''$): we have two cases:
      ($e \leq_1 e'$ and $e' \leq_1 e''$) or ($e \leq_2 e'$ and $e' \leq_2 e''$). Since $\leq_1$ and
      $\leq_2$ are partial orders, we have $e \leq_1 e''$ or $e \leq_2 e''$. Hence $e \leq e''$.
    \item Antisymmetry ($e \leq e'$ and $e' \leq e$ then $e'=e$): we have two cases: ($e \leq_1 e'$
      and $e' \leq_1 e$) or ($e \leq_2 e'$ and $e' \leq_2 e$). Since $\leq_1$ and $\leq_2$ are
      partial orders, we have $e=e'$.
    \end{itemize}

    Hence $\leq$ is a partial order.

  \item $\#$ is symmetric and irreflexive:
    \begin{itemize}
    \item Symmetric (if $e \# e'$ then $e' \# e$): we have either $e \#_1 e'$ or $e \#_2 e'$. Since
      $\#_1$ and $\#_2$ are symmetric then we have $e' \#_1 e$ or $e' \#_2 e$. Thus $e' \# e$.
    \item Irreflexive ($\neg(e \# e)$): We have either $\neg(e \#_1 e)$ or $\neg(e \#_2 e)$. Since
      $\#_1$ and $\#_2$ are irreflexive then $\neg(e \# e)$.
    \end{itemize}

    Hence $\#$ is symmetric and irreflexive.
  \end{itemize}

  We need to show that $\forall e, e', e'' \in E$:
  \begin{enumerate}
  \item $\set{e' \mid e' \leq e}$ is finite
    \begin{enumerate}
    \item Case $e \in E_1$ then $\set{e' \mid e' \leq e} = {e' \mid e' \leq_1 e}$, which is finite
      since $\es{E}_1$ is a event structure.
      \item Case $e \in E_2$ then $\set{e' \mid e' \leq e} = {e' \mid e' \leq_2 e}$, which is finite
      since $\es{E}_2$ is a event structure.
    \end{enumerate}
  \item $e \# e' \leq e'' \Rightarrow e \# e''$

    We have two cases: $e, e', e'' \in E_1$ or $e, e', e'' \in E_2$.  In both cases this holds
    because $\es{E}_1$ and $\es{E}_2$ are event structures.
  \end{enumerate}
\end{proof}

\subsubsection*{Proof of Lemma~\ref{lem:nd-es1}}
\begin{proof}
    Let
  \begin{align*}
    & \es{E}_1 = \pes{E_1}{\leq_1}{\#_1} \\
    & \es{E}_2 = \pes{E_2}{\leq_2}{\#_2} \\
    & \nd{\es{E}_1}{\es{E}_2} = \pes{E}{\leq}{\#}
  \end{align*}

  We need to show that $\leq$ is a partial order and that $\#$ is symmetric and irreflexive.
  \begin{itemize}
  \item $\leq$ is a partial order:
    \begin{itemize}
    \item Reflexivity ($e \leq e$): we have two cases: $e_1 \leq_1 e'_1$ or $e_2 \leq_2 e'_2$. Since
      $\leq_1$ and $\leq_2$ are partial orders, we are done.
    \item Transitivity ($e \leq e'$ and $e' \leq e''$ then $e \leq e''$): we have two cases:
      ($e \leq_1 e'$ and $e' \leq_1 e''$) or ($e \leq_2 e'$ and $e' \leq_2 e''$). Since $\leq_1$ and
      $\leq_2$ are partial orders, we have $e \leq_1 e''$ or $e \leq_2 e''$. Hence $e \leq e''$.
    \item Antisymmetry ($e \leq e'$ and $e' \leq e$ then $e'=e$): we have two cases: ($e \leq_1 e'$
      and $e' \leq_1 e$) or ($e \leq_2 e'$ and $e' \leq_2 e$). Since $\leq_1$ and $\leq_2$ are
      partial orders, we have $e=e'$.
    \end{itemize}

    Hence $\leq$ is a partial order.

  \item $\#$ is symmetric and irreflexive:
    \begin{itemize}
    \item Symmetric (if $e \# e'$ then $e' \# e$): we have the following cases
      \begin{enumerate}
      \item if $e \#_1 e'$ or $e \#_2 e'$. Since $\#_1$ and $\#_2$ are symmetric then we have
        $e' \#_1 e$ or $e' \#_2 e$. Thus $e' \# e$.
      \item if $e \in E_1$ and $e' \in E_2$, or vice-versa, then by Definition~\ref{def:pes-nd1} we
        also have $e' \leq e$.
      \end{enumerate}
    \item Irreflexive ($\neg(e \# e)$): We have either $\neg(e \#_1 e)$ or $\neg(e \#_2 e)$. Since
      $\#_1$ and $\#_2$ are irreflexive then $\neg(e \# e)$. There are no more cases since either
      $e \in E_1$ or $e \in E_2$.
    \end{itemize}

    Hence $\#$ is symmetric and irreflexive.
  \end{itemize}
  
  We need to show that $\forall e, e', e'' \in E$:
  \begin{enumerate}
  \item $\set{e' \mid e' \leq e}$ is finite
    \begin{enumerate}
    \item Case $e \in E_1$ then $\set{e' \mid e' \leq e} = {e' \mid e' \leq_1 e}$, which is finite
      since $\es{E}_1$ is an event structure.
      \item Case $e \in E_2$ then $\set{e' \mid e' \leq e} = {e' \mid e' \leq_2 e}$, which is finite
      since $\es{E}_2$ is an event structure.
    \end{enumerate}
  \item $e \# e' \leq e'' \Rightarrow e \# e''$

    We have two cases: $e, e', e'' \in E_1$ or $e, e', e'' \in E_2$.  In both cases this holds
    because $\es{E}_1$ and $\es{E}_2$ are event structures.
  \end{enumerate}
\end{proof}

\subsubsection*{Proof of Lemma~\ref{lem:rem-init-es1}}
\begin{proof}
  Let
  \begin{align*}
    & \es{E} = \pes{E}{\leq}{\#} \\
    & \es{E} \backslash a = \pes{E'}{\leq'}{\#'}
  \end{align*}

  We need to show that $\leq'$ is a partial order and $\#'$ is symmetric and irreflexive.  Since
  $\leq'$ and $\#'$ respectively equal $\leq$ and $\#$ for the events in $E'$, it follows directly
  that $\leq'$ is a partial order and $\#'$ is symmetric and irreflexive.
  
  We need to show that $\forall e, e', e'' \in E$:
  \begin{enumerate}
  \item $\set{e' \mid e' \leq e}$ is finite

    Since $\set{e' \mid e' \leq e}$ is finite, then so it is
    $\set{e' \mid e' \leq' e} = \set{e' \mid e' \leq e} \backslash l$.
    
  \item $e \# e' \leq e'' \Rightarrow e \# e''$

    Let $e, e', e'' \in E'$.  Then $e, e', e'' \neq a$ and $\neg (e, e', e'' \# a)$.  By
    Definition~\ref{def:rem-init1}, $e \#' e'$ entails $e \# e'$ and $e, e' \in E'$ and
    $e' \leq' e''$ entails $e' \leq e''$ and $e', e'' \in E$.  Since $\es{E}$ is an event structure,
    we have $e \# e' \leq e'' \Rightarrow e \# e''$.  Thus
    $e \#' e' \leq' e'' \Rightarrow e \#' e''$ and $e, e', e'' \in E$.
  \end{enumerate}
\end{proof}

\subsubsection*{Proof of Lemma~\ref{lem:seq-mono1}}
\begin{proof}
  Let
  \begin{align*}
    & \es{E}_1 = \pes{E_1}{\leq_1}{\#_1} \\
    & \es{E}'_1 = \pes{E'_1}{\leq'_1}{\#'_1} \\
    & \es{E}_2 = \pes{E_2}{\leq_2}{\#_2} \\
    & \es{E}'_2 = \pes{E'_2}{\leq'_2}{\#'_2} \\
    & \seq{\es{E}_1}{\es{E}_2} = \pes{E}{\leq}{\#} \\
    & \seq{\es{E}'_1}{\es{E}'_2} = \pes{E'}{\leq'}{\#'}
  \end{align*}
  such that $\es{E}_1 \sqsubseteq \es{E}'_1$ and $\es{E}_2 \sqsubseteq \es{E}'_2$.

  \begin{enumerate}
  \item We start by defining $f : E \rightarrow E'$ such that
    \[
      f(e) =
      \begin{cases}
        f_1(e) & \text{ if } e \in E_1 \\
        (f_2(e), m_1(x)) & \text{ if } e \in E_2 \times \confmax{\es{E}_1}
      \end{cases}
    \]
    where $m_1(x)$ is defined as follows: $\forall x \in \confmax{\es{E}_1}$,
    $m_1(x) \in \confmax{\es{E}'_1}$ such that $f_1[x] \subseteq m_1(x)$. If multiple are possible,
    $m_1(x)$ is arbitrarily chosen among them.

    Now we show that $f$ is injective, \ie\ if $a \neq b$ then $f(a) \neq f(b)$.
    We have two cases:
    \begin{itemize}
    \item if $a,b \in E_1$ we are done
    \item if $a, b \in E_2 \times \confmax{\es{E}_1}$

      Let $a = (e,x)$ and $b = (e,y)$.
      For $(e,x) \neq (e',y)$ we have two cases:
      \begin{itemize}
      \item $e \neq e'$

        We then have
        \begin{align*}
          f(e,x) \neq f(e', y)
          & \Leftrightarrow (f_2(e), m_1(x)) \neq (f_2(e'), m_1(y)) \\
          & \Leftrightarrow f_2(e) \neq f_2(e') \text{ or } m_1(x) \neq m_1(y)
        \end{align*}
        Since $f_2$ is injective we are done.

      \item $e = e'$ and $x \neq y$

        We then have
        \begin{align*}
          f(e,x) \neq f(e', y)
          & \Leftrightarrow (f_2(e), m_1(x)) \neq (f_2(e'), m_1(y)) \\
          & \Leftrightarrow f_2(e) \neq f_2(e') \text{ or } m_1(x) \neq m_1(y)
        \end{align*}

        Since we know that $f_2$ is an injective function, then we should have $f_2(e) =
        f_2(e')$. Otherwise we enter in a contradiction. Hence it lacks to see if
        $m_1(x) \neq m_1(y)$.  From definition of $m_1(x)$ and $m_1(y)$ we have, respectively,
        $f_1[x] \subseteq m_1(x)$ and $f_2[y] \subseteq m_2(y)$. Thus
        $f_1[x], f_1[y] \in \confES{\es{E}'_1}$. Since $f_1$ is injective and $x \neq y$ then it
        exists $e \in x$ and $e' \in y$, with $e \neq e'$, such that $f_1(e) \neq f_1(e')$. Thus
        $f_1[x] \neq f_1[y]$, and consequently $m_1(x) \neq m_1(y)$.
      \end{itemize}
    \end{itemize}

    Hence $f$ is injective.

  \item $\pi(f(a)) = \pi(a)$

    We have two cases:
    \begin{itemize}
    \item $a \in E_1$

      Then $\pi(f(a)) = \pi(f_1(a)) = \pi(a)$ since $\es{E}_1 \sqsubseteq \es{E}'_1$.

    \item $(a, x) \in E_2 \times \confmax{\es{E}_1}$

      Then $\pi(f(a,x)) = \pi(f_2(a), m_1(x)) = \pi(f_2(a)) = \pi(a)$ since
      $\es{E}_2 \sqsubseteq \es{E}'_2$.
    \end{itemize}
    
  \item $a \leq b \Leftrightarrow f(a) \leq' f(b)$
    \begin{itemize}
    \item[$\Rightarrow$] Assume $a \leq b$.
      
      By Definition~\ref{def:pes-seq1} we have:
      \begin{enumerate}
      \item $a,b \in E_1$.

        Hence $a \leq_1 b$. Since $\es{E}_1 \sqsubseteq \es{E}'_1$, then $f_1(a) \leq'_1 f_1(b)$.
        By Definition~\ref{def:pes-seq1}, $f_1(a) \leq' f_1(b)$.

      \item $(a,x), (b,x) \in E_2 \times \confmax{\es{E}_1}$

        Hence $(a,x) \leq (b,x)$ such that $a \leq_2 b$.  Since $\es{E}_2 \sqsubseteq \es{E}'_2$
        then $f_2(a) \leq'_2 f_2(b)$.  Now let $m_1(x) \in \confmax{\es{E}'_1}$ such that
        $f_1[x] \subseteq m_1(x)$.  Then $(f_2(a), m_1(x)) \leq' (f_2(b), m_1(x))$.  By
        Definition~\ref{def:pes-seq1}, $f(a,x) \leq' f(b,x)$.

      \item $a \in E_1$ and $(b,x) \in E_2 \times \confmax{\es{E}_1}$ such that $a \in x$

        Hence $a \leq (b,x)$. From $(b,x) \in E_2 \times \confmax{\es{E}_1}$ we have $b \in E_2$ and
        $x \in \confmax{\es{E}_1}$. Since $\es{E}_1 \sqsubseteq \es{E}'_1$ and
        $\es{E}_2 \sqsubseteq \es{E}'_2$, then $f_1(a) \in E'_1$ and $f_2(a) \in E'_2$.  Now let
        $m_1(x) \in \confmax{\es{E}'_1}$ such that $f_1[x] \subseteq m_1(x)$.  Since $a \in x$ then
        $f_1(a) \in f_1[x]$.  Thus $f_1(a) \in m_1(x)$. Hence
        $(f_2(b), m_1(x)) = f(b,x) \in E'_2 \times \confmax{\es{E}'_1}$.  By
        Definition~\ref{def:pes-seq1}, $f(a) \leq' f(b,x)$.
      \end{enumerate}

    \item[$\Leftarrow$] Assume $f(a) \leq' f(b)$
      \begin{enumerate}
      \item $f(a) \leq'_1 f(b)$

        We have $f(a) = f_1(a), f(b) = f_1(b) \in E'_1$. Hence $f_1(a) \leq'_1 f_1(b)$.
        Since $\es{E}_1 \sqsubseteq \es{E}'_1$ then $a \leq_1 b$.
        By Definition~\ref{def:pes-seq1}, $a \leq b$.

      \item $f(a,x) \leq' f(b,x)$

        We have
        $f(a,x) = (f_2(a), m_1(x)), f(b,x) = (f_2(b), m_1(x)) \in E'_2 \times \confmax{\es{E}'_1}$.
        From $m_1(x)$ we have $f_1[x] \subseteq m_1(x)$ such that $x \in \confmax{\es{E}_1}$.  From
        Definition~\ref{def:pes-seq1} we have $f_2(a) \leq'_2 f_2(b)$.  Since
        $\es{E}_2 \sqsubseteq \es{E}'_2$ then $a \leq_2 b$.  By Definition~\ref{def:pes-seq1}
        $(a,x) \leq (b,x)$.

      \item $f(a) \leq' f(b,x)$ such that $f(a) \in m_1(x)$.

        We have $f(a) = f_1(a) \in E'_1$ and
        $f(b,x) = (f_2(b), m_1(x)) \in E'_2 \times \confmax{\es{E}'_1}$, which gives
        $f_2(b) \in E'_2$ and $m_1(x) \in \confmax{\es{E}'_1}$.  From $m_1(x)$ we have
        $f_1[x] \subseteq m_1(x)$ such that $x \in \confmax{\es{E}_1}$. Since $f_1(a) \in f_1[x]$
        then $a \in x$.  Since $\es{E}_1 \sqsubseteq \es{E}'_1$ and $\es{E}_2 \sqsubseteq \es{E}'_2$
        then $a \in E_1$ and $b \in E_2$.  We then have $(b,x) \in E_2 \times \confmax{\es{E}_1}$.
        By Definition~\ref{def:pes-seq1} $a \leq (b,x)$.
      \end{enumerate}
    \end{itemize}

  \item $a \# b \Leftrightarrow f(a) \#' f(b)$
    \begin{itemize}
    \item[$\Rightarrow$] Assume $a \# b$.

      From Definition~\ref{def:pes-seq1} we have two cases:
      \begin{itemize}
      \item $\exists (e_1 \leq a,\ e'_1 \leq b)\ .\ e_1 \#_1 e'_1$

        Here we have two subcases:
        \begin{itemize}
        \item $a, b \in E_1$

          Since $\es{E}_1 \sqsubseteq \es{E}'_1$ we have $f_1(a) \#'_1 f_1(b)$.

        \item $a \in E_1$ and $(b,x) \in E_2 \times \confmax{\es{E}_1}$

          Then $e_1 \leq a, e'_1 \leq (b,x)$ such that $e_1 \#_1 e'_1$.  By
          Definition~\ref{def:pes-seq1}, $e_1 \leq_1 a$ and $e'_1 \leq (b,x)$ such that
          $e'_1 \in x$.  Since $\es{E}_1 \sqsubseteq \es{E}'_1$ we have
          $f_1(e_1) \leq_1 f_1(a) \Leftrightarrow f(e_1) \leq f(a)$ and $f_1(e_1) \#'_1 f_1(e'_1)$.
          From $(b,c) \in E_2 \times \confmax{\es{E}_1}$ we have $b \in E_2$ and
          $x \in \confmax{\es{E}_1}$.  Now let $m_1(x) \in \confmax{\es{E}'_1}$ such that
          $f_1[x] \subseteq m_1(x)$.  Since $e'_1 \in x$ then $f_1(e'_1) \in f_1[x]$.  Since
          $\es{E}_2 \sqsubseteq \es{E}'_2$ then $f_2(b) \in E'_2$.  Thus
          $(f_2(b), m_1(x)) = f(b,x) \in E'_2 \times \confmax{\es{E}'_1}$.  Hence
          $f(e'_1) \leq f(b,x)$.  Thus $f(a) \#' f(b,x)$.
        \end{itemize}

      \item $(a,x) \# (b,x)$ such that $a \#_2 b$

        From $(a,x), (b,x) \in E_2 \times \confmax{\es{E}_1}$ we have $a,b \in E_2$ and
        $x \in \confmax{\es{E}_1}$.  From $\es{E}_2 \sqsubseteq \es{E}'_2$ we have
        $f_2(a), f_2(b) \in E'_2$ and $f_2(a) \#'_2 f_2(b)$.  Now let
        $m_1(x) \in \confmax{\es{E}'_1}$ such that $f_1[x] \subseteq m_1(x)$.  Then
        $(f_2(a), m_1(x)) = f(a,x), (f_2(b), m_1(x)) = f(b,x) \in E'_2 \times \confmax{\es{E}'_1}$.
        Thus $f(a,x) \#' f(b,x)$.
      \end{itemize}

    \item[$\Leftarrow$] Assume $f(a) \#' f(b)$

      From Definition~\ref{def:pes-seq1} we have two cases:
      \begin{itemize}
      \item $f(a) \#' f(b)$ such that
        $\exists (f(e_1) \leq f(a), f(e'_1) \leq f(b))\ .\ f_1(e_1) \#'_1 f_1(e'_1)$

        Here we have two subcases:
        \begin{itemize}
        \item $f(a), f(b) \in E'_1$.
          
          We are done since $\es{E}_1 \sqsubseteq \es{E}'_1$.

        \item $f(a) \in E'_1$ and $f(b,x) \in E'_2 \times \confmax{\es{E}'_1}$

          Then $f(e_1) \leq f(a) \Leftrightarrow f_1(e_1) \leq f_1(a)$ and
          $f(e'_1) \leq f(b,x) \Leftrightarrow f_1(e'_1) \leq (f_2(b), m_1(x))$.  By
          Definition~\ref{def:pes-seq1}, $f_1(e_1) \leq_1 f_1(a)$ and
          $f_1(e'_1) \leq (f_2(b), m_1(x))$ such that $f_1(e'_1) \in m_1(x)$.  Since
          $\es{E}_1 \sqsubseteq \es{E}'_1$, $e_1 \leq_1 a$ and $e_1 \#_1 e'_1$.  From
          $(f_2(b), m_1(x)) \in E'_2 \times \confmax{\es{E}'_1}$ we have $f_2(b) \in E'_2$ and
          $m_1(x) \in \confmax{\es{E}'_1}$.  Since $\es{E}_2 \sqsubseteq \es{E}'_2$ then
          $b \in E_2$.  Since $f_1(e'_1) \in m_1(x)$ and $f_1[x] \subseteq m_1(x)$ then
          $e'_1 \in x$.  Thus $(b,x) \in E_2 \times \confmax{\es{E}_1}$.  Hence $e'_1 \leq (b,x)$.
          Thus $a \# (b,x)$.
        \end{itemize}

      \item $f(a,x) \#' f(b,x)$ such that $f_2(a) \#'_2 f_2(b)$

        From
        $f(a,x) = (f_2(a), m_1(x)), f(b,x) = (f_2(b), m_1(x)) \in E'_2 \times \confmax{\es{E}'_1}$
        we have that $f_2(a), f_2(b) \in E_2$ and $m_1(x) \in \confmax{\es{E'_1}}$.  From
        $\es{E}_2 \sqsubseteq \es{E}'_2$ we have $a,b \in E_2$ and $a \#_2 b$.  From $m_1(x)$ we
        know that $x \in \confmax{\es{E}_1}$.  Thus $(a,x) \# (b,x)$.
      \end{itemize}
      
    \end{itemize}
  \end{enumerate}
\end{proof}

\subsubsection*{Proof of Lemma~\ref{lem:nd-mono1}}
\begin{proof}
  Let
  \begin{align*}
    & \es{E}_1 = \pes{E_1}{\leq_1}{\#_1} \\
    & \es{E}'_1 = \pes{E'_1}{\leq'_1}{\#'_1} \\
    & \es{E}_2 = \pes{E_2}{\leq_2}{\#_2} \\
    & \es{E}'_2 = \pes{E'_2}{\leq'_2}{\#'_2} \\
    & \nd{\es{E}_1}{\es{E}_2} = \pes{E}{\leq}{\#} \\
    & \nd{\es{E}'_1}{\es{E}'_2} = \pes{E'}{\leq'}{\#'}
  \end{align*}
  such that $\es{E}_1 \sqsubseteq \es{E}'_1$ and $\es{E}_2 \sqsubseteq \es{E}'_2$.

  \begin{enumerate}
  \item We start by defining a function $f: E \rightarrow E'$ such that
    \[
      f(e) =
      \begin{cases}
        f_1(e) & \text{ if } e \in E_1 \\
        f_2(e) & \text{ if } e \in E_2 \\
      \end{cases}
    \]

    Clearly $f$ is injective, since if $e \in E_1$ then $f(e) = f_1(e)$ and $f_1$ is injective and if
    $e \in E_2$ then $f(e) = f_2(e)$ and $f_2$ is injective.

  \item $\pi(f(a)) = \pi(a)$

    We have two similar cases, $a \in E_1$ or $a \in E_2$. We only show the former.

    $\pi(f(a)) = \pi(f_1(a)) = \pi(a)$, since $\es{E}_1 \sqsubseteq \es{E}'_1$.

  \item $a \leq b \Leftrightarrow f(a) \leq' f(b)$
    \begin{itemize}
    \item[$\Rightarrow$] Assume $a \leq b$.

      By Definition~\ref{def:pes-nd1}, $a \leq_1 b$ or $a \leq_2 b$.  Since
      $\es{E}_1 \sqsubseteq \es{E}'_1$ and $\es{E}_2 \sqsubseteq \es{E}'_2$, then
      $f_1(a) \leq'_1 f_1(b)$ and $f_2(a) \leq'_2 f_2(b)$.  By Definition~\ref{def:pes-nd1} we are
      done.

    \item[$\Leftarrow$] Assume $f(a) \leq' f(b)$.

      By Definition~\ref{def:pes-nd1} $f(a) \leq'_1 f(b) \Leftrightarrow f_1(a) \leq'_1 f_1(b)$ or
      $f(a) \leq'_2 f(b) \Leftrightarrow f_2(a) \leq'_2 f_2(b)$. Since
      $\es{E}_1 \sqsubseteq \es{E}'_1$ and $\es{E}_2 \sqsubseteq \es{E}'_2$, then $a \leq_1 b$ and
      $a \leq_2 b$.  By Definition~\ref{def:pes-nd1} we are done.
    \end{itemize}

  \item $a \# b \Leftrightarrow f(a) \#' f(b)$
    \begin{itemize}
    \item[$\Rightarrow$] Assume $a \# b$.

      We have the following cases:
      \begin{itemize}
      \item If $a \#_1 b$ or $a \#_2 b$, then we are done since $\es{E}_1 \sqsubseteq \es{E}'_1$ and
        $\es{E}_2 \sqsubseteq \es{E}'_2$.
      \item If $a \in E_1$ and $b \in E_2$ (or vice-versa).

        Since $\es{E}_1 \sqsubseteq \es{E}'_1$ and $\es{E}_2 \sqsubseteq \es{E}'_2$, then
        $f_1(a) \in E'_1$ and $f_2(b) \in E'_2$.  By Definition~\ref{def:pes-nd1},
        $f_1(a) \#' f_2(b) \Leftrightarrow f(a) \#' f(b)$.
      \end{itemize}

    \item[$\Leftarrow$] Assume $f(a) \#' f(b)$.

      We have the following cases:
      \begin{itemize}
      \item If $f(a) \#'_1 f(b) \Leftrightarrow f_1(a) \#'_1 f_1(b)$ or
        $f(a) \#'_2 f(b) \Leftrightarrow f_2(a) \#'_2 f_2(b)$, we are done since
        $\es{E}_1 \sqsubseteq \es{E}'_1$ and $\es{E}_2 \sqsubseteq \es{E}'_2$.
      \item If $f(a) = f_1(a) \in E'_1$ and $f(b) = f_2(b) \in E'_2$ (or vice-versa).

        Since $\es{E}_1 \sqsubseteq \es{E}'_1$ and $\es{E}_2 \sqsubseteq \es{E}'_2$, then
        $a \in E_1$ and $b \in E_2$.  By Definition~\ref{def:pes-nd1}, $a \# b$.
      \end{itemize}
    \end{itemize}
  \end{enumerate}
\end{proof}

\subsubsection*{Proof of Lemma~\ref{lem:conc-mono1}}
\begin{proof}
  Let
  \begin{align*}
    & \es{E}_1 = \pes{E_1}{\leq_1}{\#_1} \\
    & \es{E}'_1 = \pes{E'_1}{\leq'_1}{\#'_1} \\
    & \es{E}_2 = \pes{E_2}{\leq_2}{\#_2} \\
    & \es{E}'_2 = \pes{E'_2}{\leq'_2}{\#'_2} \\
    & \conc{\es{E}_1}{\es{E}_2} = \pes{E}{\leq}{\#} \\
    & \conc{\es{E}'_1}{\es{E}'_2} = \pes{E'}{\leq'}{\#'}
  \end{align*}
  such that $\es{E}_1 \sqsubseteq \es{E}'_1$ and $\es{E}_2 \sqsubseteq \es{E}'_2$.

  \begin{enumerate}
  \item We start by defining a function $f: E \rightarrow E'$ such that
    \[
      f(e) =
      \begin{cases}
        f_1(e) & \text{ if } e \in E_1 \\
        f_2(e) & \text{ if } e \in E_2 \\
      \end{cases}
    \]

    Clearly $f$ is injective, since if $e \in E_1$ then $f(e) = f_1(e)$ and $f_1$ is injective and if
    $e \in E_2$ then $f(e) = f_2(e)$ and $f_2$ is injective.

  \item $\pi(f(a)) = \pi(a)$

    We have two similar cases, $a \in E_1$ or $a \in E_2$. We only show the former.

    $\pi(f(a)) = \pi(f_1(a)) = \pi(a)$, since $\es{E}_1 \sqsubseteq \es{E}'_1$.

  \item $a \leq b \Leftrightarrow f(a) \leq' f(b)$
    \begin{itemize}
    \item[$\Rightarrow$] Assume $a \leq b$.

      By Definition~\ref{def:pes-conc1}, $a \leq_1 b$ or $a \leq_2 b$.  Since
      $\es{E}_1 \sqsubseteq \es{E}'_1$ and $\es{E}_2 \sqsubseteq \es{E}'_2$, then
      $f_1(a) \leq'_1 f_1(b)$ and $f_2(a) \leq'_2 f_2(b)$.  By Definition~\ref{def:pes-conc1} we are
      done.

    \item[$\Leftarrow$] Assume $f(a) \leq' f(b)$.

      By Definition~\ref{def:pes-conc1} $f(a) \leq'_1 f(b) \Leftrightarrow f_1(a) \leq'_1 f_1(b)$ or
      $f(a) \leq'_2 f(b) \Leftrightarrow f_2(a) \leq'_2 f_2(b)$. Since
      $\es{E}_1 \sqsubseteq \es{E}'_1$ and $\es{E}_2 \sqsubseteq \es{E}'_2$, then $a \leq_1 b$ and
      $a \leq_2 b$.  By Definition~\ref{def:pes-conc1} we are done.
    \end{itemize}

  \item $a \# b \Leftrightarrow f(a) \#' f(b)$

    Similar to the previous case.
  \end{enumerate}
\end{proof}

\subsubsection*{Proof of Lemma~\ref{lem:op-sim-1}}
\begin{proof}
  \begin{align*}
    & \es{E}_1 \equiv \es{E}_1' \text{ and } \es{E}_2 \equiv \es{E}_2' \\
    \Rightarrow & \set{ \text{Definition~\ref{def:pes-sub1}} } \\
    & \es{E}_1 \sqsubseteq \es{E}_1' \text{ and } \es{E}_1' \sqsubseteq \es{E}_1,
      \es{E}_2 \sqsubseteq \es{E}_2' \text{ and } \es{E}_2' \sqsubseteq \es{E}_2 \\
    \Rightarrow & \set{ op \text{ monotone}} \\
    & \es{E}_1\, op\, \es{E}_2 \sqsubseteq \es{E}_1'\, op\, \es{E}_2' \text{ and }
      \es{E}_1'\, op\, \es{E}_2' \sqsubseteq \es{E}_1\, op\, \es{E}_2 \\
    \Rightarrow & \set{ \text{Definition~\ref{def:pes-sub1}} } \\
    & \es{E}_1\, op\, \es{E}_2 \equiv \es{E}_1'\, op\, \es{E}_2'
  \end{align*}
\end{proof}

\subsubsection*{Proof of Lemma~\ref{lem:seq-rem-init1}}
\begin{proof}
  To prove $(\seq{\es{E_1}}{\es{E_2}}) \backslash l \equiv \seq{(\es{E_1}\backslash l)}{\es{E_2}}$, we
  need to verify that:
  \begin{itemize}
    \item $(\seq{\es{E_1}}{\es{E_2}}) \backslash l \sqsubseteq \seq{(\es{E_1}\backslash l)}{\es{E_2}}$
    \item $\seq{(\es{E_1}\backslash l)}{\es{E_2}} \sqsubseteq (\seq{\es{E_1}}{\es{E_2}}) \backslash l$
  \end{itemize}
  
  Let
  \begin{align*}
    & \es{E}_1 = \pes{E_1}{\leq_1}{\#_1} \\
    & \es{E}_2 = \pes{E_2}{\leq_2}{\#_2} \\
    & \seq{\es{E}_1}{\es{E}_2} = \pes{E_{\seq{1}{2}}}{\leq_{\seq{1}{2}}}{\#_{\seq{1}{2}}} \\
    & (\seq{\es{E}_1}{\es{E}_2}) \backslash l = \pes{E}{\leq}{\#} \\
    & \es{E}_1 \backslash l = \pes{E_1^l}{\leq_1^l}{\#_1^l} \\
    & \seq{(\es{E}_1 \backslash l)}{\es{E}_2} = \pes{E'}{\leq'}{\#'} \\
    & l \in \init{\seq{\es{E_1}}{\es{E_2}}}
  \end{align*}

  
  \begin{itemize}
  \item $(\seq{\es{E_1}}{\es{E_2}}) \backslash l \sqsubseteq \seq{(\es{E_1}\backslash l)}{\es{E_2}}$
    \begin{enumerate}
    \item We start by defining a function $f: E \rightarrow E'$ such that
      \[
        f(e) =
        \begin{cases}
          e & \text{ if } e \in E_1 \\
          (e, x \backslash l) & \text{ if } e \in E_2 \times \confmax{\es{E}_1}
        \end{cases}
      \]
      where $x \backslash l = \{e \mid e \in x \wedge e \neq l\}$.

      First we show that $x \backslash l \in \confmax{\es{E}_1 \backslash l}$.  It is clear that
      $(x \backslash l) \cup \{l\} = x \in \confmax{\es{E}_1}$.  Furthermore, since
      $x \in \confmax{\es{E}_1}$ then $\nexists y \in \confES{\es{E}_1}$ such that $x \cchain{}\ y$.
      Now let $y \backslash l \in \confES{\es{E}_1\backslash l}$ such that
      $x \backslash l \cchain{}\ y \backslash l$. By adding $l$ to both configurations we would have
      $x \cchain{}\ y$, which means that $x \not\in \confmax{\es{E}_1}$, hence a contradiction. Thus
      $x \backslash l \in \confmax{\es{E}_1 \backslash l}$.

      Another way of seeing that $x \backslash l \in \confmax{\es{E}_1 \backslash l}$ is by noting
      that $x \backslash l$ removes the initial event from $x$, which is itself a maximal
      configuration.

      Now we show that $f$ is injective. We have two cases:
      \begin{itemize}
      \item $a,b \in E_1$

        Then we are done since the mapping is the identity.

      \item $(a, x), (b, y) \in E_2 \times \confmax{\es{E}_1}$

        From $(a, x) \neq (b, y)$ we have two cases:
        \begin{itemize}
        \item $a \neq b$

          Trivially holds because the function is the identity

        \item $a=b$ and $x \neq y$

          Then
          $f(a, x) \neq f(b, y) \Leftrightarrow (a, x \backslash l) \neq (b, y \backslash l)
          \Leftrightarrow a \neq b \text{ or } x \backslash l \neq y \backslash l$.

          Clearly it cannot be the case $a \neq b$, otherwise it would contradict our assumption
          that $a=b$, and also the fact that the identity function is injective. Hence it can only
          be $x \backslash l \neq y \backslash l$. Since $x \neq y$ then $\exists e \in x$ such that
          $e \not\in y$. Since $e \neq l$, then we have $x \backslash l \neq y \backslash l$.
        \end{itemize}
      \end{itemize}

    \item $\pi(f(a)) = \pi(a)$

      We have two cases:
      \begin{itemize}
      \item $a \in E_1$

        It follows directly that $\pi(f(a)) = \pi(a)$, since $f(a) = a$

      \item $(a,x) \in E_2 \times \confmax{\es{E}_1}$

        Then $\pi(f(a,x)) = \pi(a, x \backslash l) = \pi(a)$
      \end{itemize}
      
    \item $a \leq b \Leftrightarrow f(a) \leq' f(b)$
      \begin{itemize}         
      \item[$\Rightarrow$] Assume $a \leq b$.

        By Definition~\ref{def:rem-init1}, $a, b\ \in E$ (\ie\ $a,b \neq l$ and
        $\neg(a,b \#_{\seq{1}{2}} l)$), and $a \leq_{\seq{1}{2}} b$. We have three cases:
        \begin{enumerate}
        \item[$(i)$] $a \leq_{\seq{1}{2}} b$ is of the form $a \leq_1 b$.

          We have $a,b \in E_1$.  Note that $f(a) = a, f(b) = b \in E'$, hence $a,b \neq l$ and
          $\neg(a,b \#_{\seq{1}{2}} l)$. By Definition~\ref{def:rem-init1} $a \leq_1^l b$. By
          Definition~\ref{def:pes-seq1} $a \leq' b$.
          
        \item[$(ii)$] $a \leq_{\seq{1}{2}} b$ is of the form $(a,x) \leq_{\seq{1}{2}} (b,x)$.

          It follows directly that $(a, x \backslash l) = f(a,x) \leq' f(b,x) = (b, x \backslash l)$
          since $a,b \in E_2$ and $x \backslash l \in \confmax{\es{E}_1 \backslash l}$.
                    
        \item[$(iii)$] $a \leq_{\seq{1}{2}} b$ is of the form $a \leq_{\seq{1}{2}} (b, x)$.

          We know that $f(a) = a$ and $f(b,x) = (b, x \backslash l)$ such that
          $x \backslash l \in \confmax{\es{E}_1 \backslash l}$ and $a \in x \backslash l$ since
          $a \neq l$ and $\neg(a \#_{\seq{1}{2}} l)$.  By Definition~\ref{def:pes-seq1},
          $a \leq' (b, x \backslash l)$.
        \end{enumerate}

      \item[$\Leftarrow$] Assume $f(a) \leq' f(b)$.

        We have three cases:
        \begin{enumerate}
        \item[$(i)$] $f(a) \leq' f(b)$ comes from $f(a) \leq_1^l f(b)$.

          We know that $f(a) = a$ and $f(b) = b$.  Hence we have $a \leq_1^l b$ and, consequently,
          $a,b \neq l$ and $a,b \#_1 l$.  By Definition~\ref{def:rem-init1}, $a \leq_1 b$.  By
          Definition~\ref{def:pes-seq1} and Definition~\ref{def:rem-init1} that $a \leq b$.
                    
        \item[$(ii)$] $f(a) \leq' f(b)$ is of the form $f(a,x) \leq' f(b,x)$.

          We know that $f(a,x) = (a, x \backslash l)$ and $f(b,x) = (b, x \backslash l)$.
          Hence we have $(a, x \backslash l) \leq' (b, x \backslash l)$.
          By Definition~\ref{def:pes-seq1} we have $a \leq_2 b$.
          Again by Definition~\ref{def:pes-seq1}, we have $(a, x) \leq_{\seq{1}{2}} (b,x)$.
          Since $x \backslash l \in \confmax{\es{E}_1 \backslash l}$, we know that none of the events in $x$ are in conflict with $l$.
          Hence, $\neg( (a,x), (b,x) \#_{\seq{1}{2}} l)$.
          By Definition~\ref{def:rem-init1}, we have $(a,x) \leq (b,x)$.
          
        \item[$(iii)$] $f(a) \leq' f(b)$ is of the form $f(a) \leq' f(b,x)$.

          We know that $f(a) = a$ and $f(b,x) = (b, x \backslash l)$.
          Hence we have $a \leq' (b, x \backslash l)$, such that $a \in x \backslash l \in \confmax{\es{E}_1 \backslash l}$,  $a \in E_1^l$, and $b \in E_2$.
          By Definition~\ref{def:rem-init1}, we have $a \in E_1$, such that $a \neq l$ and $\neg( a \#_1 l)$, and $x \in \confmax{\es{E}_1}$.
          By Definition\ref{def:pes-seq1}, we have $a \leq_{\seq{1}{2}} (b,x)$, such that $\neg( (b,x) \#_{\seq{1}{2}} l)$, since $x$ has no events in conflict with $l$.
          By Definition~\ref{def:rem-init1}, we have $a \leq (b,x)$.
        \end{enumerate}
      \end{itemize}

    \item $a \# b \Leftrightarrow f(a) \#' f(b)$
      \begin{itemize}
      \item[$\Rightarrow$] Assume $a \# b$.

        By Definition~\ref{def:rem-init1} it entails $a \#_{\seq{1}{2}} b$ and $a,b \in E$ (\ie\
        $a,b \neq l$ and $\neg(a,b \#_{\seq{1}{2}} l)$).  By Definition~\ref{def:pes-seq1} it
        entails $\exists (e_1 \leq_{\seq{1}{2}} a, e'_1 \leq_{\seq{1}{2}} b)\, .\, e_1 \#_1 e'_1$ or
        $a \#_{2} b$. We have three cases:
        \begin{enumerate}
        \item[$(i)$] $a \#_1 b$.

          Then $e_1 \leq_1 a, e'_1 \leq_1 b$.
          We know that $e_1, e'_1 \neq l$  since if $e_1 = l$ then $e_1 \#_1 e'_1 \leq_1 b \Rightarrow e_1 \#_1 b$, which would entail $b \not\in E$, a contradiction.
          Furthermore, $\neg(e_1, e'_1 \#_1 l)$ with similar arguments.
          We know that $f(e) = e$ for $e \in \{e_1, e'_1, a, b\}$.
          By Definition~\ref{def:rem-init1}, $a \#_1^l b$, $e_1 \leq_1^l a$, $e'_1 \leq_1^l b$, and $e_1 \#_1^l e'_1$.
          By Definition~\ref{def:pes-seq1}, $a \#' b$.
          
        \item[$(ii)$] $a \#_{2} b$.

          We then know that $(a,x) \# (b,x)$ and that $a, b \in E_2$.
          Clearly $(a,x), (b,x) \neq l$ and $\neg( (a,x), (b,x) \#_{\seq{1}{2}} l)$.
          Furthermore, $x \backslash l \in \confmax{\es{E}_1 \backslash l}$.
          Thus $(a, x \backslash l) = f(a,x), (b, x \backslash l) = f(b,x) \in E_2 \times \confmax{\es{E}_1 \backslash l}$.
          It then follows by Definition~\ref{def:pes-seq1},  $(a, x \backslash l) \#' (b, x \backslash l)$.

        \item[$(iii)$] $a \# (b,x)$

          Then $e_1 \leq_1 a$, and $e'_1 \leq (b,x)$ with $e'_1 \in x$.
          We know that $e_1, e'_1 \neq l$  since if $e'_1 = l$ then $e'_1 \#_1 e_1 \leq_1 a \Rightarrow e'_1 \#_1 a$, which would entail $a \not\in E$, a contradiction.
          Furthermore, $\neg(e_1, e'_1 \#_1 l)$ with similar arguments.
          We know that $f(e) = e$ for $e \in \{e_1, e'_1, a\}$ and $f(b,x) = (b, x \backslash l)$, since  $x \backslash l \in \confmax{\es{E}_1 \backslash l}$.
          By Definition~\ref{def:rem-init1}, $a \#_1^l b$, $e_1 \leq_1^l a$, $e'_1 \leq_1^l b$, and $e_1 \#_1^l e'_1$.
          By Definition~\ref{def:pes-seq1}, $a \#' (b, x \backslash l)$.
        \end{enumerate}
        
      \item[$\Leftarrow$] Assume $f(a) \#' f(b)$.

        $f(a) \#' f(b)$ entails
        $\exists(f(e_1) \leq' f(a), f(e'_1) \leq' f(b))\, .\, f(e_1) \#_1^l f(e'_1)$ or
        $f(a) \#_{2} f(b)$ (note that $f(e_1) = e_1$ and $f(e'_1) = e'_1$ because
        $e_1, e'_1 \in E_1$).  We have three cases:
        \begin{enumerate}
        \item[$(i)$] $f(a) \#_1^l f(b)$.

          We know that for $e \in \{e_1, e'_1, a, b\}$, we have  $f(e) = e$, $e \leq l$, and $\neg(e \#_{1} l)$.
          We then have $a \#_1^l b$.
          By Definition~\ref{def:rem-init1}, $a \#_1 b$, $e_1 \leq_1 a$, $e'_1 \leq_1 b$, and $e_1 \#_1 e'_1$.
          By Definition~\ref{def:pes-seq1}, $a \#_{\seq{1}{2}} b$, $e_1 \leq_{\seq{1}{2}} a$, $e'_1 \leq_{\seq{1}{2}} b$, and $e_1 \#_{\seq{1}{2}} e'_1$.
          By Definition~\ref{def:rem-init1}, $a \# b$.
                    
        \item[$(ii)$] $f(a,x) \#' f(b,x)$.

          We know that $f(a,x) = (a, x \backslash l)$ and $f(b,x) = (b, x \backslash l)$.
          Then $(a, x \backslash l) \#' (b, x \backslash l)$.
          By Definition~\ref{def:pes-seq1}, $a \#_2 b$.
          Furthermore, $(x \backslash l) \cup \{l\} = x \in \confmax{\es{E}_1}$, since $x \backslash l \in \confmax{\es{E}_1 \backslash l}$.
          We then have $(a, x), (b, x) \in E_2 \times \confmax{\es{E}_1}$ and $\neg( (a,x), (b,x) \#_{\seq{1}{2}} l)$, since $x$ has no events in conflict with $l$.
          By Definition~\ref{def:pes-seq1}, $(a, x) \# (b, x)$.

        \item[$(iii)$] $f(a) \#' f(b,x)$.

          This entails $f(e_1) \leq' f(a), f(e'_1) \leq' f(b,x)\, .\, f(e_1) \#_1^l f(e'_1)$.
          We know that for $e \in \{e_1, e'_1, a\}$, we have  $f(e) = e$, $e \leq l$, and $\neg(e \#_{1} l)$.
          Furthermore, $f(b,x) = (b, x \backslash l)$.
          We then have $a \#' (b, x \backslash l)$ and $e_1 \leq' a, e'_1 \leq' (b,x \backslash l)\, .\, e_1 \#_1^l e'_1$.
          By Definition~\ref{def:pes-seq1}, $e_1 \leq_1^l a$ and $e'_1 \in x \backslash l \in \confmax{\es{E}_1 \backslash l}$.
          By Definition~\ref{def:rem-init1}, $e_1 \leq_1 a$ and $e_1 \#_1 e'_1$.
          By Definition~\ref{def:pes-seq1},  $e_1 \leq_{\seq{1}{2}} a$ and $e_1 \#_{\seq{1}{2}} e'_1$.
          Furthermore, $(x \backslash l) \cup \{l\} = x \in \confmax{\es{E}_1}$, since $x \backslash l \in \confmax{\es{E}_1 \backslash l}$.
          Thus $e'_1 \in x$ and, consequently, $e'_1 \leq_{\seq{1}{2}} (b, x)$.
          Furthermore, $\neg( (b,x) \#_{\seq{1}{2}} l)$ since $x$ has no events in conflict with $l$.
          By Definition~\ref{def:rem-init1}, $a \# (b,x)$.
        \end{enumerate}
      \end{itemize}
    \end{enumerate}

  \item $\seq{(\es{E_1}\backslash l)}{\es{E_2}} \sqsubseteq (\seq{\es{E_1}}{\es{E_2}}) \backslash l$
    \begin{enumerate}
    \item We start by defining $f: E' \rightarrow E$ such that
      \[
        f(e) =
        \begin{cases}
          e & \text{ if } e \in E_1^l \\
          (e, x \cup \{l\}) & \text{ if } e \in E_2 \times \confmax{\es{E}_1 \backslash l}
        \end{cases}
      \]

      It is straightforward to see that if $x \in \confmax{\es{E}_1 \backslash l}$ then $x \cup \{l\} \in \confmax{\es{E}_1}$.

      We now show that $f$ is injective.
      We have two cases:
      \begin{itemize}
      \item $a,b \in E_1^l$

        Then we are done since the mapping is the identity.

      \item $(a, x), (b, y) \in E_2 \times \confmax{\es{E}_1 \backslash l}$

        From $(a, x) \neq (b, y)$ we have two cases:
        \begin{itemize}
        \item $a \neq b$

          Trivially holds because the function is the identity.

        \item $a=b$ and $x \neq y$

          Then
          $f(a, x) \neq f(b, y) \Leftrightarrow (a, x \cup \{l\}) \neq (b, y \cup \{l\})
          \Leftrightarrow a \neq b \text{ or } x \cup \{l\} \neq y \cup \{l\}$.

          Clearly it cannot be the case $a \neq b$, otherwise it would contradict our assumption
          that $a=b$, and also the fact that the identity function is injective. Hence it can only
          be $x \cup \{l\} \neq y \cup \{l\}$. Since $x \neq y$ then $\exists e \in x$ such that
          $e \not\in y$. Since $e \neq l$, then we have $x \cup \{l\} \neq y \cup \{l\}$.
        \end{itemize}
      \end{itemize}
    \end{enumerate}

  \item $\pi(f(a)) = \pi(a)$

    Similar to the other case

  \item $a \leq' b \Leftrightarrow f(a) \leq f(b)$

    Similar to the other case

  \item  $a \#' b \Leftrightarrow f(a) \# f(b)$

    Similar to the other case
  \end{itemize}
\end{proof}

\subsubsection*{Proof of Lemma~\ref{lem:nd-rem-init1}}
\begin{proof}
  We need to prove
  \begin{itemize}
  \item $(\nd{\es{E_1}}{\es{E_2}}) \backslash l \sqsubseteq \es{E}_1 \backslash l$,
    $\es{E}_1 \backslash l \sqsubseteq (\nd{\es{E_1}}{\es{E_2}}) \backslash l $ when
    $l \in \init{\es{E_1}}$
  \item $(\nd{\es{E_1}}{\es{E_2}}) \backslash l \sqsubseteq \es{E}_2 \backslash l$,
    $\es{E}_2 \backslash l \sqsubseteq (\nd{\es{E_1}}{\es{E_2}}) \backslash l $ when
    $l \in \init{\es{E_2}}$
  \end{itemize}
  Since both cases are identical, we focus solely when $l \in \init{\es{E_1}}$.

  Let
  \begin{align*}
    & \es{E}_1 = \pes{E_1}{\leq_1}{\#_1} \\
    & \es{E}_2 = \pes{E_2}{\leq_2}{\#_2} \\
    & \nd{\es{E}_1}{\es{E}_2} = \pes{E}{\leq}{\#} \\
    & (\nd{\es{E}_1}{\es{E}_2}) \backslash l = \pes{E'}{\leq'}{\#'} \\
    & \es{E}_1 \backslash l = \pes{E_1^l}{\leq_1^l}{\#_1^l} \\
    & l \in \init{\nd{\es{E_1}}{\es{E_2}}}
  \end{align*}
  


  Consider $l \in \init{\es{E_1}}$.
  \begin{itemize}
  \item $(\nd{\es{E_1}}{\es{E_2}}) \backslash l \sqsubseteq \es{E}_1 \backslash l$
    \begin{enumerate}
    \item We begin by defining the function $f: E' \rightarrow E_1^l$ as the identity, which is
      injective.

    \item $\pi(f(a)) = \pi(a)$

      It follows directly since $f$ is the identity
      
    \item $e \leq' e' \Leftrightarrow e \leq_1^l e'$
      \begin{itemize}
      \item[$\Rightarrow$] Assume $e \leq' e'$.

        By Definition~\ref{def:rem-init1}, $e \leq' e'$ entails $e, e' \in E'$.  From $f$ we know
        that $e, e' \not\in E_2$.  Hence $\neg(e \leq_2 e')$.  Thus it only remains that
        $e \leq_1 e'$ with $e, e' \in E_1$. Again, from $f$ we know that $e, e' \in E_1^l$ since
        $e, e' \neq l$ and $\neg(e \# l), \neg(e' \# l)$. Thus, by Definition~\ref{def:rem-init1},
        $e \leq_1^l e'$.
        
      \item[$\Leftarrow$] Assume $e \leq_1^l e'$.

        By Definition~\ref{def:rem-init1}, $e \leq_1^l e'$ entails $e \leq_1 e'$ with
        $e, e' \in E_1^l$. Hence $e, e' \neq l$, $\neg (e \#_1 l), \neg (e' \#_1 l)$, and
        $e, e' \in E_1$.  By Definition~\ref{def:pes-nd1} we have $e, e' \in E$ and $e \leq
        e'$. Furthermore, we also have $\neg (e \# l), \neg (e' \# l)$ and consequently
        $e, e' \in E'$. Thus $e \leq' e'$.
        
      \end{itemize}

    \item $e \#' e' \Leftrightarrow e \#_1^l e'$
      \begin{itemize}
      \item[$\Rightarrow$] Assume $e \#' e'$.

        By Definition~\ref{def:rem-init1} we have $e \# e'$ and $e, e' \in E'$ such that
        $e, e' \neq l$ and $\neg(e \# l), \neg(e' \# l)$.  By Definition~\ref{def:pes-nd1},
        $e \# e'$ entails $e \#_1 e'$ or $e \#_2 e'$ or $\set{e \# e' \mid e \in E_1, e' \in E_2}$.
        Since $l \in \init{\es{E}_1}$, then $l \in E_1$ and consequently $e, e' \not\in E_2$. Thus,
        by exclusion of hypothesis we have $e \#_1 e'$ and consequently
        $\neg(e \#_1 l), \neg(e' \#_1 l)$. Thus by Definition~\ref{def:rem-init1} we have
        $e \#_1^l e'$.

      \item[$\Leftarrow$] Assume $e \#_1^l e'$.

        By Definition~\ref{def:rem-init1} we have $e \#_1 e'$ and $e, e' \in E_1^l$, such that
        $e, e' \neq l$ and $\neg(e \#_1 l), \neg(e' \#_1 l)$. Furthermore $e, e' \in E_1$.  By
        Definition~\ref{def:pes-nd1}, $e, e' \in E$ and $\neg(e \# l), \neg(e' \# l)$. By
        Definition~\ref{def:rem-init1} we have $e \#' e'$.
      \end{itemize}
    \end{enumerate}

  \item $\es{E}_1 \backslash l \sqsubseteq (\nd{\es{E_1}}{\es{E_2}}) \backslash l$

    For this case, we notice that the function $f : E' \rightarrow E$ is again the identity, which
    is injective.

    We omit the remaining cases since they are similarly proved.
  \end{itemize}
\end{proof}

\subsubsection*{Proof of Lemma~\ref{lem:conc-rem-init1}}
\begin{proof}
  It is straightforward to see that if $l \in \init{\conc{\es{E_1}}{\es{E_2}}}$ then either
  $l \in \init{\es{E}_1}$ or $l \in \init{\es{E}_2}$, by Definition~\ref{def:pes-conc1}. Let us
  focus in the case where $l \in \init{\es{E}_1}$, since the proof for the other way around is
  similar.
  
  We need to prove
  \begin{itemize}
  \item $(\conc{\es{E_1}}{\es{E_2}}) \backslash l \sqsubseteq
    \conc{(\es{E_1}\backslash l)}{\es{E_2}}$
  \item $\conc{(\es{E_1}\backslash l)}{\es{E_2}} \sqsubseteq
    (\conc{\es{E_1}}{\es{E_2}}) \backslash l$
  \end{itemize}

  Let
  \begin{align*}
    & \es{E}_1 = \pes{E_1}{\leq_1}{\#_1} \\
    & \es{E}_2 = \pes{E_2}{\leq_2}{\#_2} \\
    & \conc{\es{E}_1}{\es{E}_2} = \pes{E}{\leq}{\#} \\
    & (\conc{\es{E}_1}{\es{E}_2}) \backslash l = \pes{E'}{\leq'}{\#'} \\
    & \es{E}_1 \backslash l = \pes{E_1^l}{\leq_1^l}{\#_1^l} \\
    & \es{E}_2 \backslash l = \pes{E_2^l}{\leq_2^l}{\#_2^l} \\
    & \conc{(\es{E}_1 \backslash l)}{\es{E}_2} = \pes{E^l}{\leq^l}{\#^l} \\
    & l \in \init{\conc{\es{E_1}}{\es{E_2}}}
  \end{align*}


  \begin{itemize}
  \item $(\conc{\es{E_1}}{\es{E_2}}) \backslash l \sqsubseteq
    \conc{(\es{E_1}\backslash l)}{\es{E_2}}$
    \begin{enumerate}
    \item We start by defining $f: E' \rightarrow E^l$ as being the identity, which is injective.

    \item $\pi(f(a)) = \pi(a)$

      It follows directly since $f$ is the identity
      
    \item $e \leq' e' \Leftrightarrow e \leq^l e'$

      \begin{itemize}
      \item[$\Rightarrow$] Assume $e \leq' e'$

        By Definition~\ref{def:rem-init1}, $e \leq e'$ and $e, e' \in E'$.
        We have two cases:
        \begin{enumerate}
        \item[$(i)$] $e \leq e'$ is of the form $e \leq_1 e'$

          Then $e, e' \in E_1$.  Since $e, e' \in E'$, then $e,e' \neq l$ and
          $\neg (e \# l), \neg (e' \# l)$. By Definition~\ref{def:pes-conc1}
          $\neg (e \#_1 l), \neg (e' \#_1 l)$. By Definition~\ref{def:rem-init1}, $e, e' \in E_1^l$
          and $e \leq_1^l e'$.  By Definition~\ref{def:pes-conc1}, $e \le^l e'$.
          
        \item[$(ii)$] $e \leq e'$ is of the form $e \leq_2 e'$

          It follows directly that $e \leq^l e'$, since $l \in E_1$.
        \end{enumerate}
        
      \item[$\Leftarrow$] Assume $e \leq^l e'$
        
        By Definition~\ref{def:rem-init1} and by Definition~\ref{def:pes-conc1} we have two cases:
        \begin{enumerate}
        \item[$(i)$] $e, e' \in E_1^l$ and $e \leq^l e'$ is of the form $e \leq_1^l e'$

          Since $e, e' \in E_1^l$, by Definition~\ref{def:rem-init1} we have $e, e' \in E_1$,
          $e, e' \neq l$ and $\neg (e \#_1 l), \neg (e' \#_1 l)$.  Hence, by
          Definition~\ref{def:pes-conc1} and Definition~\ref{def:rem-init1} we have $e, e' \in E'$.
          From $e \leq_1^l e'$, we know that from Definition~\ref{def:rem-init1} we have
          $e \leq_1 e'$ and $e, e' \in E_1^l$. It then follows by Definition~\ref{def:pes-conc1} and
          Definition~\ref{def:rem-init1} that $e \leq' e'$.
          
        \item[$(ii)$] $e, e' \in E_2^l$ and $e \leq^l e'$ is of the form $e \leq_2^l e'$

          It follows directly that $e \leq' e'$, since $l \in E_1$.
        \end{enumerate}
      \end{itemize}
      
    \item $e \#' e' \Leftrightarrow e \#^l e'$

      The reasoning for this case is similar to the previous one, since the definitions are
      identical.
    \end{enumerate}

    \item $\conc{(\es{E_1}\backslash l)}{\es{E_2}} \sqsubseteq
      (\conc{\es{E_1}}{\es{E_2}}) \backslash l$

      For this case, we notice that the function $f: E^l \rightarrow E'$ is again the identity,
      which is injective.

      We omit the remaining cases since they are similarly proved.
  \end{itemize}
\end{proof}

\subsubsection*{Proof of Lemma~\ref{lem:conc-symmetric1}}
\begin{proof}
  It follows directly from Definition~\ref{def:pes-conc1}.
\end{proof}

\subsubsection*{Proof of Lemma~\ref{res:soundI-1}}
\begin{proof}
  Induction over rules in Figure~\ref{fig:op-small1}.

  \begin{itemize}
  \item $\com{skip} \xrightarrow{sk} \checkmark$

    It follows directly that
    $\mf{\checkmark} \equiv \mf{\com{skip}} \backslash sk \equiv \emptyset$.

  \item $\com{a} \xrightarrow{a} \checkmark$

    It follows directly that
    $\mf{\checkmark} \equiv \mf{\com{a}} \backslash a \equiv \emptyset$.

  \item $\seq{\com{C}_1}{\com{C}_2} \xrightarrow{l} \com{C}_2$
    \begin{align*}
      & \seq{\com{C}_1}{\com{C}_2} \xrightarrow{l} \com{C}_2 \\
      \Rightarrow & \set{\text{Figure~\ref{fig:op-small1} entails}} \\
      & \com{C}_1 \xrightarrow{l} \checkmark \\
      \Rightarrow & \set{\text{i.h.}} \\
      & \mf{\checkmark} \equiv \mf{\com{C}_1} \backslash l \\
      \Rightarrow & \set{ \text{Lemma~\ref{lem:op-sim-1}}} \\
      & \seq{\mf{\checkmark}}{\mf{\com{C}_2}} \equiv
        \seq{(\mf{\com{C}_1} \backslash l)}{\mf{\com{C}_2}} \\
      \Rightarrow & \set{
                    \seq{\mf{\checkmark}}{\mf{\com{C}_2}} \equiv \mf{\com{C}_2},
                    \text{ Lemma~\ref{lem:seq-rem-init1}}
                    } \\
      & \mf{\com{C}_2} \equiv (\seq{\mf{\com{C}_1}}{\mf{\com{C}_2}}) \backslash l \\
      \Rightarrow & \set{\text{Definition~\ref{def:den-sem1}}} \\
      & \mf{\com{C}_2} \equiv \mf{\seq{\com{C}_1}{\com{C}_2}} \backslash l \\
    \end{align*}

  \item $\seq{\com{C}_1}{\com{C}_2} \xrightarrow{l} \seq{\com{C}'_1}{\com{C}_2}$
    \begin{align*}
      & \seq{\com{C}_1}{\com{C}_2} \xrightarrow{l} \seq{\com{C}'_1}{\com{C}_2} \\
      \Rightarrow & \set{\text{Figure~\ref{fig:op-small1} entails}} \\
      & \com{C}_1 \xrightarrow{l} \com{C}'_1 \\
      \Rightarrow & \set{\text{i.h.}} \\
      & \mf{\com{C}'_1} \equiv \mf{\com{C}_1} \backslash l \\
      \Rightarrow & \set{ \text{Lemma~\ref{lem:op-sim-1}}} \\
      & \seq{\mf{\com{C}'_1}}{\mf{\com{C}_2}} \equiv \seq{(\mf{\com{C}_1} \backslash l)}{\mf{\com{C}_2}} \\
      \Rightarrow & \set{\text{Lemma~\ref{lem:seq-rem-init1}}} \\
      & \seq{\mf{\com{C}'_1}}{\mf{\com{C}_2}} \equiv (\seq{\mf{\com{C}_1}}{\mf{\com{C}_2}}) \backslash l \\
      \Rightarrow & \set{\text{Definition~\ref{def:den-sem1}}} \\
      & \mf{\seq{\com{C}'_1}{\com{C}_2}} \equiv \mf{\seq{\com{C}_1}{\com{C}_2}} \backslash l \\
    \end{align*}

  \item $\nd{\com{C}_1}{\com{C}_2} \xrightarrow{l} \checkmark$
    \begin{align*}
      & \nd{\com{C}_1}{\com{C}_2} \xrightarrow{l} \checkmark \\
      \Rightarrow & \set{\text{Figure~\ref{fig:op-small1} entails}} \\
      & \com{C}_1 \xrightarrow{l} \checkmark
        \text{ or }
        \com{C}_2 \xrightarrow{l} \checkmark \\
      \Rightarrow & \set{\text{i.h.}} \\
      & \mf{\checkmark} \equiv \mf{\com{C}_1} \backslash l
        \text{ or }
        \mf{\checkmark} \equiv \mf{\com{C}_2} \backslash l \\
      \Rightarrow & \set{\text{Lemma~\ref{lem:nd-rem-init1} for both cases,
                    \text{Definition~\ref{def:den-sem1}}
                    }} \\
      & \mf{\checkmark} \equiv \mf{\nd{\com{C}_1}{\com{C}_2}} \backslash l
    \end{align*}

  \item $\nd{\com{C}_1}{\com{C}_2} \xrightarrow{l} \com{C}'$
    \begin{align*}
      & \nd{\com{C}_1}{\com{C}_2} \xrightarrow{l} \com{C}' \\
      \Rightarrow & \set{\text{Figure~\ref{fig:op-small1} entails}} \\
      & \com{C}_1 \xrightarrow{l} \com{C}'_1
        \text{ or }
        \com{C}_2 \xrightarrow{l} \com{C}'_2 \\
      \Rightarrow & \set{\text{i.h.}} \\
      & \mf{\com{C}'_1} \equiv \mf{\com{C}_1} \backslash l
        \text{ or }
        \mf{\com{C}'_2} \equiv \mf{\com{C}_2} \backslash l \\
      \Rightarrow & \set{\text{Lemma~\ref{lem:nd-rem-init1} for both cases},
                    \text{Definition~\ref{def:den-sem1}} } \\
      & \mf{\com{C}'_1} \equiv (\mf{\nd{\com{C}_1}{\com{C}_2}}) \backslash l
        \text{ or }
        \mf{\com{C}'_2} \equiv \mf{\nd{\com{C}_1}{\com{C}_2}} \backslash l
    \end{align*}

  \item $\conc{\com{C}_1}{\com{C}_2} \xrightarrow{l} \com{C}_2$
    \begin{align*}
      & \conc{\com{C}_1}{\com{C}_2} \xrightarrow{l} \com{C}_2 \\
      \Rightarrow & \set{\text{Figure~\ref{fig:op-small1} entails}} \\
      & \com{C}_1 \xrightarrow{l} \checkmark \\
      \Rightarrow & \set{\text{i.h.}} \\
      & \mf{\checkmark} \equiv \mf{\com{C}_1} \backslash l \\
      \Rightarrow & \set{\text{Lemma~\ref{lem:op-sim-1}}} \\
      & \conc{\mf{\checkmark}}{\mf{\com{C}_2}} \equiv \conc{(\mf{\com{C}_1} \backslash l)}{\mf{\com{C}_2}} \\
      \Rightarrow & \set{\conc{\mf{\checkmark}}{\mf{\com{C}_2}} \equiv \mf{\com{C}_2} } \\
      & \mf{\com{C}_2} \equiv \conc{(\mf{\com{C}_1} \backslash l)}{\mf{\com{C}_2}} \\
      \Rightarrow & \set{\text{Lemma~\ref{lem:conc-rem-init1}},
                    \text{Definition~\ref{def:den-sem1}}} \\
      & \mf{\com{C}_2} \equiv \mf{\conc{\com{C}_1}{\com{C}_2}} \backslash l \\
    \end{align*}

  \item $\conc{\com{C}_1}{\com{C}_2} \xrightarrow{l} \conc{\com{C}'_1}{\com{C}_2}$
    \begin{align*}
      & \conc{\com{C}_1}{\com{C}_2} \xrightarrow{l} \conc{\com{C}'_1}{\com{C}_2} \\
      \Rightarrow & \set{\text{Figure~\ref{fig:op-small1} entails}} \\
      & \com{C}_1 \xrightarrow{l} \com{C}'_1 \\
      \Rightarrow & \set{\text{i.h.}} \\
      & \mf{\com{C}'_1} \equiv \mf{\com{C}_1} \backslash l \\
      \Rightarrow & \set{\text{Lemma~\ref{lem:op-sim-1}}} \\
      & \conc{\mf{\com{C}'_1}}{\mf{\com{C}_2}} \equiv \conc{(\mf{\com{C}_1} \backslash l)}{\mf{\com{C}_2}} \\
      \Rightarrow & \set{\text{Lemma~\ref{lem:conc-rem-init1}},
                    \text{Definition~\ref{def:den-sem1}}} \\
      & \mf{\conc{\com{C}'_1}{\com{C}_2}} \equiv \mf{\conc{\com{C}_1}{\com{C}_2}} \backslash l \\
    \end{align*}

  \item $\conc{\com{C}_1}{\com{C}_2} \xrightarrow{l} \com{C}_1$
    \begin{align*}
      & \conc{\com{C}_1}{\com{C}_2} \xrightarrow{l} \com{C}_1 \\
      \Rightarrow & \set{\text{Figure~\ref{fig:op-small1} entails}} \\
      & \com{C}_2 \xrightarrow{l} \checkmark \\
      \Rightarrow & \set{\text{i.h.}} \\
      & \mf{\checkmark} \equiv \mf{\com{C}_2} \backslash l \\
      \Rightarrow & \set{\text{Lemma~\ref{lem:op-sim-1}}} \\
      & \conc{\mf{\com{C}_1}}{\mf{\checkmark}} \equiv
        \conc{\mf{\com{C}_1}}{(\mf{\com{C}_2} \backslash l)} \\
      \Rightarrow & \set{\conc{\mf{\com{C}_1}}{\mf{\checkmark}} \equiv \mf{\com{C}_1} } \\
      & \mf{\com{C}_1} \equiv \conc{\mf{\com{C}_1}}{(\mf{\com{C}_2} \backslash l)} \\
      \Rightarrow & \set{\text{Lemma~\ref{lem:conc-rem-init1}},
                    \text{Definition~\ref{def:den-sem1}}} \\
      & \mf{\com{C}_1} \equiv \mf{\conc{\com{C}_1}{\com{C}_2}} \backslash l \\
    \end{align*}

  \item $\conc{\com{C}_1}{\com{C}_2} \xrightarrow{l} \conc{\com{C}_1}{\com{C}'_2}$
    \begin{align*}
      & \conc{\com{C}_1}{\com{C}_2} \xrightarrow{l} \conc{\com{C}_1}{\com{C}'_2} \\
      \Rightarrow & \set{\text{Figure~\ref{fig:op-small1} entails}} \\
      & \com{C}_2 \xrightarrow{l} \com{C}'_2 \\
      \Rightarrow & \set{\text{i.h.}} \\
      & \mf{\com{C}'_2} \equiv \mf{\com{C}_2} \backslash l \\
      \Rightarrow & \set{\text{Lemma~\ref{lem:op-sim-1}}} \\
      & \conc{\mf{\com{C}_1}}{\mf{\com{C}'_2}} \equiv \conc{\mf{\com{C}_1}}{(\mf{\com{C}_2} \backslash l)} \\
      \Rightarrow & \set{\text{Lemma~\ref{lem:conc-rem-init1}},
                    \text{Definition~\ref{def:den-sem1}}} \\
      & \mf{\conc{\com{C}_1}{\com{C}'_2}} \equiv \mf{\conc{\com{C}_1}{\com{C}_2}} \backslash l \\
    \end{align*}
  \end{itemize}
\end{proof}

\subsubsection*{Proof of Theorem~\ref{res:soundII-1}}
\begin{proof}
  Induction over the length of $\omega$, which is denoted by $|\omega|$.

  \begin{itemize}
  \item $|\omega| = 1$

    It follows directly that
    $\exists \set{l} \in \confES{\mf{\com{C}}}\, .\, \emptyset \cchain{l} \set{l}$

  \item $|\omega| > 1$
    \begin{align*}
      & \com{C} \xtwoheadrightarrow{\omega} \com{C}' \\
      \Rightarrow & \set{\text{Definition~\ref{fig:op-nstep1}}} \\
      & \com{C} \xrightarrow{l} \com{C}'' \qquad
        \com{C}'' \xtwoheadrightarrow{\omega'} \com{C}' \\
      \Rightarrow & \set{\text{Lemma~\ref{res:soundI-1}, i.h.}} \\
      & \mf{\com{C}''} \equiv \mf{\com{C}} \backslash l \qquad
        \exists y \in \confES{\mf{\com{C}''}}\, .\, \emptyset \cchain{\omega'} y \\
      \Rightarrow & \set{\text{Definition~\ref{def:rem-init1}}} \\
      & \set{l} \cup y \in \confES{\mf{\com{C}}}\, .\, \emptyset \cchain{l} \set{l} \cchain{\omega'} \set{l} \cup y = x
    \end{align*}
  \end{itemize}
\end{proof}

\subsubsection*{Proof of Lemma~\ref{res:adI-1}}
\begin{proof}
  Induction over the interpretation of commands.

  \begin{itemize}
  \item $sk \in \init{\mf{\com{skip}}}$

    Let $\com{C}' = \checkmark$.  It follows directly that $\com{skip} \xrightarrow{sk} \checkmark$
    and that $\mf{\com{skip}} \backslash sk \equiv \mf{\checkmark}$.

  \item $a \in \init{\mf{\com{a}}}$

    Let $\com{C}' = \checkmark$.  It follows directly that $\com{a} \xrightarrow{a} \checkmark$
    and that $\mf{\com{a}} \backslash a \equiv \mf{\checkmark}$.

  \item $l \in \init{\mf{\seq{\com{C}_1}{\com{C}_2}}}$

    By Definition~\ref{def:pes-seq1} we have that $l \in \init{\mf{\com{C}_1}}$.  By i.h.,
    $\exists \com{C}'$ such that $\com{C}_1 \xrightarrow{l} \com{C}'$ and
    $\mf{\com{C}_1} \backslash l = \mf{\com{C}'}$. We have two cases:
    \begin{enumerate}
    \item $\com{C}' = \checkmark$

      We have $\com{C}_1 \xrightarrow{l} \checkmark$ and
      $\mf{\com{C}_1} \backslash l \equiv \mf{\checkmark}$.  By the rules in Figure~\ref{fig:op-small1},
      $\seq{\com{C}_1}{\com{C}_2} \xrightarrow{l} \com{C}_2$.
      By Definition~\ref{def:pes-seq1}, $\seq{(\mf{\com{C}_1} \backslash l)}{\mf{\com{C}_2}} \equiv
      \seq{\mf{\checkmark}}{\mf{\com{C}_2}} \equiv \mf{\com{C}_2}$.
      
    \item $\com{C}' = \com{C}'_1$

      We have $\com{C}_1 \xrightarrow{l} \com{C}'_1$ and
      $\mf{\com{C}_1} \backslash l \equiv \mf{\com{C}'_1}$.  By the rules in Figure~\ref{fig:op-small1},
      $\seq{\com{C}_1}{\com{C}_2} \xrightarrow{l} \seq{\com{C}'_1}{\com{C}_2}$.  By
      Definition~\ref{def:pes-seq1},
      $\seq{(\mf{\com{C}_1} \backslash l)}{\mf{\com{C}_2}} \equiv \seq{\mf{\com{C}'_1}}{\mf{\com{C}_2}}$.
      By Definition~\ref{def:den-sem1}, $\mf{\seq{\com{C}'_1}{\com{C}_2}}$.
    \end{enumerate}

  \item $l \in \init{\mf{\nd{\com{C}_1}{\com{C}_2}}}$

    By Definition~\ref{def:pes-nd1} we have two cases: 
    \begin{enumerate}
    \item $l \in \init{\mf{\com{C}_1}}$

      By i.h. $\exists \com{C}'\, .\, \com{C}_1 \xrightarrow{l} \com{C}'$ and
      $\mf{\com{C}_1} \backslash l \equiv \mf{\com{C}'}$. 
      By the rules in Figure~\ref{fig:op-small1} we have two cases:
      \begin{enumerate}
      \item $\com{C}' = \checkmark$

        We have $\com{C}_1 \xrightarrow{l} \checkmark$ and
        $\mf{\com{C}_1} \backslash l \equiv \mf{\checkmark}$.  By the rules in Figure~\ref{fig:op-small1}
        we have $\nd{\com{C}_1}{\com{C}_2} \xrightarrow{l} \checkmark$.  By
        Lemma~\ref{lem:nd-rem-init1} we have
        $\mf{\checkmark} \equiv \mf{\com{C}_1} \backslash l \equiv
        (\nd{\mf{\com{C}_1}}{\mf{\com{C}_2}}) \backslash l$.
        By Definition~\ref{def:den-sem1}, $\mf{\nd{\com{C}_1}{\com{C}_2}} \backslash l$.

      \item $\com{C}' = \com{C}'_1$

        We have $\com{C}_1 \xrightarrow{l} \com{C}'_1$ and
        $\mf{\com{C}_1} \backslash l \equiv \mf{\com{C}'_1}$.  By the rules in Figure~\ref{fig:op-small1}
        we have $\nd{\com{C}_1}{\com{C}_2} \xrightarrow{l} \com{C}'_1$.  By
        Lemma~\ref{lem:nd-rem-init1} we have
        $\mf{\com{C}'_1} \equiv \mf{\com{C}_1} \backslash l \equiv (\nd{\mf{\com{C}_1}}{\mf{\com{C}_2}})
        \backslash l$.  By Definition~\ref{def:den-sem1},
        $\mf{\nd{\com{C}_1}{\com{C}_2}} \backslash l$.
      \end{enumerate}
      
    \item $l \in \init{\mf{\com{C}_2}}$
      
      By i.h. $\exists \com{C}'\, .\, \com{C}_2 \xrightarrow{l} \com{C}'$ and
      $\mf{\com{C}_2} \backslash l \equiv \mf{\com{C}'}$. 
      By the rules in Figure~\ref{fig:op-small1} we have two cases:
      \begin{enumerate}
      \item $\com{C}' = \checkmark$

        We have $\com{C}_2 \xrightarrow{l} \checkmark$ and
        $\mf{\com{C}_2} \backslash l \equiv \mf{\checkmark}$.  By the rules in Figure~\ref{fig:op-small1}
        we have $\nd{\com{C}_1}{\com{C}_2} \xrightarrow{l} \checkmark$.  By
        Lemma~\ref{lem:nd-rem-init1} we have
        $\mf{\checkmark} \equiv \mf{\com{C}_2} \backslash l \equiv
        (\nd{\mf{\com{C}_1}}{\mf{\com{C}_2}}) \backslash l$.
        By Definition~\ref{def:den-sem1}, $\mf{\nd{\com{C}_1}{\com{C}_2}} \backslash l$.

      \item $\com{C}' = \com{C}'_2$

        We have $\com{C}_2 \xrightarrow{l} \com{C}'_2$ and
        $\mf{\com{C}_2} \backslash l \equiv \mf{\com{C}'_2}$.  By the rules in Figure~\ref{fig:op-small1}
        we have $\nd{\com{C}_1}{\com{C}_2} \xrightarrow{l} \com{C}'_2$.  By
        Lemma~\ref{lem:nd-rem-init1} we have
        $\mf{\com{C}'_2} \equiv \mf{\com{C}_2} \backslash l \equiv
        (\nd{\mf{\com{C}_1}}{\mf{\com{C}_2}}) \backslash l$.
        By Definition~\ref{def:den-sem1}, $\mf{\nd{\com{C}_1}{\com{C}_2}} \backslash l$.
      \end{enumerate}
    \end{enumerate}

  \item $l \in \init{\mf{\conc{\com{C}_1}{\com{C}_2}}}$

    By Definition~\ref{def:pes-conc1} we have two cases: 
    \begin{enumerate}
    \item $l \in \init{\mf{\com{C}_1}}$

      By i.h. $\exists \com{C}'\, .\, \com{C}_1 \xrightarrow{l} \com{C}'$ and
      $\mf{\com{C}_1} \backslash l \equiv \mf{\com{C}'}$. 
      By the rules in Figure~\ref{fig:op-small1} we have two cases:
      \begin{enumerate}
      \item $\com{C}' = \checkmark$

        We have $\com{C}_1 \xrightarrow{l} \checkmark$ and
        $\mf{\com{C}_1} \backslash l \equiv \mf{\checkmark}$.  By the rules in
        Figure~\ref{fig:op-small1} we have $\conc{\com{C}_1}{\com{C}_2} \xrightarrow{l} \com{C}_2$.
        By Definition~\ref{def:pes-conc1} we have
        $\conc{(\mf{\com{C}_1}\backslash l)}{\mf{\com{C}_2}}$.
        By Lemma~\ref{lem:conc-rem-init1} we have
        $(\conc{\mf{\com{C}_1}}{\mf{\com{C}_2}}) \backslash l$.  By
        Definition~\ref{def:den-sem1}, $\mf{\conc{\com{C}_1}{\com{C}_2}} \backslash l$.

      \item $\com{C}' = \com{C}'_1$

        We have $\com{C}_1 \xrightarrow{l} \com{C}'_1$ and
        $\mf{\com{C}_1} \backslash l \equiv \mf{\com{C}'_1}$.  By the rules in
        Figure~\ref{fig:op-small1} we have
        $\conc{\com{C}_1}{\com{C}_2} \xrightarrow{l} \conc{\com{C}'_1}{\com{C}_2}$. By
        Definition~\ref{def:pes-conc1},
        $\conc{(\mf{\com{C}_1}\backslash l)}{\mf{\com{C}_2}}$.
        By Lemma~\ref{lem:conc-rem-init1} we have
        $(\conc{\mf{\com{C}_1}}{\mf{\com{C}_2}}) \backslash l$.  By
        Definition~\ref{def:den-sem1}, $\mf{\conc{\com{C}_1}{\com{C}_2}} \backslash l$.
      \end{enumerate}
      
    \item $l \in \init{\mf{\com{C}_2}}$
      
      By i.h. $\exists \com{C}'\, .\, \com{C}_2 \xrightarrow{l} \com{C}'$ and
      $\mf{\com{C}_2} \backslash l \equiv \mf{\com{C}'}$. 
      By the rules in Figure~\ref{fig:op-small1} we have two cases:
      \begin{enumerate}
      \item $\com{C}' = \checkmark$

        We have $\com{C}_2 \xrightarrow{l} \checkmark$ and
        $\mf{\com{C}_2} \backslash l \equiv \mf{\checkmark}$.  By the rules in
        Figure~\ref{fig:op-small1} we have $\conc{\com{C}_1}{\com{C}_2} \xrightarrow{l} \com{C}_1$.
        By Definition~\ref{def:pes-conc1} we have
        $\conc{\mf{\com{C}_1}}{(\mf{\com{C}_2}\backslash l)}$.
        By Lemma~\ref{lem:conc-rem-init1} we have
        $(\conc{\mf{\com{C}_1}}{\mf{\com{C}_2}}) \backslash l$.  By
        Definition~\ref{def:den-sem1}, $\mf{\conc{\com{C}_1}{\com{C}_2}} \backslash l$.

      \item $\com{C}' = \com{C}'_2$

        We have $\com{C}_2 \xrightarrow{l} \com{C}'_2$ and
        $\mf{\com{C}_2} \backslash l \equiv \mf{\com{C}'_2}$.  By the rules in
        Figure~\ref{fig:op-small1} we have
        $\conc{\com{C}_1}{\com{C}_2} \xrightarrow{l} \conc{\com{C}_1}{\com{C}'_2}$. By
        Definition~\ref{def:pes-conc1},
        $\conc{\mf{\com{C}_1}}{(\mf{\com{C}_2}\backslash l)}$.
        By Lemma~\ref{lem:conc-rem-init1} we have
        $(\conc{\mf{\com{C}_1}}{\mf{\com{C}_2}}) \backslash l$.  By
        Definition~\ref{def:den-sem1}, $\mf{\conc{\com{C}_1}{\com{C}_2}} \backslash l$.
      \end{enumerate}
    \end{enumerate}
  \end{itemize}
\end{proof}

\subsubsection*{Proof of Theorem~\ref{res:adII-1}}
\begin{proof}
  Induction over the length of $\omega$.  
  \begin{itemize}
  \item $|\omega| = 1$

    We have $\set{l} \in \confES{\mf{\com{C}}}$ such that $\emptyset \cchain{l} \set{l}$.
    Furthermore $l \in \init{\mf{\com{C}}}$.  By Lemma~\ref{res:adI-1},
    $\com{C} \xrightarrow{l} \com{C}'$ and
    $\mf{\com{C}'} \equiv \mf{\com{C}} \backslash l$.  By the rules in
    Figure~\ref{fig:op-nstep1}, $\com{C} \xtwoheadrightarrow{l} \com{C}'$.

  \item $|\omega| > 1$

    We have $x \in \confES{\mf{\com{C}}}$ such that $\emptyset \cchain{\omega} x$.  Since
    $\omega = l_0 l_1 \dots l_n$, then $\emptyset \cchain{l_0} \set{l_0} \cchain{\omega'} x$.  Hence
    $l_0 \in \init{\mf{\com{C}}}$. By Lemma~\ref{res:adI-1},
    $\com{C} \xrightarrow{l_0} \com{C}'$ and
    $\mf{\com{C}'} \equiv \mf{\com{C}} \backslash l_0$.  By
    Definition~\ref{def:rem-init1}, $\exists y \in \confES{\mf{\com{C}'}}$ such that
    $\emptyset \cchain{\omega'} y$.  By i.h. $\exists \com{C}''$ such that
    $\com{C}' \xtwoheadrightarrow{\omega'} \com{C}''$.  By the rules in Figure~\ref{fig:op-nstep1},
    $\com{C} \xtwoheadrightarrow{\omega} \com{C}''$, where $\omega = l_0 : \omega'$.
  \end{itemize}
\end{proof}

\subsubsection*{Proof of Lemma~\ref{lem:po-1}}
\begin{proof}
  Let $\es{E}_1$, $\es{E}_2$, and $\es{E}_3$ be event structures.
  \begin{itemize}
  \item Reflexivity: $\es{E}_1 \trianglelefteq \es{E}_1$

    It follows directly from Definition~\ref{def:pes-fix-order-1}.

  \item Transitivity: $\es{E}_1 \trianglelefteq \es{E}_2, \es{E}_2 \trianglelefteq \es{E}_3 \Rightarrow \es{E}_1 \trianglelefteq \es{E}_3$
    \begin{enumerate}
    \item $E_1 \subseteq E_3$

      Let $e \in E_1$. Since $\es{E}_1 \trianglelefteq \es{E}_2$ then $e \in E_2$. Since
      $\es{E}_2 \trianglelefteq \es{E}_3$ then $e \in E_3$. Hence $E_1 \subseteq E_3$.

    \item $e \leq_1 e' \Leftrightarrow e, e' \in E_1, e \leq_3 e'$
      \begin{itemize}
      \item[$\Rightarrow$] Let $e \leq_1 e'$.

        Clearly $e, e' \in E_1$. Since $\es{E}_1 \trianglelefteq \es{E}_2$ then $e \leq_2
        e'$. Furthermore $e, e' \in E_2$.  Since $\es{E}_2 \trianglelefteq \es{E}_3$ then
        $e \leq_3 e'$.
      \item[$\Leftarrow$] Let $e, e' \in E_1, e \leq_3 e'$.

        Since $\es{E}_1 \trianglelefteq \es{E}_2$ then $e, e' \in E_2$.  Since
        $\es{E}_2 \trianglelefteq \es{E}_3$ then $e \leq_2 e'$.  Since
        $\es{E}_1 \trianglelefteq \es{E}_2$ $e \leq_1 e'$.
      \end{itemize}
      
    \item $e \#_1 e' \Leftrightarrow e, e' \in E_1, e \#_3 e'$

      Similar to $\leq$, \ie\ previous point.
    \end{enumerate}

  \item Antisymmetry: $\es{E}_1 \trianglelefteq \es{E}_2, \es{E}_2 \trianglelefteq \es{E}_1 \Rightarrow \es{E}_1 = \es{E}_2$
    \begin{enumerate}
    \item $E_1 = E_2$

      Let $e \in E_1$. Since $\es{E}_1 \trianglelefteq \es{E}_2$ then $e \in E_2$.  Let $e \in
      E_2$. Since $\es{E}_2 \trianglelefteq \es{E}_1$ then $e \in E_1$.  Hence $E_1 = E_2$.

    \item $(e \leq_1 e' \Leftrightarrow e, e' \in E_1, e \leq_2 e') = (e \leq_2 e' \Leftrightarrow e, e' \in E_2, e \leq_1 e')$
      \begin{itemize}
      \item[$\Rightarrow$] Let $e \leq_1 e'$.  Clearly $e, e' \in \es{E}_1$.  From
        $\es{E}_1 \trianglelefteq \es{E}_2$, $e \leq_2 e'$.

        Let $e \leq_2 e'$.  Clearly $e, e' \in \es{E}_2$.  From
        $\es{E}_2 \trianglelefteq \es{E}_1$, $e \leq_1 e'$.

      \item[$\Leftarrow$] Let $e, e' \in E_1, e \leq_2 e'$. Since
        $\es{E}_1 \trianglelefteq \es{E}_2$ then $e \leq_1 e'$.  Let $e, e' \in E_2, e \leq_1
        e'$. Since $\es{E}_2 \trianglelefteq \es{E}_1$ then $e \leq_2 e'$.
      \end{itemize}
      Hence $\leq_1 = \leq_2$

      \item $(e \#_1 e' \Leftrightarrow e, e' \in E_1, e \#_2 e') = (e \#_2 e' \Leftrightarrow e, e' \in E_2, e \#_1 e')$

    Similar reasoning as previous point.
    \end{enumerate}
  \end{itemize}
\end{proof}

\subsubsection*{Proof of Lemma~\ref{lem:po-least-elem-1}}
\begin{proof}
  \begin{itemize}
  \item $\bot$ is an event structure

    It follows directly that all conditions in Definition~\ref{def:pes} are trivially satisfied
    because $\bot$ has no events.

  \item For any event structure $\es{E} = \pes{E}{\leq}{\#}$ we want to show
    $\bot \trianglelefteq \es{E}$.
    \begin{enumerate}
    \item $\emptyset \subseteq E$

      Trivially holds.
      
    \item $e \leq_{\bot} e' \Leftrightarrow e, e' \in \emptyset, e \leq e'$

      Since $\bot$ has no events and the causal relation is empty, it follows that $e \leq_{\bot} e'$
      and $e, e' \in \emptyset$ are false. Hence the condition trivially holds.

    \item $e \#_{\bot} e' \Leftrightarrow e, e' \in \emptyset, e \# e'$

      Similar to previous point.
    \end{enumerate}
  \end{itemize}
\end{proof}

\subsubsection*{Proof of Lemma~\ref{lem:lub-es-1}}
\begin{proof}
  It follows directly from Definition~\ref{def:lub-1} that $\leq^{\omega}$ is a partial order and
  that $\#^{\omega}$ is symmetric and irreflexive, since for every case, it exists $n \in \omega$
  such that $\leq^{\omega} = \leq_n$ and $\#^{\omega} = \#_n$, with $\leq_n$ being a partial order
  and $\#_n$ being symmetric and irreflexive.
  
  \begin{itemize}
  \item $\set{e' \mid e' \leq^\omega e}$ is finite

    By Definition~\ref{def:lub-1}, $e' \leq^\omega e'$ entails $\exists n \in \omega$ such that
    $e' \leq_n e$.  Consequently, $e, e' \in E_n$.  Furthermore, $\es{E}_n$ is an event
    structure. Hence $\set{e' \mid e' \leq_n e}$ is finite.  Thus $\set{e' \mid e' \leq^\omega e}$
    is finite.

  \item $e \#^\omega e' \leq^\omega e'' \Rightarrow e \#^\omega e''$

    By Definition~\ref{def:lub-1}, $e' \leq^\omega e'$ entails $\exists n \in \omega$ such that
    $e \#_n e' \leq_n e''$, where $\es{E}_n$ is an event structure.  Hence $e \#_n e''$. Thus
    $e \#^\omega e''$.
  \end{itemize}
\end{proof}

\subsubsection*{Proof of Lemma~\ref{lem:lub-1}}
\begin{proof}
  \begin{itemize}
  \item $\es{E}^\omega$ is an upper bound

    $\forall n \in \omega$ we need to have $\es{E}_n \trianglelefteq \es{E}^\omega$.  It follows
    directly from Definition~\ref{def:pes-fix-order-1} that
    $\forall n \in \omega\, \es{E}_n \trianglelefteq \es{E}^\omega$ since $\es{E}^\omega$ is by
    definition the union of all $\es{E}_n$.
    
  \item $\es{E}^\omega$ is the least upper bound

    Let $\es{E} = \pes{E}{\leq}{\#}$ be an upper bound of the chain. We need to show that if
    $\es{E}_n \trianglelefteq \es{E}^\omega$ and $\es{E}_n \trianglelefteq \es{E}$ then
    $\es{E}^\omega \trianglelefteq \es{E}$.
    \begin{enumerate}
    \item $E^\omega \subseteq E$

      Let $e \in E^\omega$.  By Definition~\ref{def:lub-1}, $\exists n \in \omega$ such that
      $e \in E_n$. By $\es{E}_n \trianglelefteq \es{E}$ we have $e \in E$.
      
    \item $e \leq^\omega e' \Leftrightarrow e, e' \in E^\omega$ and $e \leq e'$
      \begin{itemize}
      \item[$\Rightarrow$] Let $ e \leq^\omega e'$.

        By Definition~\ref{def:lub-1}, $\exists n \in \omega$ such that $e \leq_n e'$.  It is then
        clear that $e, e' \in E_n \subseteq E^\omega$.  Since $\es{E}_n \trianglelefteq \es{E}$ we
        have $e \leq e'$.
        
      \item[$\Leftarrow$] $e, e' \in E^\omega$ and $e \leq e'$

        By Definition~\ref{def:lub-1}, $\exists n \in \omega$ such that $e, e' \in E_n$.  Since
        $\es{E}_n \trianglelefteq \es{E}$, $e \leq_n e'$.  By Definition~\ref{def:lub-1},
        $e \leq^\omega e'$.
      \end{itemize}
      
    \item $e \#^\omega e' \Leftrightarrow e, e' \in E^\omega$ and $e \# e'$

      Similar to previous point.
    \end{enumerate}
  \end{itemize}
\end{proof}

\subsubsection*{Proof of Lemma~\ref{lem:seq-fix-mono-1}}
\begin{proof}
  Let $\es{E} = \pes{E}{\leq}{\#}$, $\es{E}_1 = \pes{E_1}{\leq_1}{\#_1}$,
  $\es{E}_2 = \pes{E_2}{\leq_2}{\#_2}$, $\seq{\es{E}}{\es{E}_1} = \pes{E^1}{\leq^1}{\#^1}$, and
  $\seq{\es{E}}{\es{E}_2} = \pes{E^2}{\leq^2}{\#^2}$, such that $\es{E}_1 \trianglelefteq \es{E}_2$.

  \begin{enumerate}
  \item $E^1 \subseteq E^2 \Leftrightarrow E \uplus (E_1 \times \confmax{\es{E}}) \subseteq E \uplus (E_2 \times \confmax{\es{E}})$

    By Definition~\ref{def:pes-seq1} we have two cases:
    \begin{enumerate}
    \item $e \in E$

      Then we are done.

    \item $e \in E_1 \times \confmax{\es{E}}$

      We know that $e$ is of the form $(e_1, x)$ where $e_1 \in E_1$ and $x \in \confmax{\es{E}}$.
      Since $\es{E}_1 \trianglelefteq \es{E}_2$ then $e_1 \in E_2$ and consequently
      $e \in \confmax{\es{E}}$.      
    \end{enumerate}
    
  \item $\forall e, e'\ .\ e \leq^1 e' \Leftrightarrow e, e' \in E^1$ and $e \leq^2 e'$
    \begin{itemize}
    \item[$\Rightarrow$]
      Assume $e \leq^1 e$.
      Clearly $e, e' \in E^1$.
      By Definition~\ref{def:pes-seq1} we have three cases:
      \begin{enumerate}
      \item $e \leq^1 e'$ is of the form $e \leq e'$

        Hence $e, e' \in E$. By Definition~\ref{def:pes-seq1} $e \leq^2 e'$.

      \item $e \leq e'$ is of the form $e \leq_{1} e'$

        We know that $e, e'$ are of the form $(e,x), (e',x) \in E_1 \times \confmax{\es{E}}$, which
        entails $e, e' \in E_1$ and $x \in \confmax{\es{E}}$. Since
        $\es{E}_1 \trianglelefteq \es{E}_2$, $e, e' \in E_2$ and $e \leq_2 e'$. By
        Definition~\ref{def:pes-seq1} we have $(e, x) \leq^2 (e', x)$.
        
      \item $e \leq^1 e'$ is of the form $e \leq^1 (e',x)$

        We know that $e \in E$, $e \in x \in \confmax{\es{E}}$, and
        $(e',x) \in E_1 \times \confmax{es{E}}$, with the last entailing $e' \in E_1$. Since
        $\es{E}_1 \trianglelefteq \es{E}_2$, $e' \in E_2$ and consequently
        $(e', x) \in E_2 \times \confmax{\es{E}}$.  By Definition~\ref{def:pes-seq1} we have
        $e \leq^2 (e',x)$.
      \end{enumerate}

    \item[$\Leftarrow$] Assume $e, e' \in E^1$ and $e \leq^2 e'$. The cases are distinguished by
      $\leq^2$.
      \begin{enumerate}
      \item $e \leq^2 e'$ is of the form $e \leq e'$

        Hence $e, e' \in E$. By Definition~\ref{def:pes-seq1}, $e \leq^1 e'$.

      \item $e \leq^2 e$ is of the form $e \leq_{2} e'$

        We know that $e, e'$ are of the form $(e,x), (e',x) \in E_2 \times \confmax{\es{E}}$, which
        entails $e, e' \in E_2$ and $x \in \confmax{\es{E}}$. Since
        $\es{E}_1 \trianglelefteq \es{E}_2$ and $e, e' \in E^1$, which entails for this case that
        $e, e' \in E_1$, then $e \leq_1 e'$. By Definition~\ref{def:pes-seq1} we have
        $(e, x) \leq^1 (e', x)$.

      \item $e \leq^2 e'$ is of the form $e \leq^2 (e',x)$

        We know that $e \in E$, $e \in x \in \confmax{\es{E}}$, and
        $(e',x) \in E_2 \times \confmax{es{E}}$, with the last entailing $e' \in E_2$. Since
        $\es{E}_1 \trianglelefteq \es{E}_2$ and $e' \in E^1$, which entails for this case that
        $e' \in E_1$, then $(e', x) \in E_1 \times \confmax{\es{E}}$.  By
        Definition~\ref{def:pes-seq1} we have $e \leq^1 (e',x)$.
      \end{enumerate}
    \end{itemize}

  \item $\forall e, e'\ .\ e \#^1 e' \Leftrightarrow e, e' \in E$ and $e \#^2 e'$
    \begin{itemize}
    \item[$\Rightarrow$] Assume $e \#^1 e'$.

      Clearly $e, e' \in E^1$.  From Definition~\ref{def:pes-seq1} we have that
      $\exists (a \leq^1 e,\ a' \leq^1 e')$ such that $a \# a'$ or $e \#_1 e'$. For the former we
      have that $a \leq^1 e,\ a' \leq^1 e'$ entails $a \leq^2 e,\ a' \leq^2 e'$. For the latter we
      have that $\es{E}_1 \trianglelefteq \es{E}_2$, hence $e \#_2 e'$.  Thus $e \#^2 e'$.

    \item[$\Leftarrow$] Assume $e, e' \in E^1$ and $e \#^2 e'$.

      By Definition~\ref{def:pes-seq1} it exists $a \leq^2 e$ and $a' \leq^2 e'$ such that $a \# a'$
      or $e \#_2 e'$. Since $e, e' \in E^1$ and $\es{E}_1 \trianglelefteq \es{E}_2$ we have
      $a \leq^1 e,\ a' \leq^1 e'$ and $e \#_1 e'$.  It then follows directly that $e \#^1 e'$.
    \end{itemize}
  \end{enumerate}
\end{proof}

\subsubsection*{Proof of Lemma~\ref{lem:conc-fix-mono-1}}
\begin{proof}
  It follows directly from Definition~\ref{def:pes-conc1}.
\end{proof}

\subsubsection*{Proof of Lemma~\ref{lem:nd-fix-mono-1}}
\begin{proof}
  It follows directly from Definition~\ref{def:pes-nd1}.
\end{proof}

\subsubsection*{Proof of Lemma~\ref{lem:cont-1}}
\begin{proof}
 \begin{itemize}
 \item $\Rightarrow$: Assume $op$ is continuous.

   We have $op(\bigsqcup_n \es{E}_n) = \bigsqcup_n op (\es{E}_n)$.  Let
   $\es{E}_1 \trianglelefteq \dots \trianglelefteq \es{E}_n \trianglelefteq \dots$ and
   $\es{E}'_1 \trianglelefteq \dots \trianglelefteq \es{E}'_n \trianglelefteq \dots$ be two
   $\omega$-chains such that
   $\es{E}_1 \trianglelefteq \es{E}'_1, \dots, \es{E}_n \trianglelefteq \es{E}'_n$. We want to show
   that
   $\es{E}_1 \trianglelefteq \es{E}'_1 \Rightarrow op \es{E}_1 \trianglelefteq op \es{E}'_1, \dots,
   \es{E}_n \trianglelefteq \es{E}'_n \Rightarrow op \es{E}_n \trianglelefteq op \es{E}'_n, \dots$
   For that we can make use of the least upper bound, \ie\
   $\bigsqcup_n \es{E}_n \trianglelefteq \bigsqcup_n \es{E}'_n \Rightarrow op (\bigsqcup_n \es{E}_n)
   \trianglelefteq \bigsqcup_n op (\es{E}'_n)$.  Since $op$ is continuous,
   $op \bigsqcup_n \es{E}_n \trianglelefteq \bigsqcup_n op \es{E}'_n$. Hence $op$ is monotonic. Now
   it lacks to show that each event of $op (\bigsqcup_n \es{E}_n)$ is an event of
   $\bigsqcup_n op (\es{E}'_n)$. But that comes freely since
   $op(\bigsqcup_n \es{E}_n) = \bigsqcup_n op (\es{E}_n)$.
   
  \item $\Leftarrow$: Assume $1.$ and $2.$ above.

    We want to show $op(\bigsqcup_n \es{E}_n) = \bigsqcup_n op (\es{E}_n)$.  Let
    $\es{E}_1 \trianglelefteq \dots \trianglelefteq \es{E}_n \trianglelefteq \dots$ be a
    $\omega$-chain.  By $1.$ we know that $op$ is monotonic, hence
    $\es{E}_n \trianglelefteq \bigsqcup_n \es{E}_n$ entails
    $op(\es{E}_n) \trianglelefteq op( \bigsqcup_n \es{E}_N)$ that leads to
    $op(\bigsqcup_n \es{E}_n) \trianglelefteq \bigsqcup_n op (\es{E}_n)$.  By $2.$, each event of
    $op(\bigsqcup_n \es{E}_n)$ is an event of $\bigsqcup_n op (\es{E}_n)$. Hence by
    Definition~\ref{def:pes-fix-order-1}, $op(\bigsqcup_n \es{E}_n) = \bigsqcup_n op (\es{E}_n)$.
 \end{itemize} 
\end{proof}

\subsubsection*{Proof of Lemma~\ref{lem:seq-cont-1}}
\begin{proof}
  By Lemma~\ref{lem:seq-fix-mono-1} we know that sequential composition is monotone w.r.t
  $\trianglelefteq$ at right.  It lacks to show that each event of
  $\seq{\es{E}}{\bigsqcup_m \es{E}_m}$ is an event of $\bigsqcup_m(\seq{\es{E}}{\es{E}_m})$.  Let
  $\es{E}_1 \trianglelefteq \dots \trianglelefteq \es{E}_m \trianglelefteq \dots$ be an
  $\omega$-chain such that $\bigsqcup_m \es{E}_m$ is its least upper bound and
  $\seq{\es{E}}{\es{E}_1} \trianglelefteq \dots \trianglelefteq \seq{\es{E}}{\es{E}_m}
  \trianglelefteq \dots$ be another $\omega$-chain with $\bigsqcup_m(\seq{\es{E}}{\es{E}_m})$ as its
  least upper bound.  Let $e$ be an event of $\seq{\es{E}}{\bigsqcup_m \es{E}_m}$. By
  Definition~\ref{def:pes-seq1} we have two cases:
  \begin{enumerate}
  \item $e$ is an event of $\es{E}$

    Then we are done, since $\forall m$, $e$ is an event of $\seq{\es{E}}{\es{E}_m}$. Hence it is an
    event of $\bigsqcup_m(\seq{\es{E}}{\es{E}_m})$.
  \item $e$ is an event of $(\bigcup_{m \in \omega} E_m) \times \confmax{\es{E}}$

    We know that $e$ is of the form $(e_m, x)$ with $e_m$ an event of $\bigsqcup_m \es{E}_m$ and
    $x \in \confmax{\es{E}}$.  The former entails $\exists m$ such that $e_m$ is an event of
    $\es{E}_m$. By Definition~\ref{def:pes-seq1} we have $(e_m, x)$ as an event of
    $\seq{\es{E}}{\es{E}_m}$. Consequently $(e_m, x)$ is an event of
    $\bigsqcup_m(\seq{\es{E}}{\es{E}_m})$.
  \end{enumerate}
  By Lemma~\ref{lem:cont-1} we are done.
\end{proof}

\subsubsection*{Proof of Lemma~\ref{lem:conc-cont-1}}
\begin{proof}
  By Lemma~\ref{lem:conc-fix-mono-1} we know that parallel composition is monotone w.r.t
  $\trianglelefteq$.  It lacks to show that each event of
  $\conc{\bigsqcup_n \es{E}_n}{\bigsqcup_m \es{E}_m}$ is an event of
  $\bigsqcup_{n,m}(\conc{\es{E}_n}{\es{E}_m})$.

  Let $\es{E}_1 \trianglelefteq \dots \trianglelefteq \es{E}_n \trianglelefteq \dots$ and
  $\es{E}'_1 \trianglelefteq \dots \trianglelefteq \es{E}'_m \trianglelefteq \dots$ be 
  $\omega$-chains with least upper bound $\bigsqcup_n \es{E}_n$ and  $\bigsqcup_m \es{E}_m$, respectively.
  Let $e$ be an event of $\conc{\bigsqcup_n \es{E}_n}{\bigsqcup_m \es{E}_m}$.
  By Definition~\ref{def:pes-conc1} we have two cases:
  \begin{enumerate}
  \item $e$ is an event of $\bigsqcup_n \es{E}_n$

    Then $\exists n \in \omega$ such that $e$ is an event of $\es{E}_n$. By
    Definition~\ref{def:pes-conc1}, $e$ is an event of $\conc{\es{E}_n}{\es{E}_m}$ and consequently
    is an event of $\bigsqcup_{n,m}(\conc{\es{E}_n}{\es{E}_m})$.

  \item $e$ is an event of $\bigsqcup_m \es{E}_m$
    Similar to previous point.
  \end{enumerate}
  By Lemma~\ref{lem:cont-1} we are done.
\end{proof}

\subsubsection*{Proof of Lemma~\ref{lem:nd-cont-1}}
\begin{proof}
  By Lemma~\ref{lem:nd-fix-mono-1} we know that non-deterministic composition is monotone w.r.t
  $\trianglelefteq$.  It lacks to show that each event of
  $\nd{\bigsqcup_n \es{E}_n}{\bigsqcup_m \es{E}_m}$ is an event of
  $\bigsqcup_{n,m}(\nd{\es{E}_n}{\es{E}_m})$.

  Let $\es{E}_1 \trianglelefteq \dots \trianglelefteq \es{E}_n \trianglelefteq \dots$ and
  $\es{E}'_1 \trianglelefteq \dots \trianglelefteq \es{E}'_m \trianglelefteq \dots$ be 
  $\omega$-chains with least upper bound $\bigsqcup_n \es{E}_n$ and  $\bigsqcup_m \es{E}_m$, respectively.
  Let $e$ be an event of $\nd{\bigsqcup_n \es{E}_n}{\bigsqcup_m \es{E}_m}$.
  By Definition~\ref{def:pes-nd1} we have two cases:
  \begin{enumerate}
  \item $e$ is an event of $\bigsqcup_n \es{E}_n$

    Then $\exists n \in \omega$ such that $e$ is an event of $\es{E}_n$. By
    Definition~\ref{def:pes-nd1}, $e$ is an event of $\nd{\es{E}_n}{\es{E}_m}$ and consequently
    is an event of $\bigsqcup_{n,m}(\nd{\es{E}_n}{\es{E}_m})$.

  \item $e$ is an event of $\bigsqcup_m \es{E}_m$
    Similar to previous point.
  \end{enumerate}
  By Lemma~\ref{lem:cont-1} we are done.
\end{proof}

\subsubsection*{Proof of Lemma~\ref{lem:fix-prop-1}}
\begin{proof}
  $\Gamma(fix(\Gamma)) = fix(\Gamma) \Leftrightarrow \Gamma(\bigsqcup_n \Gamma^n(\bot)) =
  \bigsqcup_n(\Gamma^{n}(\bot))$.  Since $\Gamma$ is continuous, we have that
  $\Gamma(\bigsqcup_n \Gamma^n(\bot)) = \bigsqcup_n \Gamma\Gamma^n(\bot) = \bigsqcup_n
  \Gamma^{n+1}(\bot)$.  We note that:
  $\bot \sqcup \bigsqcup_n \Gamma\Gamma^n(\bot) = \bigsqcup_n \Gamma^{n}(\bot)$.  Since $\bot$ is
  the `identity of the least upper bound' we have:
  $\bot \sqcup \bigsqcup_n \Gamma\Gamma^n(\bot) = \bigsqcup_n \Gamma^{n}(\bot) \Leftrightarrow
  \bigsqcup_n \Gamma\Gamma^n(\bot) = \bigsqcup_n \Gamma^{n}(\bot) \Leftrightarrow \Gamma(\bigsqcup_n
  \Gamma^n(\bot)) = \bigsqcup_n \Gamma^{n}(\bot) $.

  Now we need to show that $fix(\Gamma)$ is the least fixpoint.  Let $E$ be an event structure,
  $\Gamma(E) \trianglelefteq E$, and $\perp \trianglelefteq E$.  By the monotonic property
  $\Gamma(\perp) \trianglelefteq \Gamma(E)$.  Since $\Gamma(E) \trianglelefteq E$ then
  $\Gamma(\perp) \trianglelefteq E$.  By induction $\Gamma^{n}(E) \trianglelefteq E$.  Thus
  $fix(\Gamma) = \bigsqcup_n \Gamma^{n}(\perp) \trianglelefteq E$.  Hence $fix(\Gamma)$ is the least
  fixpoint.
\end{proof}

\subsubsection*{Proof of Lemma~\ref{lem:gamma-cont-1}}
\begin{proof}
  \begin{itemize}
  \item
    \begin{align*}
      & \Gamma^{\seq{\com{C}_1}{\com{C}_2}, \gamma}(\bigsqcup_n \es{E}_n) \\
      =& \set{ \text{Definition~\ref{def:fix-den-sem-1}}} \\
      & \mf{\seq{\com{C}_1}{\com{C}_2}}_{\gamma(X \leftarrow \bigsqcup_n \es{E}_n)} \\
      =& \set{ \text{Definition~\ref{def:fix-den-sem-1}}} \\
      & \seq{\mf{\com{C}_1}}{\mf{\com{C}_2}_{\gamma(X \leftarrow \bigsqcup_n \es{E}_n)}} \\
      =& \set{ \text{Definition~\ref{def:fix-den-sem-1}}} \\
      & \seq{\mf{\com{C}_1}}{\bigsqcup_n \Gamma^{\com{C}_2,\gamma}(\es{E}_n)} \\
      =& \set{\text{Lemma~\ref{lem:seq-cont-1}}} \\
      &= \bigsqcup_n (\seq{\mf{\com{C}_1}}{\Gamma^{\com{C}_2,\gamma}(\es{E}_n)}) \\
      =& \set{ \text{Definition~\ref{def:fix-den-sem-1}}} \\
      &= \bigsqcup_n (\seq{\mf{\com{C}_1}}{\mf{\com{C}_2}_{\gamma(X \leftarrow \es{E}_n)}})\\
      =& \set{ \text{Definition~\ref{def:fix-den-sem-1}}} \\
      & \bigsqcup_n \mf{\seq{\com{C}_1}{\com{C}_2}}_{\gamma(X \leftarrow \es{E}_n)}
    \end{align*}

  \item
    \begin{align*}
      & \Gamma^{\conc{\com{C}_1}{\com{C}_2}, \gamma}(\bigsqcup_n \es{E}_n) \\
      =& \set{ \text{Definition~\ref{def:fix-den-sem-1}}} \\
      & \mf{\conc{\com{C}_1}{\com{C}_2}}_{\gamma(X \leftarrow \bigsqcup_n \es{E}_n)} \\
      =& \set{ \text{Definition~\ref{def:fix-den-sem-1}}} \\
      & \conc{\mf{\com{C}_1}_{\gamma(X \leftarrow \bigsqcup_n \es{E}_n)}}{\mf{\com{C}_2}_{\gamma(X \leftarrow \bigsqcup_n \es{E}_n)}} \\
      =& \set{ \text{Definition~\ref{def:fix-den-sem-1}}} \\
      & \conc{\bigsqcup_n \Gamma^{\com{C}_1,\gamma}(\es{E}_n)}{\bigsqcup_n \Gamma^{\com{C}_2,\gamma}(\es{E}_n)} \\
      =& \set{\text{Lemma~\ref{lem:conc-cont-1}}} \\
      &= \bigsqcup_n (\conc{\Gamma^{\com{C}_1,\gamma}(\es{E}_n)}{\Gamma^{\com{C}_2,\gamma}(\es{E}_n)}) \\
      =& \set{ \text{Definition~\ref{def:fix-den-sem-1}}} \\
      &= \bigsqcup_n (\conc{\mf{\com{C}_1}_{\gamma(X \leftarrow \es{E}_n)}}{\mf{\com{C}_2}_{\gamma(X \leftarrow \es{E}_n)}})\\
      =& \set{ \text{Definition~\ref{def:fix-den-sem-1}}} \\
      & \bigsqcup_n \mf{\conc{\com{C}_1}{\com{C}_2}}_{\gamma(X \leftarrow \es{E}_n)}
    \end{align*}

  \item
    \begin{align*}
      & \Gamma^{\nd{\com{C}_1}{\com{C}_2}, \gamma}(\bigsqcup_n \es{E}_n) \\
      =& \set{ \text{Definition~\ref{def:fix-den-sem-1}}} \\
      & \mf{\nd{\com{C}_1}{\com{C}_2}}_{\gamma(X \leftarrow \bigsqcup_n \es{E}_n)} \\
      =& \set{ \text{Definition~\ref{def:fix-den-sem-1}}} \\
      & \nd{\mf{\com{C}_1}_{\gamma(X \leftarrow \bigsqcup_n \es{E}_n)}}{\mf{\com{C}_2}_{\gamma(X \leftarrow \bigsqcup_n \es{E}_n)}} \\
      =& \set{ \text{Definition~\ref{def:fix-den-sem-1}}} \\
      & \nd{\bigsqcup_n \Gamma^{\com{C}_1,\gamma}(\es{E}_n)}{\bigsqcup_n \Gamma^{\com{C}_2,\gamma}(\es{E}_n)} \\
      =& \set{\text{Lemma~\ref{lem:nd-cont-1}}} \\
      &= \bigsqcup_n (\nd{\Gamma^{\com{C}_1,\gamma}(\es{E}_n)}{\Gamma^{\com{C}_2,\gamma}(\es{E}_n)}) \\
      =& \set{ \text{Definition~\ref{def:fix-den-sem-1}}} \\
      &= \bigsqcup_n (\nd{\mf{\com{C}_1}_{\gamma(X \leftarrow \es{E}_n)}}{\mf{\com{C}_2}_{\gamma(X \leftarrow \es{E}_n)}})\\
      =& \set{ \text{Definition~\ref{def:fix-den-sem-1}}} \\
      & \bigsqcup_n \mf{\nd{\com{C}_1}{\com{C}_2}}_{\gamma(X \leftarrow \es{E}_n)}
    \end{align*}

  \item
    \begin{align*}
      & \Gamma^{\rec{X}{\com{C}, \gamma}}(\bigsqcup_n \es{E}_n) \\
      =& \set{ \text{Definition~\ref{def:fix-den-sem-1}}} \\
      & \mf{\rec{X}{\com{C}}}_{\gamma(X \leftarrow \bigsqcup_n \es{E}_n)} \\
      =& \set{ \text{Definition~\ref{def:fix-den-sem-1}}} \\
      & fix(\Gamma^{\com{C},\gamma(X \leftarrow \bigsqcup_n \es{E}_n)})
        = fix(\es{E} \mapsto \Gamma^{\com{C},\gamma(X \leftarrow \bigsqcup_n \es{E}_n)} (\es{E})) \\
      =& \set{ \text{Definition~\ref{def:fix-den-sem-1}}} \\
      & fix(\es{E} \mapsto \mf{\com{C}}_{\gamma(X \leftarrow \bigsqcup_n \es{E}_n, Y \leftarrow E)}) \\
      =& \set{\text{i.h.}} \\
      & fix (\es{E} \mapsto \bigsqcup_n \mf{\com{C}}_{\gamma(X \leftarrow \es{E}_n, Y \leftarrow \es{E})}) \\
      =& \set{\text{\cite[Proposition~2.1.18]{abramsky94}}} \\
      & fix (\bigsqcup_n (\es{E} \mapsto \mf{\com{C}}_{\gamma(X \leftarrow \es{E}_n, Y \leftarrow \es{E})})) \\
      =& \set{\text{$fix$ continuous}} \\
      & \bigsqcup_n fix (\es{E} \mapsto \mf{\com{C}}_{\gamma(X \leftarrow \es{E}_n, Y \leftarrow \es{E})}) \\
      =& \set{ \text{Definition~\ref{def:fix-den-sem-1}}} \\
      & \bigsqcup_n fix (\es{E} \mapsto \Gamma^{\com{C}, \gamma(X \leftarrow \es{E}_n)}(\es{E})) =
        \bigsqcup_n fix (\Gamma^{\com{C}, \gamma(X \leftarrow \es{E}_n)}) \\
      =& \set{ \text{Definition~\ref{def:fix-den-sem-1}}} \\
      & \bigsqcup_n \mf{\rec{X}{\com{C}}}_{\gamma(X \leftarrow \es{E}_n)} \\
      =& \set{ \text{Definition~\ref{def:fix-den-sem-1}}} \\
      & \bigsqcup_n \Gamma^{\rec{X}{\com{C}}, \gamma}(\es{E}_n)
    \end{align*}
  \end{itemize}
\end{proof}

\subsubsection*{Proof of Lemma~\ref{lem:1-1}}
\begin{proof}
  \begin{itemize}
  \item $\mf{\com{skip}[X \leftarrow \mf{\rec{X}{\com{C}}}_\gamma]}_\gamma$

    It follows directly that $\mf{\com{skip}}_{\gamma{(X \leftarrow \mf{\rec{X}{\com{C}}}_\gamma)}}$.

  \item $\mf{\com{a}[X \leftarrow \mf{\rec{X}{\com{C}}}_\gamma]}_\gamma$

    It follows directly that $\mf{\com{a}}_{\gamma{(X \leftarrow \mf{\rec{X}{\com{C}}}_\gamma)}}$.

  \item
    \begin{align*}
      & \mf{(\seq{\com{C}_1}{\com{C}_2})[X \leftarrow \mf{\rec{X}{\com{C}}}_\gamma]}_\gamma \\
      =& \set{\text{Definition~\ref{fig:substitution-1}}} \\
      & \mf{\seq{\com{C_1}}{(\com{C}_2[X \leftarrow \mf{\rec{X}{\com{C}}}_\gamma]})}_\gamma \\
      =& \set{\text{Definition~\ref{def:fix-den-sem-1}}} \\
      & \seq{\mf{\com{C}_1}}{(\mf{\com{C}_2}_\gamma[X \leftarrow \mf{\rec{X}{\com{C}}}_\gamma])} \\
      =& \set{\text{i.h.}} \\
      & \seq{\mf{\com{C}_1}}{\mf{\com{C}_2}_{\gamma(X \leftarrow \mf{\rec{X}{\com{C}}}_\gamma)}} \\
      =& \set{\text{Definition~\ref{def:fix-den-sem-1}}} \\
      & \mf{\seq{\com{C}_1}{\com{C}_2}}_{\gamma(X \leftarrow \mf{\rec{X}{\com{C}}}_\gamma)}
    \end{align*}

  \item
    \begin{align*}
      & \mf{(\conc{\com{C}_1}{\com{C}_2})[X \leftarrow \mf{\rec{X}{\com{C}}}_\gamma]}_\gamma \\
      =& \set{\text{Definition~\ref{fig:substitution-1}}} \\
      & \mf{\conc{\com{C_1}[X \leftarrow \mf{\rec{X}{\com{C}}}_\gamma]}{\com{C}_2[X \leftarrow \mf{\rec{X}{\com{C}}}_\gamma]}}_\gamma \\
      =& \set{\text{Definition~\ref{def:fix-den-sem-1}}} \\
      & \conc{\mf{\com{C}_1}[X \leftarrow \mf{\rec{X}{\com{C}}}_\gamma]}{\mf{\com{C}_2}_\gamma[X \leftarrow \mf{\rec{X}{\com{C}}}_\gamma]} \\
      =& \set{\text{i.h.}} \\
      & \conc{\mf{\com{C}_1}_{\gamma(X \leftarrow \mf{\rec{X}{\com{C}}}_\gamma)}}{
        \mf{\com{C}_2}_{\gamma(X \leftarrow \mf{\rec{X}{\com{C}}}_\gamma)}
        } \\
      =& \set{\text{Definition~\ref{def:fix-den-sem-1}}} \\
      & \mf{\conc{\com{C}_1}{\com{C}_2}}_{\gamma(X \leftarrow \mf{\rec{X}{\com{C}}}_\gamma)}
    \end{align*}

  \item
    \begin{align*}
      & \mf{(\nd{\com{C}_1}{\com{C}_2})[X \leftarrow \mf{\rec{X}{\com{C}}}_\gamma]}_\gamma \\
      =& \set{\text{Definition~\ref{fig:substitution-1}}} \\
      & \mf{\nd{\com{C_1}[X \leftarrow \mf{\rec{X}{\com{C}}}_\gamma]}{\com{C}_2[X \leftarrow \mf{\rec{X}{\com{C}}}_\gamma]}}_\gamma \\
      =& \set{\text{Definition~\ref{def:fix-den-sem-1}}} \\
      & \nd{\mf{\com{C}_1}[X \leftarrow \mf{\rec{X}{\com{C}}}_\gamma]}{\mf{\com{C}_2}_\gamma[X \leftarrow \mf{\rec{X}{\com{C}}}_\gamma]} \\
      =& \set{\text{i.h.}} \\
      & \nd{\mf{\com{C}_1}_{\gamma(X \leftarrow \mf{\rec{X}{\com{C}}}_\gamma)}}{
        \mf{\com{C}_2}_{\gamma(X \leftarrow \mf{\rec{X}{\com{C}}}_\gamma)}
        } \\
      =& \set{\text{Definition~\ref{def:fix-den-sem-1}}} \\
      & \mf{\nd{\com{C}_1}{\com{C}_2}}_{\gamma(X \leftarrow \mf{\rec{X}{\com{C}}}_\gamma)}
    \end{align*}

  \item
    \begin{align*}
      & \mf{(\rec{Y}{\com{C}'})[X \leftarrow \mf{\rec{X}{\com{C}}}_\gamma]}_\gamma \\
      =& \set{\text{Definition~\ref{fig:substitution-1}}} \\
      & \mf{\rec{Y}{(\com{C}'[X \leftarrow \mf{\rec{X}{\com{C}}}_\gamma])}}_\gamma \\
      =& \set{\text{Definition~\ref{def:fix-den-sem-1}}} \\
      & fix(\Gamma^{\com{C}'[X \leftarrow \mf{\rec{X}{\com{C}}}_\gamma], \gamma})
        = fix(\es{E} \mapsto \Gamma^{\com{C}'[X \leftarrow \mf{\rec{X}{\com{C}}}_\gamma], \gamma}(\es{E})) \\
      =& \set{\text{Definition~\ref{def:fix-den-sem-1}}} \\
      & fix(\es{E} \mapsto \mf{\com{C}'[X \leftarrow \mf{\rec{X}{\com{C}}}_\gamma]}_{\gamma(Y \leftarrow \es{E})}) \\
      =& \set{\text{i.h.}} \\
      & fix(\es{E} \mapsto \mf{\com{C}'}_{\gamma(Y \leftarrow \es{E}, X \leftarrow \mf{\rec{X}{\com{C}}}_\gamma)}) \\
      =& \set{\text{Definition~\ref{def:fix-den-sem-1}}} \\
      & fix(\es{E} \mapsto \Gamma^{\com{C'}, \gamma(Y \leftarrow \es{E}, X \leftarrow \mf{\rec{X}{\com{C}}}_\gamma)} (\es{E}))
        = fix(\Gamma^{\com{C'}, \gamma(X \leftarrow \mf{\rec{X}{\com{C}}}_\gamma)} ) \\
      =& \seq{\text{Definition~\ref{def:fix-den-sem-1}}} \\
      & \mf{\rec{Y}{\com{C'}}}_{\gamma(X \leftarrow \mf{\rec{X}{\com{C}}}_\gamma)}
    \end{align*}
  \end{itemize}
\end{proof}

\subsubsection*{Proof of Lemma~\ref{lem:2-1}}
\begin{proof}
  \begin{align*}
    & \mf{\rec{X}{\com{C}}}_\gamma \\
    =& \text{Definition~\ref{def:fix-den-sem-1}} \\
    & fix(\Gamma^{\com{C}, \gamma}) \\
    =& \set{\text{Lemma~\ref{lem:fix-prop-1}}} \\
    & \Gamma^{\com{C}, \gamma}(fix(\Gamma^{\com{C},\gamma})) \\
    =& \set{\text{Definition~\ref{def:fix-den-sem-1}}} \\
    & \Gamma^{\com{C}, \gamma}(\mf{\rec{X}{\com{C}}}_\gamma) \\
    =& \set{\text{Definition~\ref{def:fix-den-sem-1}}} \\
    & \mf{\com{C}}_{\gamma(X \leftarrow \mf{\rec{X}{\com{C}}}_\gamma)}
  \end{align*}
\end{proof}

\subsubsection*{Proof of Lemma~\ref{res:soundI-fix-1}}
\begin{proof}
  \begin{align*}
    & \rec{X}{\com{C}} \\
    \Rightarrow & \set{\text{Figure~\ref{fig:op-small1} entails}} \\
    & \com{C} \xrightarrow{l} \com{C}' \\
    \Rightarrow & \set{\text{i.h.}} \\
    & \mf{\com{C}'}_\gamma \equiv \mf{\com{C}}_\gamma \backslash l \\
    \Rightarrow & \set{\text{setting } \gamma = \gamma(X \leftarrow \mf{\rec{X}{\com{C}}}_\gamma)} \\
    & \mf{\com{C}'}_{\gamma(X \leftarrow \mf{\rec{X}{\com{C}}}_\gamma)} \equiv
      \mf{\com{C}}_{\gamma(X \leftarrow \mf{\rec{X}{\com{C}}}_\gamma)} \backslash l \\
    \Rightarrow & \set{\text{Lemma~\ref{lem:1-1}, Lemma~\ref{lem:2-1}}} \\
    & \mf{\com{C}'[X \leftarrow \mf{\rec{X}{\com{C}}}_\gamma]}_\gamma \equiv
      \mf{\rec{X}{\com{C}}}_\gamma \backslash l
  \end{align*}
\end{proof}

\subsubsection*{Proof of Lemma~\ref{res:adI-fix-1}}
\begin{proof}
  \begin{itemize}
  \item $l \in \init{\mf{\rec{X}{\com{C}}}_\gamma}$

    By Definition~\ref{def:fix-den-sem-1} and Definition~\ref{def:lub-1},
    $l \in \init{\mf{\com{C}}_\gamma}$.  By i.h., $\exists \com{C}'$ such that
    $\com{C} \xrightarrow{l} \com{C}'$ and
    $\mf{\com{C}'}_\gamma \equiv \mf{\com{C}}_\gamma \backslash l$.  By the rules in
    Figure~\ref{fig:op-small1} and by setting
    $\gamma = \gamma(X \leftarrow \mf{\rec{X}{\com{C}}}_\gamma)$,
    $\rec{X}{\com{C}} \xrightarrow{l} \com{C}'[X \leftarrow \rec{X}{\com{C}}]$ and    
    $\mf{\com{C}'}_{\gamma(X \leftarrow \mf{\rec{X}{\com{C}}}_\gamma)} \equiv
    \mf{\com{C}}_{\gamma(X \leftarrow \rec{X}{\com{C}})} \backslash l$
  \end{itemize}
\end{proof}

\newpage
\section*{Proofs of Section~\ref{sec:pes}}\label{app-sec:proofs-prob-es}
\begin{remark}
  Here we may found useful to write the sum in Equation~\ref{eq:prob-es-condition} in another way.
  Concretely:
  \begin{align*}
    v(y) - \sum_{\emptyset \neq I \subseteq \set{1, \dots, n}} (-1)^{|I|+1}
    v
    \left(
    \bigcup_{i \in I} x_i
    \right)
    = \sum_{I \subseteq \set{1, \dots, n}} (-1)^{|I|}
    v
    \left(y \cup
    \bigcup_{i \in I} x_i
    \right)
  \end{align*}
\end{remark}

\subsubsection*{Proof of Lemma~\ref{lem:seq-es2}}
\begin{proof}
  We begin by showing that $v(x)$ is well-defined, \ie\ that given
  $x \in \confES{\seq{\es{P}_1}{\es{P}_2}}$ such that $x = x_1 \cup (x_2 \times \{x_1\})$ and
  $x = x'_1 \cup (x'_2 \times \{x'_1\})$ then $x_1 = x'_1$ and $x_2 = x'_2$.

  Let $x \in \confES{\seq{\es{P}_1}{\es{P}_2}}$, $x_1, x'_1 \in \confmax{\es{P}_1}$,
  $x_2, x'_2 \in \confES{\es{P}_2}$, and
  $x = x_1 \cup (x_2 \times \{x_1\}) = x'_1 \cup (x'_2 \times \{x'_1\})$.

  We start by showing that $x_1 = x'_1$.  Since $x_1, x'_1 \in \confmax{\es{P}_1}$, by
  Lemma~\ref{lem:max-config-e-conflict}, distinct maximal configurations of $\es{P}_1$ have events
  in conflict:
  \[
    x_1 \neq x'_1 \Rightarrow \exists e1 \in x_1, e'_1 \in x'_1 \text{ such that } e_1 \#_1 e'_1
  \]

  By contradiction, assume that $x_1 \neq x'_1$.  Since maximal configurations have events in
  conflict, then there exist $e_1, e'_1$ such that $e_1 \#_1 e'_1$.  Furthermore, we have that
  $x_1, x'_1 \subseteq x$. Hence $e_1, e'_1 \in x$.  However, by Definition~\ref{def:pes-seq2}, the
  conflict relation $\#$ restricts to $\#_1$ on events from $E_1$. Thus $e_1 \# e'_1$.  Since $x$
  contains events in conflict, it contradicts that $x$ is a configuration (configurations are
  conflict-free by definition). Therefore, $x_1 = x'_1$.

  We now show that $x_2 = x'_2$.  We already know that $x_1 = x'_1$. Hence we can rewrite
  $x = x_1 \cup (x_2 \times \{x_1\}) = x_1 \cup (x'_2 \times \{x_1\})$.  Now, we only need to show
  that $x_2 \times \{x_1\} = x'_2 \times \{x_1\}$. For that, consider the map
  $f_{x_1} : \confES{\es{P}_2} \rightarrow \Pow({\es{E}_2 \times \{x_1\}}),\ x_2 \mapsto x_2 \times
  \{x_1\}$. Such map is injective (since we pair each configuration $x_2$ with the maximal
  configuration $x_1$). Hence, it follows from injectivity that $x_2 = x'_2$.
  
  Let
  \begin{align*}
    & \es{P}_1 = \pespq{\es{E}_1}{v_1} \\
    & \es{P}_2 = \pespq{\es{E}_2}{v_2} \\
    & \seq{\es{P}_1}{\es{P}_2} = \pespq{\es{E}}{v}
  \end{align*}
  

  By Lemma~\ref{lem:seq-es1} we know that $\es{E}$ is an event structure.
  Hence we focus solely on the valuation part.
  \begin{enumerate}
  \item $v(\emptyset) = 1$

    Since $\emptyset \in \confES{\es{P}_1}$ then $v(\emptyset) = v_1(\emptyset) = 1$.

  \item $\dropc{n}{v}{y}{x_1, \dots, x_n} \geq 0$ for all $n \geq 1$ and
    $y, x_1, \dots, x_n \in \confES{\seq{\es{P}_1}{\es{P}_2}}$ with $y \subseteq x_1, \dots, x_n$

    We have two cases based on $n$:
    \begin{enumerate}
    \item $n = 0$

      We have $\dropc{0}{v}{y}{} = v(y)$.
      By Definition~\ref{def:pes-seq2} we have two cases:
      \begin{enumerate}
      \item $y \in \confES{\seq{\es{P}_1}{\es{P}_2}}$ such that $y \in \confES{\es{P}_1}$

        It follows directly that $v(y) = v_1(y) \geq 0$.

      \item $y \in \confES{\seq{\es{P}_1}{\es{P}_2}}$ such that $y = y_1 \cup (y_2 \times \{y_1\})$
        with $y_1 \in \confmax{\es{P}_1}, y_2 \in \confES{\es{P}_{2}}$ .

        It follows directly that $v(y) = v_1(y_1) \cdot v_2(y_2) \geq 0$, since $v_1, v_2$ are
        valuations.
      \end{enumerate}
      
    \item $n > 0$

      By~\cite[Proposition 5]{winskel14} we only need to check the condition for
      $y \chain\ x_1, \dots, x_n$.
      We have three cases:
      \begin{enumerate}
      \item $y \in \confES{\es{P}_1}$ but $y \not\in \confmax{\es{P}_1}$

        By~\cite[Proposition 5]{winskel14} we know that $x_1, \dots, x_n \in \confES{\es{P}_1}$,
        since $\cchain{}\ $ is a ``single-step'' relation.  We have
        $y, x_1, \dots, x_n \in \confES{\es{P}_1}$.  It follows directly that
        $\dropc{n}{v}{y}{x_1, \dots, x_n} = \dropc{n}{v_1}{y}{x_1, \dots, x_n} \geq 0$.
        
      \item $y \in \confmax{\es{P}_1}$

        By~\cite[Proposition 5]{winskel14} we know that
        $x_1, \dots, x_n \in \confES{\seq{\es{P}_1}{\es{P}_2}}$ such that
        $x_1, \dots, x_n \not\in \confES{\es{P}_1}$, since $\cchain{}\ $ is a ``single-step''
        relation.  Hence exists $x'_1, \dots, x'_n \in \confES{\es{P}_{2}}$ such that
        $x_1 = y \cup (x'_1 \times \{y\}), \dots, x_n = y \cup (x'_n \times \{y\})$.  Furthermore
        $\bigcup_{i \in I} x_i = y \cup \bigcup_{i \in I} (x'_i \times \{y\})$ and let
        $I \subseteq \set{1, \dots, n}$.
        \begin{align*}
          \dropc{n}{v}{y}{x_1, \dots, x_n}
          & = \sum_I (-1)^{|I|} v \left(y \cup \bigcup_{i \in I} x_i \right) \\
          & = \sum_I (-1)^{|I|} v \left(y \cup \bigcup_{i \in I} (y \cup (x'_i \times \{y\})) \right) \\
          & = \sum_I (-1)^{|I|} v \left(y \cup \bigcup_{i \in I} (x'_i\times \{y\}) \right) \\
          & = \sum_I (-1)^{|I|} v_1(y) \cdot  v_2\left(\bigcup_{i \in I} x'_i \right) \\
          & = v_1(y) \cdot \sum_I (-1)^{|I|} v_2\left(\bigcup_{i \in I} x'_i \right) \\
          & = v_1(y) \cdot \dropc{n}{v_2}{\emptyset}{x'_1, \dots, x'_n}
        \end{align*}

        Since $v_1(y) \geq 0$ and $\dropc{n}{v_2}{\emptyset}{x'_1, \dots, x'_n} \geq 0$, then
        $v_1(y) \cdot \dropc{n}{v_2}{\emptyset}{x'_1, \dots, x'_n} \geq 0$.



        
      \item $y \in \confES{\seq{\es{P}_1}{\es{P}_2}}$ but $y \not\in \confES{\es{P}_1}$.

        By~\cite[Proposition 5]{winskel14} we know that $\exists y_1 \in \confmax{\es{P}_1}$ and
        $y', x'_1, \dots x'_n \in \confES{\es{P}_{2}}$ such that
        $y = y_1 \cup (y' \times \{y_1\}), x_1 = y_1 \cup (x'_1 \times \{y_1\}) ,\dots, x_n = y_1
        \cup (x'_n \times \{y_1\})$, since $\cchain{}\ $ is a ``single-step'' relation. Furthermore
        $\bigcup_{i \in I} x_i = \bigcup_{i \in I} (y_1 \cup (x'_i \times \{y_1\})) = y_1 \cup
        \bigcup_{i \in I} (x'_i \times \{y_1\})$ and let $I \subseteq \set{1, \dots, n}$.
        \begin{align*}
          \dropc{n}{v}{y}{x_1, \dots, x_n}
          & = \sum_I (-1)^{|I|} v \left(y \cup \bigcup_{i \in I} x_i \right) \\
          & = \sum_I (-1)^{|I|} v \left(
            (y_1 \cup (y' \times \{y_1\})) \cup \bigcup_{i \in I} (y_1 \cup (x'_i \times \{y_1\}))
            \right) \\          
          & = \sum_I (-1)^{|I|} v \left(
            y_1 \cup \left(
            (y' \times \{y_1\}) \cup \bigcup_{i \in I} (x'_i \times \{y_1\})
            \right)
            \right) \\
          & = \sum_I (-1)^{|I|} v_1(y_1) \cdot v_2\left(y' \cup \bigcup_{i \in I} x'_i \right) \\
          & = v_1(y) \cdot \sum_I (-1)^{|I|}  v_2\left(y' \cup \bigcup_{i \in I} x'_i \right) \\
          & = v_1(y) \cdot \dropc{n}{v_2}{y'}{x'_1, \dots, x'_n}
        \end{align*}

        Since $v_1(y) \geq 0$ and $\dropc{n}{v_2}{y'}{x'_1, \dots, x'_n} \geq 0$, then
        $v_1(y) \cdot \dropc{n}{v_2}{y'}{x'_1, \dots, x'_n} \geq 0$.
        


      \end{enumerate}
    \end{enumerate}
  \end{enumerate}
\end{proof}

\subsubsection*{Proof of Lemma~\ref{lem:pes-prob2}}
\begin{proof}
  Let
  \begin{align*}
    & \es{P}_1 = \ppes{E_1}{\leq_1}{\#_1}{v_1} \\
    & \es{P}_2 = \ppes{E_2}{\leq_2}{\#_2}{v_2} \\
    & \probC{\es{P}_1}{p}{\es{P}_2} = \ppes{E}{\leq}{\#}{v} 
  \end{align*}
  
  We need to show that $\leq$ is a partial order and that $\#$ is symmetric and irreflexive.
  \begin{itemize}
  \item $\leq$ is a partial order:
    \begin{itemize}
    \item Reflexivity ($e \leq e$): we have the following cases:
      \begin{enumerate}
      \item Case $e = \tau$. Then we are done since $\tau \leq \tau$.
      \item Case $e_1 \leq_1 e'_1$ or $e_2 \leq_2 e'_2$. Since $\leq_1$ and $\leq_2$ are partial
        orders, we are done.
      \end{enumerate}
    \item Transitivity ($e \leq e'$ and $e' \leq e''$ then $e \leq e''$): we have the following
      cases:
      \begin{enumerate}
      \item Case $e=\tau$ and ($e', e'' \in E_1$ or $e', e'' \in E_2$). Then $\tau \leq e'$ and
        $e' \leq e''$. Since $e'' \in E$ the $\tau \leq e''$.
      \item Case ($e \leq_1 e'$ and $e' \leq_1 e''$) or ($e \leq_2 e'$ and $e' \leq_2 e''$). Since
        $\leq_1$ and $\leq_2$ are partial orders, we have $e \leq_1 e''$ or $e \leq_2 e''$. Hence
        $e \leq e''$.
      \end{enumerate} 
    \item Antisymmetry ($e \leq e'$ and $e' \leq e$ then $e'=e$): we have the following cases:
      \begin{enumerate}
      \item Case ($e \leq_1 e'$ and $e' \leq_1 e$) or ($e \leq_2 e'$ and $e' \leq_2 e$). Since
        $\leq_1$ and $\leq_2$ are partial orders, we have $e=e'$.
      \item Case $e=\tau$ and ($e' \in E_1$ or $e' \in E_2$). We are done since $e' \leq \tau$ is
        not possible.
      \end{enumerate}
    \end{itemize}

    Hence $\leq$ is a partial order.

  \item $\#$ is symmetric and irreflexive:
    \begin{itemize}
    \item Symmetric (if $e \# e'$ then $e' \# e$): we have the following cases
      \begin{enumerate}
      \item if $e \#_1 e'$ or $e \#_2 e'$. Since $\#_1$ and $\#_2$ are symmetric then we have
        $e' \#_1 e$ or $e' \#_2 e$. Thus $e' \# e$.
      \item if $e \in E_1$ and $e' \in E_2$, or vice-versa, then by Definition~\ref{def:pes-prob2}
        we also have $e' \leq e$.
      \end{enumerate}
    \item Irreflexive ($\neg(e \# e)$): We have either $\neg(e \#_1 e)$ or $\neg(e \#_2 e)$. Since
      $\#_1$ and $\#_2$ are irreflexive then $\neg(e \# e)$. There are no more cases since either
      $e \in E_1$ or $e \in E_2$.
    \end{itemize}

    Hence $\#$ is symmetric and irreflexive.
  \end{itemize}
  
  Let $e, e', e'' \in E$. We have four conditions to check:
  \begin{enumerate}
  \item $\set{e' \mid e' \leq e}$ is finite

    We have three cases:
    \begin{enumerate}
    \item $e = \tau$

      It follows directly that $\set{e' \mid e' \leq \tau} = \set{\tau}$ since
      $\tau \in \init{\probC{\es{P}_1}{p}{\es{P}_2}}$.

    \item $e \in E_1$

      We have that $\set{e' \mid e' \leq e} = \set{\tau} \cup \set{e' \mid e' \leq_1 e}$.  Since
      $\es{P}_1$ is a probabilistic event structure, then we know that $\set{e' \mid e' \leq_1 e}$
      is finite.  Hence $\set{\tau} \cup \set{e' \mid e' \leq_1 e}$ is finite.
      
    \item $e \in E_2$

      We have that $\set{e' \mid e' \leq e} = \set{\tau} \cup \set{e' \mid e' \leq_2 e}$.  Since
      $\es{P}_2$ is a probabilistic event structure, then we know that $\set{e' \mid e' \leq_2 e}$
      is finite.  Hence $\set{\tau} \cup \set{e' \mid e' \leq_2 e}$ is finite.
    \end{enumerate}
    
  \item $e \# e' \leq e'' \Rightarrow e \# e''$

    Since $\tau$ is not in conflict with any event, this condition trivially holds because we either
    have $e, e', e'' \in E_1$ or $e, e', e'' \in E_2$ and $\es{P}_1, \es{P_2}$ are probabilistic
    event structures.
    
  \item $v(\emptyset) = 1$

    It follows directly from the definition.
    
  \item $\dropc{n}{v}{y}{x_1, \dots, x_n} \geq 0$ for all $n \geq 1$ and
    $y, x_1, \dots, x_n \in \confES{\probC{\es{P}_1}{p}{\es{P}_2}}$ with
    $y \subseteq x_1, \dots, x_n$

    By~\cite[Proposition 5]{winskel14} we only need to check the condition for
    $y \chain\ x_1, \dots, x_n$, \ie\ $y \cchain{e_1, \dots, e_n}\ x_1, \dots, x_n$.  We have three
    cases:
    \begin{enumerate}
    \item $y = \emptyset$

      We then have $\emptyset \cchain{\tau} \set{\tau}$.

      It follows that 
      $\dropc{1}{v}{\emptyset}{\set{\tau}} = v(\emptyset) - v(\set{\tau}) = 1 -1 = 0$
      
    \item $y \backslash \tau \in \confES{\es{P}_1}$
      We have three cases (let $1 \leq i \leq n$):
      \begin{enumerate}
      \item $\forall e_i \in E_1$

        We have $x_1 \backslash \tau, \dots, x_n \backslash \tau \in \confES{\es{P}_1}$.
        Let $I \subseteq \set{1, \dots, n}$.
        \begin{align*}
          \dropc{n}{v}{y}{x_1, \dots, x_n}
          &= \sum_I (-1)^{|I|} v(y \cup \bigcup_{i \in I} x_i) \\
          &= p \cdot \sum_I (-1)^{|I|} v_1((y \backslash \tau) \bigcup_{i \in I} (x_i \backslash \tau)) \\
          &= p \cdot \dropc{n}{v_1}{y}{x_1, \dots, x_n}
        \end{align*}

        Since $p \in ]0,1[$ and $\dropc{n}{v_1}{y}{x_1, \dots, x_n}$, because $\es{P}_1$ is a
        probabilistic event structure, then $p \cdot \dropc{n}{v_1}{y}{x_1, \dots, x_n} \geq 0$.
        
      \item $\forall e_i \in E_2$

        Since $y \backslash \tau \in \confES{\es{P}_1}$, $x_i = \set{e_i} \cup y$, and for all
        $1 \leq i \leq n$ we have $e_i \# e \in y \backslash \tau \in \confES{\es{P}_1}$, then
        $v(x_i) = 0$.  Hence
        \begin{align*}
          \dropc{n}{v}{y}{x_1, \dots, x_n}
          &= \sum_I (-1)^{|I|} v(y \cup \bigcup_{i \in I} x_i) \\
          &= v(y) \\
          &= p \cdot v_1(y \backslash \tau)
        \end{align*}

        Since $p \in ]0,1[$ and $v_1(y \backslash \tau) \geq 0$, because $\es{P}_1$ is a
        probabilistic event structure, we have $p \cdot v_1(y \backslash \tau) \geq 0$.
        
      \item $\exists e_i \in E_2$

        Since $y \backslash \tau \in \confES{\es{P}_1}$, $x_i = \set{e_i} \cup y$, and for all
        $1 \leq i \leq n$ such that $e_i \# e \in y \backslash \tau \in \confES{\es{P}_1}$, we have
        $v(x_i) = 0$.  Hence
        \begin{align*}
          \dropc{n}{v}{y}{x_1, \dots, x_n}
          &= \sum_I (-1)^{|I|} v(y \cup \bigcup_{i \in I} x_i) \\
          &= \sum_{I'} (-1)^{|I'|} v(y \cup \bigcup_{i \in I'} x_i) \\
          &= p \cdot \sum_{I'} (-1)^{|I'|} v_1((y \backslash \tau) \bigcup_{i \in I'} (x_i \backslash \tau)) \\
          &= p \cdot \dropc{m}{v_1}{y}{x_1, \dots, x_m}
        \end{align*}
        where $I' = \set{1, \dots, m}$, \ie\ $I'$ is $I$ without those $e_i \in E_2$.

        Since $p \in ]0,1[$ and $\dropc{m}{v_1}{y}{x_1, \dots, x_m}$, because $\es{P}_1$ is a
        probabilistic event structure, then $p \cdot \dropc{m}{v_1}{y}{x_1, \dots, x_m} \geq 0$.
      \end{enumerate}
    \item $y \backslash \tau \in \confES{\es{P}_2}$

      Similar to previous case.
    \end{enumerate}
  \end{enumerate}
\end{proof}

\subsubsection*{Proof of Lemma~\ref{lem:conc-es2}}
\begin{proof}
  Let
  \begin{align*}
    & \es{P}_1 = \pespq{\es{E}_1}{v_1} \\
    & \es{P}_2 = \pespq{\es{E}_2}{v_2} \\
    & \conc{\es{P}_1}{\es{P}_2} = \pespq{\es{E}}{v}
  \end{align*}
  

  By Lemma~\ref{lem:conc-es1} we know that $\es{E}$ is an event structure.
  Hence we focus solely on the valuation part.
  \begin{enumerate}
  \item $v(\emptyset) = 1$
    \begin{align*}
      v(\emptyset)
      = v_1 (\emptyset \cap E_1) \cdot v_2(\emptyset \cap E_2)
      = v_1(\emptyset) \cdot v_2(\emptyset)
      = 1 \cdot 1
      = 1
    \end{align*}
    
  \item $\dropc{n}{v}{y}{x_1, \dots, x_n} \geq 0$ for all $n \geq 1$ and
    $y, x_1, \dots, x_n \in \confES{\probC{\es{P}_1}{p}{\es{P}_2}}$ with
    $y \subseteq x_1, \dots, x_n$

    By~\cite[Proposition 5]{winskel14} we only need to check the condition for
    $y \chain\ x_1, \dots, x_n$, \ie\ $y \cchain{e_1, \dots, e_n}\ x_1, \dots, x_n$.  We want to
    show that
    $\dropc{n}{v}{y}{x_1, \dots, x_n} =
    \dropc{n_1}{v_1}{y \cap E_1}{x_1 \cap E_1, \dots, x_n \cap E_1} \cdot
    \dropc{n_2}{v_2}{y \cap E_2}{x_1 \cap E_2,\dots, x_n \cap E_2}$

    Let $I_1 \subseteq \set{1, \dots, n_1}$, $I_2 \subseteq \set{1, \dots, n_2}$, and
    $I = I_1 \uplus I_2$.
    \begin{align*}
      &\dropc{n_1}{v_1}{y \cap E_1}{x_1 \cap E_1, \dots, x_n \cap E_1} \cdot
        \dropc{n_2}{v_2}{y \cap E_2}{x_1 \cap E_2,\dots, x_n \cap E_2} \\
      =& \sum_{I_1} (-1)^{|I_1|} v_1 \left( (y \cap E_1) \cup \left( \bigcup_{i \in I_1} x_i \cap E_1 \right) \right)
         \cdot
         \sum_{I_2} (-1)^{|I_2|} v_2 \left( (y \cap E_2) \cup \left( \bigcup_{j \in I_2} x_j \cap E_2 \right) \right) \\
      =& \sum_{I_1} (-1)^{|I_1|} v_1 \left( \left( y \cup \bigcup_{i \in I_1} x_i \right) \cap E_1 \right)
         \cdot
         \sum_{I_2} (-1)^{|I_2|} v_2 \left( \left( y \cup \bigcup_{j \in I_2} x_j \right) \cap E_2 \right) \\
      =& \sum_{I_1} \sum_{I_2} (-1)^{|I_1| + |I_2|}
         v_1 \left( \left( y \cup \bigcup_{i \in I_1} x_i \right) \cap E_1 \right)
         v_2 \left( \left( y \cup \bigcup_{j \in I_2} x_j \right) \cap E_2 \right) \\
      =& \sum_{I_1,I_2} (-1)^{|I_1| + |I_2|}
         v_1 \left( \left( y \cup \bigcup_{i \in I_1} x_i \right) \cap E_1 \right)
         v_2 \left( \left( y \cup \bigcup_{j \in I_2} x_j \right) \cap E_2 \right) \\
      =& \sum_{I_1,I_2} (-1)^{|I_1| + |I_2|}
         v \left( y \cup \bigcup_{i \in (I_1 \uplus I_2)} x_i \right)  \\
      =& \sum_{I} (-1)^{|I|}
         v \left( y \cup \bigcup_{i \in I} x_i \right) \\
      =& \dropc{n}{v}{y}{x_1, \dots, x_n} \geq 0
    \end{align*}
    
  \end{enumerate}
\end{proof}

\subsubsection*{Proof of Lemma~\ref{lem:rem-init-es2}}
\begin{proof}
  Let
  \begin{align*}
    & P = \pespq{\es{E}}{v} \\
    & P \backslash a = \pespq{\es{E}'}{v'}
  \end{align*}
  By Lemma~\ref{lem:rem-init-es1} we know that $\es{E'}$ is an event structure.
  Hence we focus solely on the valuation part.
  \begin{enumerate}
  \item $v'(\emptyset) = 1$
    \begin{align*}
      v'(\emptyset)
      = \dfrac{v_1(\emptyset \cup \set{a})}{v_1(\set{a})}
      = \dfrac{v_1(\set{a})}{v_1(\set{a})}
      = 1
    \end{align*}
    
  \item $\dropc{n}{v}{y}{x_1, \dots, x_n} \geq 0$ for all $n \geq 1$ and
    $y, x_1, \dots, x_n \in \confES{\probC{\es{P}_1}{p}{\es{P}_2}}$ with
    $y \subseteq x_1, \dots, x_n$

    By~\cite[Proposition 5]{winskel14} we only need to check the condition for
    $y \chain\ x_1, \dots, x_n$, \ie\ $y \cchain{e_1, \dots, e_n}\ x_1, \dots, x_n$.

    Let $I \subseteq \set{1, \dots, n}$.
    \begin{align*}
      \dropc{n}{v}{y}{x_1, \dots, x_n} \geq 0
      &\Leftrightarrow
        \sum_I (-1)^{|I|} v \left( y \cup \bigcup_{i \in I} x_i \right) \geq 0\\
      &\Leftrightarrow
        \sum_I (-1)^{|I|} \dfrac{v_1 \left( \left( y \cup \bigcup_{i \in I} x_i \right) \cup \set{a} \right)}{v_1(\set{a})}  \geq 0 \\
      &\Leftrightarrow
        \sum_I (-1)^{|I|} v_1 \left( \left( y \cup \bigcup_{i \in I} x_i \right) \cup \set{a} \right) \geq 0 \\
      &\Leftrightarrow
        \dropc{n}{v_1}{y \cup \set{a}}{x_1 \cup \set{a}, \dots, x_n \cup \set{a}} \geq 0
    \end{align*}
  \end{enumerate}
\end{proof}

\subsubsection*{Proof of Lemma~\ref{lem:seq-mono2}}
\begin{proof}
  Let
  \begin{align*}
  & \es{P}_1 = \pespq{\es{E}_1}{v_1} \\
  & \es{P}'_1 = \pespq{\es{E}'_1}{v'_1} \\
  & \es{P}_2 = \pespq{\es{E}_2}{v_2} \\
  & \es{P}'_2 = \pespq{\es{E}'_2}{v'_2} \\
  & \seq{\es{P}_1}{\es{P}_2} = \pespq{\es{E}}{v} \\
  & \seq{\es{P}'_1}{\es{P}'_2} = \pespq{\es{E}'}{v'}
  \end{align*}
  Due to
  Lemma~\ref{lem:seq-mono1}, we only need to show
  $\forall x \in \confES{\seq{\es{P}_1}{\es{P}_2}},
  y \in \confES{\seq{\es{P}'_1}{\es{P}'_2}}\, .\,
  f[x] \subseteq y \Rightarrow v(x) \geq v'(y)$.

  Let $x \in \confES{\seq{\es{P}_1}{\es{P}_2}}$ and $y \in \confES{\seq{\es{P}'_1}{\es{P}'_2}}$ such
  that $f[x] \subseteq y$. We have three cases:
  \begin{enumerate}
  \item $x \in \confES{\seq{\es{P}_1}{\es{P}_2}}$ such that $x \in \confES{\es{P}_1}$ and
    $y \in \confES{\seq{\es{P}'_1}{\es{P}'_2}}$ such that $y \in \confES{\es{P}'_1}$

    It follows directly that $v(x) \geq v'(y) \Leftrightarrow v_1(x) \geq v'_1(y)$, since
    $v(x) = v'_1(x)$, $v'(y) = v'_1(y)$ and $\es{P}_1 \sqsubseteq \es{P}'_1$.
    
  \item $x \in \confES{\seq{\es{P}_1}{\es{P}_2}}$ such that $x \in \confES{\es{P}_1}$ and
    $y \in \confES{\seq{\es{P}'_1}{\es{P}'_2}}$ such that
    $\exists y_1 \in \confmax{\es{P}'_1}, y_2 \in \confES{\es{P}'_2}$ such that
    $y = y_1 \cup (y_2 \times \{y_1\})$.

    We know that $v(x) = v_1(x)$ and that $v'(y) = v'_1(y_1) \cdot v'_2(y_2)$.  Since
    $\es{P}_1 \sqsubseteq \es{P}'_1$ then $f_1[x] \subseteq y_1$ and $v_1(x) \geq v'_1(y_1)$.  It
    then follows directly that
    $v(x) = v_1(x) \geq v'_1(y_1) \geq v'_1(y_1) \cdot v'_2(y_2) = v'(y)$.
    
  \item $x \in \confES{\seq{\es{P}_1}{\es{P}_2}}$ such that
    $\exists x_1 \in \confES{\es{P}_1}, x_2 \in \confES{\es{P}_2}$ such that
    $x = x_1 \cup (x_2 \times \{x_1\})$ and $y \in \confES{\seq{\es{P}'_1}{\es{P}'_2}}$ such that
    $\exists y_1 \in \confmax{\es{P}'_1}, y_2 \in \confES{\es{P}'_2}$ such that
    $y = y_1 \cup (y_2 \times \{y_1\})$

    We know that $v(x) = v_1(x_1) \cdot v_2(x_2)$ and $v'(y) = v'_1(y_1) \cdot v'_2(y_2)$.  Since
    $\es{P}_1 \sqsubseteq \es{P}'_1$ then $f_1[x_1] \subseteq y_1$ and
    $v_1(x_1) \geq v'_1(y_1)$, and $\es{P}_2 \sqsubseteq \es{P}'_2$ then
    $f_2[x_2] \subseteq y_2$ and $v_2(x_2) \geq v'_2(y_2)$.

    Furthermore,
    \begin{align*}
      v_1(x_1) \geq v'_1(y_1)
      &\Leftrightarrow
        v'_1(y_1) \leq v_1(x_1) \\
      &\Leftrightarrow
        \dfrac{v'_1(y_1)}{v_1(x_1)} \leq 1 \\
      &\Leftrightarrow
        \left(
        \dfrac{v'_1(y_1)}{v_1(x_1)} = 1
        \right)
        \text{ or }
        \left(
        \dfrac{v'_1(y_1)}{v_1(x_1)} < 1
        \right)
    \end{align*}
  \end{enumerate}

  Now we show that $v(x) \geq v'(y)$.
  \begin{align*}
    v(x) \geq v'(y)
    &\Leftrightarrow
      v_1(x_1) \cdot v_2(x_2) \geq v'_1(y_1) \cdot v'_2(y_2) \\
    &\Leftrightarrow
      v_2(x_2) \geq \dfrac{v'_1(y_1)}{v_1(x_1)} \cdot v'_2(y_2)
  \end{align*}
  We have two cases:
  \begin{enumerate}
  \item $\dfrac{v'_1(y_1)}{v_1(x_1)} = 1$

    $v_2(x_2) \geq \dfrac{v'_1(y_1)}{v_1(x_1)} \cdot v'_2(y_2)
    \Leftrightarrow v_2(x_2) \geq v'_2(y_2)$ and we are done.
    
  \item $\dfrac{v'_1(y_1)}{v_1(x_1)} < 1$

    Since $v_2(x_2) \geq v'_2(y_2)$ and $v'_2(y_2) \geq \dfrac{v'_1(y_1)}{v_1(x_1)} \cdot v'_2(y_2)$
    it follows that
    $v_2(x_2) \geq v'_2(y_2) \geq \dfrac{v'_1(y_1)}{v_1(x_1)} \cdot v'_2(y_2)$
    
  \end{enumerate}

\end{proof}

\subsubsection*{Proof of Lemma~\ref{lem:probc-mono2}}
\begin{proof}
  Let
  \begin{align*}
    & \es{P}_1 = \pespq{\es{E}_1}{v_1} \\
    & \es{P}'_1 \pespq{\es{E}'_1}{v'_1} \\
    & \es{P}_2 = \pespq{\es{E}_2}{v_2} \\
    & \es{P}'_2 = \pespq{\es{E}'_2}{v'_2} \\
    & \probC{\es{P}_1}{p}{\es{P}_2} = \pespq{\es{E}}{v} \\
    & \probC{\es{P}'_1}{p}{\es{P}'_2} = \pespq{\es{E}'}{v'}
  \end{align*}
  
  \begin{enumerate}
  \item We start by defining $f : E \rightarrow E'$ such that
    \[
      f(e) =
      \begin{cases}
        e & \text{ if } e = \tau \in \init{\probC{\es{P}_1}{p}{\es{P}_2}} \\
        f_1(e) & \text{ if } e \in E_1 \\
        f_2(e) & \text{ if } e \in E_2
      \end{cases}
    \]

    It is straightforward to see that $f$ is injective.
    
  \item $\pi(f(e)) = \pi(e)$

    If $e \in E_1$ or $e \in E_2$ we are done, since $\es{P}_1 \sqsubseteq \es{P}'_1$ and $\es{P}_2 \sqsubseteq \es{P}'_2$.
    If $e = \tau$, then $\pi(f(e)) = \pi(e)$, and we are done.
    
  \item $e \leq e' \Leftrightarrow f(e) \leq' f(e')$

    If $e,e' \in E_1$ or $e,e' \in E_2$ we are done, since $\es{P}_1 \sqsubseteq \es{P}'_1$ and $\es{P}_2 \sqsubseteq \es{P}'_2$.
    If $e = \tau$ then $f(e)$. Consequently $e \leq e' \Leftrightarrow e \leq' f(e')$, which is trivially satisfied by Definition~\ref{def:pes-prob2}.
    
  \item $e \# e' \Leftrightarrow f(e) \#' f(e')$

    Since $\tau$ is not in conflict with any event, this case is similar to that of Lemma~\ref{lem:nd-mono1}.
    
  \item $\forall x \in \confES{\probC{\es{P}_1}{p}{\es{P}_2}},
    y \in \confES{\probC{\es{P}_1}{p}{\es{P}_2}}\, .\,
    f[x] \subseteq y \Rightarrow v(x) \geq v'(y)$

    Let $x \in \confES{\probC{\es{P}_1}{p}{\es{P}_2}}$ and
    $y \in \confES{\probC{\es{P}_1}{p}{\es{P}_2}}$ such that $f[x] \subseteq y$.
    We have two cases:
    \begin{enumerate}
    \item $x \backslash \tau \in \confES{\es{P}_1}$ and $y \backslash \tau \in \confES{\es{P}'_1}$

      It follows directly that $v(x) \geq v'(y)$, since $\es{P}_1 \sqsubseteq \es{P}'_1$ and
      $v(x) = p \cdot v_1(x \backslash \tau) \geq p \cdot v'_1(y \backslash \tau) = v'(y)$.
      
    \item $x \backslash \tau \in \confES{\es{P}_2}$ and $y \backslash \tau \in \confES{\es{P}'_2}$

      It follows directly that $v(x) \geq v'(y)$, since $\es{P}_2 \sqsubseteq \es{P}'_2$ and
      $v(x) = (1-p) \cdot v_2(x \backslash \tau) \geq (1-p) \cdot v'_2(y \backslash \tau) = v'(y)$.
    \end{enumerate}
  \end{enumerate}
\end{proof}

\subsubsection*{Proof of Lemma~\ref{lem:conc-mono2}}
\begin{proof}
  Let
  \begin{align*}
    & \es{P}_1 = \pespq{\es{E}_1}{v_1} \\
    & \es{P}'_1 = \pespq{\es{E}'_1}{v'_1} \\
    & \es{P}_2 = \pespq{\es{E}_2}{v_2} \\
    & \es{P}'_2 = \pespq{\es{E}'_2}{v'_2} \\
    & \conc{\es{P}_1}{\es{P}_2} = \pespq{\es{E}}{v} \\
    & \conc{\es{P}'_1}{\es{P}'_2} = \pespq{\es{E}'}{v'}
  \end{align*}
  Due to Lemma~\ref{lem:conc-mono1}, we only need to show
  $\forall x \in \confES{\conc{\es{P}_1}{\es{P}_2}}, y \in \confES{\conc{\es{P}'_1}{\es{P}'_2}}\,
  .\, f[x] \subseteq y \Rightarrow v(x) \geq v'(y)$.

  Let $x \in \confES{\conc{\es{P}_1}{\es{P}_2}}$ and $y \in \confES{\conc{\es{P}'_1}{\es{P}'_2}}$
  such that $f[x] \subseteq y$.  Since $\es{P}_1 \sqsubseteq \es{P}'_1$ then
  $\forall x_1 \in \confES{\es{P}_1}, y_1 \in \confES{\es{P}'_1}$ such that
  $f_1[x_1] \subseteq y_1$ we have $v_1(x_1) \geq v'_1(y_1)$ and that
  $\es{P}_2 \sqsubseteq \es{P}'_2$ entails
  $\forall x_2 \in \confES{\es{P}_2}, y_2 \in \confES{\es{P}'_2}$ such that
  $f_2[x_2] \subseteq y_2$ we have $v_2(x_2) \geq v'_2(y_2)$.  By
  Definition~\ref{def:pes-conc2}, $v(x) = v_1(x_1) \cdot v_2(x_2)$ and
  $v'(y) = v'_1(y_1) \cdot v'_2(y_2)$, where $x_1 = x \cap E_1$, $x_2 = x \cap E_1$,
  $y_1 = y \cap E'_1$, and $y_2 = y \cap E'_2$.  We then have:
  \begin{align*}
    v(x)
    = v_1(x_1) \cdot v_2(x_2)
    \geq v'_1(y_1) \cdot v'_2(y_2)
    = v'(y)
  \end{align*}
\end{proof}

\subsubsection*{Proof of Lemma~\ref{lem:seq-rem-init2}}
\begin{proof}
  Let
  \begin{align*}
    & \es{P}_1 = \pespq{\es{E}_1}{v_1} \\
    & \es{P}_2 = \pespq{\es{E}_2}{v_2} \\
    & \seq{\es{P}_1}{\es{P}_2} = \pespq{\es{E}_{\seq{1}{2}}}{v_{\seq{1}{2}}} \\
    & (\seq{\es{P}_1}{\es{P}_2}) \backslash l = \pespq{\es{E}}{v} \\
    & \es{P}_1 \backslash l = \pespq{\es{E}_1^l}{v_1^l} \\
    & \seq{(\es{P}_1 \backslash l)}{\es{P}_2} = \pespq{\es{E}'}{v'} \\
    & l \in \init{\seq{\es{P}_1}{\es{P}_2}}
  \end{align*}

  Due to Lemma~\ref{lem:seq-rem-init1} we only need to show
  \begin{enumerate}
  \item $\forall x \in \confES{(\seq{\es{P}_1}{\es{P}_2}) \backslash l},
    y \in \confES{\seq{(\es{P}_1 \backslash l)}{\es{P}_2}}\, .\,
    f[x] \subseteq y \Rightarrow v(x) \geq v'(y)$

    Let $x \in \confES{(\seq{\es{P}_1}{\es{P}_2}) \backslash l}$ and
    $y \in \confES{\seq{(\es{P}_1 \backslash l)}{\es{P}_2}}$ such that
    $f[x] \subseteq y$.  We have two cases:
    \begin{enumerate}
    \item $\set{l} \cup x \in \confES{\seq{\es{P}_1}{\es{P}_2}}$ such that
      $\set{l} \cup x \in \confES{\es{P}_1}$
      \begin{align*}
        v(x)
        = \dfrac{v_{\seq{1}{2}} (\set{l} \cup x)}{v_{\seq{1}{2}}(\set{l})}
        = \dfrac{v_1 (\set{l} \cup x)}{v_1(\set{l})}
        = v_1^l(x)
        = v'(x)
      \end{align*}

      Since $\seq{(\es{P}_1 \backslash l)}{\es{P}_2}$ is a probabilistic event structure, then for
      $x, y \in \confES{\seq{(\es{P}_1 \backslash l)}{\es{P}_2}}$ such that
      $f[x] \subseteq y$, we have $v'(x) \geq v'(y)$, since
      $\dropc{1}{v'}{x}{y} \geq 0 \Leftrightarrow v'(x) - v'(y) \geq 0 \Leftrightarrow v'(x) \geq
      v'(y)$.

      Hence $v(x) \geq v'(y)$.
      
    \item $\set{l} \cup x \in \confES{\seq{\es{P}_1}{\es{P}_2}}$ such that
      $\exists (\set{l} \cup x_1) \in \confmax{\es{P}_1}, x_2 \in \confES{\es{P}_2}$ where
      $\set{l} \cup x = (\set{l} \cup x_1) \cup (x_2 \times \{\set{l} \cup x_1\})$
      \begin{align*}
        v(x)
        = \dfrac{v_{\seq{1}{2}} (\set{l} \cup x)}{v_{\seq{1}{2}}(\set{l})}
        = \dfrac{v_1(\set{l} \cup x_1) \cdot v_2(x_2)}{v_1(\set{l})}
        = v_1^l(x_1) \cdot v_2(x_2)
        = v'(x)
      \end{align*}
    \end{enumerate}

    Since $\seq{(\es{P}_1 \backslash l)}{\es{P}_2}$ is a probabilistic event structure, we obtain
    $v(x) \geq v'(y)$.
    
  \item $\forall x \in \confES{\seq{(\es{P}_1 \backslash l)}{\es{P}_2}},
    y \in \confES{(\seq{\es{P}_1}{\es{P}_2}) \backslash l} \, .\,
    f[x] \subseteq y \Rightarrow v'(x) \geq v(y)$

    Let $x \in \confES{\seq{(\es{P}_1 \backslash l)}{\es{P}_2}}$ and
    $y \in \confES{(\seq{\es{P}_1}{\es{P}_2}) \backslash l}$ such that $f[x] \subseteq y$.  We have
    two cases:
    \begin{enumerate}
    \item $x \in \confES{\seq{\es{P}_1 \backslash l}{\es{P}_2}}$ such that
      $x \in \confES{\es{P}_1 \backslash l}$
      \begin{align*}
        v'(x)
        = v_1^l(x)
        = \dfrac{v_1 (\set{l} \cup x)}{v_1(\set{l})}
        = \dfrac{v_{\seq{1}{2}} (\set{l} \cup x)}{v_{\seq{1}{2}}(\set{l})}
        = v(x)
      \end{align*}

      Since $(\seq{\es{P}_1}{\es{P}_2}) \backslash l$ is a probabilistic event structure, we obtain
    $v'(x) \geq v(y)$.
      
    \item $x \in \confES{\seq{\es{P}_1 \backslash l}{\es{P}_2}}$ such that
      $\exists x_1 \in \confmax{\es{P}_1 \backslash l}, x_2 \in \confES{\es{P}_2}$ such that
      $x = x_1 \cup (x_2 \times \{x_1\})$
      \begin{align*}
        v'(x)
        = v_1^l(x_1) \cdot v_2(x_2)
        = \dfrac{v_1(\set{l} \cup x_1) \cdot v_2(x_2)}{v_1(\set{l})}
        = \dfrac{v_{\seq{1}{2}} (\set{l} \cup x)}{v_{\seq{1}{2}}(\set{l})}
        = v(x)
      \end{align*}

      Since $(\seq{\es{P}_1}{\es{P}_2}) \backslash l$ is a probabilistic event structure, we obtain
      $v'(x) \geq v(y)$.
    \end{enumerate}    
  \end{enumerate}
\end{proof}

\subsubsection*{Proof of Lemma~\ref{lem:conc-rem-init2}}
\begin{proof}
  Let
  \begin{align*}
    & \es{P}_1 = \pespq{\es{E}_1}{v_1} \\
    & \es{P}_2 = \pespq{\es{E}_2}{v_2} \\
    & \conc{\es{P}_1}{\es{P}_2} = \pespq{\es{E}}{v} \\
    & (\conc{\es{P}_1}{\es{P}_2}) \backslash l = \pespq{\es{E}'}{v'} \\
    & \es{P}_1 \backslash l = \pespq{\es{E}_1^l}{v_1^l} \\
    & \es{P}_2 \backslash l = \pespq{\es{E}_2^l}{v_2^l} \\
    & \conc{(\es{P}_1 \backslash l)}{\es{P}_2} = \pespq{\es{E}^l}{v^l} \\
    & l \in \init{\conc{\es{P}_1}{\es{P}_2}}
  \end{align*}

  Due to Lemma~\ref{lem:seq-rem-init1}, and similarly to it, we focus when $l \in \init{\es{E}_1}$
  and only show
  \begin{enumerate}
  \item $\forall x \in \confES{(\conc{\es{P}_1}{\es{P}_2}) \backslash l},
    y \in \confES{\conc{(\es{P}_1 \backslash l)}{\es{P}_2}}\, .\,
    f[x] \subseteq y \Rightarrow v'(x) \geq v^l(y)$
    \begin{align*}
      v'(x)
      &= \dfrac{v(\set{l} \cup x)}{v(\set{l})} 
        = \dfrac{v_1((\set{l} \cup x) \cap E_1) \cdot v_2((\set{l} \cup x) \cap E_2)}{v_1(\set{l})} \\
      &= \dfrac{v_1(\set{l} \cup (x \cap E_1)) \cdot v_2(x \cap E_2)}{v_1(\set{l})}
        = v_1^l(x \cap E_1) \cdot v_2(x \cap E_2)
        = v^l(x)
    \end{align*}

    Since $\conc{(\es{P}_1 \backslash l)}{\es{P}_2}$ is a probabilistic event structure, we obtain
    $v'(x) \geq v^l(y)$.
    
  \item $\forall x \in \confES{\conc{(\es{P}_1 \backslash l)}{\es{P}_2}} ,
    y \in \confES{(\conc{\es{P}_1}{\es{P}_2}) \backslash l}\, .\,
    f[x] \subseteq y \Rightarrow v^l(x) \geq v'(y)$
    \begin{align*}
      v^l(x)
      &= v_1^l(x \cap E_1) \cdot v_2(x \cap E_2)
        = \dfrac{v_1(\set{l} \cup (x \cap E_1)) \cdot v_2(x \cap E_2)}{v_1(\set{l})}\\
      &= \dfrac{v_1((\set{l} \cup x) \cap E_1) \cdot v_2((\set{l} \cup x) \cap E_2)}{v_1(\set{l})}
        = \dfrac{v(\set{l} \cup x)}{v(\set{l})}
        = v'(x)
    \end{align*}

    Since $(\conc{\es{P}_1}{\es{P}_2}) \backslash l$ is a probabilistic event structure, we obtain
    $v^l(x) \geq v'(y)$.
  \end{enumerate}
\end{proof}

\subsubsection*{Proof of Lemma~\ref{lem:conc-symmetric2}}
\begin{proof}
  It follows directly from Definition~\ref{def:pes-conc2}.
\end{proof}

\subsubsection*{Proof of Lemma~\ref{lem:prob-init}}
\begin{proof}
  \begin{itemize}
  \item $sk \in \init{\mf{\com{skip}}}$

    It follows directly that $v(\set{sk}) = 1$.

  \item $a \in \init{\mf{\com{a}}}$

    It follows directly that $v(\set{a}) = 1$.

  \item $\tau \in \init{\mf{\probC{\com{C}_1}{p}{\com{C}_2}}}$

    It follows directly that $v(\set{\tau}) = 1$.

  \item $l' \in \init{\mf{\seq{\com{C}_1}{\com{C}_2}}}$

    By Definition~\ref{def:pes-seq2} we have $l' \in \init{\mf{\com{C}_1}}$.  By i.h.,
    $v(\set{l'}) = 1$ and since $l' \in \init{\mf{\seq{\com{C}_1}{\com{C}_2}}}$ we are done.

  \item $l' \in \init{\mf{\conc{\com{C}_1}{\com{C}_2}}}$

    By Definition~\ref{def:pes-conc2} we have $l' \in \init{\mf{\com{C}_1}}$ or
    $l' \in \init{\mf{\com{C}_2}}$.  By i.h., $v(\set{l'}) = 1$ for both cases. Since
    $l' \in \init{\mf{\conc{\com{C}_1}{\com{C}_2}}}$ we are done.
  \end{itemize}
\end{proof}

\subsubsection*{Proof of Lemma~\ref{lem:sound-aux-seq-2}}
\begin{proof}
  Follows directly from the respective definitions.
\end{proof}

\subsubsection*{Proof of Lemma~\ref{lem:sound-aux-conc-2}}
\begin{proof}
  Let
  \begin{align*}
    \es{P} &= \ppes{E}{\leq}{\#}{v} \\
    \es{P}_i &= \ppes{E_i}{\leq_i}{\#_i}{v_i} \\
    \conc{\es{P}_i}{\es{P}} &= \ppes{E'_i}{\leq'_i}{\#'_i}{v'_i} \\
    \sum_i p_i \cdot \es{P}_i &= \ppes{E^1}{\leq^1}{\#^1}{v^1} \\
    \conc{(\sum_i p_i \cdot \es{P}_i)}{\es{P}} &= \ppes{E^2}{\leq^2}{\#^2}{v^2} \\
    \sum_i p_i \cdot (\conc{\es{P}_i}{\es{P}}) &= \ppes{E^3}{\leq^3}{\#^3}{v^3}
  \end{align*}

  \begin{enumerate}
  \item
    $\conc{\left( \sum_i p_i \cdot \es{P}_i \right)}{\es{P}} \sqsubseteq
    \sum_i p_i \cdot (\conc{\es{P}_i}{\es{P}})$

    We know that $E^2 = E^1 \uplus E = (\{\tau\} \uplus \biguplus_{i} E_i) \uplus E$ and
    $E^3 = \{\tau\} \uplus \biguplus_{i}(E_i \uplus E)$.  Note that with
    $\biguplus_{i}(E_i \uplus E)$, for each $i$ we are making a copy of $E$.  We then denote the
    events of $E$ from $\biguplus_{i}(E_i \uplus E)$ as $e^i$, for each $i$ (it is important to
    highlight that the plain events of $e$ and $e^i$ are the same) and we write $E^{(i)}$ for the
    $P$-copy inside branch $i$.

    \begin{enumerate}
    \item We begin by defining, for a fixed $i$, $f_i : E^2 \rightarrow E^3$ such that
      \[
        f_i(e) =
        \begin{cases}
          e & \text{ if } e \in \{\tau\} \uplus E_i  \subseteq E^1 \\
          e^i & \text{ if } e \in E
        \end{cases}
      \]

      Since $i$ is fixed, it is straightforward to see that $f_i$ is injective ($f_i$ is essentially
      the identity map, where for each $i$ it creates a copy of the events in $E$).

    \item $\pi(f_i(e)) = \pi(e)$

      If $e \in \{\tau\} \uplus E_i$ then we are done since $f_i(e) = e$.
      If $e \in E$, then $\pi(f_i(e)) = \pi(e^i) = \pi(e)$, since $e^i$ is a copy of $e$ for a fixed $i$.
      
    \item $e \leq^2 e' \implies f_i(e) \leq^3 f_i(e')$

      Assume $e \leq^2 e'$. By Definition~\ref{def:pes-conc2}, $e \leq^1 e'$ or $e \leq e'$.  If
      $e \leq e'$, then by Definition~\ref{def:pes-conc2} and for a given $i$,
      $f_i(e) \leq'_i f_i(e')$.  By Remark~\ref{rem:probc-n} $f_i(e) \leq^3 f_i(e')$.

      If $e \leq^1 e'$, then by Remark~\ref{rem:probc-n}, exists $i$ such that $e \leq_i e'$.  By
      Definition~\ref{def:pes-conc2}, $e \leq'_i e'$. By Remark~\ref{rem:probc-n}
      $f_i(e) \leq^3 f_i(e')$.
      
    \item $e \#^2 e' \implies f_i(e) \#^3 f_i(e')$

      Similar to the previous case.

    \item $\forall x \in \confES{\conc{(\sum_i p_i \cdot \es{P}_i)}{\es{P}}},
      y \in \confES{\sum_i p_i \cdot (\conc{\es{P}_i}{\es{P}})}\, .\,
      f_i[x] \subseteq y \Rightarrow v^2(x) \geq v^3(y)$

      Let $x \in \confES{\conc{(\sum_i p_i \cdot \es{P}_i)}{\es{P}}}$,
      $y \in \confES{\sum_i p_i \cdot (\conc{\es{P}_i}{\es{P}})}$, such that
      $f_i[x] \subseteq y$.
      
      By Definition~\ref{def:pes-conc2} and Remark~\ref{rem:probc-n} we have:
      \[
        v^2(x) =
        \begin{cases}
          p_i v_i((x \cap E^1) \backslash \tau) \cdot v(x \cap E)
          & \text{ if } (x \cap E^1) \backslash \tau \in \confES{\es{P}_i} \\
          v(x \cap E)
          & \text{ if } (x \cap E^1) = \emptyset \text{ or } (x \cap E^1 = \{\tau\})
        \end{cases}
      \]

      By Remark~\ref{rem:probc-n} we have:
      \[
        v^3(x) =
        \begin{cases}
          p_i v'_i(y \backslash \tau)
          & \text{ if } y \backslash \tau \in \confES{\conc{\es{P}_i}{\es{P}}}\\
          1
          & \text{ if } (y = \{\tau\}) \text{ or } (y = \emptyset)
        \end{cases}
      \]

      We have the following cases:
      \begin{enumerate}
      \item $y = \emptyset$

        Thus $v^3(y) = 1$.  Since $f_i[x] \subseteq y$, then $x = \emptyset$. Thus
        $v^2(x) = v(\emptyset \cap E) = v(\emptyset) = 1$.

      \item $y = \{\tau\}$

        Thus $v^3(y) = 1$.  Since $f_i[x] \subseteq y$, then $x = \emptyset$ or $x = \{\tau\}$.  If
        $x = \emptyset$ we are done.  If $x = \{\tau\}$, then
        $v^2(\{\tau\}) = v(\{\tau\} \backslash E)$, which entails $\{\tau\} \cap E^1 = \{\tau\}$.
        Then $\{\tau\} \cap E = \emptyset$ since $\tau \not\in E$.  Thus
        $v(\{\tau\} \cap E) = v(\emptyset) = 1$.

      \item $y \backslash \tau \in \confES{\conc{\es{P}_i}{\es{P}}}$

        Thus $v^3(y) = p_i \cdot v'_i(y \backslash \tau)$.  By Definition~\ref{def:pes-conc2},
        $v'_i(y \backslash \tau) = v_i(y \backslash \tau \cap E_i) \cdot v(y \backslash \tau \cap
        E^{(i)})$. Since $\tau \not\in E_i$ and $\tau \not\in E$, then
        $y\backslash \tau \cap E_i = y \cap E_i$ and
        $y \backslash \tau \cap E^{(i)} = y \cap E^{(i)}$.  Thus
        $p_i \cdot v'_i(y \backslash \tau) = p_i \cdot v_i(y \cap E_i) \cdot v(y \cap E^{(i)})$.

        Since $f_i[x] \subseteq y$, if $x = \emptyset$ or $x = \{\tau\}$ we know that $v^2(x) = 1$.
        Hence $1 \geq p_i \cdot v_i(y \cap E_i) \cdot v(y \cap E)$, since $v_i$ and $v$ are
        configuration-valuations.

        It lacks to consider the case where
        $v^2(x) = p_i v_i((x \cap E^1)\backslash \tau) \cdot v(x \cap E)$.  To show that
        $p_i v_i((x \cap E^1)\backslash \tau) \cdot v(x \cap E) \geq p_i \cdot v_i(y \cap E_i) \cdot
        v(y \cap E^{(i)})$, it suffices to show that
        $v_i((x \cap E^1)\backslash \tau) \geq v_i(y \cap E_i)$ and
        $v(x \cap E) \geq v(y \cap E^{(i)})$, since product is monotone in each factor.

        We argue first that $v(x \cap E) \geq v(y \cap E^{(i)})$. We know that
        $f_i[x_i \cap E] = \{ f_i(e) \mid e \in x \cap E \} = f_i[x_i] \cap E^{(i)}$. This set is
        composed by the events of $E$ that belong in $x$ in the $i$-th copy. Thus, we have
        $f_i[x_i \cap E] = f_i[x_i] \cap E^{(i)} \subseteq y \cap E^{(i)}$, from
        $f_i[x_i] \subseteq y$.  Since $P$ is a probabilistic event structure, then
        $v(f_i[x \cap E]) \geq v(y \cap E^{(i)})$. Hence $v(f_i[x \cap E]) = v(x \cap E)$. Thus
        $v(x \cap E) \geq v(y \cap E^{(i)})$.  With similar arguments, we have
        $v_i((x \cap E^1)\backslash \tau) \geq v_i(y \cap E_i)$.  Thus
        $p_i v_i((x \cap E^1)\backslash \tau) \cdot v(x \cap E) \geq p_i \cdot v_i(y \cap E_i) \cdot
        v(y \cap E^{(i)})$.
      \end{enumerate}
    \end{enumerate}

  \item $x \in \confmax{\conc{\left( \sum_i p_i \cdot \es{P}_i \right)}{\es{P}}}$ iff
    $x \in \confmax{\sum_i p_i \cdot (\conc{\es{P}_i}{\es{P}})}$

    For both cases, it is relevant to notice the following: let $\es{P}_1, \es{P}_2$ be two
    probabilistic event structures such that $x_1 \in \confES{\es{P}_1}$ and
    $x_2 \in \confES{\es{P}_2}$. Then $x = x_1 \cup x_2 \in \confES{\conc{\es{P}_1}{\es{P}_2}}$,
    since by Definition~\ref{def:pes-conc2} there is no conflict between events of $\es{P}_1$ and
    events of $\es{P}_2$.
    \begin{itemize}
    \item[$\Leftarrow$] If $x \in \confmax{\conc{\left( \sum_i p_i \cdot \es{P}_i \right)}{\es{P}}}$
      then $x \in \confmax{\sum_i p_i \cdot (\conc{\es{P}_i}{\es{P}})}$

      Let $x \in \confmax{\conc{\left( \sum_i p_i \cdot \es{P}_i \right)}{\es{P}}}$.  We can
      represent $x$ as follows: $x = x \cap (E \uplus E_i) = (x \cap E) \uplus (x \cap E_i)$, for
      $E_i$ in $\biguplus_i E_i$ and where $x \cap E \in \confmax{\es{P}}$ and
      $x \cap E_i \in \confmax{\sum_i p_i \cdot \es{P}_i}$.
      Hence, it follows directly that
      $(x \cap E) \uplus (x \cap E_i) = x \cap (E \uplus E_i) = x
      \in \confmax{\sum_i p_i \cdot (\conc{\es{P}_i}{\es{P}})}$.
      
    \item[$\Rightarrow$] If $x \in \confmax{\sum_i p_i \cdot (\conc{\es{P}_i}{\es{P}})}$ then
      $x \in \confmax{\conc{\left( \sum_i p_i \cdot \es{P}_i \right)}{\es{P}}}$

      Let $x \in \confmax{\sum_i p_i \cdot (\conc{\es{P}_i}{\es{P}})}$.  Hence
      $\exists x_i \in \confmax{\sum_i p_i \cdot \es{P}_i}, y \in \confmax{\es{P}}$ such that
      $x = x_i \cup y$. Since $x_i \in \confmax{\sum_i p_i \cdot \es{P}_i}$, then
      $\exists i\, .\, x_i \in \confmax{\es{P}_i}$. We then have
      $x_i \cup y = x \in \confmax{\conc{\es{P}_i}{\es{P}}}$ and consequently
      $x \in \confmax{\conc{\left( \sum_i p_i \cdot \es{P}_i \right)}{\es{P}}}$.
    \end{itemize}
  \end{enumerate}
\end{proof}

\subsubsection*{Proof of Lemma~\ref{lem:sound-aux-seq-conc-2}}
\begin{proof}
  \begin{enumerate}
  \item $\com{C} \equiv \seq{\com{C}_1}{\com{C}_2}$

    We know that
    $\seq{\com{C}_1}{\com{C}_2} \rightarrow \sum_i p_i \cdot (\tau, \seq{\com{C}_i}{\com{C}_2})$.
    By the rules in Figure~\ref{fig:op-small2} we have
    $\com{C}_1 \rightarrow \sum_i p_i \cdot (\tau, \com{C}_i)$.  By i.h. we have
    $x_1 \in \confmax{\mf{\com{C}_1}_{\gamma}}$ and
    $x_1 \in \confmax{\sum_i p_i \cdot \mf{\com{C}_i}_{\gamma}}$ such that
    $\exists \mf{\com{C}_i}_{\gamma}\, .\, v_1(x_1) = v_i(x_1)$.  Let
    $x_2 \in \confmax{\mf{\com{C}_2}_{\gamma}}$.  By Definition~\ref{def:pes-seq2} we have
    $x_1 \cup (x_2 \times \{x_1\}) \in \confmax{\mf{\seq{\com{C}_1}{\com{C}_2}}_{\gamma}}$ and
    $x_1 \cup (x_2 \times \{x_1\}) \in \confmax{\seq{(\sum_i p_i \cdot
        \mf{\com{C}_i}_{\gamma})}{\mf{\com{C}_2}_{\gamma}}}$, and
    $v(x) = v_1(x_1) \cdot v_2(x_2) = v_i(x_1) \cdot v_2(x_2)$.  By Lemma~\ref{lem:sound-aux-seq-2}
    we have
    $x_1 \cup (x_2 \times \{x_1\}) \in \confmax{\sum_i p_i \cdot
      \mf{\seq{\com{C}_i}{\com{C}_2}}_{\gamma}}$.

  \item $\com{C} \equiv \conc{\com{C}_1}{\com{C}_2}$

    We know that
    $\conc{\com{C}_1}{\com{C}_2} \rightarrow \sum_i p_i \cdot (\tau, \conc{\com{C}_i}{\com{C}_2})$.
    By the rules in Figure~\ref{fig:op-small2} we have
    $\com{C}_1 \rightarrow \sum_i p_i \cdot (\tau, \com{C}_i)$.  By i.h. we have
    $x_1 \in \confmax{\mf{\com{C}_1}_{\gamma}}$ and
    $x_1 \in \confmax{\sum_i p_i \cdot \mf{\com{C}_i}_{\gamma}}$ such that
    $\exists \mf{\com{C}_i}_{\gamma}\, .\, v_1(x_1) = v_i(x_1)$.  Let
    $x_2 \in \confmax{\mf{\com{C}_2}_{\gamma}}$.  By Definition~\ref{def:pes-conc2} we have
    $x_1 \cup x_2 \in \confmax{\mf{\conc{\com{C}_1}{\com{C}_2}}_{\gamma}}$ and
    $x_1 \cup x_2 \in \confmax{\conc{(\sum_i p_i \cdot
        \mf{\com{C}_i}_{\gamma})}{\mf{\com{C}_2}}_{\gamma}}$, and
    $v(x) = v_1(x_1) \cdot v_2(x_2) = v_i(x_1) \cdot v_2(x_2)$.  By Lemma~\ref{lem:sound-aux-conc-2}
    we have $x_1 \cup x_2 \in \confmax{\sum_i p_i \cdot \mf{\conc{\com{C}_i}{\com{C}_2}}}_{\gamma}$.

    A similar reasoning is applied when
    $\conc{\com{C}_1}{\com{C}_2} \rightarrow \sum_j p_j \cdot (\tau, \conc{\com{C}_2}{\com{C}_j})$.
  \end{enumerate}
\end{proof}

\subsubsection*{Proof of Lemma~\ref{lem:prog-2}}
\begin{proof}
  Induction over $\com{C}$.
  \begin{itemize}
  \item $\com{C} \equiv \com{skip}$.

    It follows directly that $\com{skip} \twoheadrightarrow 1 \cdot (sk, \checkmark)$

  \item $\com{C} \equiv \com{a}$.

    It follows directly that $\com{a} \twoheadrightarrow 1 \cdot (a, \checkmark)$

  \item $\com{C} \equiv \probC{\com{C}_1}{p}{\com{C}_2}$

    By Figure~\ref{fig:op-small2},
    $\probC{\com{C}_1}{p}{\com{C}_2} \rightarrow p \cdot (\tau, \com{C}_1) + (1-p) \cdot (\tau,
    \com{C}_2)$.  By i.h., $\exists$ $\sum_n p_n (\omega_n, \com{C}_n)$,
    $\sum_m p_m (\omega_m, \com{C}_m)$ s.t.
    $\com{C}_1 \twoheadrightarrow \sum_n p_n (\omega_n, \com{C}_n)$ and
    $\com{C}_2 \twoheadrightarrow \sum_m p_m (\omega_m, \com{C}_m)$.  By Figure~\ref{fig:op-nstep2},
    $\probC{\com{C}_1}{p}{\com{C}_2} \twoheadrightarrow \sum_n p_n (\tau : \omega_n, \com{C}_n) + \sum_m
    p_m (\tau : \omega_m, \com{C}_m)$.

  \item $\com{C} \equiv \seq{\com{C}_1}{\com{C}_2}$

    According to Figure~\ref{fig:op-small2} we have three cases:
    \begin{enumerate}
    \item $\seq{\com{C}_1}{\com{C}_2} \rightarrow 1 \cdot (l, \com{C}_2)$

      By i.h., $\exists$ $\sum_n p_n (\omega_n, \com{C}_n)$ s.t.
      $\com{C}_2 \twoheadrightarrow \sum_n p_n (\omega_n, \com{C}_n)$.  By
      Figure~\ref{fig:op-nstep2},
      $\seq{\com{C}_1}{\com{C}_2} \twoheadrightarrow \sum_n p_n (l : \omega_n, \com{C}_n)$.

    \item $\seq{\com{C}_1}{\com{C}_2} \rightarrow 1 \cdot (l, \seq{\com{C}_1'}{\com{C}_2})$

      By i.h., $\exists$ $\sum_n p_n (\omega_n, \com{C}_n)$ s.t.
      $\seq{\com{C}_1'}{\com{C}_2} \twoheadrightarrow \sum_n p_n (\omega_n, \com{C}_n)$.  By
      Figure~\ref{fig:op-nstep2},
      $\seq{\com{C}_1}{\com{C}_2} \twoheadrightarrow \sum_n p_n (l : \omega_n, \com{C}_n)$.

    \item
      $\seq{\com{C}_1}{\com{C}_2} \rightarrow \sum_i p_i \cdot (\tau, \seq{\com{C}_i}{\com{C}_2})$

      By i.h., $\forall i$, $\exists$ $\sum_n p_n (\omega_{i n}, \com{C}_{i n})$ s.t.
      $\seq{\com{C}_i}{\com{C}_2} \twoheadrightarrow \sum_n p_n (\omega_{i n}, \com{C}_{i n})$.  By
      Figure~\ref{fig:op-nstep2},
      $\seq{\com{C}_1}{\com{C}_2} \twoheadrightarrow \sum_i p_i \sum_n p_n
      (\tau : \omega_{i n}, \com{C}_{i n})$.
    \end{enumerate}

  \item $\com{C} \equiv \conc{\com{C}_1}{\com{C}_2}$

    According to Figure~\ref{fig:op-small2} we have three cases:
    \begin{enumerate}
    \item $\conc{\com{C}_1}{\com{C}_2} \rightarrow 1 \cdot (l, \com{C}_2)$

      By i.h., $\exists$ $\sum_n p_n (\omega_n, \com{C}_n)$ s.t.
      $\com{C}_2 \twoheadrightarrow \sum_n p_n (\omega_n, \com{C}_n)$.  By
      Figure~\ref{fig:op-nstep2},
      $\conc{\com{C}_1}{\com{C}_2} \twoheadrightarrow \sum_n p_n (l : \omega_n, \com{C}_n)$.

    \item $\conc{\com{C}_1}{\com{C}_2} \rightarrow 1 \cdot (l, \conc{\com{C}_1'}{\com{C}_2})$

      By i.h., $\exists$ $\sum_n p_n (\omega_n, \com{C}_n)$ s.t.
      $\conc{\com{C}_1'}{\com{C}_2} \twoheadrightarrow \sum_n p_n (\omega_n, \com{C}_n)$.  By
      Figure~\ref{fig:op-nstep2},
      $\conc{\com{C}_1}{\com{C}_2} \twoheadrightarrow \sum_n p_n (l : \omega_n, \com{C}_n)$.

    \item
      $\conc{\com{C}_1}{\com{C}_2} \rightarrow \sum_i p_i \cdot (\tau, \conc{\com{C}_i}{\com{C}_2})$

      By i.h., $\forall i$, $\exists$ $\sum_n p_n (\omega_{i n}, \com{C}_{i n})$ s.t.
      $\conc{\com{C}_i}{\com{C}_2} \twoheadrightarrow \sum_n p_n (\omega_{i n}, \com{C}_{i n})$.  By
      Figure~\ref{fig:op-nstep2},
      $\conc{\com{C}_1}{\com{C}_2} \twoheadrightarrow \sum_i p_i \sum_n p_n
      (\tau : \omega_{i n}, \com{C}_{i n})$.
    \end{enumerate}

  \item $\com{C} \equiv \rec{X}{\com{D}}$

    According to Figure~\ref{fig:op-small2} we have two cases:
    \begin{enumerate}
    \item $\rec{X}{\com{D}} \rightarrow 1 \cdot (l, \com{D}'[X \leftarrow \rec{X}{\com{D}}])$

      By i.h., $\exists$ $\sum_n p_n (\omega_n, \com{D}_n)$ s.t.
      $\com{D}'[X \leftarrow \rec{X}{\com{D}}] \twoheadrightarrow \sum_n p_n (\omega_n, \com{D}_n)$.
      By Figure~\ref{fig:op-nstep2},
      $\rec{X}{\com{D}} \twoheadrightarrow \sum_n p_n (l : \omega_n, \com{D}_n)$.

    \item $\rec{X}{\com{D}} \rightarrow \sum_i p_i (\tau, \com{D}_i [X \leftarrow \rec{X}{\com{D}}])$

      By i.h.,$\forall i\ \exists\ \sum_n p_n (\omega_{i n}, \com{D}_{i n})$ s.t.
      $\com{D}_i [X \leftarrow \rec{X}{\com{D}}] \twoheadrightarrow \sum_n p_n (\omega_{i n},
      \com{D}_{i n})$.  By Figure~\ref{fig:op-nstep2},
      $\rec{X}{\com{D}} \twoheadrightarrow \sum_i p_i \sum_n p_n (\tau : \omega_{i n}, \com{D}_{i n})$.
    \end{enumerate}
  \end{itemize}
\end{proof}

\subsubsection*{Proof of Lemma~\ref{res:soundI-2}}
\begin{proof}
  Induction over rules in Figure~\ref{fig:op-small2}.

  \begin{itemize}
  \item $\com{skip} \rightarrow 1 \cdot (sk, \checkmark)$

    It follows directly that $\mf{\checkmark} \equiv \mf{\com{skip}} \backslash sk \equiv \emptyset$.

  \item $\com{a} \rightarrow 1 \cdot (a, \checkmark)$

    It follows directly that $\mf{\checkmark} \equiv \mf{\com{a}} \backslash a \equiv \emptyset$.

  \item $\probC{\com{C}_1}{p}{\com{C}_2} \rightarrow p \cdot (\tau, \com{C}_1) + (1-p) \cdot (\tau, \com{C}_2)$

    It follows directly that
    $p \cdot \mf{\com{C}_1} + (1-p) \cdot \mf{\com{C}_2} = \mf{\probC{\com{C}_1}{p}{\com{C}_2}}$.

  \item $\seq{\com{C}_1}{\com{C}_2} \rightarrow 1 \cdot (l, \com{C}_2)$
    \begin{align*}
      & \seq{\com{C}_1}{\com{C}_2} \xrightarrow{l} \com{C}_2 \\
      \Rightarrow & \set{ \text{Figure~\ref{fig:op-small2} entails}} \\
      & \com{C}_1 \xrightarrow{l} \checkmark \\
      \Rightarrow & \set{\text{i.h.}} \\
      & \mf{\checkmark} \equiv \mf{\com{C}_1} \backslash l \\
      \Rightarrow & \set{ \text{Lemma~\ref{lem:seq-mono2}}} \\
      & \seq{\mf{\checkmark}}{\mf{\com{C}_2}} \equiv
        \seq{(\mf{\com{C}_1} \backslash l)}{\mf{\com{C}_2}} \\
      \Rightarrow & \set{
                    \seq{\mf{\checkmark}}{\mf{\com{C}_2}} \equiv \mf{\com{C}_2},
                    \text{ Lemma~\ref{lem:seq-rem-init2}}
                    } \\
      & \mf{\com{C}_2} \equiv (\seq{\mf{\com{C}_1}}{\mf{\com{C}_2}}) \backslash l \\
      \Rightarrow & \set{\text{Definition~\ref{def:den-sem2}}} \\
      & \mf{\com{C}_2} \equiv \mf{\seq{\com{C}_1}{\com{C}_2}} \backslash l \\
    \end{align*}

  \item $\seq{\com{C}_1}{\com{C}_2} \rightarrow 1 \cdot (l, \seq{\com{C}'_1}{\com{C}_2})$
    \begin{align*}
      & \seq{\com{C}_1}{\com{C}_2} \xrightarrow{l} \seq{\com{C}'_1}{\com{C}_2} \\
      \Rightarrow & \set{ \text{Figure~\ref{fig:op-small2} entails}} \\
      & \com{C}_1 \xrightarrow{l} \com{C}'_1 \\
      \Rightarrow & \set{\text{i.h.}} \\
      & \mf{\com{C}'_1} \equiv \mf{\com{C}_1} \backslash l \\
      \Rightarrow & \set{ \text{Lemma~\ref{lem:seq-mono2}}} \\
      & \seq{\mf{\com{C}'_1}}{\mf{\com{C}_2}} \equiv \seq{(\mf{\com{C}_1} \backslash l)}{\mf{\com{C}_2}} \\
      \Rightarrow & \set{\text{Lemma~\ref{lem:seq-rem-init2}}} \\
      & \seq{\mf{\com{C}'_1}}{\mf{\com{C}_2}} \equiv (\seq{\mf{\com{C}_1}}{\mf{\com{C}_2}}) \backslash l \\
      \Rightarrow & \set{\text{Definition~\ref{def:den-sem2}}} \\
      & \mf{\seq{\com{C}'_1}{\com{C}_2}} \equiv \mf{\seq{\com{C}_1}{\com{C}_2}} \backslash l \\
    \end{align*}

  \item $\seq{\com{C}_1}{\com{C}_2} \rightarrow \sum_i p_i \cdot (\tau, \seq{\com{C}_i}{\com{C}_2})$
    \begin{align*}
      & \seq{\com{C}_1}{\com{C}_2} \rightarrow \sum_i p_i \cdot (\tau, \seq{\com{C}_i}{\com{C}_2}) \\
      \Rightarrow &  \set{ \text{Figure~\ref{fig:op-small2} entails}} \\
      & \com{C}_1 \rightarrow \sum_i p_i \cdot (\tau, \com{C}_i) \\
      \Rightarrow & \set{\text{i.h.}} \\
      & \mf{\com{C}_1} \sqsubseteq \sum_i p_i \cdot \mf{\com{C}_i} \\
      \Rightarrow & \set{ \text{Lemma~\ref{lem:seq-mono2}}} \\
      & \mf{\seq{\com{C}_1}{\com{C}_2}} \sqsubseteq \seq{(\sum_i p_i \cdot \mf{\com{C}_i})}{\mf{\com{C}_2}} \\
      \Rightarrow & \set{\text{Lemma~\ref{lem:sound-aux-seq-2}}} \\
      & \mf{\seq{\com{C}_1}{\com{C}_2}} \sqsubseteq \sum_i p_i \cdot \mf{\seq{\com{C}_i}{\com{C}_2}}
    \end{align*}

  \item $\conc{\com{C}_1}{\com{C}_2} \rightarrow 1 \cdot (l, \com{C}_2)$
    \begin{align*}
      & \conc{\com{C}_1}{\com{C}_2} \xrightarrow{l} \com{C}_2 \\
      \Rightarrow & \set{ \text{Figure~\ref{fig:op-small2} entails}} \\
      & \com{C}_1 \xrightarrow{l} \checkmark \\
      \Rightarrow & \set{\text{i.h.}} \\
      & \mf{\checkmark} \equiv \mf{\com{C}_1} \backslash l \\
      \Rightarrow & \set{\text{Lemma~\ref{lem:conc-mono2}}} \\
      & \conc{\mf{\checkmark}}{\mf{\com{C}_2}} \equiv \conc{(\mf{\com{C}_1} \backslash l)}{\mf{\com{C}_2}} \\
      \Rightarrow & \set{\conc{\mf{\checkmark}}{\mf{\com{C}_2}} \equiv \mf{\com{C}_2} } \\
      & \mf{\com{C}_2} \equiv \conc{(\mf{\com{C}_1} \backslash l)}{\mf{\com{C}_2}} \\
      \Rightarrow & \set{\text{Lemma~\ref{lem:conc-rem-init2}},
                    \text{Definition~\ref{def:den-sem2}}} \\
      & \mf{\com{C}_2} \equiv \mf{\conc{\com{C}_1}{\com{C}_2}} \backslash l \\
    \end{align*}

  \item $\conc{\com{C}_1}{\com{C}_2} \rightarrow 1 \cdot (l, \conc{\com{C}'_1}{\com{C}_2})$
    \begin{align*}
      & \conc{\com{C}_1}{\com{C}_2} \xrightarrow{l} \conc{\com{C}'_1}{\com{C}_2} \\
      \Rightarrow & \set{ \text{Figure~\ref{fig:op-small2} entails}} \\
      & \com{C}_1 \xrightarrow{l} \com{C}'_1 \\
      \Rightarrow & \set{\text{i.h.}} \\
      & \mf{\com{C}'_1} \equiv \mf{\com{C}_1} \backslash l \\
      \Rightarrow & \set{\text{Lemma~\ref{lem:conc-mono2}}} \\
      & \conc{\mf{\com{C}'_1}}{\mf{\com{C}_2}} \equiv \conc{(\mf{\com{C}_1} \backslash l)}{\mf{\com{C}_2}} \\
      \Rightarrow & \set{\text{Lemma~\ref{lem:conc-rem-init2}},
                    \text{Definition~\ref{def:den-sem2}}} \\
      & \mf{\conc{\com{C}'_1}{\com{C}_2}} \equiv \mf{\conc{\com{C}_1}{\com{C}_2}} \backslash l \\
    \end{align*}

  \item $\conc{\com{C}_1}{\com{C}_2} \rightarrow \sum_i p_i \cdot (\tau, \conc{\com{C}_i}{\com{C}_2})$
    \begin{align*}
      & \conc{\com{C}_1}{\com{C}_2} \rightarrow \sum_i p_i \cdot (\tau, \conc{\com{C}_i}{\com{C}_2}) \\
      \Rightarrow &  \set{ \text{Figure~\ref{fig:op-small2} entails}} \\
      & \com{C}_1 \rightarrow \sum_i p_i \cdot (\tau, \com{C}_i) \\
      \Rightarrow & \set{\text{i.h.}} \\
      & \mf{\com{C}_1} \sqsubseteq \sum_i p_i \cdot \mf{\com{C}_i} \\
      \Rightarrow & \set{ \text{Lemma~\ref{lem:conc-mono2}}} \\
      & \mf{\conc{\com{C}_1}{\com{C}_2}} \sqsubseteq \conc{(\sum_i p_i \cdot \mf{\com{C}_i})}{\mf{\com{C}_2}} \\
      \Rightarrow & \set{\text{Lemma~\ref{lem:sound-aux-conc-2}, Definition~\ref{def:den-sem2}}} \\
      & \mf{\conc{\com{C}_1}{\com{C}_2}} \sqsubseteq \sum_i p_i \cdot \mf{\conc{\com{C}_i}{\com{C}_2}}
    \end{align*}

  \item $\conc{\com{C}_1}{\com{C}_2} \xrightarrow{l} \com{C}_1$
    \begin{align*}
      & \conc{\com{C}_1}{\com{C}_2} \xrightarrow{l} \com{C}_1 \\
      \Rightarrow & \set{ \text{Figure~\ref{fig:op-small2} entails}} \\
      & \com{C}_2 \xrightarrow{l} \checkmark \\
      \Rightarrow & \set{\text{i.h.}} \\
      & \mf{\checkmark} \equiv \mf{\com{C}_2} \backslash l \\
      \Rightarrow & \set{\text{Lemma~\ref{lem:conc-mono2}}} \\
      & \conc{\mf{\com{C}_1}}{\mf{\checkmark}} \equiv
        \conc{\mf{\com{C}_1}}{(\mf{\com{C}_2} \backslash l)} \\
      \Rightarrow & \set{\conc{\mf{\com{C}_1}}{\mf{\checkmark}} \equiv \mf{\com{C}_1} } \\
      & \mf{\com{C}_1} \equiv \conc{\mf{\com{C}_1}}{(\mf{\com{C}_2} \backslash l)} \\
      \Rightarrow & \set{\text{Lemma~\ref{lem:conc-rem-init2}},
                    \text{Definition~\ref{def:den-sem2}}} \\
      & \mf{\com{C}_1} \equiv \mf{\conc{\com{C}_1}{\com{C}_2}} \backslash l \\
    \end{align*}

  \item $\conc{\com{C}_1}{\com{C}_2} \xrightarrow{l} \conc{\com{C}_1}{\com{C}'_2}$
    \begin{align*}
      & \conc{\com{C}_1}{\com{C}_2} \xrightarrow{l} \conc{\com{C}_1}{\com{C}'_2} \\
      \Rightarrow & \set{ \text{Figure~\ref{fig:op-small2} entails}} \\
      & \com{C}_2 \xrightarrow{l} \com{C}'_2 \\
      \Rightarrow & \set{\text{i.h.}} \\
      & \mf{\com{C}'_2} \equiv \mf{\com{C}_2} \backslash l \\
      \Rightarrow & \set{\text{Lemma~\ref{lem:conc-mono2}}} \\
      & \conc{\mf{\com{C}_1}}{\mf{\com{C}'_2}} \equiv \conc{\mf{\com{C}_1}}{(\mf{\com{C}_2} \backslash l)} \\
      \Rightarrow & \set{\text{Lemma~\ref{lem:conc-rem-init2}},
                    \text{Definition~\ref{def:den-sem2}}} \\
      & \mf{\conc{\com{C}_1}{\com{C}'_2}} \equiv \mf{\conc{\com{C}_1}{\com{C}_2}} \backslash l \\
    \end{align*}

  \item $\conc{\com{C}_1}{\com{C}_2} \rightarrow \sum_j p_j \cdot (\tau, \conc{\com{C}_1}{\com{C}_j})$
    \begin{align*}
      & \conc{\com{C}_1}{\com{C}_2} \rightarrow \sum_j p_j \cdot (\tau, \conc{\com{C}_1}{\com{C}_j}) \\
      \Rightarrow &  \set{ \text{Figure~\ref{fig:op-small2} entails}} \\
      & \com{C}_2 \rightarrow \sum_j p_j \cdot (\tau, \com{C}_j) \\
      \Rightarrow & \set{\text{i.h.}} \\
      & \mf{\com{C}_2} \sqsubseteq \sum_j p_j \cdot \mf{\com{C}_j} \\
      \Rightarrow & \set{ \text{Lemma~\ref{lem:conc-mono2}}} \\
      & \mf{\conc{\com{C}_1}{\com{C}_2}} \sqsubseteq \conc{\mf{\com{C}_1}}{(\sum_j p_j \cdot \mf{\com{C}_j})} \\
      \Rightarrow & \set{\text{Lemma~\ref{lem:sound-aux-conc-2}, Definition~\ref{def:den-sem2}}} \\
      & \mf{\conc{\com{C}_1}{\com{C}_2}} \sqsubseteq \sum_j p_j \cdot \mf{\conc{\com{C}_1}{\com{C}_j}}
    \end{align*}
  \end{itemize}
\end{proof}

\subsubsection*{Proof of Theorem~\ref{res:soundII-2}}
\begin{proof}
  Induction over the size of $\omega_0$.
  \begin{itemize}
  \item $|\omega_0| = 1$

    We have that $\com{C} \twoheadrightarrow 1 \cdot (l, \checkmark)$.  It follows directly that
    $\{ l \} \in \confmax{\mf{\com{C}}}$ such that $\emptyset \cchain{l} \{ l \}$ and
    $v(\{ l \}) = 1$.

  \item $|\omega| > 1$

    Let us rewrite $p_0 (\omega_0, \checkmark) + \sum_k p_k (\omega_{k}, \com{C}_k)$ as
    $\sum_{ij} p_{ij} (\omega_{ij}, \com{C}_{ij})$
    \begin{align*}
      & \com{C} \twoheadrightarrow \sum_{ij} p_{ij} (\omega_{ij}, \com{C}_{ij}) \\
      \Rightarrow & \{ \text{Figure~\ref{fig:op-nstep2} entails}  \} \\
      & \com{C} \rightarrow \sum_i p_i (l', \com{C}_i)
        \qquad
        \forall i\, \com{C}_i \twoheadrightarrow \sum_j p_j (\omega_{ij}', \com{C}_{ij})
    \end{align*}

    We have two cases:
    \begin{enumerate}
    \item Case $l' \neq \tau$
      \begin{align*}
        & \com{C} \rightarrow 1 \cdot (l', \com{C}')
          \qquad \com{C'} \twoheadrightarrow p_0 (\omega_{0}', \checkmark)
          + \sum_{j \neq 0} p_j (\omega_j', \com{C}_j) \\
        \Rightarrow & \{ \text{Lemma~\ref{res:soundI-2}, i.h.}  \} \\
        & \mf{\com{C}} \backslash l' \equiv \mf{\com{C}'}
          \qquad
          \exists x_0 \in \confmax{\mf{\com{C}'}} \text{ such that }
          \emptyset \cchain{\omega_j'} x_0'
          \text{ and } p_0 = v'(x_0') \\
        \Rightarrow & \{ \text{Definition~\ref{def:rem-init2}} \} \\
        & \{ l' \} \cup x_0' \in \confmax{\mf{\com{C}}} \text{ such that }
          \emptyset \cchain{l'} \{ l' \} \cchain{\omega_j'} \{ l' \} \cup x_0'
          \text{ and }
          p_0 = v(\{ l' \} \cup x_0')
      \end{align*}
    \item Case $l' = \tau$
      \begin{align*}
        & \com{C} \rightarrow \sum_i p_i (l', \com{C}_i)
        \qquad
          \exists i\, \com{C}_i \twoheadrightarrow
          p_0' (\omega_{i0}', \checkmark) +
          \sum_{j \neq 0} p_j (\omega_{ij}', \com{C}_{ij}) \\
        \Rightarrow & \{ \text{Lemma~\ref{res:soundI-2}, i.h.} \} \\
        & \mf{\com{C}} \sqsubseteq \sum_i p_i \mf{\com{C}_i} \\
        & \exists i, \exists x_{i0} \in \confmax{\mf{\com{C}_i}} \text{ such that }
          \emptyset \cchain{\omega_{i0}'} x_{i0} \text{ and } p'_0 = v_i(x_{i0})
      \end{align*}
      Now we have two sub-cases:
      \begin{enumerate}
      \item Case $\com{C} = \probC{\com{C}_1}{p}{\com{C}_2}$

        By Definition~\ref{def:rem-init2},
        $\exists i, \exists \set{\tau} \cup x_{i0} \in \confmax{\mf{\com{C}}}$ such that 
        \begin{align*}
          &
            \emptyset \cchain{\tau} \set{\tau}
            \cchain{\omega_{i0}'} \set{\tau} \cup x_{i0}',\
            v_i(\set{\tau} \cup x_{i0}') = p_i \cdot p_0' = p_0
        \end{align*}
      \item Case $\com{C} = \seq{\com{C}_1}{\com{C}_2}$ or $\com{C} = \conc{\com{C}_1}{\com{C}_2}$

        By Remark~\ref{rem:probc-n} and Lemma~\ref{lem:sound-aux-seq-conc-2}
        $\exists i, \exists \set{\tau} \cup x_{i0} \in \confmax{\mf{\com{C}}}$ such that
        \begin{align*}
          &
            \emptyset \cchain{\tau} \set{\tau}
            \cchain{\omega_{i0}'} \set{\tau} \cup x_{i0}',\
            v_i(\set{\tau} \cup x_{i0}') = p_i \cdot p_0' = p_0
        \end{align*}

      \end{enumerate}
    \end{enumerate}
  \end{itemize}

\end{proof}

\subsubsection*{Proof of Lemma~\ref{res:adI-2}}
\begin{proof}

  \begin{itemize}
  \item $sk \in \init{\mf{\com{skip}}}$

    Let $\com{C}' = \checkmark$.  It follows directly that $\com{skip} \rightarrow 1 \cdot (sk, \checkmark)$
    and that $\mf{\com{skip}} \backslash sk \equiv \mf{\checkmark}$.

  \item $a \in \init{\mf{\com{a}}}$

    Let $\com{C}' = \checkmark$.  It follows directly that $\com{a} \rightarrow 1 \cdot (a, \checkmark)$
    and that $\mf{\com{a}} \backslash a \equiv \mf{\checkmark}$.

  \item $\tau \in \init{\mf{\probC{\com{C}_1}{p}{\com{C}_2}}}$

    By Definition~\ref{def:pes-prob2} we have that
    $\mf{\probC{\com{C}_1}{p}{\com{C}_2}} = p \cdot \mf{\com{C}_1} + (1-p) \cdot \mf{\com{C}_2}$,
    hence $\tau \in \init{p \cdot \mf{\com{C}_1} + (1-p) \cdot \mf{\com{C}_2}}$.  Let
    $l \in \init{\mf{\com{C}_1}}$ By Definition~\ref{def:pes-prob2} and Lemma~\ref{lem:prob-init} we
    have that $v(\set{\tau, l}) = p \cdot v_1(\set{l}) = p$.  It then follows directly that
    $\probC{\com{C}_1}{p}{\com{C}_2} \rightarrow
    p \cdot (\tau, \com{C}_1) + (1-p) \cdot (\tau, \com{C}_2)$ and
    $\mf{\probC{\com{C}_1}{p}{\com{C}_2}} = p \cdot \mf{\com{C}_1} + (1-p) \cdot \mf{\com{C}_2}$.
    Similarly we do the same when $l \in \init{\mf{\com{C}_2}}$.

  \item $l' \in \init{\mf{\seq{\com{C}_1}{\com{C}_2}}}$

    We have two cases:
    \begin{enumerate}
    \item $l' \neq \tau$

      By Definition~\ref{def:pes-seq2} we have that $l \in \init{\mf{\com{C}_1}}$.  By i.h.,
      $\exists \com{C}'$ such that $\com{C}_1 \rightarrow 1 \cdot (l, \com{C}')$ and
      $\mf{\com{C}_1} \backslash l \equiv \mf{\com{C}'}$. We have two cases:
      \begin{enumerate}
      \item $\com{C}' = \checkmark$

        We have $\com{C}_1 \rightarrow 1 \cdot (l, \checkmark)$ and
        $\mf{\com{C}_1} \backslash l \equiv \mf{\checkmark}$.  By the rules in Figure~\ref{fig:op-small2},
        $\seq{\com{C}_1}{\com{C}_2} \rightarrow 1 \cdot (l,\com{C}_2)$.
        By Definition~\ref{def:pes-seq2}, $\seq{(\mf{\com{C}_1} \backslash l)}{\mf{\com{C}_2}} \equiv
        \seq{\mf{\checkmark}}{\mf{\com{C}_2}} = \mf{\com{C}_2}$.
        
      \item $\com{C}' = \com{C}'_1$

        We have $\com{C}_1 \xrightarrow{l} \com{C}'_1$ and
        $\mf{\com{C}_1} \backslash l \equiv \mf{\com{C}'_1}$.  By the rules in Figure~\ref{fig:op-small2},
        $\seq{\com{C}_1}{\com{C}_2} \rightarrow 1 \cdot (l, \seq{\com{C}'_1}{\com{C}_2})$.  By
        Definition~\ref{def:pes-seq2},
        $\seq{(\mf{\com{C}_1} \backslash l)}{\mf{\com{C}_2}} \equiv \seq{\mf{\com{C}'_1}}{\mf{\com{C}_2}}$.
        By Definition~\ref{def:den-sem2}, $\mf{\seq{\com{C}'_1}{\com{C}_2}}$.
      \end{enumerate}
      
    \item $l' = \tau$

      We have $\tau \in \init{\mf{\seq{\com{C}_1}{\com{C}_2}}}$, which by
      Definition~\ref{def:pes-seq2} gives us that $\tau \in \init{\mf{\com{C}_1}}$.  By i.h.,
      $\exists \com{C}', \com{C}''$ such that
      $\com{C}_1 \rightarrow p \cdot (\tau, \com{C}') + (1-p) \cdot (\tau, \com{C}'')$ with
      $p = v(\set{\tau, e'})$ and $e' \in \init{\com{C}'}$, and
      $\mf{\com{C}_1} \sqsubseteq p \cdot \mf{\com{C}'} + (1-p) \cdot \mf{\com{C}''}$.  By the rules
      in Figure~\ref{fig:op-small2}, we have
      $\seq{\com{C}_1}{\com{C}_2} \rightarrow
      p \cdot (\tau, \seq{\com{C}'}{\com{C}_2}) + (1-p) \cdot (\tau, \seq{\com{C}''}{\com{C}_2})$.
      By Lemma~\ref{lem:seq-mono2},
      $\seq{\mf{\com{C}_1}}{\mf{\com{C}_2}} \sqsubseteq
      \seq{(p \cdot \mf{\com{C}'} + (1-p) \cdot \mf{\com{C}''})}{\com{C}_2}$.
      By Lemma~\ref{lem:sound-aux-seq-2} and Definition~\ref{def:den-sem2},
      $\mf{\seq{\com{C}_1}{\com{C}_2}}
      \sqsubseteq p \cdot \mf{\seq{\com{C}'}{\com{C}_2}}
      + (1-p) \cdot \mf{\seq{\com{C}''}{\com{C}_2}}$.
    \end{enumerate}

  \item $l \in \init{\mf{\conc{\com{C}_1}{\com{C}_2}}}$

    We have two cases:
    \begin{enumerate}
    \item $l' \neq \tau$

      By Definition~\ref{def:pes-conc2} we have two cases: 
      \begin{enumerate}
      \item $l \in \init{\mf{\com{C}_1}}$

        By i.h. $\exists \com{C}'\, .\, \com{C}_1 \rightarrow 1 \cdot (l, \com{C}')$ and
        $\mf{\com{C}_1} \backslash l \equiv \mf{\com{C}'}$. 
        By the rules in Figure~\ref{fig:op-small2} we have two cases:
        \begin{enumerate}
        \item $\com{C}' = \checkmark$

          We have $\com{C}_1 \rightarrow 1 \cdot (l, \checkmark)$ and
          $\mf{\com{C}_1} \backslash l \equiv \mf{\checkmark}$.  By the rules in
          Figure~\ref{fig:op-small1} we have
          $\conc{\com{C}_1}{\com{C}_2} \rightarrow 1 \cdot (l, \com{C}_2)$.  By
          Definition~\ref{def:pes-conc2}, $\conc{(\mf{\com{C}_1}\backslash l)}{\mf{\com{C}_2}}$.
          By Lemma~\ref{lem:conc-rem-init2} we have
          $(\conc{\mf{\com{C}_1}}{\mf{\com{C}_2}}) \backslash l$.
          By Definition~\ref{def:den-sem2}, $\mf{\conc{\com{C}_1}{\com{C}_2}} \backslash l$.

        \item $\com{C}' = \com{C}'_1$

          We have $\com{C}_1 \rightarrow 1 \cdot (l, \com{C}'_1)$ and
          $\mf{\com{C}_1} \backslash l \equiv \mf{\com{C}'_1}$.  By the rules in
          Figure~\ref{fig:op-small2} we have
          $\conc{\com{C}_1}{\com{C}_2} \rightarrow 1 \cdot (l, \conc{\com{C}'_1}{\com{C}_2})$. By
          Definition~\ref{def:pes-conc2}, $\conc{(\mf{\com{C}_1}\backslash l)}{\mf{\com{C}_2}}$.
          By Lemma~\ref{lem:conc-rem-init2} we have
          $(\conc{\mf{\com{C}_1}}{\mf{\com{C}_2}}) \backslash l$.  By Definition~\ref{def:den-sem2},
          $\mf{\conc{\com{C}_1}{\com{C}_2}} \backslash l$.
        \end{enumerate}
        
      \item $l \in \init{\mf{\com{C}_2}}$
        
        By i.h. $\exists \com{C}'\, .\, \com{C}_2 \rightarrow 1 \cdot (l, \com{C}')$ and
        $\mf{\com{C}_2} \backslash l \equiv \mf{\com{C}'}$.  By the rules in
        Figure~\ref{fig:op-small2} we have two cases:
        \begin{enumerate}
        \item $\com{C}' = \checkmark$

          We have $\com{C}_2 \rightarrow 1 \cdot (l, \checkmark)$ and
          $\mf{\com{C}_2} \backslash l \equiv \mf{\checkmark}$.  By the rules in
          Figure~\ref{fig:op-small2} we have
          $\conc{\com{C}_1}{\com{C}_2} \rightarrow 1 \cdot (l, \com{C}_1)$.  By
          Definition~\ref{def:pes-conc2}, $\conc{\mf{\com{C}_1}}{(\mf{\com{C}_2}\backslash l)}$.
          By Lemma~\ref{lem:conc-rem-init2} we have
          $(\conc{\mf{\com{C}_1}}{\mf{\com{C}_2}}) \backslash l$.
          By Definition~\ref{def:den-sem2}, $\mf{\conc{\com{C}_1}{\com{C}_2}} \backslash l$.

        \item $\com{C}' = \com{C}'_2$

          We have $\com{C}_2 \rightarrow 1 \cdot (l, \com{C}'_2)$ and
          $\mf{\com{C}_2} \backslash l \equiv \mf{\com{C}'_2}$.  By the rules in
          Figure~\ref{fig:op-small2} we have
          $\conc{\com{C}_1}{\com{C}_2} \rightarrow 1 \cdot (l, \conc{\com{C}_1}{\com{C}'_2})$. By
          Definition~\ref{def:pes-conc2}, $\conc{\mf{\com{C}_1}}{(\mf{\com{C}_2}\backslash l)}$.
          By Lemma~\ref{lem:conc-rem-init2} we have
          $(\conc{\mf{\com{C}_1}}{\mf{\com{C}_2}}) \backslash l$.  By Definition~\ref{def:den-sem2},
          $\mf{\conc{\com{C}_1}{\com{C}_2}} \backslash l$.
        \end{enumerate}
      \end{enumerate}
      
    \item $l' = \tau$

      We have $\tau \in \init{\mf{\conc{\com{C}_1}{\com{C}_2}}}$, which by
      Definition~\ref{def:pes-conc2} entails $\tau \in \init{\mf{\com{C}_1}}$ or
      $\tau \in \init{\mf{\com{C}_2}}$.
      We have two cases:
      \begin{enumerate}
      \item $\tau \in \init{\mf{\com{C}_1}}$

        By i.h., $\exists \com{C}', \com{C}''$ such that
        $\com{C}_1 \rightarrow p \cdot (\tau, \com{C}') + (1-p) \cdot (\tau, \com{C}'')$ with
        $p = v(\set{\tau, e'})$ and $e' \in \init{\mf{\com{C}'}}$, and
        $\mf{\com{C}_1} \sqsubseteq p \cdot \mf{\com{C}'} + (1-p) \cdot \mf{\com{C}''}$.  By the rules
        in Figure~\ref{fig:op-small2}, we have
        $\conc{\com{C}_1}{\com{C}_2} \rightarrow
        p \cdot (\tau, \conc{\com{C}'}{\com{C}_2}) + (1-p) \cdot (\tau, \conc{\com{C}''}{\com{C}_2})$.
        By Lemma~\ref{lem:conc-mono2},
        $\conc{\mf{\com{C}_1}}{\mf{\com{C}_2}} \sqsubseteq
        \conc{(p \cdot \mf{\com{C}'} + (1-p) \cdot \mf{\com{C}''})}{\mf{\com{C}_2}}$.
        By Lemma~\ref{lem:sound-aux-conc-2} and Definition~\ref{def:den-sem2},
        $\mf{\conc{\com{C}_1}{\com{C}_2}} \sqsubseteq
        p \cdot \mf{\conc{\com{C}'}{\com{C}_2}} + (1-p) \cdot \mf{\conc{\com{C}''}{\com{C}_2}}$.
        
      \item $\tau \in \init{\mf{\com{C}_2}}$

        By i.h., $\exists \com{C}', \com{C}''$ such that
        $\com{C}_2 \rightarrow p \cdot (\tau, \com{C}') + (1-p) \cdot (\tau, \com{C}'')$ with
        $p = v(\set{\tau, e'})$ and $e' \in \init{\mf{\com{C}'}}$, and
        $\mf{\com{C}_2} \sqsubseteq p \cdot \mf{\com{C}'} + (1-p) \cdot \mf{\com{C}''}$.  By the rules
        in Figure~\ref{fig:op-small2}, we have
        $\conc{\com{C}_1}{\com{C}_2} \rightarrow
        p \cdot (\tau, \conc{\com{C}_1}{\com{C}'}) + (1-p) \cdot (\tau, \conc{\com{C}_1}{\com{C}''})$.
        By Lemma~\ref{lem:conc-mono2},
        $\conc{\mf{\com{C}_1}}{\mf{\com{C}_2}} \sqsubseteq
        \conc{\com{C}_1}{(p \cdot \mf{\com{C}'} + (1-p) \cdot \mf{\com{C}''})}$.
        By Lemma~\ref{lem:sound-aux-conc-2} and Definition~\ref{def:den-sem2},
        $\mf{\conc{\com{C}_1}{\com{C}_2}} \sqsubseteq
        p \cdot \mf{\conc{\com{C}_1}{\com{C}'}} + (1-p) \cdot \mf{\conc{\com{C}_1}{\com{C}''}}$.
      \end{enumerate}
    \end{enumerate}
  \end{itemize}
\end{proof}

\subsubsection*{Proof of Theorem~\ref{res:adII-2}}
\begin{proof}
  Induction over the size of $\omega_{x_0}$.
  \begin{itemize}
  \item $|\omega_{x_0}| = 1$

    We have $\{ l \} \in \confmax{\mf{\com{C}}}$.  It follows directly that
    $\com{C} \twoheadrightarrow 1 \cdot (l, \checkmark)$ and $v(\{l\}) = 1$.
    
  \item $|\omega_{x_0}| > 1$

    We have $x_0 \in \confmax{\mf{\com{C}}}$.
    We know that $\omega_{x_0} = l_0 l_1 \dots l_n$.
    Hence $\emptyset \cchain{l_0} \{l_0\} \cchain{\omega_{x_0'}} \{l_0\} \cup x_0'$.
    We then have $l_0 \in \init{\mf{\com{C}}}$.
    By Lemma~\ref{res:adI-2} we have two cases:
    \begin{enumerate}
    \item $l_0 \neq \tau$

      Hence $\com{C} \rightarrow 1 \cdot (l_0, \com{C}')$ and
      $\mf{\com{C}} \backslash l_0 \equiv \mf{\com{C}'}$.  By Definition~\ref{def:rem-init2},
      $x_0 \backslash l_0 \in \confmax{\mf{\com{C}'}}$ such that
      $\emptyset \cchain{\omega_{x_0'}} x_0 \backslash l_0$.  By i.h.,
      $\com{C}' \twoheadrightarrow v(x_0 \backslash l_0) (\omega_{x_0'}, \checkmark) + \sum_k p_k
      (\omega_k, \com{C}_k)$, for some $p_k, \omega_k, \com{C}_k$.  By Figure~\ref{fig:op-nstep2},
      $\com{C} \twoheadrightarrow v(x_0 \backslash l_0) (l_0 : \omega_{x_0'}, \checkmark) + \sum_k
      p_k (l_0 : \omega_k, \com{C}_k)$, for some $p_k, \omega_k, \com{C}_k$.

    \item $l_0 = \tau$

      Hence $\com{C} \rightarrow p \cdot (\tau, \com{C}') + (1-p) \cdot (\tau, \com{C}'')$ and
      $\mf{\com{C}} \sqsubseteq p \mf{\com{C}'} + (1-p) \mf{\com{C}''}$.
      Now we have two sub-cases:
      \begin{enumerate}
      \item Case $\com{C} = \probC{\com{C}'}{p}{\com{C}''}$

        We then have $\mf{\com{C}} = p \mf{\com{C}'} + (1-p) \mf{\com{C}''}$, by
        Definition~\ref{def:den-sem2}, and consequently
        $x_0 \in \confmax{p \mf{\com{C}'} + (1-p) \mf{\com{C}''}}$.  By
        Definition~\ref{def:pes-prob2}, $x_0 \backslash \tau \in \confmax{\mf{\com{C}'}}$ or
        $x_0 \backslash \tau \in \confmax{\mf{\com{C}''}}$.  We only consider the former, since the
        latter has a similar reasoning.  By i.h.,
        $\com{C}' \twoheadrightarrow v(x_1) (w_{x_1}, \checkmark) + \sum_n p_n (\omega_n,
        \com{C}_n)$, for some $p_n, \omega_n, \com{C}_n$.  By Lemma~\ref{lem:prog-2},
        $\com{C}'' \twoheadrightarrow \sum_m p_m (\omega_m, \com{C}_m)$.
        By Figure~\ref{fig:op-nstep2},
        \begin{align*}
          \com{C} \twoheadrightarrow
          & p \cdot \left(
            v(x_1) (\tau : w_{x_1}, \checkmark) + \sum_n p_n (\tau : \omega_n, \com{C}_n)
            \right) \\
          & + (1-p) \cdot \sum_m p_m (\tau : \omega_m, \com{C}_m)
        \end{align*}
        
      \item Case $\com{C} = \seq{\com{C}_1}{\com{C}_2}$ or $\com{C} = \conc{\com{C}_1}{\com{C}_2}$

        By Lemma~\ref{lem:sound-aux-seq-conc-2},
        $x_0 \in \confmax{p \mf{\com{C}'} + (1-p) \mf{\com{C}''}}$.  By
        Definition~\ref{def:pes-prob2}, $x_0 \backslash \tau \in \confmax{\mf{\com{C}'}}$ or
        $x_0 \backslash \tau \in \confmax{\mf{\com{C}''}}$.  We only consider the former, since the
        latter has a similar reasoning.  By i.h.,
        $\com{C}' \twoheadrightarrow v(x_1) (w_{x_1}, \checkmark) + \sum_n p_n (\omega_n,
        \com{C}_n)$, for some $p_n, \omega_n, \com{C}_n$.  By Lemma~\ref{lem:prog-2},
        $\com{C}'' \twoheadrightarrow \sum_m p_m (\omega_m, \com{C}_m)$.
        By Figure~\ref{fig:op-nstep2},
        \begin{align*}
          \com{C} \twoheadrightarrow
          & p \cdot \left(
            v(x_1) (\tau : w_{x_1}, \checkmark) + \sum_n p_n (\tau : \omega_n, \com{C}_n)
            \right) \\
          & + (1-p) \cdot \sum_m p_m (\tau : \omega_m, \com{C}_m)
        \end{align*}
      \end{enumerate}
    \end{enumerate}
  \end{itemize}
\end{proof}

\subsubsection*{Proof of Lemma~\ref{lem:po-2}}
\begin{proof}
  Due to Lemma~\ref{lem:po-1} we only need to check the condition of the valuations.  Consider
  $\es{P}_1 = (\es{E}_1, v_1)$, $\es{P}_2 = (\es{E}_2, v_2)$, and $\es{P}_3 = (\es{E}_3, v_3)$ to be
  probabilistic event structures.
  \begin{itemize}
  \item Reflexivity: $\es{P}_1 = \es{P}_1$

    We want to show that $\forall x \in \confES{\es{P}_1}\, .\, v_1(x) = v_1(x)$.
    It holds straightforwardly.
    
  \item Transitivity: $\es{P}_1 \trianglelefteq \es{P}_2,\ \es{P}_2 \trianglelefteq \es{P}_3 \Rightarrow \es{P}_1 \trianglelefteq \es{P}_3$

    We want to show $\forall x \in \confES{\es{P}_1}\, .\, v_1(x) = v_3(x)$.
    From $\es{P}_1 \trianglelefteq \es{P}_2$, $\forall x \in \confES{\es{P}_1}\, .\, v_1(x) = v_2(x)$.
    From $\es{P}_2 \trianglelefteq \es{P}_3$, $\forall x \in \confES{\es{P}_2}\, .\, v_2(x) = v_3(x)$.
    Hence, $\forall x \in \confES{\es{P}_1}\, .\, v_1(x) = v_3(x)$.
    
  \item Antisymmetry: $\es{P}_1 \trianglelefteq \es{P}_2,\ \es{P}_2 \trianglelefteq \es{P}_1 \Rightarrow \es{P}_1 = \es{P}_2$

    We want to show $\forall x \in \confES{\es{P}_1}, \confES{\es{P}_2}\, .\, v_1(x) = v_2(x)$.
    From $\es{P}_1 \trianglelefteq \es{P}_2$, $\forall x \in \confES{\es{P}_1}\, .\, v_1(x) = v_2(x)$.
    From $\es{P}_2 \trianglelefteq \es{P}_1$, $\forall x \in \confES{\es{P}_2}\, .\, v_2(x) = v_1(x)$.
    Hence, $\forall x \in \confES{\es{P}_1}, \confES{\es{P}_2}\, .\, v_1(x) = v_2(x)$.
  \end{itemize}
\end{proof}

\subsubsection*{Proof of Lemma~\ref{lem:po-least-elem-2}}
\begin{proof}
  We first show that $\bot$ is a probabilistic event structure. From Lemma~\ref{lem:po-least-elem-1}
  $\bot = \pes{\emptyset}{\emptyset}{\emptyset}$ is an event structure. It lacks to see the
  conditions on the valuations.  It follows directly the definition that $v(\emptyset) =
  1$. Furthermore the only configuration in $\confES{\bot}$ is $\emptyset$. Hence we trivially have
  that $v(\emptyset) \geq 0$.

  To show that $\bot$ is the least element, consider any probabilistic event structure $\es{P}$. We
  need to show that $\bot \trianglelefteq \es{P}$. Due to Lemma~\ref{lem:po-least-elem-1} we focus
  solely on the valuations. Since the empty configuration is the only one in $\confES{\emptyset}$
  and since $\es{P}$ is a probabilistic event structure it holds that
  $v_{\bot}(\emptyset) = 1 = v(\emptyset)$.
\end{proof}

\subsubsection*{Proof of Lemma~\ref{lem:lub-es-2}}
\begin{proof}
  Due to Lemma~\ref{lem:lub-es-1} we focus only on the valuation part, where we have two conditions to verify:
  \begin{itemize}
  \item $v^\omega(\emptyset) = 1$

    From Definition~\ref{def:lub-2} we know that
    $\exists n \in \omega\, .\, v^\omega(\emptyset) = v_n(\emptyset) = 1$.

  \item $\forall y, x_1, \dots, x_m \in \confES{\es{P}^\omega}$ such that
    $y \subseteq x_1, \dots, x_m$,
    $\sum_{I \subseteq \set{1, \dots, m}} (-1)^{|I|} v^\omega(y \cup \cup_{i \in I} x_i) \geq 0$

    Following~\cite[Propostion 5]{winskel14} we only need to focus on $y \cchain{} x_1, \dots, x_m$.
    From Definition~\ref{def:lub-2} we know it $\exists n \in \omega\, .\, v^\omega(y) = v_n(y)$. We
    then have three cases, depending if the events are in $E_n$, in $E_{n+1}$, or in both.
    \begin{enumerate}
    \item the events are in $E_n$

      We know that
      $\sum_{\emptyset \neq I \subseteq \set{1, \dots, m}} (-1)^{|I+1|} v^\omega(\cup_{i \in I} x_i)
      = \sum_{\emptyset \neq I \subseteq \set{1, \dots, m}} (-1)^{|I+1|} v_n(\cup_{i \in I} x_i)$
      and consequently
      $v_n(y) - \sum_{\emptyset \neq I \subseteq \set{1, \dots, m}} (-1)^{|I+1|} v_n(\cup_{i\in I}
      x_i) \geq 0$, since $\es{P}_n$ is a probabilistic event structure
      
    \item the events are in $E_{n+1}$

      We know that $v_n(y) = v_{n+1}(y)$ since $\es{P}_n \trianglelefteq \es{P}_{n+1}$. Furthermore
      \[
        \sum_{\emptyset \neq I \subseteq \set{1, \dots, m}} (-1)^{|I+1|} v^\omega(\cup_{i \in I}
        x_i) = \sum_{\emptyset \neq I \subseteq \set{1, \dots, m}} (-1)^{|I+1|} v_{n+1}(\cup_{i \in
          I} x_i)
      \]
      and consequently
      \[
        v_{n+1}(y) - \sum_{\emptyset \neq I \subseteq \set{1, \dots, m}} (-1)^{|I+1|}
        v_{n+1}(\cup_{i\in I} x_i) \geq 0
      \]
      since $\es{P}_{n+1}$ is a probabilistic event structure
      
    \item the events are in both

      Since $\es{P}_n \trianglelefteq \es{P}_{n+1}$ we know that
      $\forall x \in \confES{\es{P}_n}\, .\, v_n(x) = v_{n+1}(x)$, which leads us to the previous
      case.
    \end{enumerate}
  \end{itemize}
\end{proof}

\subsubsection*{Proof of Lemma~\ref{lem:lub-2}}
\begin{proof}
  Due to Lemma~\ref{lem:lub-1} we focus only on the valuations.
  \begin{itemize}
  \item $\es{P}^\omega$ is an upper bound

    $\forall n \in \omega$ we need to have $\es{P}_n \trianglelefteq \es{P}^\omega$.  It follows
    directly from Definition~\ref{def:pes-fix-order-2} that $\forall n \in \omega$ we have
    $\es{P}_n \trianglelefteq \es{P}^\omega$ since by Definition~\ref{def:lub-2},
    $\forall x \in \confES{\es{P}^\omega},\, \exists n \in \omega\, .\, v^\omega(x) = v_n(x)$.
    
  \item $\es{P}^\omega$ is the least upper bound

    Let $\es{P} = (\es{E}, v)$ be an upper bound of the chain.  We need to show that if
    $\es{P}_n \trianglelefteq \es{P}^\omega$ and $\es{P}_n \trianglelefteq \es{P}$ then
    $\es{P}^\omega \trianglelefteq \es{P}$.  From $\es{P}_n \trianglelefteq \es{P}^\omega$,
    $\forall x \in \confES{\es{P}_n}\, .\, v_n(x) = v^\omega(x)$.  From Definition~\ref{def:lub-2},
    $\forall x \in \confES{\es{P}^\omega},\ \exists n \in \omega\, .\, v^\omega(x) = v_n(x)$.  From
    $\es{P}_n \trianglelefteq \es{P}$, $\forall x \in \confES{\es{P}_n}\, .\, v_n(x) = v(x)$.  Thus
    $\forall x \in \confES{\es{P}^\omega}, \exists n \in \omega\, .\, v^\omega(x) = v_n(x) = v(x)$.
  \end{itemize}
\end{proof}

\subsubsection*{Proof of Lemma~\ref{lem:seq-fix-mono-2}}
\begin{proof}
  Let
  \begin{align*}
    & \es{P} = (\es{E}, v) \\
    & \es{P}_1 = (\es{E}_1, v_1) \\
    & \es{P}_2 = (\es{E}_2, v_2) \\
    & \seq{\es{P}}{\es{P}_1} = (\es{E}^1, v^1) \\
    & \seq{\es{P}}{\es{P}_2} = (\es{E}^2, v^2)
  \end{align*}
  We want to show $\forall x \in \confES{\seq{\es{P}}{\es{P}_1}}\, .\, v^1(x) = v^2(x)$.  Due to
  Lemma~\ref{lem:seq-fix-mono-1} we only focus on the valuations.
  According to Definition~\ref{def:pes-seq2} we have two cases:
  \begin{enumerate}
  \item $x \in \confES{\seq{\es{P}}{\es{P}_1}}$ such that $x \in \confES{\es{P}}$

    Then we are done because $v^1(x) = v(x) = v^2(x)$.
    
  \item $x \in \confES{\seq{\es{P}}{\es{P}_1}}$ such that
    $\exists y \in \confmax{\es{P}},\, y' \in \confES{\es{P}_1}\, .\, x = y \cup (y' \times \{y\})$

    Then we have $v^1(x) = v(y) \cdot v_1(y')$.  Since $\es{P}_1 \trianglelefteq \es{P}_2$,
    $\forall y' \in \confES{\es{P}_1}\, .\, v_1(y') = v_2(y')$.  Then
    $v(y) \cdot v_1(y') = v(y) \cdot v_2(y') = v^2(x)$.  Hence $v^1(x) = v^2(x)$.
  \end{enumerate}
\end{proof}

\subsubsection*{Proof of Lemma~\ref{lem:conc-fix-mono-2}}
\begin{proof}
  Let
  \begin{align*}
    & \es{P}_1 = (\es{E}_1, v_1) \\
    & \es{P}_2 = (\es{E}_2, v_2) \\
    & \es{P}'_1 = (\es{E}'_1, v'_1) \\
    & \es{P}'_2 = (\es{E}'_2, v'_2) \\
    & \conc{\es{P}_1}{\es{P}_2} = (\es{E}, v) \\
    & \conc{\es{P}'_1}{\es{P}'_2} = (\es{E}', v')
  \end{align*}
  We want to show $\forall x \in \confES{\conc{\es{P}_1}{\es{P}_2}}\, .\, v(x) = v'(x)$.
  Due to Lemma~\ref{lem:conc-fix-mono-1} we only focus on the valuations.  

  Let $x \in \confES{\conc{\es{P}_1}{\es{P}_2}}\, .\, v(x) = v_1(x \cap E_1) \cdot v_2(x \cap E_2)$,
  such that $x_1 = x \cap E_1$ and $x_2 = x \cap E_2$.  Since $\es{P}_1 \trianglelefteq \es{P}'_1$
  and $\es{P}_2 \trianglelefteq \es{P}'_2$ then
  $\forall x_1 \in \confES{\es{P}_1}\, .\, v_1(x_1) = v'_1(x_1)$ and
  $\forall x_2 \in \confES{\es{P}_2}\, .\, v_2(x_2) = v'_1(x_2)$, respectively.  Hence
  $v(x) = v_1(x_1) \cdot v_2(x_2) = v'_1(x_1) \cdot v'_2(x_2) = v'(x)$.
\end{proof}

\subsubsection*{Proof of Lemma~\ref{lem:probc-fix-mono-2}}
\begin{proof}
  Let
  \begin{align*}
    & \es{P}_1 = \pespq{\es{E}_1}{v_1} \\
    & \es{P}'_1 = \pespq{\es{E}'_1}{v'_1} \\
    & \es{P}_2 = \pespq{\es{E}_2}{v_2} \\
    & \es{P}'_2 = \pespq{\es{E}'_2}{v'_2} \\
    & \probC{\es{P}_1}{p}{\es{P}_2} = \pespq{\es{E}}{v} \\
    & \probC{\es{P}'_1}{p}{\es{P}'_2} = \pespq{\es{E}'}{v'}
  \end{align*}

  The conditions to check are:
  \begin{enumerate}
  \item $E \subseteq E'$     
  \item $\forall e, e'\, .\, e \leq e' \Leftrightarrow e,e' \in E \wedge e \leq' e'$
  \item $\forall e, e'\, .\, e \# e' \Leftrightarrow e,e' \in E \wedge e \#' e'$
  \item $\forall x \in \confES{\probC{\es{P}_1}{p}{\es{P}_2}}\, .\, v(x) \geq v'(y)$
  \end{enumerate}

  The first three conditions follow directly from Definition~\ref{def:pes-prob2}.  Hence we focus on
  the last one.

  Let $x \in \confES{\probC{\es{P}_1}{p}{\es{P}_2}}$.  We have two cases:
  \begin{enumerate}
  \item $x \backslash \tau \in \confES{\es{P}_1}$

    It follows directly that $v(x) = v'(x)$, since $\es{P}_1 \trianglelefteq \es{P}'_1$ and
    $v(x) = p \cdot v_1(x \backslash \tau) = p \cdot v'_1(x \backslash \tau) = v'(x)$.
    
  \item $x \backslash \tau \in \confES{\es{P}_2}$

    It follows directly that $v(x) = v'(x)$, since $\es{P}_2 \trianglelefteq \es{P}'_2$ and
    $v(x) = (1-p) \cdot v_2(x \backslash \tau) = (1-p) \cdot v'_2(y \backslash \tau) = v'(x)$.
  \end{enumerate}  
\end{proof}

\subsubsection*{Proof of Lemma~\ref{lem:seq-cont-2}}
\begin{proof}
  By Lemma~\ref{lem:seq-fix-mono-2}, the sequential composition is monotone at right.  Furthermore,
  showing that each event of $\seq{\es{P}}{\bigsqcup_m \es{P}_m}$ is an event of
  $\bigsqcup_m(\seq{\es{P}}{\es{P}_m})$ is already done in Lemma~\ref{lem:seq-cont-1}.  By
  Lemma~\ref{lem:cont-1} we are done.
\end{proof}

\subsubsection*{Proof of Lemma~\ref{lem:conc-cont-2}}
\begin{proof}
  By Lemma~\ref{lem:conc-fix-mono-2}, the parallel composition is monotone at right.
  Furthermore, showing that each event of $\conc{\bigsqcup_n \es{P}_n}{\bigsqcup_m \es{P}_m}$ is an
  event of $\bigsqcup_{n,m}(\conc{\es{P}_n}{\es{P}_m})$ is already done in
  Lemma~\ref{lem:conc-cont-1}.  By Lemma~\ref{lem:cont-1} we are done.
\end{proof}

\subsubsection*{Proof of Lemma~\ref{lem:probc-cont-2}}
\begin{proof}
  By Lemma~\ref{lem:probc-fix-mono-2} we know that the probabilistic choice is monotone.  It lacks
  to show that each event of $\probC{\bigsqcup_n \es{P}_n}{p}{\bigsqcup_m \es{P}_m}$ is an event of
  $\bigsqcup_{n,m}(\probC{\es{P}_n}{p}{\es{P}_m})$.

  Let $\es{P}_1 \trianglelefteq \dots \trianglelefteq \es{P}_n \trianglelefteq \dots$ and
  $\es{P}'_1 \trianglelefteq \dots \trianglelefteq \es{P}'_m \trianglelefteq \dots$ be
  $\omega$-chains with least upper bound $\bigsqcup_n \es{P}_n$ and $\bigsqcup_m \es{P}_m$,
  respectively.  Let $e$ be an event of $\probC{\bigsqcup_n \es{P}_n}{p}{\bigsqcup_m \es{P}_m}$.
  By Definition~\ref{def:pes-prob2} we have three cases:
  \begin{enumerate}
  \item $e = \tau$

    It follows directly from Definition~\ref{def:pes-prob2} that $\tau$ is an event of
    $\probC{\es{P}_n}{p}{\es{P}_m}$.  Consequently it is an event of
    $\bigsqcup_{n,m}(\probC{\es{P}_n}{p}{\es{P}_m})$.
    
  \item $e$ is an event of $\bigsqcup_n \es{P}_n$

    By Definition~\ref{def:lub-2}, $\exists n \in \omega\, .\, e$ is an event of $\es{P}_n$. By
    Definition~\ref{def:pes-prob2}, $e$ is an event of $\probC{\es{P}_n}{p}{\es{P}_m}$ and
    consequently it is an event of $\bigsqcup_{n,m}(\probC{\es{P}_n}{p}{\es{P}_m})$.
    
  \item $e$ is an event of $\bigsqcup_m \es{P}_m$

    Similar to the previous point.
  \end{enumerate}
  By Lemma~\ref{lem:cont-1} we are done.
\end{proof}

\subsubsection*{Proof of Lemma~\ref{lem:gamma-cont-2}}
\begin{proof}
  We only do for the probabilistic choice, since for the remaining cases the prove is the same as in
  Lemma~\ref{lem:gamma-cont-1}.
  \begin{align*}
    & \Gamma^{\probC{\com{C}_1}{p}{\com{C}_2}, \gamma}(\bigsqcup_n \es{P}_n) \\
    =& \set{ \text{Definition~\ref{def:fix-den-sem-2}}} \\
    & \mf{\probC{\com{C}_1}{p}{\com{C}_2}}_{\gamma(X \leftarrow \bigsqcup_n \es{P}_n)} \\
    =& \set{ \text{Definition~\ref{def:fix-den-sem-2}}} \\
    & \probC{\mf{\com{C}_1}_{\gamma(X \leftarrow \bigsqcup_n \es{P}_n)}}{p}{\mf{\com{C}_2}_{\gamma(X \leftarrow \bigsqcup_n \es{P}_n)}} \\
    =& \set{ \text{Definition~\ref{def:fix-den-sem-2}}} \\
    & \probC{\bigsqcup_n \Gamma^{\com{C}_1,\gamma}(\es{P}_n)}{p}{\bigsqcup_n \Gamma^{\com{C}_2,\gamma}(\es{P}_n)} \\
    =& \set{\text{Lemma~\ref{lem:probc-cont-2}}} \\
    &= \bigsqcup_n (\probC{\Gamma^{\com{C}_1,\gamma}(\es{P}_n)}{p}{\Gamma^{\com{C}_2,\gamma}(\es{P}_n)}) \\
    =& \set{ \text{Definition~\ref{def:fix-den-sem-2}}} \\
    &= \bigsqcup_n (\probC{\mf{\com{C}_1}_{\gamma(X \leftarrow \es{P}_n)}}{p}{\mf{\com{C}_2}_{\gamma(X \leftarrow \es{P}_n)}})\\
    =& \set{ \text{Definition~\ref{def:fix-den-sem-2}}} \\
    & \bigsqcup_n \mf{\probC{\com{C}_1}{p}{\com{C}_2}}_{\gamma(X \leftarrow \es{P}_n)}
  \end{align*}
\end{proof}

\subsubsection*{Proof of Lemma~\ref{lem:1-2}}
\begin{proof}
  We only show the probabilistic choice, since the proof for the other cases is in
  Lemma~\ref{lem:1-1}.
  \begin{align*}
    & \mf{(\probC{\com{C}_1}{p}{\com{C}_2})[X \leftarrow \mf{\rec{X}{\com{C}}}_\gamma]}_\gamma \\
    =& \set{\text{Adaptation of Definition~\ref{fig:substitution-1}}} \\
    & \mf{\probC{\com{C_1}[X \leftarrow \mf{\rec{X}{\com{C}}}_\gamma]}{p}{\com{C}_2[X \leftarrow \mf{\rec{X}{\com{C}}}_\gamma]}}_\gamma \\
    =& \set{\text{Definition~\ref{def:fix-den-sem-2}}} \\
    & \probC{\mf{\com{C}_1}[X \leftarrow \mf{\rec{X}{\com{C}}}_\gamma]}{p}{\mf{\com{C}_2}_\gamma[X \leftarrow \mf{\rec{X}{\com{C}}}_\gamma]} \\
    =& \set{\text{i.h.}} \\
    & \probC{\mf{\com{C}_1}_{\gamma(X \leftarrow \mf{\rec{X}{\com{C}}}_\gamma)}}{p}{
      \mf{\com{C}_2}_{\gamma(X \leftarrow \mf{\rec{X}{\com{C}}}_\gamma)}
      } \\
    =& \set{\text{Definition~\ref{def:fix-den-sem-2}}} \\
    & \mf{\probC{\com{C}_1}{p}{\com{C}_2}}_{\gamma(X \leftarrow \mf{\rec{X}{\com{C}}}_\gamma)}
  \end{align*}
\end{proof}

\subsubsection*{Proof of Lemma~\ref{lem:fix-conf-max-2}}
\begin{proof}
  \begin{align*}
    & \rec{X}{\com{C}} \rightarrow \sum_i p_i \cdot (\tau, \com{C}_i[X \leftarrow \rec{X}{\com{C}}])\\
    \Rightarrow & \set{\text{rules in Figure~\ref{fig:op-small2}}} \\
    & \com{C} \rightarrow \sum_i p_i (\tau, \com{C}_i) \\
    \Rightarrow & \set{\text{i.h.}} \\
    & x \in \confmax{\mf{C}_\gamma} \text{ and }
      x \in \confmax{\sum_i p_i \cdot \mf{\com{C}_i}_\gamma} \text{ s.t. }
      \exists \mf{\com{C}_i}_\gamma\, .\, v_i(x) = v(x) \\
    \Rightarrow & \set{\gamma = \gamma(X \leftarrow \mf{\rec{X}{\com{C}})}_\gamma} \\
    & x \in \confmax{\mf{C}_{\gamma(X \leftarrow \mf{\rec{X}{\com{C}}}_\gamma)}} \text{ and }
      x \in \confmax{\sum_i p_i \cdot \mf{\com{C}_i}_{\gamma(X \leftarrow \mf{\rec{X}{\com{C}}}_\gamma)}} \\
    \Rightarrow & \set{\text{Lemma~\ref{lem:2-1} and Lemma~\ref{lem:1-2}}} \\
    & x \in \confmax{\rec{X}{\com{C}}}_\gamma \text{ and }
      x \in \confmax{\sum_i p_i \mf{\com{C}_i[X \leftarrow \mf{\rec{X}{\com{C}}}_\gamma]}_\gamma}
  \end{align*}
\end{proof}

\subsubsection*{Proof of Lemma~\ref{res:soundI-fix-2}}
\begin{proof}
  \begin{itemize}
  \item
    \begin{align*}
      & \rec{X}{\com{C}} \rightarrow 1 \cdot (l, \com{C}'[x \leftarrow \rec{X}{\com{C}}]) \\
      \Rightarrow & \set{\text{Figure~\ref{fig:op-small2}}} \\
      & \com{C} \rightarrow 1 \cdot (l, \com{C}') \\
      \Rightarrow & \set{\text{i.h.}} \\
      & \mf{\com{C}}_\gamma \backslash l \equiv \mf{\com{C}'}_\gamma \\
      \Rightarrow & \set{\gamma =  \gamma(X \leftarrow \mf{\rec{X}{\com{C}}}_\gamma)} \\
      & \mf{\com{C}}_{\gamma(X \leftarrow \mf{\rec{X}{\com{C}}}_\gamma)} \backslash l
        \equiv \mf{\com{C}'}_{\gamma(X \leftarrow \mf{\rec{X}{\com{C}}}_\gamma)} \\
      \Rightarrow & \set{\text{Lemma~\ref{lem:2-1}, Lemma~\ref{lem:1-2}}} \\
      & \mf{\rec{X}{\com{C}}}_\gamma \backslash l \equiv \mf{\com{C}'[X \leftarrow \mf{\rec{X}{\com{C}}}_\gamma]}_\gamma
    \end{align*}

  \item
    \begin{align*}
      & \rec{X}{\com{C}} \rightarrow \sum_i p_i \cdot (\tau, \com{C}_i[x \leftarrow \rec{X}{\com{C}}]) \\
      \Rightarrow & \set{\text{Figure~\ref{fig:op-small2}}} \\
      & \com{C} \rightarrow \sum_i p_i \cdot (\tau, \com{C}_i) \\
      \Rightarrow & \set{\text{i.h.}} \\
      & \mf{\com{C}}_\gamma \sqsubseteq \sum_i p_i \mf{\com{C}_i}_\gamma \\
      \Rightarrow & \set{\gamma =  \gamma(X \leftarrow \mf{\rec{X}{\com{C}}}_\gamma)} \\
      & \mf{\com{C}}_{\gamma(X \leftarrow \mf{\rec{X}{\com{C}}}_\gamma)}
        \equiv \sum_I p_i \cdot \mf{\com{C}_i}_{\gamma(X \leftarrow \mf{\rec{X}{\com{C}}}_\gamma)} \\
      \Rightarrow & \set{\text{Lemma~\ref{lem:2-1}, Lemma~\ref{lem:1-2}}} \\
      & \mf{\rec{X}{\com{C}}}_\gamma \equiv
        \sum_i p_i \mf{\com{C}_i[X \leftarrow \mf{\rec{X}{\com{C}}}_\gamma]}_\gamma
    \end{align*}    
  \end{itemize}
\end{proof}

\subsubsection*{Proof of Theorem~\ref{res:soundII-fix-2}}
\begin{proof}
  We only need to add the following sub-case when the size of the word is bigger than one and the
  transition is made by $\tau$.
  \begin{itemize}
  \item Case $\com{C} = \rec{X}{\com{D}}$

    By Remark~\ref{rem:probc-n} and Lemma~\ref{lem:fix-conf-max-2}
    \begin{align*}
      &\forall i, 
        \exists \set{\tau} \cup x_{i0} \in \confmax{\mf{\rec{X}{\com{D}}}}\, .\,
        \emptyset \cchain{\tau} \set{\tau}
        \cchain{\omega_{i0}'} \set{\tau} \cup x_{i0}',\
        v_i(\set{\tau} \cup x_{i0}') = p_i \cdot p_0' = p_0
    \end{align*}  
  \end{itemize}
\end{proof}

\subsubsection*{Proof of Lemma~\ref{res:adI-fix-2}}
\begin{proof}
  \begin{itemize}
  \item $l' \neq \tau \in \init{\mf{\rec{X}{\com{C}}}_\gamma}$

    By Definition~\ref{def:fix-den-sem-2} and Definition~\ref{def:lub-2},
    $l' \in \init{\mf{\com{C}}_\gamma}$.  By i.h., $\exists \com{C}'$ such that
    $\com{C} \rightarrow 1 \cdot (l', \com{C}')$ and
    $\mf{\com{C}}_\gamma \backslash l' \equiv \mf{\com{C}'}_\gamma$.  By Figure~\ref{fig:op-small2}
    and by letting $\gamma =  \gamma(X \leftarrow \mf{\rec{X}{\com{C}}}_\gamma)$ and Lemma~\ref{lem:2-1} and
    Lemma~\ref{lem:1-2},
    $\rec{X}{\com{C}} \rightarrow 1 \cdot (l', \com{C}'[X \leftarrow \rec{X}{\com{C}}])$ and
    $\mf{\rec{X}{\com{C}}}_\gamma \backslash l' \equiv \mf{\com{C}'[X \leftarrow
      \mf{\rec{X}{\com{C}}}_\gamma]}_\gamma$.
    
  \item $l' = \tau \in \init{\mf{\rec{X}{\com{C}}}_\gamma}$

    By Definition~\ref{def:fix-den-sem-2} and Definition~\ref{def:lub-2},
    $l' \in \init{\mf{\com{C}}_\gamma}$. By i.h. $\exists \com{C}', \com{C}''$,
    $\exists e \in \init{\mf{\com{C}'}_\gamma}$ such that
    $\com{C} \rightarrow p \cdot (\tau, \com{C}') + (1-p) \cdot (\tau, \com{C}'')$ with
    $p = v(\set{\tau, e})$ and
    $\mf{\com{C}}_\gamma \sqsubseteq p \cdot \mf{\com{C}'}_\gamma + (1-p) \cdot
    \mf{\com{C}''}_\gamma$.  By Figure~\ref{fig:op-small2} and by letting
    $\gamma = \gamma(X \leftarrow \mf{\rec{X}{\com{C}}}_\gamma)$ and Lemma~\ref{lem:2-1} and
    Lemma~\ref{lem:1-2},
    $\rec{X}{\com{C}} \rightarrow p \cdot (\tau, \com{C}'[X \leftarrow \rec{X}{\com{C}}]) + (1-p)
    \cdot (\tau, \com{C}''[X \leftarrow \rec{X}{\com{C}}])$ and
    $\mf{\rec{X}{\com{C}}}_\gamma \sqsubseteq p \cdot \mf{\com{C}'[X \leftarrow
      \mf{\rec{X}{\com{C}}}_\gamma]}_\gamma + (1-p) \cdot \mf{\com{C}''[X \leftarrow
      \mf{\rec{X}{\com{C}}}_\gamma}_\gamma$.
  \end{itemize}
\end{proof}

\subsubsection*{Proof of Theorem~\ref{res:adII-fix-2}}
\begin{proof}
  We only need to add the following sub-case when the size of the word is bigger than one and the
  transition is made by $\tau$.
  \begin{itemize}
  \item Case $\com{C} = \rec{X}{\com{D}}$

    By Lemma~\ref{lem:fix-conf-max-2},
    $x_0 \in \confmax{p \mf{\com{C}'[X \leftarrow \rec{X}{\com{D}}]} + (1-p) \mf{\com{C}''[X
        \leftarrow \rec{X}{\com{D}}]}}$.  By Definition~\ref{def:pes-prob2},
    $x_0 \backslash \tau \in \confmax{\mf{\com{C}'[X \leftarrow \rec{X}{\com{D}}]}}$ or
    $x_0 \backslash \tau \in \confmax{\mf{\com{C}''[X \leftarrow \rec{X}{\com{D}}]}}$.  We only
    consider the former, since the latter has a similar reasoning.  By i.h.,
    $\com{C}'[X \leftarrow \rec{X}{\com{D}}] \twoheadrightarrow v(x_1) (w_{x_1}, \checkmark) +
    \sum_n p_n (\omega_n, \com{C}_n)$, for some $p_n, \omega_n, \com{C}_n$.  By
    Lemma~\ref{lem:prog-2},
    $\com{C}''[X \leftarrow \rec{X}{\com{D}}] \twoheadrightarrow \sum_m p_m (\omega_m, \com{C}_m)$.
    By Figure~\ref{fig:op-nstep2},
    \begin{align*}
      \com{C} \twoheadrightarrow
      p \cdot \left(
      v(x_1) (\tau : w_{x_1}, \checkmark) + \sum_n p_n (\tau : \omega_n, \com{C}_n)
      \right) +
      (1-p) \cdot \sum_m p_m (\tau : \omega_m, \com{C}_m)
    \end{align*}
  \end{itemize}
\end{proof}

\newpage
\section*{Proofs of Section~\ref{sec:qes}}\label{app-sec:proofs-unit-es}
\subsubsection*{Proof of Lemma~\ref{lem:seq-es3}}
\begin{proof}
  Let
  \begin{align*}
    & \es{U}_1 = \qpes{E_1}{\leq_1}{\#_1}{Q_1} \\
    & \es{U}_2 = \qpes{E_2}{\leq_2}{\#_2}{Q_2} \\
    & \seq{\es{U}_1}{\es{U}_2} = \qpes{E}{\leq}{\#}{Q}
  \end{align*}  

  Due to Lemma~\ref{lem:seq-es1} we only need show the conditions added in the definition of unitary
  event structures.
  \begin{enumerate}
  \item $\forall e, e' \in E,\, \econc{e}{e'} \Rightarrow [Q(e), Q(e')] = 0$

    Since $\econc{e}{e'}$ only if $e, e' \in E_1$ or $e, e' \in E_2 \times \confmax{\es{U}_1}$ we
    are done.
    
  \item $\minconflict$ is transitive

    It follows directly since $\minconflict$ only occurs between events of the same set of events.
    
  \item $\forall e \in E  \sum_{e' \in [e]} Q(e')$ is unitary
    
    We have two cases, since there is no minimal conflict between events in $E_1$ and
    $E_2 \times \confmax{\es{U}_1}$:
    \begin{enumerate}
    \item $\forall e \in E_1$

      Since $E_1$ is a unitary event structure, we are done.

    \item $\forall e \in E_2 \times \confmax{\es{U}_1}$

      Since $E_2$ is a unitary event structure, we are done.
    \end{enumerate}
  \end{enumerate}
\end{proof}

\subsubsection*{Proof of Lemma~\ref{lem:conc-es3}}
\begin{proof}
  Let
  \begin{align*}
    & \es{U}_1 = \qpes{E_1}{\leq_1}{\#_1}{Q_1} \\
    & \es{U}_2 = \qpes{E_2}{\leq_2}{\#_2}{Q_2} \\ 
    & \conc{\es{U}_1}{\es{U}_2} = \qpes{E}{\leq}{\#}{Q}
  \end{align*}

  Due to Lemma~\ref{lem:conc-es1} we only need show the conditions added in the definition of unitary
  event structures.
  \begin{enumerate}
  \item $\forall e, e' \in E,\, \econc{e}{e'} \Rightarrow [Q(e), Q(e')] = 0$

    We have two cases:
    \begin{enumerate}
    \item $e, e' \in E_1$ or $e, e' \in E_2$

      The condition trivially holds, since $\es{U}_1$ and $\es{U}_2$ are unitary event structures.
      
    \item $e \in E_1$ and $e' \in E_2$

      It follows directly from Definition~\ref{def:pes-conc3}.
    \end{enumerate}

  \item $\minconflict$ is transitive

    It follows directly since the parallel composition does not create new conflicts and that the
    conflict relation is inherited from $\es{U}_1$ and $\es{U}_2$ which are unitary event
    structures.

  \item $\forall e \in E, \sum_{e' \in [e]} Q(e')$ is unitary

    Since there is no minimal conflict between events in $E_1$ and $E_2$, it follows directly that
    if $\forall e \in E_1$ or $\forall e \in E_2$ the condition holds since $\es{U}_1$ and
    $\es{U}_2$ are unitary event structures.    
  \end{enumerate}
\end{proof}

\subsubsection*{Proof of Lemma~\ref{lem:pes-meas3}}
\begin{proof}
  Let
  \begin{align*}
    & \es{U}_1 = \qpes{E_1}{\leq_1}{\#_1}{Q_1} \\
    & \es{U}_2 = \qpes{E_2}{\leq_2}{\#_2}{Q_2} \\
    & \meas{n}{\es{U}_1}{\es{U}_2} = \qpes{E}{\leq}{\#}{Q}
  \end{align*}

    We need to show that $\leq$ is a partial order and that $\#$ is symmetric and irreflexive.
  \begin{itemize}
  \item $\leq$ is a partial order:
    \begin{itemize}
    \item Reflexivity ($e \leq e$): we have the following cases:
      \begin{enumerate}
      \item Case $e \in \{\tau_0^n, \tau_1^n\}$. It follows directly from
        Definition~\ref{def:pes-meas3} that either $\tau_0^n \leq \tau_0^n$ or
        $\tau_1^n \leq \tau_1^n$.
      \item Case $e_1 \leq_1 e'_1$ or $e_2 \leq_2 e'_2$. Since $\leq_1$ and $\leq_2$ are partial
        orders, we are done.
      \end{enumerate}
    \item Transitivity ($e \leq e'$ and $e' \leq e''$ then $e \leq e''$): we have the following cases:
      \begin{enumerate}
      \item Case ($e \leq_1 e'$ and $e' \leq_1 e''$) or ($e \leq_2 e'$ and $e' \leq_2 e''$). Since
        $\leq_1$ and $\leq_2$ are partial orders, we have $e \leq_1 e''$ or $e \leq_2 e''$. Hence
        $e \leq e''$.
      \item Case $e = \tau_0^n$ and $e', e'' \in E_1$. We have $\tau_0^n \leq e'$ and
        $e' \leq e''$. Since $e'' \in E_1$ then $\tau_n^0 \leq e''$.
      \item Case $e = \tau_1^n$ and $e', e'' \in E_2$. We have $\tau_1^n \leq e'$ and
        $e' \leq e''$. Since $e'' \in E_2$ then $\tau_n^1 \leq e''$.
      \end{enumerate}
    \item Antisymmetry ($e \leq e'$ and $e' \leq e$ then $e'=e$): we have two cases: ($e \leq_1 e'$
      and $e' \leq_1 e$) or ($e \leq_2 e'$ and $e' \leq_2 e$). Since $\leq_1$ and $\leq_2$ are
      partial orders, we have $e=e'$. Since there are no more cases we are done.
    \end{itemize}

    Hence $\leq$ is a partial order.

  \item $\#$ is symmetric and irreflexive:
    \begin{itemize}
    \item Symmetric (if $e \# e'$ then $e' \# e$): we have the following cases
      \begin{enumerate}
      \item Case $e \#_1 e'$ or $e \#_2 e'$. Since $\#_1$ and $\#_2$ are symmetric then we have
        $e' \#_1 e$ or $e' \#_2 e$. Thus $e' \# e$.
      \item Case ($e = \tau_0^n$ or $e \in E_1$) and ($e' = \tau_1^n$ or $e' \in E_2$). It follows
        directly from Definition~\ref{def:pes-meas3}.
      \item Case ($e = \tau_1^n$ or $e \in E_2$) and ($e' = \tau_0^n$ or $e' \in E_1$). It follows
        directly from Definition~\ref{def:pes-meas3}.
      \end{enumerate}
    \item Irreflexive ($\neg(e \# e)$): We have either $\neg(e \#_1 e)$ or $\neg(e \#_2 e)$. Since
      $\#_1$ and $\#_2$ are irreflexive then $\neg(e \# e)$. There are no more cases since either
      $e \in E_1$ or $e \in E_2$.
    \end{itemize}

    Hence $\#$ is symmetric and irreflexive.
  \end{itemize}
  
  We need to prove:
  \begin{enumerate}
  \item $\set{e' \mid e' \leq e}$ is finite

    We have four cases:
    \begin{enumerate}
    \item $e = \tau_0^n$

      It follows directly that $\set{e' \mid e' \leq \tau_0^n} = \set{\tau_0^n}$ since
      $\tau_0^n \in \init{\meas{n}{\es{U}_1}{\es{U}_2}}$.
      
    \item $e = \tau_1^n$

      It follows directly that $\set{e' \mid e' \leq \tau_1^n} = \set{\tau_1^n}$ since
      $\tau_1^n \in \init{\meas{n}{\es{U}_1}{\es{U}_2}}$.
      
    \item $e \in E_1$

      We have that $\set{e' \mid e' \leq e} = \set{\tau_0^n} \cup \set{e' \mid e' \leq_1 e}$.  Since
      $\es{U}_1$ is a unitary event structure, then we know that $\set{e' \mid e' \leq_1 e}$ is
      finite.  Hence $\set{\tau_0^n} \cup \set{e' \mid e' \leq_1 e}$ is finite.
      
    \item $e \in E_2$

      We have that $\set{e' \mid e' \leq e} = \set{\tau_1^n} \cup \set{e' \mid e' \leq_2 e}$.  Since
      $\es{U}_2$ is a unitary event structure, then we know that $\set{e' \mid e' \leq_2 e}$ is
      finite.  Hence $\set{\tau_1^n} \cup \set{e' \mid e' \leq_2 e}$ is finite.
    \end{enumerate}
    
  \item $e \# e' \leq e'' \Rightarrow e \# e''$
    It follows directly by Definition~\ref{def:pes-meas3} that $e \# e''$.

      

      
  \item $\econc{e}{e'} \Rightarrow [Q(e), Q(e')] = 0$

    The concurrent events are either in $\es{U}_1$ or in $\es{U}_2$, which are unitary event
    structures, hence the condition trivially holds.
    
  \item $\minconflict$ is transitive

    It follows directly since the conflict relation is inherited from $\es{U}_1, \es{U}_2$, which
    are unitary event structures, and from the fact that the new events, $\tau_0^n$ and $\tau_1^n$,
    are in minimal conflict between them, \ie\ $\tau_0^n \minconflict \tau_1^n$.

  \item $\forall e \in E, \sum_{e' \in [e]} Q(e')$ is unitary

    We have two cases (since if $e_1 \in E_1, e_2 \in E_2$ then $\neg(e_1 \minconflict e_2)$):
    \begin{enumerate}
    \item $e = \tau_0^n, e' = \tau_1^n$ or vice-versa

      It follows directly from Definition~\ref{def:pes-meas3} that $Q(\tau_0^n) + Q(\tau_1^n) = Id$,
      which is unitary.

    \item $\forall e \in E_1$ or $\forall e \in E_2$

      It follows directly from $\es{U}_1$ and $\es{U}_2$ being unitary events structures.
    \end{enumerate}
  \end{enumerate}
\end{proof}

\subsubsection*{Proof of Lemma~\ref{lem:rem-init-es3}}
\begin{proof}
  Let $\es{U} = \qpes{E}{\leq}{\#}{Q}$ and $\es{U} \backslash a = \qpes{E'}{\leq'}{\#'}{Q'}$.

  Due to Lemma~\ref{lem:rem-init-es1} we only need to check the conditions added in the definition
  of unitary event structures.
  \begin{enumerate}
  \item $\forall e, e' \in E',\, \econc{e}{e'} \Rightarrow [Q'(e), Q'(e')] = 0$

    It follows directly from Definition~\ref{def:rem-init3} that
    $[Q'(e), Q'(e')] = [Q|_{E'} (e), Q|_{E'} (e')] = 0$

  \item $\minconflict$ is transitive

    It follows directly since the conflict relation $\#'$ is the restriction of $\#$ to the events
    of $E'$.
    
  \item $\forall e \in E, \sum_{e' \in [e]} Q(e')$ is unitary

    It follows directly from Definition~\ref{def:rem-init3} that
    $\sum_{e' \in [e]} Q'(e') = \sum_{e' \in [e]} Q|_{E'}(e')$ which is unitary.

  \end{enumerate}
\end{proof}

\subsubsection*{Proof of Lemma~\ref{lem:seq-mono3}}
\begin{proof}
  Let
  \begin{align*}
    & \es{U}_1 = \qpes{E_1}{\leq_1}{\#_1}{Q_1} \\
    & \es{U}'_1 = \qpes{E'_1}{\leq'_1}{\#'_1}{Q'_1} \\
    & \es{U}_2 = \qpes{E_2}{\leq_2}{\#_2}{Q_2} \\
    & \es{U}'_2 = \qpes{E'_2}{\leq'_2}{\#'_2}{Q'_2} \\
    & \seq{\es{U}_1}{\es{U}_2} = \qpes{E}{\leq}{\#}{Q} \\
    & \seq{\es{U}'_1}{\es{U}'_2} = \qpes{E'}{\leq'}{\#'}{Q'}
  \end{align*}
  such that $\es{U}_1 \sqsubseteq \es{U}'_1$ and $\es{U}_2 \sqsubseteq \es{U}'_2$.

  Due to Lemma~\ref{lem:seq-mono1} we only need to show
  $\forall e \in E\ .\ Q(e) = Q'(f(e))$.

  Let $e \in E$.
  We have two cases:
  \begin{enumerate}
  \item $e \in E_1$

    Since $\es{U}_1 \sqsubseteq \es{U}'_1$, it follows directly that
    $Q(e) = Q_1(e) = Q_1'(e) = Q'(e)$.
    
  \item $e=(e_2,x) \in E_2 \times \confmax{\es{U}_1}$

    By Definition~\ref{def:pes-seq3}, we know that $Q(e_2, x) = Q_{2}(e_2)$.  Since
    $\es{U}_2 \sqsubseteq \es{U}'_2$, then $Q_2(e_2) = Q'_2(e_2)$.  By
    Definition~\ref{def:pes-seq3}, we have that $Q'_2(e_2) = Q'(e_2, m_1(x))$.
  \end{enumerate}



    

\end{proof}

\subsubsection*{Proof of Lemma~\ref{lem:conc-mono3}}
\begin{proof}
  Let
  \begin{align*}
    & \es{U}_1 = \qpes{E_1}{\leq_1}{\#_1}{Q_1} \\
    & \es{U}'_1 = \qpes{E'_1}{\leq'_1}{\#'_1}{Q'_1} \\
    & \es{U}_2 = \qpes{E_2}{\leq_2}{\#_2}{Q_2} \\
    & \es{U}'_2 = \qpes{E'_2}{\leq'_2}{\#'_2}{Q'_2} \\
    & \conc{\es{U}_1}{\es{U}_2} = \qpes{E}{\leq}{\#}{Q} \\
    & \conc{\es{U}'_1}{\es{U}'_2} = \qpes{E'}{\leq'}{\#'}{Q'}
  \end{align*}
  such that $\es{U}_1 \sqsubseteq \es{U}'_1$ and $\es{U}_2 \sqsubseteq \es{U}'_2$.

  Due to Lemma~\ref{lem:conc-mono1} we only need to show
  $\forall e \in E\ .\ Q(e) = Q'(f(e))$.

  Let $e \in E$.  If $e \in E_1$ then by Definition~\ref{def:pes-conc3} we have $Q(e) = Q_1(e)$,
  which by $\es{U}_1 \sqsubseteq \es{U}_1'$ gives $Q_1(e) = Q_1'(e)$ that by
  Definition~\ref{def:pes-conc3} gives $Q_1'(e) = Q'(e)$.  Similarly when $e \in E_2$.
\end{proof}

\subsubsection*{Proof of Lemma~\ref{lem:meas-mono3}}
\begin{proof}
  Let
  \begin{align*}
    & \es{U}_1 = \qpes{E_1}{\leq_1}{\#_1}{Q_1} \\
    & \es{U}'_1 = \qpes{E'_1}{\leq'_1}{\#'_1}{Q'_1} \\
    & \es{U}_2 = \qpes{E_2}{\leq_2}{\#_2}{Q_2} \\
    & \es{U}'_2 = \qpes{E'_2}{\leq'_2}{\#'_2}{Q'_2} \\
    & \meas{n}{\es{U}_1}{\es{U}_2} = \qpes{E}{\leq}{\#}{Q} \\
    & \meas{n}{\es{U}'_1}{\es{U}'_2} = \qpes{E'}{\leq'}{\#'}{Q'}
  \end{align*}
  such that $\es{U}_1 \sqsubseteq \es{U}'_1$ and $\es{U}_2 \sqsubseteq \es{U}'_2$.

  We have to show that:
  \begin{enumerate}
  \item We start by defining $f : E \rightarrow E'$ such that
    \[
      f(e) =
      \begin{cases}
        e & \text{ if } e \in \init{\meas{n}{\es{U}_1}{\es{U}_2}} \\
        f_1(e) & \text{ if } e \in E_1 \\
        f_2(e) & \text{ if } e \in E_2
      \end{cases}
    \]

    It is straightforward to see that $f$ is injective
    
  \item $\pi(f(e)) = \pi(e)$

    If $e \in E_1$ or $e \in E_2$, we are done since $\es{U}_1 \sqsubseteq \es{U'}_1$ and $\es{U}_2 \sqsubseteq \es{U'}_2$.
    If $e \in \init{\meas{n}{\es{U}_1}{\es{U}_2}}$, then $\pi(f(e)) = \pi(e)$, since $f(e) = e$.
    
  \item $e \leq e' \Leftrightarrow f(e) \leq' f(e')$

    If $e \in E_1$ or $e \in E_2$, we are done since $\es{U}_1 \sqsubseteq \es{U'}_1$ and $\es{U}_2 \sqsubseteq \es{U'}_2$.
    If $e \in \init{\meas{n}{\es{U}_1}{\es{U}_2}}$, since $f(e) = e$. Consequently $e \leq e' \Leftrightarrow e \leq' f(e')$, which is trivially satisfied by Definition~\ref{def:pes-meas3}.
    
  \item $e \# e' \Leftrightarrow f(e) \#' f(e')$

    This case is similar to that of Lemma~\ref{lem:nd-mono1}.

  \item $\forall e \in E\ .\ Q(e) = Q'(f(e))$

    If $e = \tau_0^n$ or $e = \tau_1^n$ then we are done, since $f(e) = e$ and, consequently,
    $Q(e) = Q'(e)$. If $e \in E_1$ or $e \in E_2$ we are done since $\es{U}_1 \sqsubseteq \es{U}'_1$ and
    $\es{U}_2 \sqsubseteq \es{U}'_2$.
  \end{enumerate}

  The first three conditions follow directly from Definition~\ref{def:pes-meas3}.  For the last
  condition we argue as follows: 
\end{proof}

\subsubsection*{Proof of Lemma~\ref{lem:seq-rem-init3}}
\begin{proof}
  Let
  \begin{align*}
    & \es{U}_1 = \qpes{E_1}{\leq_1}{\#_1}{Q_1} \\
    & \es{U}_2 = \qpes{E_2}{\leq_2}{\#_2}{Q_2} \\
    & \seq{\es{U}_1}{\es{U}_2} = \qpes{E_{\seq{1}{2}}}{\leq_{\seq{1}{2}}}{\#_{\seq{1}{2}}}{Q_{\seq{1}{2}}} \\
    & (\seq{\es{U}_1}{\es{U}_2}) \backslash l = \qpes{E}{\leq}{\#}{Q} \\
    & \es{U}_1 \backslash l = \qpes{E_1^l}{\leq_1^l}{\#_1^l}{Q_1^l} \\
    & \seq{(\es{U}_1 \backslash l)}{\es{U}_2} = \qpes{E'}{\leq'}{\#'}{Q'} \\ 
    & l \in \init{\seq{\es{U}_1}{\es{U}_2}}
  \end{align*}
 
  Due to Lemma~\ref{lem:seq-rem-init1} we focus only on the quantum part.
  \begin{itemize}
  \item $(\seq{\es{U}_1}{\es{U}_2}) \backslash l \sqsubseteq \seq{(\es{U}_1\backslash l)}{\es{U}_2}$

    We need to show $\forall e \in E\ .\ Q(e) = Q'(f(e))$.  Let $e \in E$.  By
    Definition~\ref{def:rem-init3}, $Q(e) = Q_{\seq{1}{2}}|_{E} (e)$.
    By Definition~\ref{def:pes-seq3} we have two cases:
    \begin{itemize}
    \item $e \in E_1$
   
      By Definition~\ref{def:pes-seq3} and since $e \neq l$ and $\neg (e \# l)$, which gives
      $e \in E_1^l$, then $Q_{\seq{1}{2}}|_{E} (e) = Q_1|_{E_1^l}(e)$.  By
      Definition~\ref{def:rem-init3}, $Q_1|_{E_1^l}(e) = Q_1^l(e)$.  By
      Definition~\ref{def:pes-seq3}, $Q_1^l(e) = Q'(e)$.
      
    \item $e=(e_2,x) \in E_2 \times \confmax{\es{U}_1}$

      By Definition~\ref{def:pes-seq3} and since $l \not\in E_2$,
      $Q_{\seq{1}{2}}|_{E} (e_2,x) = Q_{2}(e_2)$.  Hence, again by
      Definition~\ref{def:pes-seq3}, $Q_2(e_2) = Q'(e_2, m_1(x))$.
    \end{itemize}
  \item $\seq{(\es{U}_1\backslash l)}{\es{U}_2} \sqsubseteq (\seq{\es{U}_1}{\es{U}_2}) \backslash l$

    Similar reasoning to the previous bullet.
  \end{itemize}
\end{proof}

\subsubsection*{Proof of Lemma~\ref{lem:meas-rem-init3}}
\begin{proof}
  Let
  \begin{align*}
    & \es{U}_1 = \qpes{E_1}{\leq_1}{\#_1}{Q_1} \\ 
    & \es{U}_2 = \qpes{E_2}{\leq_2}{\#_2}{Q_2} \\ 
    & \meas{n}{\es{U}_1}{\es{U}_2} = \qpes{E}{\leq}{\#}{Q} \\
    & (\meas{n}{\es{U}_1}{\es{U}_2}) \backslash l = \qpes{E'}{\leq'}{\#'}{Q'}
  \end{align*}
  
  Let us focus on the case where $l = \tau_0^n$, since the reasoning when $l = \tau_1^n$ is equal to
  the one below.
  \begin{itemize}
  \item $(\meas{n}{\es{U}_1}{\es{U}_2}) \backslash \tau_0^n \sqsubseteq \es{U}_1$
    \begin{enumerate}
    \item We start by defining $f : E' \rightarrow E_1$ as being the identity, which is injective.

    \item $\pi(f(e)) = \pi(e)$

      It follows directly since $f$ is the identity.

    \item $e \leq' e' \Leftrightarrow f(e) \leq_1 f(e')$

      Since $f$ is the identity, $e \leq' e' \Leftrightarrow e \leq_1 e'$.  It then follows directly
      from Definition~\ref{def:rem-init3} and Definition~\ref{def:pes-meas3}.

    \item $e \#' e' \Leftrightarrow f(e) \#_1 f(e')$

      Since $f$ is the identity, $e \#' e' \Leftrightarrow e \#_1 e'$.  It then follows directly
      from Definition~\ref{def:rem-init3} and Definition~\ref{def:pes-meas3}.

    \item $\forall e \in E'\ .\ Q'(e) = Q_1(f(e))$

      Since $f$ is the identity, $Q'(e) = Q_1(e)$.  Let $e \in E'$.  By
      Definition~\ref{def:rem-init3}, $Q'(e) = Q|_{E'}(e)$.  Since $\tau_0^n \leq e$ we know that
      $e \in E_1$, hence by Definition~\ref{def:pes-meas3}, $Q|_{E'}(e) = Q_1(e)$.
    \end{enumerate}

  \item $\es{U}_1 \sqsubseteq (\meas{n}{\es{U}_1}{\es{U}_2}) \backslash \tau_0^n$
    \begin{enumerate}
    \item We start by defining $f : E_1 \rightarrow E'$ as being the identity, which is injective.

    \item $\pi(f(e)) = \pi(e)$

      It follows directly since $f$ is the identity.

    \item $e \leq_1 e' \Leftrightarrow f(e) \leq' f(e')$

      Since $f$ is the identity, $e \leq_1 e' \Leftrightarrow e \leq' e'$.  It then follows directly
      from Definition~\ref{def:rem-init3} and Definition~\ref{def:pes-meas3}.

    \item $e \#_1 e' \Leftrightarrow f(e) \#' f(e')$

      Since $f$ is the identity, $e \#_1 e' \Leftrightarrow e \#' e'$.  It then follows directly
      from Definition~\ref{def:rem-init3} and Definition~\ref{def:pes-meas3}.

    \item $\forall e \in E_1\ .\ Q_1(e) = Q'(f(e))$

      Since $f$ is the identity, $Q_1(e) = Q'(e)$.  Let $e \in E_1$.  By
      Definition~\ref{def:pes-meas3}, $Q(e) = Q_1(e)$.  Since $e \neq l$ and $\neg(e \# l)$, then
      $Q'(e) = Q(e)$.  Thus $Q_1(e) = Q'(e)$.
    \end{enumerate}
  \end{itemize}
\end{proof}

\subsubsection*{Proof of Lemma~\ref{lem:conc-rem-init3}}
\begin{proof}
  Let $l \in \init{\es{E}_1}$ and
  \begin{align*}
    & \es{U}_1 = \qpes{E_1}{\leq_1}{\#_1}{Q_1} \\
    & \es{U}_2 = \qpes{E_2}{\leq_2}{\#_2}{Q_2} \\
    & \conc{\es{U}_1}{\es{U}_2} = \qpes{E_{\conc{1}{2}}}{\leq_{\conc{1}{2}}}{\#_{\conc{1}{2}}}{Q_{\conc{1}{2}}} \\
    & (\conc{\es{U}_1}{\es{U}_2}) \backslash l = \qpes{E}{\leq}{\#}{Q} \\
    & \es{U}_1 \backslash l = \qpes{E_1^l}{\leq_1^l}{\#_1^l}{Q_1^l} \\
    & \es{U}_2 \backslash l = \qpes{E_2^l}{\leq_2^l}{\#_2^l}{Q_2^l} \\
    & \conc{(\es{U}_1 \backslash l)}{\es{U}_2} = \qpes{E'}{\leq'}{\#'}{Q'} \\ 
    & l \in \init{\conc{\es{U}_1}{\es{U}_2}}
  \end{align*}
 
  Due to Lemma~\ref{lem:conc-rem-init1} we focus only on the quantum part.  
  \begin{itemize}
  \item $(\conc{\es{U}_1}{\es{U}_2}) \backslash l \sqsubseteq \conc{(\es{U}_1\backslash l)}{\es{U}_2}$

    We need to show $\forall e \in E\ .\ Q(e) = Q'(f(e))$.
    Let $e \in E$.
    By Definition~\ref{def:rem-init3}, $Q(e) = Q_{\conc{1}{2}}|_{E}(e)$.
    By Definition~\ref{def:pes-conc3} we have two cases:
    \begin{enumerate}
    \item $e \in E_1$

      Since $l \in \init{\es{U}_1}$ then $e \neq l$ and $\neg(e \# l)$.  Thus $e \in E_1^l$.
      Consequently, $Q_{\conc{1}{2}}|_{E}(e) = Q|_{E_1^l}(e)$.  By Definition~\ref{def:rem-init3},
      $Q1|_{E_1^l}(e) = Q_1^l(e)$.  By Definition~\ref{def:pes-conc3}, $Q_1^l(e) = Q'(e)$.

    \item  $e \in E_2$

      Then $Q_{\conc{1}{2}}|_{E}(e) = Q_2(e)$.  By Definition~\ref{def:pes-conc3}, $Q_2(e) = Q'(e)$.
    \end{enumerate}
    
  \item $\conc{(\es{U}_1\backslash l)}{\es{U}_2} \sqsubseteq (\conc{\es{U}_1}{\es{U}_2}) \backslash l$

    We need to show $\forall e \in E'\ .\ Q'(e) = Q(e)$.
    By Definition~\ref{def:pes-conc3} we have two cases:
    \begin{enumerate}
    \item $e \in E_1^l$

      Then $Q'(e) = Q_1^l(e)$.  By Definition~\ref{def:rem-init3}, $Q_1^l(e) = Q_1|_{E_1^l}(e)$.  By
      Definition~\ref{def:pes-conc3}, $Q_1|_{E_1^l}(e) = Q_{\conc{1}{2}}|_{E_1^l \uplus E_2}$.  By
      Definition~\ref{def:rem-init3}, $Q_{\conc{1}{2}}|_{E_1^l \uplus E_2} = Q(e)$.

    \item $e \in E_2^l$

      Then $Q'(e) = Q_2(e)$. By Definition~\ref{def:pes-conc3}, $Q_2(e) = Q_{\conc{1}{2}}(e)$.  By
      Definition~\ref{def:rem-init3}, $Q_{\conc{1}{2}}(e) = Q(e)$.
    \end{enumerate}
  \end{itemize}
\end{proof}

\subsubsection*{Proof of Lemma~\ref{lem:conc-symmetric3}}
\begin{proof}
  It follows directly from Definition~\ref{def:pes-conc3}.
\end{proof}

\subsubsection*{Proof of Lemma~\ref{res:soundI-3}}
\begin{proof}
  Induction over rules in Figure~\ref{fig:op-small3}.

  \begin{itemize}
  \item $\com{skip} \xrightarrow{sk} \checkmark$

    It follows directly that $\mf{\checkmark} \equiv \mf{\com{skip}} \backslash sk \equiv \emptyset$.

  \item $\com{a} \xrightarrow{a} \checkmark$

    It follows directly that $\mf{\checkmark} \equiv \mf{\com{a}} \backslash a \equiv \emptyset$.

  \item $\meas{n}{\com{C}_1}{\com{C}_2} \xrightarrow{\tau_0^n} \com{C}_1$

    It follows directly since
    $\seq{\es{P}_0^n}{\mf{\com{C}_1}}
    \sqsubseteq \seq{\es{P}_0^n}{\mf{\com{C}_1}} + \seq{\es{P}_1^n}{\mf{\com{C}_2}}
    = \mf{\meas{n}{\com{C}_1}{\com{C}_2}}$.

  \item $\meas{n}{\com{C}_1}{\com{C}_2} \xrightarrow{\tau_1^n} \com{C}_2$

    It follows directly since
    $\seq{\es{P}_1^n}{\mf{\com{C}_2}}
    \sqsubseteq \seq{\es{P}_0^n}{\mf{\com{C_1}}} + \seq{\es{P}_1^n}{\mf{\com{C_2}}}
    = \mf{\meas{n}{\com{C}_1}{\com{C}_2}}$.

  \item $\seq{\com{C}_1}{\com{C}_2} \xrightarrow{l} \com{C}_2$
    \begin{align*}
      & \seq{\com{C}_1}{\com{C}_2} \xrightarrow{l} \com{C}_2 \\
      \Rightarrow & \set{ \xrightarrow{l} \text{ entails}} \\
      & \com{C}_1 \xrightarrow{l} \checkmark \\
      \Rightarrow & \set{\text{i.h.}} \\
      & \mf{\checkmark} \equiv \mf{\com{C}_1} \backslash l \\
      \Rightarrow & \set{ \text{Lemma~\ref{lem:seq-mono3}}} \\
      & \seq{\mf{\checkmark}}{\mf{\com{C}_2}} \equiv
        \seq{(\mf{\com{C}_1} \backslash l)}{\mf{\com{C}_2}} \\
      \Rightarrow & \set{
                    \seq{\mf{\checkmark}}{\mf{\com{C}_2}} \equiv \mf{\com{C}_2},
                    \text{ Lemma~\ref{lem:seq-rem-init3}}
                    } \\
      & \mf{\com{C}_2} \equiv (\seq{\mf{\com{C}_1}}{\mf{\com{C}_2}}) \backslash l \\
      \Rightarrow & \set{\text{Definition~\ref{def:den-sem3}}} \\
      & \mf{\com{C}_2} \equiv \mf{\seq{\com{C}_1}{\com{C}_2}} \backslash l \\
    \end{align*}

  \item $\seq{\com{C}_1}{\com{C}_2} \xrightarrow{l} \seq{\com{C}'_1}{\com{C}_2}$
    \begin{align*}
      & \seq{\com{C}_1}{\com{C}_2} \xrightarrow{l} \seq{\com{C}'_1}{\com{C}_2} \\
      \Rightarrow & \set{ \xrightarrow{l} \text{ entails}} \\
      & \com{C}_1 \xrightarrow{l} \com{C}'_1 \\
      \Rightarrow & \set{\text{i.h.}} \\
      & \mf{\com{C}'_1} \equiv \mf{\com{C}_1} \backslash l \\
      \Rightarrow & \set{ \text{Lemma~\ref{lem:seq-mono3}}} \\
      & \seq{\mf{\com{C}'_1}}{\mf{\com{C}_2}} \equiv \seq{(\mf{\com{C}_1} \backslash l)}{\mf{\com{C}_2}} \\
      \Rightarrow & \set{\text{Lemma~\ref{lem:seq-rem-init3}}} \\
      & \seq{\mf{\com{C}'_1}}{\mf{\com{C}_2}} \equiv (\seq{\mf{\com{C}_1}}{\mf{\com{C}_2}}) \backslash l \\
      \Rightarrow & \set{\text{Definition~\ref{def:den-sem3}}} \\
      & \mf{\seq{\com{C}'_1}{\com{C}_2}} \equiv \mf{\seq{\com{C}_1}{\com{C}_2}} \backslash l \\
    \end{align*}

  \item $\conc{\com{C}_1}{\com{C}_2} \xrightarrow{l} \com{C}_2$
    \begin{align*}
      & \conc{\com{C}_1}{\com{C}_2} \xrightarrow{l} \com{C}_2 \\
      \Rightarrow & \set{ \xrightarrow{l} \text{ entails}} \\
      & \com{C}_1 \xrightarrow{l} \checkmark \\
      \Rightarrow & \set{\text{i.h.}} \\
      & \mf{\checkmark} \equiv \mf{\com{C}_1} \backslash l \\
      \Rightarrow & \set{\text{Lemma~\ref{lem:conc-mono3}}} \\
      & \conc{\mf{\checkmark}}{\mf{\com{C}_2}} \equiv \conc{(\mf{\com{C}_1} \backslash l)}{\com{C}_2} \\
      \Rightarrow & \set{\conc{\mf{\checkmark}}{\mf{\com{C}_2}} \equiv \mf{\com{C}_2} } \\
      & \mf{\com{C}_2} \equiv \conc{(\mf{\com{C}_1} \backslash l)}{\com{C}_2} \\
      \Rightarrow & \set{\text{Lemma~\ref{lem:conc-rem-init3}},
                    \text{Definition~\ref{def:den-sem3}}} \\
      & \mf{\com{C}_2} \equiv \mf{\conc{\com{C}_1}{\com{C}_2}} \backslash l \\
    \end{align*}

  \item $\conc{\com{C}_1}{\com{C}_2} \xrightarrow{l} \conc{\com{C}'_1}{\com{C}_2}$
    \begin{align*}
      & \conc{\com{C}_1}{\com{C}_2} \xrightarrow{l} \conc{\com{C}'_1}{\com{C}_2} \\
      \Rightarrow & \set{ \xrightarrow{l} \text{ entails}} \\
      & \com{C}_1 \xrightarrow{l} \com{C}'_1 \\
      \Rightarrow & \set{\text{i.h.}} \\
      & \mf{\com{C}'_1} \equiv \mf{\com{C}_1} \backslash l \\
      \Rightarrow & \set{\text{Lemma~\ref{lem:conc-mono3}}} \\
      & \conc{\mf{\com{C}'_1}}{\mf{\com{C}_2}} \equiv \conc{(\mf{\com{C}_1} \backslash l)}{\com{C}_2} \\
      \Rightarrow & \set{\text{Lemma~\ref{lem:conc-rem-init3}},
                    \text{Definition~\ref{def:den-sem3}}} \\
      & \mf{\conc{\com{C}'_1}{\com{C}_2}} \equiv \mf{\conc{\com{C}_1}{\com{C}_2}} \backslash l \\
    \end{align*}

  \item $\conc{\com{C}_1}{\com{C}_2} \xrightarrow{l} \com{C}_1$
    \begin{align*}
      & \conc{\com{C}_1}{\com{C}_2} \xrightarrow{l} \com{C}_1 \\
      \Rightarrow & \set{ \xrightarrow{l} \text{ entails}} \\
      & \com{C}_2 \xrightarrow{l} \checkmark \\
      \Rightarrow & \set{\text{i.h.}} \\
      & \mf{\checkmark} \equiv \mf{\com{C}_2} \backslash l \\
      \Rightarrow & \set{\text{Lemma~\ref{lem:conc-mono3}}} \\
      & \conc{\mf{\com{C}_1}}{\mf{\checkmark}} \equiv
        \conc{\mf{\com{C}_1}}{(\mf{\com{C}_2} \backslash l)} \\
      \Rightarrow & \set{\conc{\mf{\com{C}_1}}{\mf{\checkmark}} \equiv \mf{\com{C}_1} } \\
      & \mf{\com{C}_1} \equiv \conc{\mf{\com{C}_1}}{(\mf{\com{C}_2} \backslash l)} \\
      \Rightarrow & \set{\text{Lemma~\ref{lem:conc-rem-init3}},
                    \text{Definition~\ref{def:den-sem3}}} \\
      & \mf{\com{C}_1} \equiv \mf{\conc{\com{C}_1}{\com{C}_2}} \backslash l \\
    \end{align*}

  \item $\conc{\com{C}_1}{\com{C}_2} \xrightarrow{l} \conc{\com{C}_1}{\com{C}'_2}$
    \begin{align*}
      & \conc{\com{C}_1}{\com{C}_2} \xrightarrow{l} \conc{\com{C}_1}{\com{C}'_2} \\
      \Rightarrow & \set{ \xrightarrow{l} \text{ entails}} \\
      & \com{C}_2 \xrightarrow{l} \com{C}'_2 \\
      \Rightarrow & \set{\text{i.h.}} \\
      & \mf{\com{C}'_2} \equiv \mf{\com{C}_2} \backslash l \\
      \Rightarrow & \set{\text{Lemma~\ref{lem:conc-mono3}}} \\
      & \conc{\mf{\com{C}_1}}{\mf{\com{C}'_2}} \equiv \conc{\com{C}_1}{(\mf{\com{C}_2} \backslash l)} \\
      \Rightarrow & \set{\text{Lemma~\ref{lem:conc-rem-init3}},
                    \text{Definition~\ref{def:den-sem3}}} \\
      & \mf{\conc{\com{C}_1}{\com{C}'_2}} \equiv \mf{\conc{\com{C}_1}{\com{C}_2}} \backslash l \\
    \end{align*}


    

  \end{itemize}
\end{proof}

\subsubsection*{Proof of Theorem~\ref{res:soundII-3}}
\begin{proof}
  \begin{itemize}
  \item $|\omega| = 1$

    It follows directly that
    $\exists \set{l} \in \confES{\mf{\com{C}}}\, .\, \emptyset \cchain{l} \set{l}$

  \item $|\omega| > 1$
    \begin{align*}
      & \com{C} \xtwoheadrightarrow{\omega} \com{C}' \\
      \Rightarrow & \set{\text{Definition~\ref{fig:op-nstep3}}} \\
      & \com{C} \xrightarrow{l} \com{C}'' \qquad
        \com{C}'' \xtwoheadrightarrow{\omega'} \com{C}' \\
      \Rightarrow & \set{\text{Lemma~\ref{res:soundI-3}, i.h.}} \\
      & \mf{\com{C}''} \equiv \mf{\com{C}} \backslash l \qquad
        \exists y \in \confES{\mf{\com{C}''}}\, .\, \emptyset \cchain{\omega'} y \\
      \Rightarrow & \set{\text{Definition~\ref{def:rem-init3}}} \\
      & \set{l} \cup y \in \confES{\mf{\com{C}}}\, .\, \emptyset \cchain{l} \set{l} \cchain{\omega'} \set{l} \cup y = x
    \end{align*}
  \end{itemize}
\end{proof}

\subsubsection*{Proof of Lemma~\ref{res:adI-3}}
\begin{proof}
  Induction over the interpretation of commands.

  \begin{itemize}
  \item $sk \in \init{\com{skip}}$

    Let $\com{C}' = \checkmark$.  It follows directly that $\com{skip} \xrightarrow{sk} \checkmark$
    and that $\mf{\com{skip}} \backslash sk \equiv \mf{\checkmark}$.

  \item $a \in \init{\com{a}}$

    Let $\com{C}' = \checkmark$.  It follows directly that $\com{a} \xrightarrow{a} \checkmark$
    and that $\mf{\com{a}} \backslash a \equiv \mf{\checkmark}$.

  \item $l \in \init{\meas{n}{\com{C}_1}{\com{C}_2}}$

    By Definition~\ref{def:den-sem3}, $\mf{{\meas{n}{\com{C}_1}{\com{C}_2}}} =
    \seq{\es{P}_0^n}{\mf{\com{C}_1}} + \seq{\es{P}_1^n}{\mf{\com{C}_2}}$.
    We have two cases:
    \begin{enumerate}
    \item $l' = \tau_0^n$

      By Lemma~\ref{lem:meas-rem-init3},
      $(\seq{\es{P}_0^n}{\mf{\com{C}_1}} + \seq{\es{P}_1^n}{\mf{\com{C}_2}}) \backslash \tau_0^n 
      \equiv \mf{\com{C}_1}$.  Furthermore
      $\meas{n}{\com{C}_1}{\com{C}_2} \xrightarrow{\tau_0^n} \com{C}_1$.
      
    \item $l' = \tau_1^n$

      By Lemma~\ref{lem:meas-rem-init3},
      $(\seq{\es{P}_0^n}{\mf{\com{C}_1}} + \seq{\es{P}_1^n}{\mf{\com{C}_2}}) \backslash \tau_1^n
      \equiv \mf{\com{C}_2}$.  Furthermore
      $\meas{n}{\com{C}_1}{\com{C}_2} \xrightarrow{\tau_1^n} \com{C}_2$.
    \end{enumerate}
    
  \item $l' \in \init{\seq{\com{C}_1}{\com{C}_2}}$

    By Definition~\ref{def:pes-seq3} we have that $l' \in \init{\mf{\com{C}_1}}$.  By i.h.,
    $\exists \com{C}'$ such that $\com{C}_1 \xrightarrow{l} \com{C}'$ and
    $\mf{\com{C}_1} \backslash l' \equiv \mf{\com{C}'}$. We have two cases:
    \begin{enumerate}
    \item $\com{C}' = \checkmark$

      We have $\com{C}_1 \xrightarrow{l'} \checkmark$ and
      $\mf{\com{C}_1} \backslash l' \equiv \mf{\checkmark}$.  By the rules in Figure~\ref{fig:op-small3},
      $\seq{\com{C}_1}{\com{C}_2} \xrightarrow{l'} \com{C}_2$.
      By Definition~\ref{def:pes-seq3}, $\seq{(\mf{\com{C}_1} \backslash l')}{\mf{\com{C}_2}} \equiv
      \seq{\mf{\checkmark}}{\mf{\com{C}_2}} \equiv \mf{\com{C}_2}$.
      
    \item $\com{C}' = \com{C}'_1$

      We have $\com{C}_1 \xrightarrow{l'} \com{C}'_1$ and
      $\mf{\com{C}_1} \backslash l' \equiv \mf{\com{C}'_1}$.  By the rules in Figure~\ref{fig:op-small3},
      $\seq{\com{C}_1}{\com{C}_2} \xrightarrow{l'} \seq{\com{C}'_1}{\com{C}_2}$.  By
      Definition~\ref{def:pes-seq3},
      $\seq{(\mf{\com{C}_1} \backslash l')}{\mf{\com{C}_2}} \equiv \seq{\mf{\com{C}'_1}}{\mf{\com{C}_2}}$.
      By Definition~\ref{def:den-sem3}, $\mf{\seq{\com{C}'_1}{\com{C}_2}}$.
    \end{enumerate}

  \item $l' \in \init{\conc{\com{C}_1}{\com{C}_2}}$

    By Definition~\ref{def:pes-conc3} we have two cases: 
    \begin{enumerate}
    \item $l' \in \init{\mf{\com{C}_1}}$

      By i.h. $\exists \com{C}'\, .\, \com{C}_1 \xrightarrow{l'} \com{C}'$ and
      $\mf{\com{C}_1} \backslash l' \equiv \mf{\com{C}'}$. 
      By the rules in Figure~\ref{fig:op-small3} we have two cases:
      \begin{enumerate}
      \item $\com{C}' = \checkmark$

        We have $\com{C}_1 \xrightarrow{l'} \checkmark$ and
        $\mf{\com{C}_1} \backslash l' \equiv \mf{\checkmark}$.  By the rules in
        Figure~\ref{fig:op-small3} we have $\conc{\com{C}_1}{\com{C}_2} \xrightarrow{l'} \com{C}_2$.
        By Definition~\ref{def:pes-conc3}, $\conc{(\mf{\com{C}_1}\backslash l')}{\mf{\com{C}_2}}$.
        By Lemma~\ref{lem:conc-rem-init3} we have
        $(\conc{\mf{\com{C}_1}}{\mf{\com{C}_2}}) \backslash l'$.  By Definition~\ref{def:den-sem3},
        $\mf{\conc{\com{C}_1}{\com{C}_2}} \backslash l'$.

      \item $\com{C}' = \com{C}'_1$

        We have $\com{C}_1 \xrightarrow{l'} \com{C}'_1$ and
        $\mf{\com{C}_1} \backslash l' \equiv \mf{\com{C}'_1}$.  By the rules in Figure~\ref{fig:op-small3}
        we have $\conc{\com{C}_1}{\com{C}_2} \xrightarrow{l'} \conc{\com{C}'_1}{\com{C}_2}$. By
        Definition~\ref{def:pes-conc3}, $\conc{(\mf{\com{C}_1}\backslash l')}{\mf{\com{C}_2}}$.
        By Lemma~\ref{lem:conc-rem-init3} we have
        $(\conc{\mf{\com{C}_1}}{\mf{\com{C}_2}}) \backslash l'$.  By Definition~\ref{def:den-sem3},
        $\mf{\conc{\com{C}_1}{\com{C}_2}} \backslash l'$.
      \end{enumerate}
      
    \item $l' \in \init{\mf{\com{C}_2}}$
      
      By i.h. $\exists \com{C}'\, .\, \com{C}_2 \xrightarrow{l} \com{C}'$ and
      $\mf{\com{C}_2} \backslash l' = \mf{\com{C}'}$. 
      By the rules in Figure~\ref{fig:op-small3} we have two cases:
      \begin{enumerate}
      \item $\com{C}' = \checkmark$

        We have $\com{C}_2 \xrightarrow{l'} \checkmark$ and
        $\mf{\com{C}_2} \backslash l' \equiv \mf{\checkmark}$.  By the rules in Figure~\ref{fig:op-small3}
        we have $\conc{\com{C}_1}{\com{C}_2} \xrightarrow{l'} \com{C}_1$.  By
        Definition~\ref{def:pes-conc3}, $\conc{\mf{\com{C}_1}}{(\mf{\com{C}_2}\backslash l')}$.
        By Lemma~\ref{lem:conc-rem-init3} we have
        $(\conc{\mf{\com{C}_1}}{\mf{\com{C}_2}}) \backslash l'$.
        By Definition~\ref{def:den-sem3}, $\mf{\conc{\com{C}_1}{\com{C}_2}} \backslash l'$.

      \item $\com{C}' = \com{C}'_2$

        We have $\com{C}_2 \xrightarrow{l'} \com{C}'_2$ and
        $\mf{\com{C}_2} \backslash l' \equiv \mf{\com{C}'_2}$.  By the rules in Figure~\ref{fig:op-small3}
        we have $\conc{\com{C}_1}{\com{C}_2} \xrightarrow{l'} \conc{\com{C}_1}{\com{C}'_2}$. By
        Definition~\ref{def:pes-conc3}, $\conc{\mf{\com{C}_1}}{(\mf{\com{C}_2}\backslash l')}$.
        By Lemma~\ref{lem:conc-rem-init3} we have
        $(\conc{\mf{\com{C}_1}}{\mf{\com{C}_2}}) \backslash l'$.  By Definition~\ref{def:den-sem3},
        $\mf{\conc{\com{C}_1}{\com{C}_2}} \backslash l'$.
      \end{enumerate}
    \end{enumerate}



      

  \end{itemize}
\end{proof}

\subsubsection*{Proof of Theorem~\ref{res:adII-3}}
\begin{proof}
  Induction over the length of $\omega$.
  
  \begin{itemize}
  \item $|\omega| = 1$

    We have $\set{l'} \in \confES{\com{C}}$ such that $\emptyset \cchain{l'} \set{l'}$.  Furthermore
    $l' \in \init{\mf{\com{C}}}$.  By Lemma~\ref{res:adI-3}, $\com{C} \xrightarrow{l'} \com{C}'$ and
    $\mf{\com{C}'} \equiv \mf{\com{C}}\backslash l'$.  By the rules in Figure~\ref{fig:op-nstep3},
    $\com{C} \xtwoheadrightarrow{l'} \com{C}'$.

  \item $|\omega| > 1$

    We have $x \in \confES{\mf{C}}$ such that $\emptyset \cchain{\omega} x$.  Since
    $\omega = l_0 l_1 \dots l_n$, then $\emptyset \cchain{l_0} \set{l_0} \cchain{\omega'} x$.  Hence
    $l_0 \in \init{\mf{\com{C}}}$. By Lemma~\ref{res:adI-3}, $\com{C} \xrightarrow{l_0} \com{C}'$
    and $\mf{\com{C}'} \equiv \mf{\com{C}} \backslash l_0$.  By Definition~\ref{def:rem-init3},
    $\exists y \in \confES{\mf{\com{C}'}}$ such that $\emptyset \cchain{\omega'} y$.  By
    i.h. $\exists \com{C}''$ such that $\com{C}' \xtwoheadrightarrow{\omega'} \com{C}''$.  By the
    rules in Figure~\ref{fig:op-nstep3}, $\com{C} \xtwoheadrightarrow{\omega} \com{C}''$, where
    $\omega = l_0 : \omega'$.

  \end{itemize}
\end{proof}

\subsubsection*{Proof of Lemma~\ref{lem:E-tilde}}
\begin{proof}
  We show that $\tilde{\leq}$ is a partial order and that $\tilde{\#}$ is symmetric and irreflexive.
  \begin{itemize}
  \item $\tilde{\leq}$ is a partial order:
    \begin{itemize}
    \item Reflexivity follows directly from Definition~\ref{def:E-tilde}.
    \item Transitivity we note that the pairs are compose of the same event. Hence transitivity
      holds.
    \item Antisymmetry holds with the same argument as in transitivity.
    \end{itemize}
  \item $\tilde{\#}$ is symmetric and irreflexive since $\tilde{\#}$ equals $\#$ for the events in
    $\tilde{E}$.
  \end{itemize}
  
  We show $\tilde{\es{U}}_y$ obeys the conditions of a unitary event structure.
  \begin{itemize}
  \item $\set{e' \mid e' \tilde{\leq} e}$ is finite

    Trivially holds because every $e \in \tilde{E}$ is only causally related to itself.
    
  \item $e \tilde{\#} e' \tilde{\leq} e'' \Rightarrow e \tilde{\#} e''$

    Trivially holds because every $e \in \tilde{E}$ is only causally related to itself.  Hence
    $e' \tilde{\leq} e''$ is by definition $e' \tilde{\leq} e'$. It then follows directly
    $e \tilde{\#} e' \tilde{\leq} e' \Rightarrow e \tilde{\#} e'$.
    
  \item $\econc{e}{e'} \Rightarrow [\tilde{Q}(e), \tilde{Q}(e')] = 0$

    It follows directly since if $\econc{e}{e'}$ in $\tilde{{\es{U}}}$, then $\econc{e}{e'}$ in
    $\es{U}$, in which $[Q(e), Q(e')] = 0$.  Hence
    $[\tilde{Q}(e), \tilde{Q}(e')] = [Q|_{\tilde{E}}(e) ,Q|_{\tilde{E}}(e')] = 0$

  \item $\minconflict$ is transitive

    It follows directly from the fact that $\tilde{\#}$ is inherited from $\es{U}$.

  \item $\forall e \in \tilde{E}, \sum_{e' \in [e]} Q(e')$ is unitary

    Let $ e \in \tilde{E}$.  By definition we have
    $\tilde{Q} = Q|_{\tilde{E}}$.  It then follows that
    $\sum_{e' \in [e]} \tilde{Q}([e]) = \sum_{e' \in [e]} Q|_{\tilde{E}}(e')$ is unitary.
  \end{itemize}
\end{proof}

\subsubsection*{Proof of Lemma~\ref{lem:E-hat}}
\begin{proof}
  We show that $\hat{\leq}$ is a partial order and that $\hat{\#}$ is symmetric and irreflexive.
  \begin{itemize}
  \item $\hat{\leq}$ is a partial order:
    \begin{itemize}
    \item Reflexivity follows directly from Definition~\ref{def:E-hat}.
    \item Transitivity we note that the pairs are compose of the same event. Hence transitivity
      holds.
    \item Antisymmetry holds with the same argument as in transitivity.
    \end{itemize}
  \item $\hat{\#}$ is symmetric and irreflexive since $\hat{\#}$ is the empty relation.
  \end{itemize}  
  
  We show $\hat{\es{U}}$ obeys the conditions of a unitary event structure.  For convenience we
  sometimes write $U$ instead of $\sum_{e' \in [e]} Q(e')$.
  \begin{itemize}
  \item $\set{[e'] \mid [e'] \hat{\leq} [e]}$ is finite

    Trivially holds because every $[e] \in \hat{E}$ is only causally related to itself.
    
  \item $[e] \hat{\#} [e'] \hat{\leq} [e''] \Rightarrow [e] \hat{\#} [e'']$

    Trivially holds because the conflict relation is empty.
    
  \item $\econc{[e]}{[e']} \Rightarrow [\hat{Q}([e]), \hat{Q}([e'])] = 0$

    We have three cases:
    \begin{enumerate}
    \item $|[e]| = |[e']| = 1$

      It follows directly that $[\hat{Q}([e]), \hat{Q}([e'])] = [Q(e), Q(e')] = 0$

    \item $|[e]| = 1$ and $|[e']| > 1$

      We have
      $[\hat{Q}([e]), \hat{Q}([e'])] =
      [Q(e), U] =
      0$, since the event $e$ and all the events in $[e']$ are concurrent.

    \item $|[e]| > 1$ and $|[e']| > 1$

      We have
      $[\hat{Q}([e]), \hat{Q}([e'])] =
      [U_1, U_2] =
      0$, since all the events in $[e]$ are concurrent with the events in $[e']$.
    \end{enumerate}

  \item $\minconflict$ is transitive

    It follows directly since $\hat{\#} = \emptyset$.
    
  \item $\forall [e] \in \hat{E}, \sum_{[e'] \in [[e]]} \hat{Q}([e'])$ is unitary

    Since $\hat{\#} = \emptyset$ then $\sum_{[e'] \in [[e]]} \hat{Q}([e']) = \hat{Q}([e'])$ since
    $|[[e]]| = 1$.  By Definition~\ref{def:E-hat} we have two cases:
    \begin{enumerate}
    \item $|[e']| = 1$

      Then $\hat{Q}([e']) = Q(e')$ which is a unitary.
      
    \item $|[e]| > 1$

      Then $\hat{Q}([e']) = U$ which is a unitary.
    \end{enumerate}
  \end{itemize}
\end{proof}

\subsubsection*{Proof of Lemma~\ref{lem:map-es}}
\begin{proof}
  We show that $proj_{\tilde{\es{E}}}$ satisfies the conditions to be a map of event structures.
  \begin{itemize}
  \item $\forall x \in \confES{\tilde{\es{E}}} \Rightarrow
    proj_{\tilde{\es{E}}}(x) \in \confES{\hat{\es{E}}}$

    It follows straightforwardly since $\hat{\#}$ is empty and $\hat{\leq}$ is the trivial one.
    
  \item $\forall (e \neq e') \in x \in \confES{\tilde{\es{E}}}$, if $proj_{\tilde{\es{E}}}$ is
    defined in both then $proj_{\tilde{\es{E}}}(e) \neq proj_{\tilde{\es{E}}}(e')$

    Let $(e \neq e') \in x \in \confES{\tilde{\es{E}}}$.  Since $proj_{\tilde{\es{E}}}$ is total,
    $proj_{\tilde{\es{E}}}$ is defined in both and since $e, e' \in x$ then $\neg(e \tilde{\#} e')$.
    Hence it follows straightforwardly that $proj_{\tilde{\es{E}}}(e) \neq proj_{\tilde{\es{E}}}(e')$.
  \end{itemize}  
\end{proof}

\subsubsection*{Proof of Lemma~\ref{lem:po-3}}
\begin{proof}
  Due to Lemma~\ref{lem:po-1} we only focus on the condition of the quantum operators.  Let
  $\es{U}_1 = \qpes{E_1}{\leq_1}{\#_1}{Q_1}$, $\es{U}_2 = \qpes{E_2}{\leq_2}{\#_2}{Q_2}$, and
  $\es{U}_3 = \qpes{E_3}{\leq_3}{\#_3}{Q_3}$ be quantum event structures.
  \begin{itemize}
  \item Reflexivity: $\es{U}_1 \trianglelefteq \es{U}_1$

    It follows directly that $\forall e \in E_1\, .\, Q_1(e) = Q_1(e)$
    
  \item Transitivity : $\es{U}_1 \trianglelefteq \es{U}_2,\, \es{U}_2 \trianglelefteq \es{U}_3 \Rightarrow \es{U}_1 \trianglelefteq \es{U}_3$

    From $\es{U}_1 \trianglelefteq \es{U}_2$, $\forall e \in E_1\, .\, Q_1(e) = Q_2(e)$.  From
    $\es{U}_2 \trianglelefteq \es{U}_3$, $\forall e \in E_2\, .\, Q_2(e) = Q_3(e)$.  Hence
    $\forall e \in E_1\, .\, Q_1(e) = Q_2(e) = Q_3(e)$.
    
  \item Antisymmetry: $\es{U}_1 \trianglelefteq \es{U}_2,\, \es{U}_2 \trianglelefteq \es{U}_1 \Rightarrow \es{U}_1 = \es{U}_2$

    From Lemma~\ref{lem:po-1} we know that $E_1 = E_2$.  Hence it follows directly that
    $\es{U}_1 = \es{U}_2$.
  \end{itemize}
\end{proof}

\subsubsection*{Proof of Lemma~\ref{lem:po-least-elem-3}}
\begin{proof}
  We begin by showing that $\bot$ is a unitary event structure. We already know that
  $\pes{\emptyset}{\emptyset}{\emptyset}$ is an event structure, hence it lacks to verify the
  conditions regarding the quantum operator. However, since there are no events, the conditions
  trivially holds.

  To show that $\bot$ is the least element, consider any unitary event structure $\es{U}$. We need
  to show $\bot \trianglelefteq \es{U}$.  Due to Lemma~\ref{lem:po-least-elem-1} we only focus on
  the quantum operator. We need to show that for every event in $\bot$, its mapping through $!$ and
  $Q$ is the same. We show it by contradiction. Thus, we need to find an event $e \in \emptyset$
  such that its mapping through $!$ and $Q$ is not the same. However, there are no events in
  $\bot$. Thus the condition holds.
\end{proof}

\subsubsection*{Proof of Lemma~\ref{lem:lub-es-3}}
\begin{proof}
  Due to Lemma~\ref{lem:lub-es-1} we focus on the quantum operator condition.
  \begin{itemize}
  \item $\forall e, e' \in E^\omega\, ,\ \econc{e}{e'} \Rightarrow [Q^\omega(e), Q^\omega(e')] = 0$

    Let $e, e' \in E^\omega$ such that $\econc{e}{e'}$. We have two cases:
    \begin{enumerate}
    \item $e, e' \in E_n$

      By Definition~\ref{def:lub-3} we have $Q^\omega(e) = Q_n(e)$ and $Q^\omega(e') = Q_n(e')$ and
      since $\es{U}_n$ is a unitary event structure we are done.

    \item $e \in E_n$ and $e' \in E_m$ such that $\es{U}_n \trianglelefteq \es{U}_m$

      By Definition~\ref{def:lub-3} we have $Q^\omega(e) = Q_n(e)$ and $Q^\omega(e') = Q_m(e')$.
      From $\es{U}_n \trianglelefteq \es{U}_m$, we have that $Q_n(e) = Q_m(e)$ and since $\es{U}_m$
      is a unitary event structure, then we are done.
    \end{enumerate}
    
  \item $\minconflict$ is transitive

    We want to show that for $e, e', e'' \in E^\omega$ if $e \minconflict e'$ and
    $e' \minconflict e''$ then $e \minconflict e''$.
    According to Definition~\ref{def:lub-3} we have three cases:
    \begin{enumerate}
    \item $e, e', e'' \in E_n$

      Then we are done, since $\es{U}_n$ is a unitary event structure.
      
    \item $e \in E_n$ and $e', e'' \in E_m$ such that $\es{U}_n \trianglelefteq \es{U}_m$

      From $\es{U}_n \trianglelefteq \es{U}_m$ we know that $E_n \subseteq E_m$, hence $e \in E_m$.
      Since $\es{U}_m$ is a unitary event structure we are done.
      
    \item $e \in E_n$, $e' \in E_m$, and $e'' \in E_k$ such that $\es{U}_n \trianglelefteq \es{U}_m \trianglelefteq \es{U}_k$

      From $\es{U}_n \trianglelefteq \es{U}_m \trianglelefteq \es{U}_k$ we have that $E_n \subseteq E_m \subseteq E_k$.
      Hence, $e, e' \in E_k$.
      Since $\es{U}_k$ is a unitary event structure we are done.
    \end{enumerate}
    
  \item $\forall e \in E^\omega,\ |[e]| > 1 \Rightarrow \sum_{e' \in [e]} Q(e') = U$

    Let $e \in E^\omega\, .\, |[e]| > 1$.  By Definition~\ref{def:lub-3}, $\exists n \in \omega$
    such that $e \in E_n$ and $Q^\omega(e) = Q_n(e)$. Since $\es{U}_n$ is a unitary event structure
    we have $\sum_{e' \in [e]} Q^\omega(e') = \sum_{e' \in [e]} Q_n(e') = U$.
  \end{itemize}
\end{proof}

\subsubsection*{Proof of Lemma~\ref{lem:lub-3}}
\begin{proof}
  Due to Lemma~\ref{lem:lub-1} we only need to focus on the quantum condition.
  \begin{itemize}
  \item $\es{U}^\omega$ is an upper bound

    $\forall n \in \omega$ we need to have $\es{U}_n \trianglelefteq \es{U}^\omega$.  We need to
    check that $\forall e \in E_n\, .\, Q_n(e) = Q^\omega(e)$.  It follows directly from
    Definition~\ref{def:lub-3} that $\exists n \in \omega\, .\, e \in E_n$ and
    $Q^\omega(e) = Q_n(e)$.
    
  \item $\es{U}^\omega$ is a least upper bound

    Let $\es{U} = \qpes{E}{\leq}{\#}{Q}$ be an upper bound of the chain.  We need to show that if
    $\es{U}_n \trianglelefteq \es{U}^\omega$ and $\es{U}_n \trianglelefteq \es{U}$ then
    $\es{U}^\omega \trianglelefteq \es{U}$.  From $\es{U}_n \trianglelefteq \es{U}^\omega$,
    $\forall e \in E_n\, .\, Q_n(e) = Q^\omega(e)$.  By Definition~\ref{def:lub-3},
    $\exists n \in \omega\, .\, e \in E_n$ and $Q^\omega(e) = Q_n(e)$.  From
    $\es{U}_n \trianglelefteq \es{U}$, $\forall e \in E_n\, .\, Q_n(e) = Q(e)$.  Thus
    $\forall e \in \es{U}^\omega,\ \exists n \in \omega\, .\, e \in E_n$ and
    $Q^\omega(e) = Q_n(e) = Q(e)$.
  \end{itemize}
\end{proof}

\subsubsection*{Proof of Lemma~\ref{lem:seq-fix-mono-3}}
\begin{proof}
  Let
  \begin{align*}
    & \es{U} = \qpes{E}{\leq}{\#}{Q} \\
    & \es{U}_1 = \qpes{E_1}{\leq_1}{\#_1}{Q_1} \\
    & \es{U}_2 = \qpes{E_2}{\leq_2}{\#_2}{Q_2} \\
    & \seq{\es{U}}{\es{U}_1} = \qpes{E^1}{\leq^1}{\#^1}{Q^1} \\ 
    & \seq{\es{U}}{\es{U}_2} = \qpes{E^2}{\leq^2}{\#^2}{Q^2}
  \end{align*}
  such that $\es{U}_1 \trianglelefteq \es{U}_2$.

  We want to show $\forall e \in E^1\, .,\ Q^1(e) = Q^2(e)$.  Due to Lemma~\ref{lem:seq-fix-mono-1}
  we focus solely on the quantum condition.  Let $e \in E^1$.  By Definition~\ref{def:pes-seq3} we
  have two cases:
  \begin{enumerate}
  \item $e \in E$

    It follows directly that $Q^1(e) = Q(e) = Q^2(e)$.
    
  \item $e=(e_1,x) \in E_1 \times \confmax{\es{U}}$

    We have $Q^1(e) =  Q_1(e_1)$.
    From $\es{U}_1 \trianglelefteq \es{U}_2$, $Q_1(e_1) = Q_2(e_1)$.
    By Definition~\ref{def:pes-seq3}, $Q_2(e_1) = Q^2(e)$.
    Thus $Q^1(e) = Q^2(e)$.
  \end{enumerate}
\end{proof}

\subsubsection*{Proof of Lemma~\ref{lem:conc-fix-mono-3}}
\begin{proof}
  Let
  \begin{align*}
    & \es{U}_1 = \qpes{E_1}{\leq_1}{\#_1}{Q_1} \\ & \es{U}_2 = \qpes{E_2}{\leq_2}{\#_2}{Q_2} \\
    & \es{U}'_1 = \qpes{E'_1}{\leq'_1}{\#'_1}{Q'_1} \\ & \es{U}_2 = \qpes{E'_2}{\leq'_2}{\#'_2}{Q'_2} \\
    & \conc{\es{U}_1}{\es{U}_2} = \qpes{E}{\leq}{\#}{Q} \\ 
    & \conc{\es{U}'_1}{\es{U}'_2} = \qpes{E'}{\leq'}{\#'}{Q'}
  \end{align*}
  such that $\es{U}_1 \trianglelefteq \es{U}'_1$ and $\es{U}_2 \trianglelefteq \es{U}'_2$.

  We want to show that $\forall e \in E\, .\, Q(e) = Q'(e)$.  Due to Lemma~\ref{lem:conc-fix-mono-1}
  we focus solely on the quantum condition.  Let $e \in E$.  By Definition~\ref{def:pes-conc3} we
  have two cases:
  \begin{enumerate}
  \item $e \in E_1$

    We know that $Q(e) = Q_1(e)$.  Since $\es{U}_1 \trianglelefteq \es{U}'_1$, $Q_1(e) = Q'_1(e)$.
    By Definition~\ref{def:pes-conc3}, $Q'_1(e) = Q'(e)$.  Thus $Q(e) = Q'(e)$.

  \item $e \in E_2$

    Similar to the previous.
  \end{enumerate}
\end{proof}

\subsubsection*{Proof of Lemma~\ref{lem:meas-fix-mono-3}}
\begin{proof}
  Let
  \begin{align*}
    & \es{U}_1 = \qpes{E_1}{\leq_1}{\#_1}{Q_1} \\ & \es{U}_2 = \qpes{E_2}{\leq_2}{\#_2}{Q_2} \\
    & \es{U}'_1 = \qpes{E'_1}{\leq'_1}{\#'_1}{Q'_1} \\ & \es{U}_2 = \qpes{E'_2}{\leq'_2}{\#'_2}{Q'_2} \\
    & \meas{n}{\es{U}_1}{\es{U}_2} = \qpes{E}{\leq}{\#}{Q} \\
    & \meas{n}{\es{U}'_1}{\es{U}'_2} = \qpes{E'}{\leq'}{\#'}{Q'}
  \end{align*}
  such that $\es{U}_1 \trianglelefteq \es{U}'_1$ and $\es{U}_2 \trianglelefteq \es{U}'_2$.

  The conditions to check are:
  \begin{enumerate}
  \item $E \subseteq E'$     
  \item $\forall e, e'\, .\, e \leq e' \Leftrightarrow e,e' \in E \wedge e \leq' e'$
  \item $\forall e, e'\, .\, e \# e' \Leftrightarrow e,e' \in E \wedge e \#' e'$
  \item $\forall e \in E\, .\, Q(e) = Q'(e)$
  \end{enumerate}

  The first three conditions follow directly from Definition~\ref{def:pes-meas3}.  Hence we focus on
  the last one.

  Let $e \in E$.
  We have the four cases:
  \begin{enumerate}
  \item $e = \tau_0^n$

    By Definition~\ref{def:pes-meas3} we are done since, $Q(\tau_0^n) = Q'(\tau_0^n)$.
    
  \item $e = \tau_1^n$

    By Definition~\ref{def:pes-meas3} we are done since, $Q(\tau_1^n) = Q'(\tau_1^n)$.
    
  \item $e \in E_1$

    We know that $Q(e) = Q_1(e)$.
    From $\es{U}_1 \trianglelefteq \es{U}'_1$, $Q_1(e) = Q'_1(e)$.
    By Definition~\ref{def:pes-meas3}, $Q'_1(e) = Q'(e)$.
    Thus $Q(e) = Q'(e)$.
    
  \item $e \in E_2$

    Similar to the previous point.
  \end{enumerate}
\end{proof}

\subsubsection*{Proof of Lemma~\ref{lem:seq-cont-3}}
\begin{proof}
  Similar to Lemma~\ref{lem:seq-cont-1}.
\end{proof}

\subsubsection*{Proof of Lemma~\ref{lem:conc-cont-3}}
\begin{proof}
  Similar to Lemma~\ref{lem:conc-cont-1}.  
\end{proof}

\subsubsection*{Proof of Lemma~\ref{lem:meas-cont-3}}
\begin{proof}
  By Lemma~\ref{lem:meas-fix-mono-3}, the measurement is monotone.  Then we only need to show that
  each event of $\meas{q}{\bigsqcup_n \es{U}_n}{\bigsqcup_m \es{U}_m}$ is an event of
  $\bigsqcup_{n,m}(\meas{q}{\es{U}_n}{\es{U}_m})$.  Let $e$ be an event of
  $\meas{q}{\bigsqcup_n \es{U}_n}{\bigsqcup_m \es{U}_m}$. We have four cases:
  \begin{enumerate}
  \item $e = \tau_0^n$

    It follows directly from Definition~\ref{def:pes-meas3}.
    
  \item $e = \tau_1^n$

    It follows directly from Definition~\ref{def:pes-meas3}.
    
  \item $e$ is an event of $\bigsqcup_n \es{U}_n$

    By Definition~\ref{def:lub-3}, $\exists n \in \omega$ such that $e$ is an event of $\es{U}_n$.
    By Definition~\ref{def:pes-meas3}, $e$ is an event of $\meas{q}{\es{U}_n}{\es{U}_m}$.  By
    Definition~\ref{def:lub-3}, $e$ is an event of $\bigsqcup_{n,m}(\meas{q}{\es{U}_n}{\es{U}_m})$.
    
  \item $e$ is an event of $\bigsqcup_m \es{U}_m$

    Similar to the previous point.
  \end{enumerate}
\end{proof}

\subsubsection*{Proof of Lemma~\ref{res:soundI-fix-3}}
\begin{proof}
  \begin{itemize}
  \item $\qwhi{n}{\com{C}} \xrightarrow{\tau_0^n} \checkmark$

    We know that
    \[
      \mf{\qwhi{n}{\com{C}}} = \mf{\meas{n}{\checkmark}{\seq{\com{C}}{\qwhi{n}{\com{C}}}}}
    \]
    Hence
    \[
      \mf{\qwhi{n}{\com{C}}} \backslash \tau_0^n =
      \mf{\meas{n}{\checkmark}{\seq{\com{C}}{\qwhi{n}{\com{C}}}}} \backslash \tau_0^n
    \]
    By Lemma~\ref{lem:meas-rem-init3}, we have $\mf{\checkmark}$.  Hence
    $\mf{\qwhi{n}{\com{C}}} \backslash \tau_0^n = \mf{\checkmark}$.
    
  \item $\qwhi{n}{\com{C}} \xrightarrow{\tau_1^n} \seq{\com{C}}{\qwhi{n}{\com{C}}}$

    We know that
    \[
      \mf{\qwhi{n}{\com{C}}} = \mf{\meas{n}{\checkmark}{\seq{\com{C}}{\qwhi{n}{\com{C}}}}}\]
    Hence
    \[
      \mf{\qwhi{n}{\com{C}}} \backslash \tau_1^n =
      \mf{\meas{n}{\checkmark}{\seq{\com{C}}{\qwhi{n}{\com{C}}}}} \backslash \tau_1^n
    \]
    By Lemma~\ref{lem:meas-rem-init3}, we have $\mf{\seq{\com{C}}{\qwhi{n}{\com{C}}}}$.  Hence
    $\mf{\qwhi{n}{\com{C}}} \backslash \tau_0^n = \mf{\seq{\com{C}}{\qwhi{n}{\com{C}}}}$.
  \end{itemize}
\end{proof}

\subsubsection*{Proof of Lemma~\ref{res:adI-fix-3}}
\begin{proof}
  \begin{itemize}
  \item $l \in \init{\mf{\qwhi{n}{\com{C}}}}$

    Since $\mf{\qwhi{n}{\com{C}}} = \mf{\meas{n}{\checkmark}{\seq{\com{C}}{\qwhi{n}{\com{C}}}}}$, we have
    that $l = \tau_0^n$ or $l = \tau_1^n$.
    We have two cases:
    \begin{enumerate}
    \item $l = \tau_0^n$

      Let $\com{C}' = \checkmark$.  We know that
      \[
        \mf{\qwhi{n}{\com{C}}} = \mf{\meas{n}{\checkmark}{\seq{\com{C}}{\qwhi{n}{\com{C}}}}}
      \]
      Hence
      \[
        \mf{\meas{n}{\checkmark}{\seq{\com{C}}{\qwhi{n}{\com{C}}}}} \backslash \tau_0^n
      \]
      By Lemma~\ref{lem:meas-rem-init3}, we have $\mf{\checkmark}$.  Thus it follows directly that
      $\qwhi{n}{\com{C}} \xrightarrow{\tau_0^n} \checkmark$.
      
    \item $l = \tau_1^n$

      Let $\com{C}' = \seq{\com{C}}{\qwhi{n}{\com{C}}}$.  We know that
      \[
        \mf{\qwhi{n}{\com{C}}} = \mf{\meas{n}{\checkmark}{\seq{\com{C}}{\qwhi{n}{\com{C}}}}}
      \]
      Hence
      \[
        \mf{\meas{n}{\checkmark}{\seq{\com{C}}{\qwhi{n}{\com{C}}}}} \backslash \tau_1^n
      \]
      By Lemma~\ref{lem:meas-rem-init3}, we have $\mf{\checkmark}$.  Thus it follows directly that
      $\qwhi{n}{\com{C}} \xrightarrow{\tau_1^n} \seq{\com{C}}{\qwhi{n}{\com{C}}}$.
    \end{enumerate}
  \end{itemize}
\end{proof}

\end{document}